\documentclass [11pt,twoside]{report}
\usepackage{slashed}
\usepackage{color}

\usepackage{graphicx}
\usepackage{amssymb}
\usepackage{amsmath}
\usepackage{CJK}
\usepackage{CJKnumb}
\usepackage{mathrsfs}
\usepackage{bbm}
\usepackage{amsfonts}
\usepackage[nomove]{cite}
\usepackage{latexsym}
\usepackage{indentfirst}
\usepackage{longtable}
\newcommand{\yihao}    {\fontsize{26pt}{36pt}    \selectfont}
\newcommand{\erhao}    {\fontsize{22pt}{28pt}    \selectfont}
\newcommand{\sanhao}   {\fontsize{16pt}{24pt}    \selectfont}
\newcommand{\xiaosan}  {\fontsize{15pt}{22pt}    \selectfont}

\newcommand{\double}{\baselineskip 1.5 \baselineskip}

\def\nb{\nonumber}

\def\H3p{H_3^+}

\def\H{{\cal H}}

\def\IP{\relax{\rm I\kern-.18em P}}

\newcommand{\la}{\lambda}

\newcommand{\beq}{\begin{equation}}
\newcommand{\be}{\begin{equation}}
\newcommand{\eeq}{\end{equation}}
\newcommand{\ee}{\end{equation}}

\newenvironment{proof}{Proof:}

\def\endofproof {\hfill{$\Box$}\\}

\newcommand{\void}[1]{}

\def\beq{\begin{equation}}
\def\eeq{\end{equation}}
\def\bea{\begin{eqnarray}}
\def\eea{\end{eqnarray}}
\def\beqs{\begin{equation*}}
\def\eeqs{\end{equation*}}
\def\beas{\begin{eqnarray*}}
\def\eeas{\end{eqnarray*}}
\def\lmd{\lambda}
\def\la{\langle}
\def\ra{\rangle}
\def\bt{\beta}

\newcommand{\ppt}[1]{{\partial\over \partial #1}}
\newcommand{\td}[1]{\tilde{#1}}

\def\ra{\rightarrow}

\input epsf
\newcount\figno
\def\fig#1#2#3{
\par\begingroup\parindent=0pt\leftskip=1cm\rightskip=1cm\parindent=0pt
\baselineskip=15pt
\global\advance\figno by 1
\epsfxsize=#3
\centerline{\epsfbox{#2}}
\vskip 12pt
{\bf \small 图 \the\figno:} {\small #1}\par
\endgroup\par
}
\def\figlabel#1{\xdef#1{\the\figno
\mbox{ }}}
\def\encadremath#1{\vbox{\hrule\hbox{\vrule\kern8pt\vbox{\kern8pt
\hbox{$\displaystyle #1$}\kern8pt}
\kern8pt\vrule}\hrule}}
\def\endofproof {\hfill{$\Box$}\\}

\makeatletter
\renewcommand{\fnum@figure}[1]{\textbf{\figurename~\thefigure}\hspace{1.0ex} }
\renewcommand{\fnum@table}[1]{\textbf{\tablename~\thetable}\hspace{1.0ex} \sffamily}

\makeatother

\makeatletter
\def\@evenhead{
\vbox{\hbox
to\textwidth{\rlap{\textrm{\thepage}}\hfil{\leftmark}\llap{ }\hfil\mbox{}}
\protect\vspace{2truemm}\relax \hrule depth0pt height0.15truemm
width\textwidth }
}
\def\@oddhead{
\vbox{\hbox to\textwidth{\mbox{}\hfil \rlap{ }{\rightmark}
\hfil\llap{\textrm{\thepage}}} \protect\vspace{2truemm}\relax
\hrule depth0pt height0.15truemm width\textwidth }
}
\def\@evenfoot{}
\def\@oddfoot{}
\makeatother

\makeatletter
\renewcommand{\chaptername}{Chapter \@arabic\c@chapter}
\renewcommand{\@makechapterhead}[1]{%
  \vspace*{-\baselineskip}%
  {\normalfont \centering\Huge\bfseries%
  \chaptername \quad #1\par\nobreak%
  \vspace{1.5\baselineskip}
 }}
\renewcommand{\@makeschapterhead}[1]{%
  \vspace*{-\baselineskip}%
  {\normalfont\centering\Huge\bfseries #1\par\nobreak%
  \vspace{1.5\baselineskip}
}}
\makeatother

\newcommand{\gae}{\begin{array}{c}\,\vspace{-0.5em}\hspace{-0.2em}\sim\vspace{-1.7em}\\>
\end{array}}
\newcommand{\lae}{\begin{array}{c}\,\vspace{-0.5em}\hspace{-0.2em}\sim\vspace{-1.7em}\\<
\end{array}}

\def\to{\rightarrow}
\newcommand{\ba}{\begin{array}}
\newcommand{\ea}{\end{array}}

\newcommand{\eref}[1]{{(\ref{#1})}}
\newtheorem{lemma}{Lemma}
\newtheorem{theorem}{Theorem}
\newtheorem{corollary}{Corollary}
\newcommand{\cref}[1]{{{\bf{Corollary~}\ref{#1}}}}
\newcommand{\tref}[1]{{{\bf{Theorem~}\ref{#1}}}}
\newcommand{\sref}[1]{{{Section~\ref{#1}}}}
\newcommand{\lref}[1]{{{\bf{Lemma~}\ref{#1}}}}

\newcommand{\ke}{k_{\mbox{\tiny off}}}
\newcommand{\Le}{L_{\mbox{\tiny off}}}
\newcommand{\Ve}{V_{\mbox{\tiny off}}}

\begin{document}



\double
\date{}

\begin{titlepage}
\newpage

\newpage
\thispagestyle{empty}
\includegraphics[height=4.5cm,origin=c]{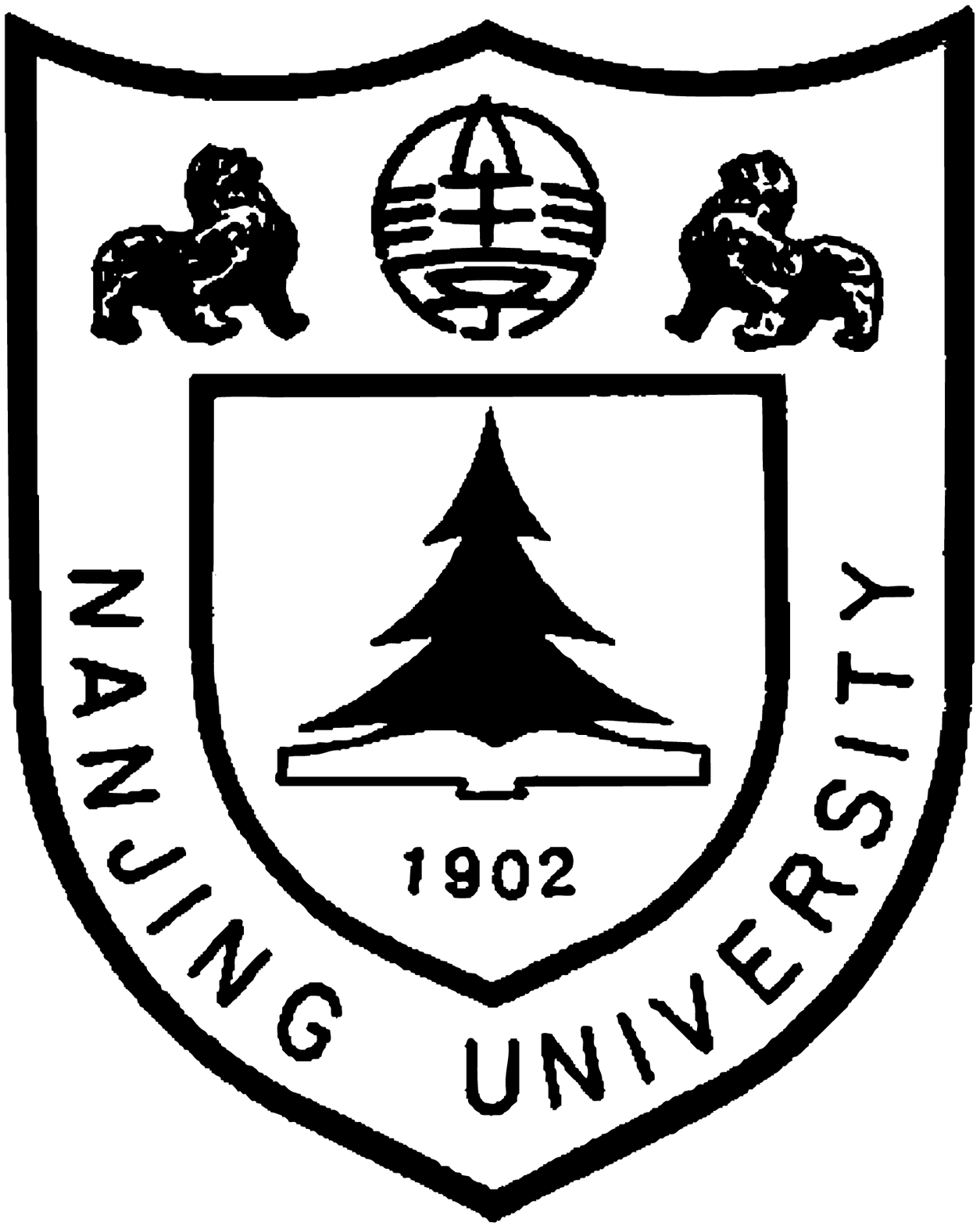}\hfill
      \begin{tabular}[b]{r}
        \yihao Nanjing University\\[0.3cm]
        \sanhao A Dissertation for Ph. D. Degree\\[1.8cm]\\
      \end{tabular}

      \begin{center}
          \parbox[t][5cm][c]{\textwidth}{\erhao
           \begin{center} {\textbf{\textsf{}}}\end{center} }

         \parbox[t][3cm][c]{\textwidth}{ {\xiaosan
            \begin{center}
              \begin{tabular}{rl}
              Title: & A Tale Of Two Amplitudes In High Energy Physics\\
                Candidate: & Yun Zhang\\
                Supervisor: & Yeuk-Kwan Edna Cheung, Ph. D.\\
               & Konstantin G. Savvidy, Ph. D.\\
               & Gang Chen, Ph. D.\\
                Speciality: & Theoretical Physics
                \end{tabular}
            \end{center} } }
          \vfill
          \parbox[t][4cm][t]{\textwidth}{
            \begin{center} {\sanhao \textbf{\textsf{~School of Physics, Nanjing University~\\Nanjing, P. R.
China \\ \vspace{1.cm} May 18, 2010}}} \end{center} }
      \end{center}

\end{titlepage}



\pagenumbering{roman}

\vspace{-0.2cm}


\centerline{\textbf{\Large{Abstract}}}
\vspace*{0.5cm}

I will describe my work on two classes of scattering amplitudes in high energy physics. The thesis consists of two parts: one is on proton Compton scattering in a Unified Proton-Delta theory, and the other is on the computation of scattering amplitudes in Yang-Mills theory.

We study proton Compton scattering in the first resonance region in an effective field theory approach, in an effort to better understand the proton's electromagnetic properties--its electric and magnetic polarizabilities--done in collaboration with Dr. Konstantin G. Savvidy. The consistent electrodynamic interaction of spin 3/2 field, which respects current conservation, has recently been developed by Savvidy in his generalized Rarita-Schwinger theory for spin 3/2 particle, resolving the long standing superluminal propagation problem of the old Rarita-Schwinger theory. Proton and $\Delta^+$ are naturally unified in this generalized Rarita-Schwinger theory with proton being the spin 1/2 component and $\Delta^+$ being the spin 3/2 component. To describe proton Compton scattering, we introduce six non-minimal electromagnetic interactions--with their coefficients being called "form factors"--and bare polarizabilities in an effective Lagrangian, consistent with the requirement of gauge invariance. We express proton and Delta magnetic moments in terms of the form factors. We then compute the proton Compton scattering amplitudes, and obtain the total electric and magnetic polarizabilities in terms of the bare ones and the form factors. We also study the approximation of the amplitudes around the Delta pole. Using experimental data, we obtain the best fit values for the form factors and bare polarizabilities. As a prediction, we derive Delta magnetic moment from the best fit values of the parameters.


After some background preparation, we present our joint work with Dr. Gang Chen on the study of boundary behavior of off-shell Yang-Mills amplitudes with a pair of external momenta complexified. In Feynman gauge, we introduce a set of ''reduced vertices'' which can effectively capture the boundary behavior up to the first two leading orders and can, in turn, greatly simplify subsequent analysis. Boundary behavior of amplitudes with two adjacent legs complexified can be read off from the reduced vertices. We then prove a theorem on permutation sum for a given color ordering, and use it to analyze the boundary behavior of amplitudes with two non-adjacent legs complexified. Based on the boundary behaviors, we construct off-shell Britto-Cachazo-Feng-Witten (BCFW) recursion relations for general tree level amplitudes. As applications, we calculate off-shell amplitudes and study relations between off-shell amplitudes.

Finally, we study the recursion relations for off-shell Yang-Mills amplitudes at tree and one loop levels as deduced from imposing complexified Ward identity, also in collaboration with Dr. Gang Chen. It is based on a previous work by Gang Chen in which Ward identity is used to derive a recursion relation for calculating tree level boundary terms. We extend his work to derive recursion relations of the full scattering amplitudes at both tree and one loop levels. Using Feynman rules, we explicitly prove Ward identity at tree and one loop levels. We then give recursion relations for general N-point off-shell amplitudes. We calculate three and four point one loop off-shell amplitudes as applications of our method.

\vspace*{0.5cm}

\newpage

\tableofcontents

\setcounter{page}{1}

\chapter{Backgrounds for Proton Compton Scattering}

\pagenumbering{arabic}

Proton is the particle which makes up the greatest fraction of the matter in the visible universe. Its properties have been extensively studied, but nevertheless it still has some important parameters not precisely measured and not well predicted from theory. Among them, the electric and magnetic polarizabilities, right behind the more fundamental electromagnetic properties of the proton as are the electric charge and magnetic moment, continue to attract much investigation in recent years. The electric and magnetic polarizabilities measure the induced electric and magnetic multipole moments of the proton in electric and magnetic fields. 

Proton's electromagnetic properties can be studied by scattering proton with the electromagnetic gauge field, ie. photon. When the incident photon is off shell, the process is called virtual Compton scattering \cite{Guichon:1995pu,Pasquini:2001yy,Roche:2000ng}. In this thesis, we only concern the real Compton scattering with on shell photon scattering with proton. As the energy of the incident photon increases, there are various peaks emerging in the Compton scattering cross section, as the intermediate particles approaching their mass shells. The first and highest peak encountered is from $\Delta^+(1232{\mbox MeV}, J^P=\frac{3}{2}^+)$), which is the lightest excited baryon resonance. The cross section arrives at its first peak at about 300 MeV of incident photon energy, drops to the valley at about 450 MeV, and then rises to the later resonances. A configuration of the proton Compton scattering cross section data points is shown in Figure \ref{dataconfig}. We will study the Compton scattering with incident photon energy less than about 455 MeV, which is called the first resonance region.

\begin{figure}[h]
\centerline{\includegraphics[height=10cm,width=8cm]{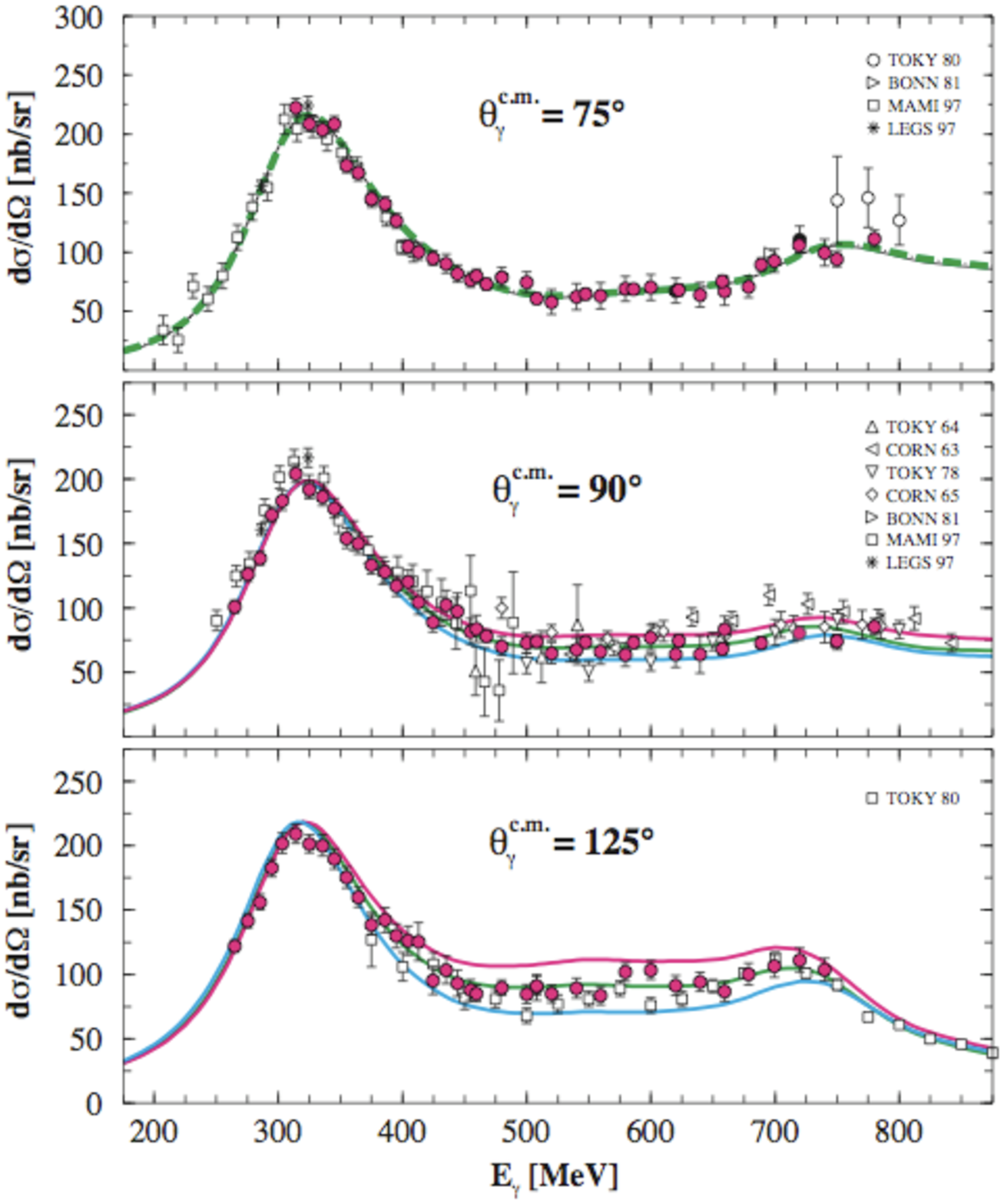}}
\caption{A configuration of the proton Compton scattering cross section data points. Picture is taken from \cite{Schumacher:2013hu}.}
\label{dataconfig}
\end{figure}

$\Delta^+$ has the same quark constitution (uud) as the proton and is only slightly, less than 300 MeV, heavier than the proton. The quarks are aligned in the same direction in $\Delta^+$, and the d quark is anti aligned with u quarks in proton, which causes the differences of spins and masses of proton and $\Delta^+$. In proton Compton scattering, the d quark alignment is easily flipped by the photon, exciting the $\Delta^+$ degree of freedom from proton. This process is non-negligible even at low energy ($\lae 100$ MeV) due to the small mass gap between proton and $\Delta^+$. Thus, $\Delta^+$ is the most important degree of freedom to consider in proton Compton scattering, besides the proton itself, in the first resonance region. Some minor contributions are from the meson exchange channels, like $\pi^0$ or $\sigma$ meson, which are easy to deal with and are seldom disputed about.


High quality proton Compton scattering data in the first resonance range have been available from various laboratories \cite{Baranov1974,Hallin1993,Zieger92,Hunger97,MacGibbon1995,Blanpied01,Olmos2001,MAMI2001}, especially in the recent two decades. Also theoretically there have been many progresses in proton Compton scattering. The most phenomenologically successful description until now has been that based on dispersion theory \cite{Lvov1997, Hearn:1962zz, Pfeil:1974ib, Lvov:1980wp, Drechsel:1999rf,Lvov:1979zd, Guiasu:1978ak,GellMann:1954db,Baldin1960,Pasquini2007}, in which scattering amplitudes are constructed as analytic functions of momenta and scattering angle. Then these functions are calculated based on elaborate understanding of the low energy scattering processes.
Another alternative is to use an effective theory \cite{Feuster:1998cj, Pascalutsa1995,Scholten1996, Kondratyuk2001}, like the chiral perturbation theory 
\cite{Pagels1974, Weinberg1978, Gasser1983, Gasser1987,McGovern:2012ew}. In this approach, people use an effective Lagrangian, do field theory calculations of proton Compton scattering cross section based on the Lagrangian, and fit parameters of the theory.

In the following, I will give an overview of some backgrounds for proton Compton scattering, including the notion of polarizabilities, and dispersion theory and effective field theory approaches to proton Compton scattering. The Rarita-Schwinger theory of spin 3/2 fields and its generalization by Konstantin G. Savvidy are included in the introduction of the latter approach.

\section{Polarizabilities}\label{sec:polarizabilities}

Proton has one positive unit of electric charge. Proton's magnetic moment in the nuclear unit $\mu_N$ is about 2.79, which has been measured to the nine-th digit \cite{Mohr:2000ie,Mohr:2005zz,Mohr:2008fa,Mohr:2012tt}. Right behind these two more fundamental electromagnetic parameters are the polarizabilities, which are measured relatively poorly and are among the motivations of many experiments and theoretical works \cite{Federspiel:1991yd,Mukhopadhyay:1993zx,MacGibbon1995,Guichon:1995pu,Tonnison:1998mi,Holstein:1999uu,Roche:2000ng,Pasquini:2000ue,Blanpied01,Olmos2001,Laveissiere:2004nf,Pasquini2007,Gorchtein:2009qq,Fonvieille:2012cd,McGovern:2012ew}. 

Polarizabilities measure the induced electric and magnetic multipole moments of proton in electromagnetic field. Polarizabilities affect the Compton scattering since the proton can be polarized by the photon and its inner structure responses to the photon. In static electric field $\vec E$ and magnetic field $\vec H$, with the induced electric dipole moment $\vec p$ and magnetic dipole moment $\vec \mu$ for proton, one can define the static electric polarizability $\alpha_E$ and magnetic polarizability $\bt_M$ by:
\beq
\vec p=4\pi \alpha_E \vec E, ~~\vec \mu=4\pi \beta_M \vec H.
\eeq
The effective Hamiltonian is \cite{Lvov93}:
\beq
{\cal H}_{\rm Int}^{(2)}=-2\pi(\alpha_E \vec{E}^2+\beta_M \vec{H}^2),
\label{stapolHam}
\eeq
such that:
\beqs
\vec p=-\frac{\delta {\cal H}_{\rm Int}^{(2)}}{\delta \vec E}=4\pi \alpha_E \vec E, ~~\vec \mu=-\frac{\delta {\cal H}_{\rm Int}^{(2)}}{\delta \vec H}=4\pi \beta_M \vec H.
\eeqs
The "(2)" in ${\cal H}_{\mbox{\tiny Int}}^{(2)}$ means that this is the Hamiltonian with two derivatives of the photon field, ie. second order of photon energy, remembering that $\vec E$ and $\vec H$ both contain one derivative of the photon field.

At very low energy, proton interacts with photon as a point particle, and only its charge and magnetic moment affect the scattering amplitudes \cite{Low1954}. As photon energy increases, the static polarizabilities first enter the amplitudes at second order of incident photon energy $\omega$. In lab frame where initial proton is at rest, the contribution from static polarizabilities to the amplitudes is:
\beq
{\cal A}=4\pi \, \alpha_E  \, \omega \omega'  \, {\vec \epsilon\,}'\cdot {\vec\epsilon}+4\pi \, \beta_M  \, {\vec\epsilon\,}'\times{\vec k}'\cdot {\vec\epsilon}\times{\vec k},
\eeq
where $\omega$, $\vec \epsilon$ and $\vec k$ are the energy, polarization vector and momentum of incoming photon respectively, and those with prime for the outgoing photon. At scattering angle $\theta$, $\omega'=\omega[1+\frac{\omega}{M}(1-\cos\theta)]^{-1}$. Converted to cross section in the lab frame, the static polarizabilities enter as:
\beq
\frac{d \sigma}{d \Omega_{\mbox{\tiny lab}}}=\left(\frac{d \sigma}{d \Omega}\right)_{\mbox{\tiny Powell}}-\frac{e^2 \omega^2}{4 \pi M}(\frac{\alpha_E+\beta_M}{2}(1+\cos\theta)^2+\frac{\alpha_E-\beta_M}{2}(1-\cos\theta)^2)+\mathcal O(\omega^3),
\eeq
where $\left(\frac{d \sigma}{d \Omega}\right)_{\mbox{\tiny Powell}}$ is the Compton cross section of a Dirac point particle with anomalous magnetic moment included \cite{Powell1949, Low1954, GellMann1954}.


In \cite{0611327}, they proposed an effective Lagrangian to model polarizabilities for a Dirac particle:
\beq
{\cal L}_{\mbox{\tiny pol}}=-\frac{i\pi}{M}~(\bar\psi \, \gamma^\mu  \,\partial_\nu  \,\psi-\partial_\nu  \,\bar\psi  \,\gamma^\mu \psi)(\alpha_E ~ F_{\mu\rho} \, F^{\rho\nu}-\beta_M ~  \tilde{F}_{\mu\rho} \, \tilde{F}^{\rho\nu}).
\eeq

In non static electric and magnetic fields, there are richer structures described by more polarizabilities. The most general effective interactions between the electromagnetic field and the proton up to $\mathcal O(\omega^4)$ are given by
\cite{Babusci:1998ww}:
\begin{eqnarray}
{\mathcal H}_{\rm Int} &=& {\mathcal H}_{\rm Int}^{(2)}+ 
{\mathcal H}_{\rm Int}^{(3)} + {\mathcal H}_{\rm Int}^{(4)}+\cdots.\nb\\
{\mathcal H}_{\rm Int}^{(3)} &=& -{4 \pi} 
\left[{\gamma_{E1}\over 2}  {\vec{\sigma}}\cdot{\vec E}\times{\dot{\vec E}} +
      {\gamma_{M1}\over 2}  {\vec{\sigma}}\cdot{\vec H}\times{\dot{\vec H}} -
       \gamma_{E2} {{\sigma_i}E_{ij}}H_j +
       \gamma_{M2} {\sigma_i}{H_{ij}}E_j
\right], \label{spinpol}\\ 
{\mathcal H}_{\rm Int}^{(4)} &=& -{4 \pi\over 2} 
\left[\alpha_{E\nu} {\dot{\vec E}}^2 +
       \beta_{M\nu} {\dot{\vec H}}^2 +
     {1\over 6}( \alpha_{E2} E_{ij}^2+
                 \beta_{M2} {H_{ij}}^2 )
\right]\label{fourthpol}, 
\end{eqnarray}
where $E_{ij} = (\nabla_iE_j + \nabla_jE_i)/2$, and similarly for
$H_{ij}$. ${\mathcal H}_{\rm Int}^{(2)}$ is given in (\ref{stapolHam}). The parameters $\alpha_{E\nu}$ and $\beta_{M\nu}$ in (\ref{fourthpol}) are dispersive corrections to the static polarizabilities $\alpha_E$ and $\beta_M$ when the electric and magnetic fields are time dependent:
\bea
&&\vec p(\omega)=4\pi(\alpha_E+\alpha_{E\nu}\omega^2+\cdots)\vec E(\omega),\nb\\
&&\vec \mu(\omega)=4\pi(\bt_M+\bt_{M\nu}\omega^2+\cdots)\vec H(\omega).
\eea
The parameters $\alpha_{E2}$ and $\beta_{M2}$ in (\ref{fourthpol}) are quadrupole polarizabilities that measure the induced electric quadrupole moment $Q_{ij}$ and magnetic quadrupole moment $M_{ij}$ in the non static electric and magnetic fields:
\beq
Q_{ij}=\frac{\delta {\cal H}_{\rm Int}^{(4)}}{\delta E_{ij}}=\frac{1}{6}4\pi \alpha_{E2}E_{ij},~~M_{ij}=\frac{\delta {\cal H}_{\rm Int}^{(4)}}{\delta H_{ij}}=\frac{1}{6}4\pi\bt_{M2}H_{ij}.
\eeq

The polarizabilities $\gamma_{E1}$, $\gamma_{M1}$, $\gamma_{E2}$ and $\gamma_{M2}$ in (\ref{spinpol}) are called spin polarizabilities, since the interactions in (\ref{spinpol}) depend on the spin of the proton through the Pauli matrices $\vec \sigma$. To directly measure the spin polarizabilities, one needs Compton scattering with polarized photon and proton, which is hard to operate. More often people have been using sum rules, which shall be introduced below, to extract the spin polarizabilities as well as other polarizabilities. Of particular experimental 
interests are the forward $\gamma_0$ and backward $\gamma_\pi$
spin polarizabilities:
\begin{eqnarray}
\gamma_0 &=&  - \gamma_{E1} - \gamma_{M1} - \gamma_{E2}-\gamma_{M2} \\
\gamma_\pi &=&  - \gamma_{E1} + \gamma_{M1} + \gamma_{E2}-\gamma_{M2}.
\end{eqnarray}
The forward spin polarizability $\gamma_0$ appears as the coefficient of the spin dependent term in ${\mathcal O}(\omega^3)$ of the forward scattering amplitudes (final photon and initial photon are moving in the same direction):
\beq
i 4\pi \gamma_0 \omega^3\vec\sigma\cdot (\vec \epsilon\,'^*\times\vec\epsilon\,).
\eeq
Similarly, $\gamma_\pi$ appears as the coefficient of the spin dependent term in ${\mathcal O}(\omega^3)$ of the backward Compton scattering amplitudes (final photon and initial photon are moving in opposite directions).

By definition, one can extract polarizabilities from Compton scattering by a low energy expansion of the Compton cross section. This approach is convincing only for photon energies well below 100 MeV, which places harsh requirements on the precision of the experiments. One can also theoretically predict polarizabilities, for example in Heavy Baryon Chiral Perturbation Theory \cite{Holstein:1999uu,Bernard:1992qa}, or in Skyrme model \cite{Nyman:1984ys,Kim:1997hq}, etc., and some results obtained are consistent with experimental values. Another successful approach of extracting polarizabilities has been through sum rules, introduced below, in which ideas of dispersion theory are applied.



\subsection{Sum Rules}
At forward angle the Compton scattering amplitudes take the general form \cite{Drechsel:2009ch}:
\beq
\mathcal A=4\pi \vec\epsilon\,'^*\cdot \vec\epsilon f(\omega)+4 i \pi \vec\sigma\cdot(\vec\epsilon\,'^*\times\vec\epsilon\,)g(\omega),
\eeq
where $f$ and $g$ are respectively even and odd functions of the lab frame incident photon energy $\omega$, due to crossing symmetry. Denoting $\mathcal A_{3/2}$ as the amplitudes when the initial photon and proton have paralell helicity, and $\mathcal A_{1/2}$ for opposite helicity, then $f(\omega)=(\mathcal A_{1/2}+\mathcal A_{3/2})/2$ and $g(\omega)=(\mathcal A_{1/2}-\mathcal A_{3/2})/2$. Corresponding to $\mathcal A_{1/2}$ and $\mathcal A_{3/2}$ one can define the helicity dependent cross sections $\sigma_{1/2}$ and $\sigma_{3/2}$. Total cross section $\sigma_{\rm tot}=(\sigma_{1/2}+\sigma_{3/2})/2$, and transverse cross section $\sigma_{\rm trans}=(\sigma_{1/2}-\sigma_{3/2})/2$.

The real parts of $f$ and $g$ can be expressed by dispersion integrals \cite{Drechsel:2009ch}:
\bea
&&{\rm Re}\, f(\omega)=f(0)+\frac{\omega^2}{2\pi^2} P\int_{\omega_{\rm thr}}^\infty \frac{\sigma_{\rm tot}(\omega')d\omega'}{\omega'^2-\omega^2},\nb\\
&&{\rm Re}\, g(\omega)=\frac{\omega}{2\pi^2} P\int_{\omega_{\rm thr}}^\infty \frac{\omega'\sigma_{\rm trans}(\omega')d\omega'}{\omega'^2-\omega^2},
\label{forwardDR}
\eea
where $\omega_{\rm thr}\simeq 150$ MeV is the threshold for pion photoproduction.

On the other hand, to $\mathcal O(\omega^3)$, $f$ and $g$ are determined by proton's charge, anomalous magnetic moment $\kappa$, static polarizabilities and the forward spin polarizability:
\beq
f(\omega)=-\frac{e^2}{4\pi M}+(\alpha_E+\beta_M)\omega^2+\mathcal O(\omega^4),~g(\omega)=-\frac{e^2\kappa^2}{8\pi M^2}\omega+\gamma_0\omega^3+\mathcal O(\omega^5).
\label{fgpolarizability}
\eeq
Comparing (\ref{forwardDR}) and (\ref{fgpolarizability}), one obtains the following sum rules:
\bea
\alpha_E+\bt_M&=&\frac{1}{2\pi^2}\int_{\omega_{\rm thr}}^\infty \frac{\sigma_{\rm tot}(\omega')d\omega'}{\omega'^2},\\
\frac{\pi e^2 \kappa^2}{2M^2}&=&-\int_{\omega_{\rm thr}}^\infty \frac{\sigma_{\rm trans}(\omega')d\omega'}{\omega'},\\
\gamma_0&=&\frac{1}{4\pi^2}\int_{\omega_{\rm thr}}^\infty \frac{\sigma_{\rm trans}(\omega')d\omega'}{\omega'^3},
\eea
which are the sum rules of Baldin \cite{Baldin1960}, Gerasimov-Drell-Hearn (GDH) \cite{Gerasimov1966,Drell1966} and GellMann-Goldberger-Thirring (GGT) \cite{GellMann1954,GellMann:1954db}, respectively.

Other sum rules for $\alpha_E-\beta_M$ and $\gamma_\pi$ etc. are discussed in eg. \cite{Holstein:1994tw,Schumacher:2013hu,L'vov:1998ez,L'vov:1998vg,Bernabeu:1998wf}.





\section{Dispersion Theory Approach to Proton Compton Scattering}
A major and most phenomenologically successful approach to proton Compton scattering hitherto has been that based on dispersion theory \cite{Lvov1997, Hearn:1962zz, Pfeil:1974ib,Lvov:1980wp, Drechsel:1999rf,Lvov:1979zd, Guiasu:1978ak,GellMann:1954db,Baldin1960,Pasquini2007,Schumacher:2013hu}, in which the Compton scattering amplitudes are studied based on elaborate understanding of the low energy scattering processes.


Using the orthogonal basis suggested by Prange \cite{Prange1958}, one can decompose the proton Compton scattering amplitudes as:
\bea
\mathcal A=&&\bar u'(p) \epsilon'^{*\mu}\left[-\frac{P_\mu' P_\nu'}{P'^2}(T_1+\gamma\cdot K T_2)-\frac{N_\mu N_\nu}{N^2}(T_3+\gamma\cdot K T_4)\right.\nb\\
&&\left.+i \frac{P_\mu' N_\nu-P_\nu' N_\mu}{P'^2K^2}\gamma_5 T_5+i \frac{P_\mu' N_\nu+P_\nu' N_\mu}{P'^2K^2}\gamma_5 \gamma\cdot K T_6\right]\epsilon^\nu u(p),
\label{Tiexpansion}
\eea
with the definitions for $P_\mu'$, $P$, $K$ and $N_\mu$:
\bea
&&P_\mu'=P_\mu-K_\mu \frac{P\cdot K}{K^2},~~ P=\frac{1}{2}(p+p'),~~K=\frac{1}{2}(k'+k),\nb\\
&&N_\mu=\epsilon_{\mu\alpha\beta\gamma}P'^\alpha Q^\beta K^\gamma,~~Q=\frac{1}{2}(p-p')=\frac{1}{2}(k'-k).
\eea
In the expressions, $p$ and $p'$ are the momenta of initial and final protons; $k$ and $k'$ are the momenta of initial and final photons.

The six $T_i$'s are functions of photon energy and scattering angle, or equivalently functions of $\nu=\frac{s-u}{4M}=\omega+\frac{t}{4M}$ and $t=(k-k')^2$. Due to crossing symmetry, $T_2(\nu,t)$ and $T_4(\nu,t)$ are odd functions of $\nu$ and the other four are even functions. $T_i$ have singularities related to poles from one particle exchanges, and those related to inelastic thresholds in the s, u and t channels \cite{Logunov1958, Hearn1961}, which are due to the fact that some denominators in (\ref{Tiexpansion}) vanish at forward or backward angles:
\bea
K^2&=&-\frac{t}{4}=\frac{1}{8s}(s-M^2)^2(1-\cos\theta),\nb\\
P'^2 K^2&=&\frac{1}{4}(su-M^4)=-\frac{1}{8s}(s-M^2)^2(1+\cos\theta),\nb\\
N^2&=&P'^2(K^2)^2\sim\sin^2\theta.
\eea
A set of linear combinations \cite{Lvov:1980wp} of $T_i$ are free from the constraints from the inelastic thresholds:
\beq
\begin{array}{ll}
A_1=\frac{1}{t}[T_1+T_3+\nu(T_2+T_4)],&A_2=\frac{1}{t}[2T_5+\nu(T_2+T_4)],\\
A_3=\frac{M^2}{M^4-su}[T_1-T_3-\frac{t}{4\nu}(T_2-T_4)],&A_4=\frac{M^2}{M^4-su}[2MT_6-\frac{t}{4\nu}(T_2-T_4)],\\
A_5=\frac{1}{4\nu}[T_2+T_4],&A_6=\frac{1}{4\nu}[T_2-T_4].\end{array}
\eeq
$A_i(\nu,t)$, called invariant amplitudes,\footnote{There are other conventions for defining the invariant amplitudes $A_i$, see eg. \cite{Bardeen1968,Hearn:1962zz,Pfeil:1974ib}. The relations between their conventions and the convention here for $A_i$ are given in eg. \cite{Babusci:1998ww,Lvov1997}.} are all even functions of $\nu$. In terms of $A_i$, the lab frame Compton amplitudes take the following form:
{\allowdisplaybreaks{}
\bea
\mathcal A=&&\frac{1}{\sqrt{1-t/4M^2}}\left\{2M\,\vec\epsilon\,'^*\cdot\vec\epsilon \,\omega\omega'\left[(1-\frac{t}{4M^2})(-A_1-A_3)-\frac{\nu^2}{M^2}A_5-A_6\right]\right.\nb\\
&&+2M (\hat k'\times\vec\epsilon\,')^*\cdot (\hat k\times\vec\epsilon) \,\omega\omega'\left[(1-\frac{t}{4M^2})(A_1-A_3)+\frac{\nu^2}{M^2}A_5-A_6\right]\nb\\
&&-2i\vec\sigma\cdot\vec\epsilon\,'^*\times\vec\epsilon\, \nu\omega\omega'(A_5+A_6)+2i\vec\sigma\cdot (\hat k'\times\vec\epsilon\,')^*\times (\hat k\times\vec\epsilon) \,\nu\omega\omega'(A_5-A_6)\nb\\
&&+i\vec\sigma\cdot \hat k (\hat k'\times\vec\epsilon\,')^*\cdot\vec\epsilon \,\omega^2\omega'\left[A_2+(1-\frac{\omega'}{M})A_4+\frac{\nu}{M}A_5+A_6\right]\nb\\
&&-i\vec\sigma\cdot \hat k' \,\vec\epsilon\,'^*\cdot (\hat k\times\vec\epsilon) \,\omega\omega'^2\left[A_2+(1+\frac{\omega'}{M})A_4-\frac{\nu}{M}A_5+A_6\right]\nb\\
&&-i\vec\sigma\cdot \hat k' \,\vec\epsilon\,'^*\cdot (\hat k\times\vec\epsilon) \,\omega\omega'^2\left[-A_2+(1-\frac{\omega'}{M})A_4-\frac{\nu}{M}A_5+A_6\right]\nb\\
&&\left.+i\vec\sigma\cdot \hat k (\hat k'\times\vec\epsilon\,')^*\cdot\vec\epsilon \,\omega^2\omega'\left[-A_2+(1+\frac{\omega'}{M})A_4+\frac{\nu}{M}A_5+A_6\right]\right\},
\eea}
where $\hat k$ and $\hat k'$ are the unit vectors in directions of $\vec k$ and $\vec k'$.

Fixed-t dispersion relations for $A_i$ can be formulated as \cite{Lvov1997,Petrunkin1981:278,Lvov:1980wp}:
\bea
&&{\rm Re}\ A_i(\nu,t)=A_i^{\rm pole}(\nu,t)+A_i^{\rm int}(\nu,t)+A_i^{\rm as}(\nu,t).
\label{fixtDR}\\
&&A_i^{\rm int}(\nu,t)=\frac{2}{\pi}P\int_{\nu_{\rm thr}(t)}^{\nu_{\rm max}(t)} {\rm Im} A_i(\nu',t)\frac{\nu'd\nu'}{\nu'^2-\nu^2}.\label{Aiint}\\
&&A_i^{\rm as}(\nu,t)=\frac{1}{\pi}{\rm Im}\int_{\nu_{\nu'=\rm max}(t)e^{i\phi},0<\phi<\pi} A_i(\nu',t)\frac{\nu'd\nu'}{\nu'^2-\nu^2}.\label{Aias}
\eea
The first term in (\ref{fixtDR}), also called Born term $A_i^B$ in some literatures, is completely determined by the electric charge and magnetic moment of proton. It receives singular contribution from the proton in the intermediate state. 
\bea
&&A_i^{\rm pole}(\nu,t)=\frac{M e^2 r_i(t)}{(s-M^2)(u-M^2)},\nb\\
&&r_1=-2+(\kappa^2+2 \kappa)\frac{t}{4M^2},~r_2=2\kappa+2+(\kappa^2+2\kappa)\frac{t}{4M^2},\nb\\
&&r_3=r_5=\kappa^2+2\kappa,~r_4=\kappa^2,~r_6=-\kappa^2-2\kappa-2.
\eea

The second term in (\ref{fixtDR}), ie. (\ref{Aiint}) is labeled integral contribution. It is the usual dispersion integral taken between the pion photoproduction threshold $\nu_{\rm thr}(t)=\omega_{\rm thr}+t/4M$ with $\omega_{\rm thr}\simeq 150$ MeV, and a maximum energy $\nu_{\rm max}$. Below $\nu_{\rm max}$, ${\rm Im} A_i$ can be evaluated using optical theorem and known amplitudes of meson photoproduction \cite{Lvov1997}. In \cite{Lvov1997} they use the maximum energy: $\omega_{\rm max}=\nu_{\rm max}(t)-t/4M=1.5$ GeV.

The last term in (\ref{fixtDR}) is the contribution from Compton scattering with energy above $\nu_{\rm max}$. From (\ref{Aias}) it is seen that the dependence on energy for all the $A_i$'s are neglect-able when $\nu^2\ll \nu_{\rm \max}^2$. As found in \cite{Lvov1997}:
\beq
A_{1,2}\sim \nu^{\alpha(t)},~A_{3,5,6}\sim \nu^{\alpha(t)-2},~A_4\sim  \nu^{\alpha(t)-3}\ {\mbox as}\ \nu\to\infty,
\eeq
with $\alpha(t)\le 1$ being the Regge pole trajectory. Thus $A_{3,4,5,6}$ vanish when $\nu\to\infty$ and satisfy unsubtracted dispersion relations:
\beq
A_i^{\rm as}(\nu,t)=\frac{2}{\pi}P\int^\infty_{\nu_{\rm max}(t)} {\rm Im} A_i(\nu',t)\frac{\nu'd\nu'}{\nu'^2-\nu^2}.
\eeq
It is expected that $A_{3,4,5,6}^{\rm as}$ are small, and $A_{3,4,5,6}$ are saturated at low energy $\nu'\ll \nu_{\rm max}$.

$A_{1,2}^{\rm as}$ are assumed to receive contributions from exchanges of $\sigma$ and $\pi$ mesons \cite{Lvov1997}:
\bea
&&A_1^{\rm as}(t)\simeq A_1^{\sigma(t)}=\frac{g_{\sigma NN}F_{\sigma\gamma\gamma}}{t-m_\sigma^2},\nb\\
&&A_2^{\rm as}(t)\simeq A_2^{\pi^0(t)}=\frac{g_{\pi NN}F_{\pi^0\gamma\gamma}}{t-m_{\pi^0}^2} F_\pi(t),
\eea
where $F_\pi(t)=(\Lambda_\pi^2-m_\pi^2)/(\Lambda_\pi^2-t)$ with the cutoff parameter $\Lambda_\pi\simeq 0.7$ GeV.



\section{Effective Field Theory Approach to Proton Compton Scattering}
Besides dispersion theory, another major approach to proton Compton scattering is effective field theory approach \cite{Feuster:1998cj, Pascalutsa1995,Scholten1996, Kondratyuk2001,Pagels1974, Weinberg1978, Gasser1983, Gasser1987,McGovern:2012ew}, in which people use an effective Lagrangian to describe the involved degrees of freedom. Since it is necessary to include the spin 3/2 particle $\Delta^+$ in low energy proton Compton scattering, a spin 3/2 field theory is required. The most common theory is Rarita-Schwinger theory \cite{Rarita:1941mf}.

\subsection{Rarita-Schwinger Spin 3/2 Theory}\label{oldRarSch}
The spin 3/2 field can be represented by a spinor with Lorentz index: $\psi_\mu^a$ (the spinor index a is usually suppressed). The free Lagrangian can be written as:
\bea
&&{\mathcal L}= - {\bar \psi}_\lambda \, [p_\mu \, \Gamma^{\mu\lambda}{}_\rho-m \, \Theta^\lambda{}_\rho] \, \psi^\rho,\nonumber\\
&&\Gamma^{\mu\lambda}{}_\rho=\gamma^\mu \, \eta^\lambda{}_\rho+\xi \, (\gamma^\lambda \,  \eta^\mu{}_\rho+\gamma_\rho \,  \eta^{\lambda\mu})+\zeta \, \gamma^\lambda \, \gamma^\mu \, \gamma_\rho,\nonumber\\
&&\Theta^\lambda{}_\rho=\eta^\lambda{}_\rho- z  \, \gamma^\lambda \, \gamma_\rho.
\label{oldRS}
\eea
Since $\psi_\mu$ has 16 degrees of freedom, in general it can be decomposed into a spin 3/2 component and two spin 1/2 components. Historically, the spin 1/2 components were not welcome and people got rid of them by adjusting the parameters $\zeta$, $\xi$ and $z$ such that the masses of the spin 1/2 components are infinite. In this way, only the spin 3/2 component can become on shell and be a physical particle. The constraints on the parameters are:
\bea
&&z=3\xi^2+3\xi+1,\label{zxirel}\\
&&\zeta=\frac{3\xi^2+2\xi+1}{2}\label{zetaxirel},
\eea
which reduces the parameters to a single one $\xi$. The spin 3/2 component satisfies:
\beq
p^\mu \psi(p)_\mu=0,~~\gamma^\mu \psi(p)_\mu=0.
\label{spin32property}
\eeq
One can further make a point transformation:
\beq
\psi_\mu\rightarrow \psi_\mu+\lambda\gamma_\mu\gamma_\nu\psi^\nu,
\label{pointtrans}
\eeq
which does not affect the spin 3/2 component due to (\ref{spin32property}). Under (\ref{pointtrans}) the form of the Lagrangian (\ref{oldRS}) is not changed, but just with $\xi$ transformed to $ \xi(1-4\lmd)-2\lmd$. Thus $\xi$ can be set to a preferred value and the theory has no free parameters. There are two common choices for the parameters:
\bea
&&\xi=-1,~~\zeta=1,~~z=1,\nb\\
\mbox{or}~~&&\xi=-\frac{1}{3},~~\zeta=\frac{1}{3},~~z=\frac{1}{3}.
\eea

This theory has a fatal problem that when it is minimally coupled to the electromagnetic field, there exist wave function solutions that propagate faster than light \cite{Velo:1969bt,Velo:1970ur}. This superluminal propagation problem is solved recently by Konstantin G. Savvidy \cite{Kostas10}.

\subsection{A Modified Rarita-Schwinger Theory for Spin 3/2}\label{newRarSch}
In \cite{Kostas10}, Konstantin G. Savvidy found that by a modification of the Rarita-Schwinger theory (\ref{oldRS}), the superluminal propagation problem can be solved. He abandoned the idea that there should not be on shell spin 1/2 components besides the spin 3/2 one. Instead, he kept one of the spin 1/2 components by relaxing the condition (\ref{zxirel}), and $z$ became a free parameter in his modified Rarita-Schwinger theory. Using the transformation (\ref{pointtrans}), one can fix $\xi$ for convenience. By choosing:
\beq
\xi=2z-1,~~\zeta=6z^2-4z+1,
\eeq
the wave function $u_2^\mu$ of the spin 1/2 component is transverse, ie. satisfying $p_\mu u_2(p)^\mu=0$. With this choice of parameters, the mass $M$ of the spin 1/2 component is related to the spin 3/2 component mass m as:
\beq
M=\frac{m}{6z-2}.
\eeq


By minimal substitution $p_\mu \to p_\mu + e \, A_\mu$ in (\ref{oldRS}), Konstantin G. Savvidy obtained the electromagnetic interaction in the modified theory:
\beq
{\mathcal L}_I=e \, {\bar \psi}_\lambda  \, \Gamma^{\mu\lambda}{}_\rho  \, \psi^\rho \, A_\mu,
\eeq
with $\Gamma^{\mu\lambda}{}_\rho$ given in (\ref{oldRS}). Ward identity, the requirement from gauge invariance, is satisfied by this interaction vertex:
\beq
-i  \, k_\mu  \, \Gamma^{\mu\lambda}{}_\rho=S^\lambda{}_\rho(p+k)^{-1}-S^\lambda{}_\rho(p)^{-1}.
\label{wardid}
\eeq

In \cite{Kostas10}, it is found that wave function solutions with superluminal propagation are only turned on when $z=\frac{1}{3}$, which corresponds to the old Rarita-Schwinger theory.


\subsection{Examples of Effective Field Theory Approach}
Peccei \cite{Peccei1968, Peccei1969} introduced a workable effective theory of the $\Delta^+$, ie. the old Rarita-Schwinger theory reviewed in Section \ref{oldRarSch}, including the lowest order effective coupling to the photon:
\beq
{\cal L}\propto \bar\psi_{\Delta^+}^\lambda(i \gamma_\alpha g_{\lmd\bt}-i \gamma_\bt g_{\lmd\alpha}-\frac{1}{2}\gamma_\lmd \sigma_{\alpha\bt})\gamma_5 \psi_p F^{\alpha\bt}+h.c.\ \ .
\label{Pecceiint}
\eeq
Pascalutsa and Scholten \cite{Pascalutsa1995} extended the Peccei approach with additional allowed forms of the $\gamma \Delta^+ p$ interactions:
\bea
&&{\cal L}_{\gamma \Delta^+ p}={\cal L}_{\gamma \Delta^+ p}^1+{\cal L}_{\gamma \Delta^+ p}^2+{\cal L}_{\gamma \Delta^+ p}^3,\nb\\
&&{\cal L}_{\gamma \Delta^+ p}^1=\frac{iG_1}{2M_p}\bar\psi_{\Delta^+}^\alpha \Theta_{\alpha\mu}(z_1)\gamma_\nu\gamma_5 T_3 \psi_p F^{\nu\mu}+h.c.\ ,\nb\\
&&{\cal L}_{\gamma \Delta^+ p}^2=\frac{-G_2}{(2M_p)^2}\bar\psi_{\Delta^+}^\alpha \Theta_{\alpha\mu}(z_2)\gamma_5 T_3\partial_\nu \psi_p F^{\nu\mu}+h.c.\ ,\nb\\
&&{\cal L}_{\gamma \Delta^+ p}^3=\frac{-G_3}{(2M_p)^2}\bar\psi_{\Delta^+}^\alpha \Theta_{\alpha\mu}(z_3)\gamma_5 T_3 \psi_p \partial_\nu F^{\nu\mu}+h.c.\ ,
\label{Pascalutsaint}
\eea
where $\Theta_{\alpha\bt}(z)=g_{\alpha\bt}-(z+\frac{1}{2})\gamma_\alpha\gamma_\bt$, and $T_3$ is the third component of $\frac{1}{2}\to \frac{3}{2}$ isospin transition operator. Then they fit to experimental data with four parameters $G_1, G_2, z_1$ and $z_2$ ($G_3$ and $z_3$ do not contribute in real Compton scattering).

Pascalutsa and Scholten also discussed in detail the spin 3/2 propagator which is necessary for calculating the contributions of the $\Delta^+$ exchange diagrams. They found that better agreement is obtained if spin 1/2 contributions are kept in the propagator, which led them to conclude that the spin 1/2 contributions may be remnants of some high-mass excited states of the nucleon.

In (\ref{Pecceiint}) and (\ref{Pascalutsaint}), the electromagnetic interactions of $\Delta^+$ and proton are all introduced through couplings to the strength tensor $F_{\mu\nu}$, which is crucial for them to maintain gauge invariance.

\section{Motivations of Our Work}

We intend to study proton Compton scattering in the effective field theory approach. In previous works like \cite{Peccei1968, Peccei1969, Pascalutsa1995} they used the old Rarita-Schwinger theory for $\Delta^+$ in the intermediate state. However, this approach is undermined by the fact that the old Rarita-Schwinger theory is pathological when coupled to electromagnetic fields, as mentioned in Section \ref{oldRarSch}. We looked for a solution to this problem.

Fortunately, as we have mentioned, in \cite{Kostas10} Konstantin G. Savvidy proposed a modified Rarita-Schwinger theory, which is a well grounded theory for charged spin 3/2 particles. He also derived the minimal electromagnetic interaction which satisfies gauge invariance. At the same time, since in his modified Rarita-Schwinger theory there is an additional spin 1/2 component besides the spin 3/2 one, we can use the two components to model proton and $\Delta^+$ respectively. That is to say, we unify proton and $\Delta^+$ in an effective spin 3/2 theory. This unification is appealing since proton and $\Delta^+$ have the same constituent quarks and can transform from one into the other by absorbing or emitting a photon or even a neutral pion (which itself carries no spin and charge).

In the next chapter, I will describe our work \cite{Zhang:2013uaa} on proton Compton scattering based on the modified Rarita-Schwinger theory \cite{Kostas10}, with Dr. Konstantin G. Savvidy as my advisor. We introduce several non minimal electromagnetic interactions to describe proton and $\Delta^+$ magnetic moments and proton-$\Delta^+$ M1 transition amplitudes etc.. We study the approximation of the amplitudes around the $\Delta^+$ peak. We also introduce an effective Lagrangian for the polarzabilities. Then we calculate proton Compton scattering cross section based on the effective Lagrangian, and fit to current experimental data. Main advantages of our work over previous works are as explained in the last paragraph: first, our approach is based on a well grounded spin 3/2 theory for $\Delta^+$; second, we naturally unify proton and $\Delta^+$ in one theory. Additionally, it turns out that we can make a prediction of $\Delta^+$ magnetic moment.

\chapter{Proton Compton Scattering In Unified Proton-$\Delta^+$ Model}

\section{Introduction}
Throughout the decades, there have been many experimental and theoretical works \cite{Mohr:2000ie,Mohr:2005zz,Mohr:2008fa,Mohr:2012tt, Baranov1974, Hallin1993,Zieger92,Hunger97, MacGibbon1995,Blanpied01,Olmos2001,MAMI2001,Lvov1997, Hearn:1962zz, Pfeil:1974ib, Lvov:1980wp, Drechsel:1999rf,Lvov:1979zd, Guiasu:1978ak,GellMann:1954db,Baldin1960, Pasquini2007,Feuster:1998cj, Pascalutsa1995,Scholten1996, Kondratyuk2001,Pagels1974, Weinberg1978, Gasser1983, Gasser1987,McGovern:2012ew,Peccei1968, Peccei1969} aimed to study proton's electromagnetic properties. Among them, the electric and magnetic polarizabilities, right behind the electric charge and magnetic moment, are neither well measured or profoundly understood theoretically. The electric and magnetic polarizabilities measure the induced electric and magnetic multipole moments of the proton in electric and magnetic fields. They can be measured in proton Compton scattering. The lightest baryon resonance $\Delta^+(1232{\mbox MeV}, J^P=\frac{3}{2}^+)$) plays a very important role in low energy proton Compton scattering.

In an effective field theory approach, people have been using the old Rarita-Schwinger theory reviewed in Section \ref{oldRarSch} to describe $\Delta^+$ \cite{Peccei1968, Peccei1969,Pascalutsa1995}. However the old Rarita-Schwinger theory suffers from superluminal propagation problem. In \cite{Kostas10}, Konstantin G. Savvidy, my advisor of this work, proposed a modified Rarita-Schwinger theory that is free of the superluminal propagation problem, reviewed in Section \ref{newRarSch}:
\bea
&&{\mathcal L}= - {\bar \psi}_\lambda \, [p_\mu \, \Gamma^{\mu\lambda}{}_\rho-m \, \Theta^\lambda{}_\rho] \, \psi^\rho,\nonumber\\
&&\Gamma^{\mu\lambda}{}_\rho=\gamma^\mu \, \eta^\lambda{}_\rho+\xi \, (\gamma^\lambda \,  \eta^\mu{}_\rho+\gamma_\rho \,  \eta^{\lambda\mu})+\zeta \, \gamma^\lambda \, \gamma^\mu \, \gamma_\rho,\nonumber\\
&&\Theta^\lambda{}_\rho=\eta^\lambda{}_\rho- z  \, \gamma^\lambda \, \gamma_\rho,\nonumber\\
&&\xi=2\, z-1,\ \ \ \ \zeta= 6\, z^2-4\,z+1.
\eea

In this chapter, we will study proton Compton scattering, using this modified Rarita-Schwinger theory to describe $\Delta^+$. At the same time, we can unify proton and $\Delta^+$ in this modified Rarita-Schwinger theory. The relation between proton mass M and $\Delta^+$ mass m is $M=\frac{m}{6z-2}$. This unification is very natural since proton and $\Delta^+$ have the same quark constituents. Unification of proton and $\Delta^+$ is a distinct feature of our work, compared to previous works in effective field theory approach.

For phenomenological applications, it is necessary to be able to tune the electromagnetic properties of the system to the experimentally observed form. First, the theory developed in \cite{Kostas10} 
is sufficient to describe the minimal interaction of this physical system with the photon, but in the real world it is expected that in the photon intermediate energy range of 100-300 MeV the Compton scattering process is dominated by contributions from the anomalous magnetic moment as well as the polarizabilities of the nucleon. Of these additional contributions, anomalous magnetic moment contributes to the amplitudes already at the linear order while polarizabilities start out at the second order. Thus cross-section can grow at first quadratically and then quartically with energy. This is in contrast to the minimal coupling in QED, where the Klein-Nishina cross-section is essentially constant in the relevant energy range. Thus, the minimal electromagnetic interaction of this model is not enough.

We investigate the non-minimal electromagnetic interactions, with their coefficients called form factors, in Section \ref{elemagint}.  Whereas in the Dirac theory there is only the anomalous magnetic moment in addition to the charge, in the spin 3/2 case there are five form-factors for coupling the momentum-independent fermion bilinears directly to the EM field strength.  The form factors are in principle arbitrary functions of transferred momentum-squared. In the present case of real Compton scattering, transferred momentum-squared is fortunately zero such that the form factors are constants and are jointly determined from experimental data. In this section we also introduce an effective Lagrangian to model the static electric and magnetic polarizabilities. In Section \ref{transmatrix}, we discuss the p-p, $\Delta^+$-$\Delta^+$, p-$\Delta^+$ transition matrices and derive the formulae of the proton and $\Delta^+$ magnetic moments. In Section \ref{comscaandpol}, we calculate proton Compton scattering cross section, express electric and magnetic polarizabilites in terms of the form factors, and also analyze the behavior of the amplitudes around the $\Delta^+$ pole. The interactions in Section \ref{elemagint} can also be used for $\Delta^+\to p\,\gamma$ decay process, described in Section \ref{sec:decayamp}. Then we fit the Compton scattering data in Section \ref{datafit}, and extract polarizabilities using expressions derived in Section \ref{comscaandpol}. Since one combination of the form factors gives $\Delta^+$ magnetic moment, our fitting result also provides a prediction of $\Delta^+$ magnetic moment, the experiments on which suffer from large uncertainties. Finally in Section \ref{conclu} we conclude.

\section{Electromagnetic Interaction}\label{elemagint}
The minimal electromagnetic interaction has been derived in \cite{Kostas10}:
\beq
{\mathcal L}_I=e \, {\bar \psi}_\lambda  \, \Gamma^{\mu\lambda}{}_\rho  \, \psi^\rho \, A_\mu,
\eeq
where $\Gamma^{\mu\lambda}{}_\rho=\gamma^\mu  \, \eta^\lambda{}_\rho+\xi \, (\gamma^\lambda  \, \eta^\mu{}_\rho+\gamma_\rho  \, \eta^{\lambda\mu})+\zeta \, \gamma^\lambda \, \gamma^\mu \, \gamma_\rho$.

The propagator is \cite{Kostas10}:
\bea
\label{ModRSpropagator}
-i S(p)&=
 & \frac{(\slashed{p}+m)\, \Pi_3  }{p^2-m^2+i\,\epsilon}  
- \frac{(\slashed{p}+M) \, \Pi_{11} }{p^2-M^2 - i\,\epsilon} \, \frac{2M^2}{m^2} \nonumber\\ 
&+&  
\left[ \vspace{30pt} \Pi_{22} - (\Pi_{21}+\Pi_{12})\,/B +  \Pi_{11}\,3/B^2 \right] \, \frac{3}{2 \,(M +2 \, m)}~~,\nonumber\\
B &=& \frac{3\,m}{2 \, M +m} ~~.
\eea
where the standard spin projection operators $\Pi$ can be found for example in \cite{Benmerrouche:1989uc}. 

Ward identity is satisfied:
\beq
-i  \, k_\mu  \, \Gamma^{\mu\lambda}{}_\rho=S^\lambda{}_\rho(p+k)^{-1}-S^\lambda{}_\rho(p)^{-1}.
\label{wardid}
\eeq

\subsection{Non-Minimal Electromagnetic Interactions}
For phenomenological application to proton Compton scattering, minimal interaction alone does not suffice. A well known fact for Dirac theory is that it allows for two electromagnetic form factors, one of which is charge and the other describes the anomalous magnetic moment. 
We add the following as yet undetermined non-minimal interactions to the vertex:
\beq
{\tilde \Gamma}^{\mu\lambda}{}_\rho= \Gamma^{\mu\lambda}{}_\rho+ \frac{i}{2M} \, \sum\limits_{n} F_n(k^2) \, (\Gamma_n{})^{\mu\lambda}{}_\rho,
\label{genvertex}
\eeq
where the $F_n(k^2)$ are form factors. If amplitudes are to be gauge invariant, the
Ward identity (\ref{wardid}) should still hold for ${\tilde \Gamma}^{\mu\lambda}{}_\rho$. For that, it is sufficient to set $k_\mu (\Gamma_n){}^{\mu\lambda}{}_\rho=0$ and thus $\Gamma_n$ can be of the form:
\beq
(\Gamma_n)^{\mu\lambda}{}_\rho=(\Sigma_n)^{\mu\nu\lambda}{}_\rho \, k_\nu,
\label{nonminvertex}
\eeq
where $(\Sigma_n)^{\mu\nu\lambda}{}_\rho$ is antisymmetric in $\mu$ and $\nu$.

Antisymmetric tensors live in the $(1,0)\oplus(0,1)$ representation of Lorentz group, and we can count the number of these representations in the product representation of the two matter fields. Representation for $\psi_\lambda$(or ${\bar \psi}_\lambda$) is a product of that for a vector field and that for a spinor field:
\beq
(\frac{1}{2},\frac{1}{2})\otimes[(\frac{1}{2},0)\oplus(0,\frac{1}{2})]=(1,\frac{1}{2})\oplus(0,\frac{1}{2})\oplus(\frac{1}{2},0)\oplus(\frac{1}{2},1).
\eeq
The vertexes live in the tensor product of the above reducible representations, and
\bea
&&[(1,\frac{1}{2})\oplus(0,\frac{1}{2})\oplus(\frac{1}{2},0)\oplus(\frac{1}{2},1)]\otimes[(1,\frac{1}{2})\oplus(0,\frac{1}{2})\oplus(\frac{1}{2},0)\oplus(\frac{1}{2},1)]\nonumber\\
&& \supset 5\ [(1,0)\oplus(0,1)].
\eea
This tells us that there are five antisymmetric tensors and we have been able to explicitly construct them as:
\bea
&&(\Sigma_1)^{\mu\nu\lambda}{}_\rho=-\frac{1}{2}\tau^{\mu\nu\lambda}{}_\rho\ ,\nonumber\\
&&(\Sigma_2)^{\mu\nu\lambda}{}_\rho=\sigma^{\mu\nu}\eta^\lambda{}_\rho\ ,\nonumber\\
&&(\Sigma_3)^{\mu\nu\lambda}{}_\rho=-\frac{1}{9}\gamma^\lambda\sigma^{\mu\nu}\gamma_\rho\ ,\nonumber\\
&&(\Sigma_4)^{\mu\nu\lambda}{}_\rho=\frac{1}{12}(\gamma^\lambda\gamma^\mu \eta^\nu{}_\rho-\gamma^\lambda\gamma^\nu \eta^\mu{}_\rho+\gamma^\mu\gamma_\rho \eta^{\nu\lambda}-\gamma^\nu\gamma_\rho \eta^{\mu\lambda})\ ,\nonumber\\
&&(\Sigma_5)^{\mu\nu\lambda}{}_\rho=\frac{-i}{12}(\gamma^\lambda\gamma^\mu \eta^\nu{}_\rho-\gamma^\lambda\gamma^\nu \eta^\mu{}_\rho-\gamma^\mu\gamma_\rho \eta^{\nu\lambda}+\gamma^\nu\gamma_\rho \eta^{\mu\lambda}),
\label{antisymtensor}
\eea
where $\tau^{\mu\nu\lambda}{}_\rho$ and $\sigma^{\mu\nu}$ are generators of the Lorentz transformations for spin 1 and 1/2 respectively, $\tau^{\mu\nu\lambda}{}_\rho=i(\eta^{\mu\lambda}\eta^\nu{}_\rho-\eta^{\nu\lambda}\eta^\mu{}_\rho)$ and $\sigma^{\mu\nu}=\frac{i}{4}(\gamma^\mu\gamma^\nu-\gamma^\nu\gamma^\mu)$. 
The coefficients in (\ref{antisymtensor}) are chosen such that the form factors $F_n$ enter with equal weights in proton magnetic moment in (\ref{magmoment}).

These tensors satisfy the requirement of Hermiticity:
\beq
[{\bar \psi}(p_1)_\lambda ~ (\Sigma_n)^{\mu\nu\lambda}{}_\rho  ~ \psi^\rho(p_2) ~(p_1-p_2)_\nu]^\dagger = - {\bar \psi}(p_2)_\lambda  ~ (\Sigma_n)^{\mu\nu\lambda}{}_\rho  ~  \psi^\rho(p_1) ~ (p_2-p_1)_\nu,
\eeq
which implies that the form factors  $F_n(k^2)$ are real.

At higher order in momenta there are a small number of additional possibilities for non minimal interactions \cite{Chen:2011ve}. The pure spin 3/2 field has three form factors other than charge \cite{chensav}, while only the first and second ones in (\ref{antisymtensor}) contribute to pure spin 3/2, because $\gamma_\rho \, u_4^\rho=0$ for spin 3/2 wave function $u_4^\rho$.
 Thus we add one more tensor:
\beq
(\Sigma_6)^{\mu\nu\lambda}{}_\rho=\frac{1}{M^2} \,  k^\lambda \, \sigma^{\mu\nu} \, k_\rho.
\label{sixthtensor}
\eeq

The form factors $F_i(k^2)$ which appear as coefficients in (\ref{genvertex}) are scalar functions of transferred-momentum squared. For real Compton scattering, the photon is on-shell, ie. $k^2$=0, so these form factors are taken to be constants in what follows. 


\subsection{Bare Polarizabilities Effective Lagrangian}

Expanding Compton scattering cross section in lab frame at low energy, static polarizabilities $\alpha_E$ and $\beta_M$ first enter at second order:
\beq
\frac{d \sigma}{d \Omega_{\mbox{\tiny lab}}}=\left(\frac{d \sigma}{d \Omega}\right)_{\mbox{\tiny Powell}}-\frac{e^2 \omega^2}{4 \pi M}(\frac{\alpha_E+\beta_M}{2}(1+\cos\theta)^2+\frac{\alpha_E-\beta_M}{2}(1-\cos\theta)^2)+\mathcal O(\omega^3).
\label{polarizability}
\eeq
Thus, polarizabilities are here defined in the standard way as in the literature, by comparing the theoretical predictions and experimental data to the Powell cross section $\left(\frac{d \sigma}{d \Omega}\right)_{\mbox{\tiny Powell}}$, which is the differential cross section of a Dirac point particle with anomalous magnetic moment included \cite{Powell1949, Low1954, GellMann1954}. A different definition would result if the Klein-Nishina result for the Dirac point particle without anomalous magnetic moment is taken as the basis for comparison. Such difference has sometimes led to confusion in the literature, but has been satisfactorily resolved by separating the contributions due to the anomalous magnetic moment. 
The non-minimal interaction vertices presented in this section contribute to the effective polarizabilities $\alpha_E$ and $\beta_M$ as we shall see in Section \ref{comscaandpol}. 
In addition to these vertices, we may include also an effective four-point contact interaction that can contribute directly to the polarizabilities. Inspired by the effective Lagrangian proposed in \cite{0611327}, we include the following interaction Lagrangian to model ''bare" polarizabilities:
\beq
{\cal L}_{\mbox{\tiny pol}}=\frac{i\pi}{M}~(\bar\psi_\lambda \, \Gamma^{\mu\lambda}{}_\rho  \,\partial_\nu  \,\psi^\rho-\partial_\nu  \,\bar\psi_\lambda  \,\Gamma^{\mu\lambda}{}_\rho \psi^\rho)(\alpha_{EB} ~ F_{\mu\rho} \, F^{\rho\nu}+\beta_{MB} ~  \tilde{F}_{\mu\rho} \, \tilde{F}^{\rho\nu}).
\label{poleffLag}
\eeq
This Lagrangian is not unique, but other candidates contribute identically to the cross section up to the second order of the incident photon energy. 

The two coefficients $\alpha_{EB}$ and $\beta_{MB}$ we call bare polarizabilities.
The contribution of $\alpha_{EB}$ and $\beta_{MB}$ to the cross section is of the form in (\ref{polarizability}). In the 
low energy limit where the proton is at rest before and after scattering, and the photon frequency tends to zero, the Hamiltonian corresponding to (\ref{poleffLag}) is: 
\beq
{\cal H}_{\mbox{\tiny pol}}=-2\pi(\alpha_{EB} |\vec{E}|^2+\beta_{MB} |\vec{H}|^2),
\label{barepol}
\eeq
in agreement with (\ref{stapolHam}).

At higher orders in photon energy it is possible to define and extract more general polarizabilities, refer to Section \ref{sec:polarizabilities}.

\section{Magnetic Moments and the $\gamma\,N\,\Delta^+$ Transition Matrix}\label{transmatrix}
To calculate the magnetic moments, let $A^\mu=(0,\vec{A})$, say $A^\mu=(0,A_x,0,0)$, and $\vec{p}_1-\vec{p}_2=q\hat{z}$, with $\vec{p}_1\to 0$ and $\vec{p}_2\to 0$. The definitions of proton and $\Delta^+$ magnetic moments $\mu_p$ and $\mu_{\Delta^+}$ are as following:
\bea
&&e \, \bar{u}_2(p_1,\sigma_1)  \, \tilde{\Gamma}^\mu   \, u_2(p_2,\sigma_2)   \, A_\mu = 
                  2  \,M   \, \mu_p  \left(J_y^{\left(\frac{1}{2}\right)}\right)_{\sigma_1\sigma_2}(-i q A_x)+\mathcal O(q^2),\nonumber\\
&&e \, \bar{u}_4(p_1,\sigma_1)  \, \tilde{\Gamma}^\mu   \, u_4(p_2,\sigma_2)   \, A_\mu=
                          -2  \,m   \, \mu_{\Delta^+}  \left(J_y^{\left(\frac{3}{2}\right)}\right)_{\sigma_1\sigma_2}(-i \,q \, A_x)+\mathcal O(q^2).
\label{magmom}
\eea
In this equation, $u_2$ and $u_4$ are respectively the spin 1/2 and 3/2 positive frequency wave functions; $\vec{J}^{\left(\frac{1}{2}\right)}$ and $\vec{J}^{\left(\frac{3}{2}\right)}$ are standard quantum mechanical spin operators for spin 1/2 and 3/2. 

In units of $\mu_N=\frac{e}{2M}$ (M is proton mass), the magnetic moments of the proton and $\Delta^+$ as defined by (\ref{magmom}) are:
\bea
&&\frac{\mu_p}{\mu_N}=1+\lambda_p=1+\frac{4M(m+M)}{3m^2}+\frac{2M^2}{3m^2}(F_1+F_2+F_3+F_5),\nonumber\\
&&\frac{\mu_{\Delta^+}}{\mu_N}=\frac{M}{m}+(-\frac{1}{2}F_1+F_2).
\label{magmoment}
\eea

When all the form factors are set to zero, $\mu_{\Delta^+}=\frac{e}{2m}$, so the g-factor of $\Delta^+$ is $\frac{2}{3}$, which agrees with expectations for that of an elementary spin 3/2 particle \cite{Belinfante:1953zz}. However, even in the minimally coupled theory, the  spin 1/2 particle still has an anomalous magnetic moment. Proton magnetic moment $\mu_p=2.79$ is well measured and acts as a constraint on the form factors through (\ref{magmoment}).  Intriguingly, the actual value is close to that of the minimally coupled theory, so that $F_1+F_2+F_3+F_5$ is approximately zero.



$F_6$ does not contribute to the magnetic moments because it is higher order in the soft photon momentum. $F_4$ does not enter $\mu_{\Delta^+}$ due to $\gamma_\rho u_4^\rho=0$ for spin 3/2 solution $u_4^\rho$. $F_4$ does not appear in the proton magnetic moment, either, as can be shown from the equation of motion. 


In the limit of degenerate masses for proton and $\Delta^+$ (where in reality the mass gap is indeed small: $\frac{|m-M|}{M}\sim 0.3$), we can calculate the transition amplitudes between slowly moving proton and $\Delta^+$. We take $\Delta^+$ at rest: $p_1=(m,0,0,0)$ and proton momentum $p_2=(\sqrt{M^2+k^2},0,0,k)$ and work in the degenerate limit $M \to m$. We take the photon to be left polarized $A^{\tiny{L}} = \frac{1}{\sqrt{2}}(0,1,i,0)$ or right polarized 
$A^{\tiny{R}}=\frac{1}{\sqrt{2}}(0,1,-i,0)$ and we calculate $e \, \bar{u}_4 \,(p_1,\sigma_1) \,\tilde{\Gamma}^\mu  \,u_2(p_2,\sigma_2)  \, A^{L/R\mu}$. For small k, the transitions are $\mathcal O(k)$ and the results are: 
\bea
&&e ~\bar{u}_4(p_1,\sigma_1)~\tilde{\Gamma}^\mu ~u_2(p_2,\sigma_2) ~A_{L\mu}=e \,G \,k~ \left(\begin{array}{cc} 1& 0\\ 0 & \frac{1}{\sqrt{3}} \\ 0& 0        \\0&0\end{array}\right)_{\sigma_1,\sigma_2}+\mathcal O(k^2),\nonumber\\
&&e ~\bar{u}_4(p_1,\sigma_1)~\tilde{\Gamma}^\mu ~u_2(p_2,\sigma_2) ~A_{R\mu}=e \,G \,k~\left(\begin{array}{cc} 0&0    \\ 0&0 \\ -\frac{1}{\sqrt{3}} & 0\\ 0 & -1\end{array}\right)_{\sigma_1,\sigma_2}+\mathcal O(k^2),\label{Gcoupling}
\eea
where\footnote{for $e \bar{u}_2\tilde{\Gamma}^\mu u_4 e_{\mu}$, the combination appearing is $G^\dagger$.} $G=\frac{1}{4\sqrt{6}}(2F_1+8F_2+F_5+8-i F_4)$ determines the magnetic transition amplitudes between the proton and the $\Delta^+$. $G$ is an important parameter and will also show up in the next section.

\section{Compton Scattering Cross Section}\label{comscaandpol}

\begin{figure}[htdp]
\center{
\includegraphics[height=4.5cm,width=13cm]{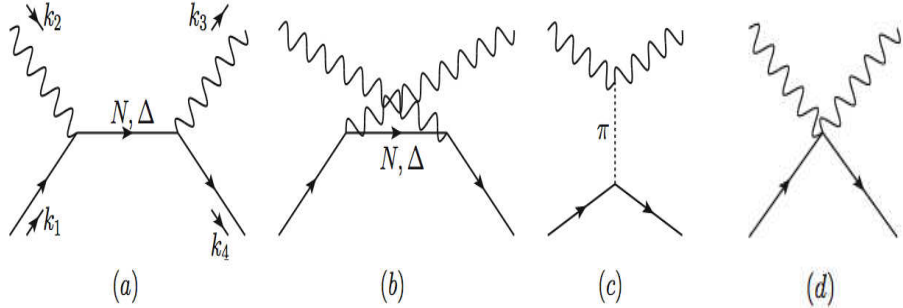}}
\caption{Tree-level Feynman diagrams for proton Compton scattering. In (a) and (b) the intermediate particle is proton or $\Delta^+$, and in (c) the $\pi^0$ meson is exchanged. (d) is the diagram for the contact interaction in (\ref{poleffLag}). For appropriate photon incident energy, the intermediate $\Delta^+$ is approximately on-shell, and around this energy, there is the characteristic peak in the cross-section which is dominated by the $\Delta^+$ contribution.}
\label{picture1}
\end{figure}

At tree level, the Feynman diagrams for proton Compton scattering are shown in Figure.\ref{picture1}, through exchanges of proton, $\Delta^+$ and $\pi^0$, and a contact interaction from (\ref{poleffLag}). For s and u channels (diagrams (a) and (b) in Figure. \ref{picture1}), the vertices and propagator were given in the previous section in (\ref{genvertex}, \ref{nonminvertex}, \ref{antisymtensor}, \ref{sixthtensor}, \ref{ModRSpropagator}). 


For pion exchange t-channel diagram (c),
 there is no contribution from $\Delta^+$, and we use the familiar Dirac spinor for proton wave function. 
The relevant interaction Lagrangian is:
\beq
{\cal L}_{\mbox{\tiny int}}=i \, g_\pi \, \bar u\,  \gamma^5 \, u \, \pi^0 - \frac{1}{8}\, F_{\pi\gamma\gamma} \, \epsilon_{\mu\nu\rho\lambda}\, F^{\mu\nu}\, F^{\rho\lambda} \, \pi^0.
\eeq

The fourth diagram (d) is according to (\ref{poleffLag}).

In addition to the diagrams in Figure. \ref{picture1}, strong interactions contribute through the pion one-loop diagrams as in Figure. \ref{picture2}. Above the pion photoproduction threshold, the diagram (a) in Figure. \ref{picture2} contributes to the  imaginary part of the self-energy of $\Delta^+$ and determines the line-shape of the $\Delta^+$ resonance.  In principle, the imaginary part depends on c.m. momentum squared s, and all the diagrams in Figure. \ref{picture2} should be taken into account at one-loop order \cite{McGovern:2012ew}. 
For our purposes, we make an estimate of this effect by setting the imaginary part of $\Delta^+$ mass m to be the observed width at the resonance, ie. we analytically extend the amplitudes by substituting m with $(m_0-i\frac{\Gamma}{2})\sim (1210-50i)$ MeV everywhere it appears in the amplitudes. 
In the spirit of the conserved vector current, despite modifications of both the vertex and the propagator, with this procedure Ward identitiy is preserved due to the analyticity of   (\ref{wardid}) in $m$. 

\begin{figure}[htdp]
\setlength{\unitlength}{1cm}
\center{\includegraphics[height=7cm,width=11cm]{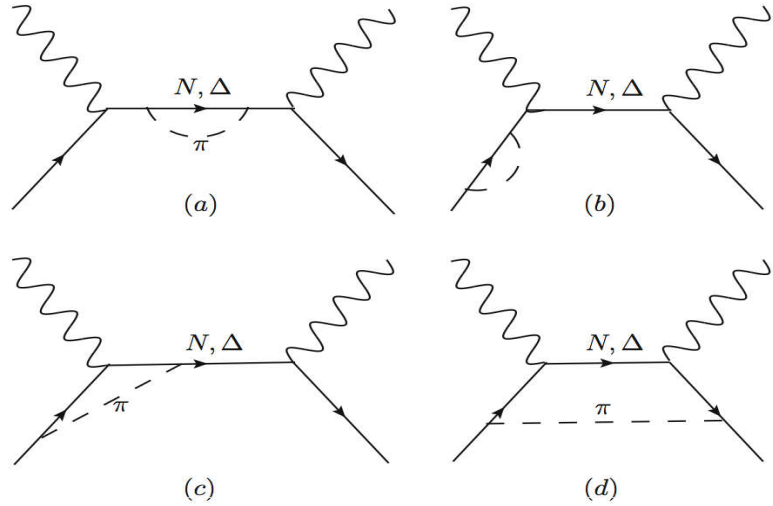}}
\caption{The one-loop level  pion corrections: (a) and (b) are the self-energy diagrams; (c) and (d) are pion vertex corrections. More diagrams emerge at higher loop levels.}
\label{picture2}
\end{figure}

It is verified that our results of amplitudes satisfy Low's theorem \cite{Low1954}, namely that for the low energy Compton scattering on spin 1/2 particles, the amplitudes expanded to first order of photon energy are completely determined by the mass, electric charge and magnetic moment of the spin 1/2 particle. 
According to the theorem, in lab frame with photon incident energy $\omega$, the amplitudes are:
\bea
&&{\mathcal A}_{\sigma_1,\sigma_4,\lambda_2,\lambda_3}\nonumber\\
=&&\frac{e^2}{M} \vec{\epsilon}_{\lambda_2}\cdot \vec{\epsilon}^*_{\lambda_3}\delta_{\sigma_1 \sigma_2}-\frac{i e \omega}{M}(2 \mu_p-\frac{e}{2M})( \vec{\epsilon}^*_{\lambda_3}\times \vec{\epsilon}_{\lambda_2})\cdot \vec{\sigma}_{\sigma_1\sigma_2}\nonumber\\
&&+\frac{ie\mu_p}{\omega M}(\vec{\epsilon}_{\lambda_2}\cdot \vec{k}_3(\vec{\epsilon}^*_{\lambda_3}\times \vec{k}_3)-\vec{\epsilon}^*_{\lambda_3}\cdot \vec{k}_2(\vec{\epsilon}_{\lambda_2}\times \vec{k}_2))\cdot \vec{\sigma}_{\sigma_1\sigma_2}\nonumber\\
&&+\frac{2i\mu_p^2}{\omega}((\vec{\epsilon}^*_{\lambda_3}\times \vec{k}_3)\times(\vec{\epsilon}_{\lambda_2}\times \vec{k}_2))\cdot \vec{\sigma}_{\sigma_1\sigma_2}+O(\omega^2).
\label{lowtheorem}
\eea


Expansion of the cross section to second order in photon energy $\omega$ is exactly in the form of (\ref{polarizability}) with $\alpha_E \pm \beta_M$: 
\allowdisplaybreaks{
\bea
&&\alpha_E+\beta_M \nonumber\\
=&&\frac{e^2}{288 \pi m^4 (m-M) (m+M)(2m+M)}\nonumber\\
&&\left[\right.\left(-32 M^4-64 m M^3-96 m^2 M^2-192 m^3 M\right)+\left(-64 m^4-96 M m^3-32 M^2 m^2\right) F_1\nonumber\\
&&+\left(64 m^4-480 M m^3-320 M^2 m^2-32 M^3 m\right) F_2+\left(-64 M m^3-32 M^2 m^2\right) F_5\nonumber\\
&&+\left(-3 m^4-12 M m^3-17 M^2 m^2+8 M^4\right) F_1^2+\left(-256 M m^3-128 M^2 m^2\right) F_2^2\nonumber\\
&&+\left(-4 M m^3-2 M^2 m^2\right) F_4^2+\left(6 M^4-8 m^2 M^2-4 m^3 M\right) F_5^2\nonumber\\
&&+\left(-60 m^4-72 M m^3-52 M^2 m^2-24 M^3 m+16 M^4\right) F_1 F_2\nonumber\\
&&+\left(16 M^4+8 m M^3-16 m^2 M^2-8 m^3 M\right) F_1 F_3\nonumber\\
&&+\left(-6 m^4-16 M m^3-22 M^2 m^2+4 M^3 m+16 M^4\right) F_1 F_5\nonumber\\
&&+\left(32 M^4-32 m M^3-32 m^2 M^2+32 m^3 M\right) F_2 F_3\nonumber\\
&&+\left(24 M^4-24 m M^3-56 m^2 M^2-40 m^3 M\right) F_2 F_5+\left(8 M^4-8 m^2 M^2\right) F_3 F_5\left.\right]\nonumber\\
&&+\alpha_{EB}+\beta_{MB},\nonumber\\
   &&\nonumber\\
&&\alpha_E-\beta_M\nonumber\\
=&&\frac{e^2}{288\pi m^4 M (m-M)(m+M)(2m+M)}\nonumber\\
&&\left[\right.\left(-32 M^5-192 m M^4-224 m^2 M^3+320 m^3 M^2+512 m^4 M\right)\nonumber\\
&&+\left(96 m^5+240 M m^4+192 M^2 m^3-80 M^3 m^2-192 M^4 m-64 M^5\right) F_1\nonumber\\
&&+\left(-192 m^5+608 M m^4+672 M^2 m^3-32 M^3 m^2-224 M^4 m-64 M^5\right) F_2\nonumber\\
&&+\left(-64 M^5-128 m M^4+64 m^2 M^3+128 m^3 M^2\right) F_3\nonumber\\
&&+\left(-64 M^5-160 m M^4+192 m^3 M^2+128 m^4 M\right) F_5\nonumber\\
&&+\left(-24 M^5-32 m M^4+15 m^2 M^3+44 m^3 M^2+21 m^4 M\right) F_1^2\nonumber\\
&&+\left(-32 M^5-32 m M^4+24 m^2 M^3+160 m^3 M^2+264 m^4 M\right) F_2^2\nonumber\\
&&+\left(32 m^2 M^3-32 M^5\right) F_3^2+\left(6 M m^4+2 M^2 m^3-2 M^3 m^2\right) F_4^2\nonumber\\
&&+\left(-26 M^5-16 m M^4+24 m^2 M^3+18 m^3 M^2+6 m^4 M\right) F_5^2\nonumber\\
&&+\left(-48 M^5-88 m M^4+60 m^2 M^3+184 m^3 M^2+84 m^4 M\right) F_1 F_2\nonumber\\
&&+\left(-48 M^5-24 m M^4+48 m^2 M^3+24 m^3 M^2\right)F_1 F_3\nonumber\\
&&+\left(-48 M^5-44 m M^4+42 m^2 M^3+56 m^3 M^2+18 m^4 M\right) F_1 F_5\nonumber\\
&&+\left(-32 M^5-64 m M^4+32 m^2 M^3+64 m^3 M^2\right) F_2 F_3\nonumber\\
&&+\left(-40 M^5-72 m M^4+32 m^2 M^3+104 m^3 M^2+72 m^4 M\right) F_2 F_5\nonumber\\
&&+\left(-56 M^5-16 m M^4+56 m^2 M^3+16 m^3 M^2\right) F_3 F_5\left.\right]\nonumber\\
&&+\alpha_{EB}-\beta_{MB}.
\label{polarizabilityexpressnew}
\eea
}

$\alpha_{EB}$ and $\beta_{MB}$ are bare polarizabilities defined in (\ref{barepol}). Substitute $m=1232$ MeV and $M=938$ MeV in (\ref{polarizabilityexpressnew}) we have approximately:
\bea
\alpha_E+\beta_M\sim && -0.860-0.556 F_1-1.789 F_2-0.240 F_5\nonumber\\
&&-0.069 F_1^2-0.961 F_2^2-0.015 F_4^2-0.020 F_5^2\nonumber\\
&&-0.536 F_1 F_2-0.023 F_1 F_3-0.085 F_1 F_5+0.009 F_2 F_3\nonumber\\
&&-0.234 F_2 F_5-0.007 F_3 F_5+\alpha_{EB}+\beta_{MB}\ (10^{-4} fm^3),\nonumber\\
&&\nonumber\\
\alpha_E-\beta_M\sim &&1.894+1.284 F_1+2.602 F_2+0.202 F_3+0.650 F_5\nonumber\\
&&+0.146 F_1^2+1.339 F_2^2+0.028 F_3^2+0.023 F_4^2+0.064 F_5^2\nonumber\\
&&+0.728 F_1 F_2+0.069 F_1 F_3+0.177 F_1 F_5+0.101 F_2 F_3\nonumber\\
&&+0.447 F_2 F_5+0.067 F_3 F_5+\alpha_{EB}-\beta_{MB}\ (10^{-4} fm^3).
\label{numpolarizability}
\eea



\subsection{Amplitudes at the $\Delta^+$ Pole}

In reality, $\Delta^+$ and proton have small mass gap:
\beq
m=M (1+x- i y),
\eeq
where $x\sim 0.3$ and $y\sim 0.05\sim \frac{1}{2}x^2$. Around the $\Delta^+$ resonance, the photon momentum divided by M, $q/M$, is of the same order as x, we can approximate the amplitudes to lowest non-trivial order of $q/M$, x and y.

First, at the $\Delta^+$ pole position, the contribution to the pole mainly comes from the first term in the propagator in s-channel where the momentum propagated is the sum of the momenta of initial proton and initial photon $k_1+k_2$. In this case, the denominator contributing to the pole is $(k_1+k_2)^2-m^2=(E_{CM}-m)(E_{CM}+m)$. We multiply the amplitudes by $(E_{CM}-m)$ and then expand them with respect to $q/M$, x, and y. We define $q/M= r_1 x$ and $y=r_2 x^2$. Around the peak, $r_1\sim 1$ and $r_2 \sim 0.5$. Then we can  expand the amplitudes multiplied by $(E_{CM}-m)$ with respect to x alone, and finally set $r_1=1$ (at the peak) and $r_2=\frac{y}{x^2}$. In center of mass frame, we rotate the proton wave functions to make them polarized along proton's direction of moving. That is, for proton moving in direction $\theta$ with respect to z-axis, we have the polarized proton wave functions $\tilde u_2$ by rotating the directly boosted wave functions $u_2$: 
\bea
&&\tilde{u}_2(p,\frac{1}{2}) =\cos\Big(\frac{\theta}{2}\Big) ~u_2(p,\frac{1}{2})+\sin\Big(\frac{\theta}{2}\Big) ~u_2(p,-\frac{1}{2}),\nonumber\\
&&\tilde{u}_2(p,-\frac{1}{2}) =-\sin\Big(\frac{\theta}{2}\Big) ~u_2(p,\frac{1}{2})+\cos\Big(\frac{\theta}{2}\Big) ~u_2(p,-\frac{1}{2}).
\label{polarizedproton}
\eea
Then the approximate amplitudes are:
\bea
&&{\cal A}_{ppRR}=\frac{(2 i r_2+1 -2|G|^2) x^2\cos^3\frac{\theta}{2}}{2(E_{CM}-m)},\nonumber\\
&&{\cal A}_{ppRL}=\frac{(-2 i r_2-1 -2|G|^2) x^2\cos\frac{\theta}{2}\sin^2\frac{\theta}{2}}{2(E_{CM}-m)},\nonumber\\
&&{\cal A}_{ppLR}=\frac{(-2 i r_2-1 -2|G|^2) x^2\cos\frac{\theta}{2}\sin^2\frac{\theta}{2}}{2(E_{CM}-m)},\nonumber\\
&&{\cal A}_{ppLL}=\frac{(6 i r_2+3 +2|G|^2+3(2 i r_2+1-2|G|^2)\cos\theta) x^2\cos\frac{\theta}{2}}{12(E_{CM}-m)},\nonumber\\
&&{\cal A}_{pmRR}=\frac{(2 i r_2+1 -2|G|^2) x^2\cos^2\frac{\theta}{2}\sin\frac{\theta}{2}}{2(E_{CM}-m)},\nonumber\\
&&{\cal A}_{pmRL}=\frac{(-2 i r_2-1 -2 |G|^2) x^2\sin^3\frac{\theta}{2}}{2(E_{CM}-m)},\nonumber\\
&&{\cal A}_{pmLR}=\frac{(-6 i r_2-3 +2|G|^2+3(2 i r_2+1+2|G|^2)\cos\theta) x^2\sin\frac{\theta}{2}}{12(E_{CM}-m)},\nonumber\\
&&{\cal A}_{pmLL}=\frac{(2 i r_2+1 -2|G|^2) x^2\cos^2\frac{\theta}{2}\sin\frac{\theta}{2}}{2(E_{CM}-m)}.
\label{peakbehavior}
\eea
In the above expressions, $x=\frac{\mbox{\tiny Re}(m)-M}{M}$, $r_2=\frac{y}{x^2}=-\frac{\mbox{\tiny Im}(m)}{M x^2}$. $G$ is the combination of form factors which describes the $\gamma N \Delta^+$ transition amplitudes as described in the previous section in (\ref{Gcoupling}). For the subscripts of the amplitudes, the first/second p/m stands for the final/initial proton polarization and the first/second L/R for final/initial photon polarization. Here only 8 of the 16 amplitudes are given, since the other 8 are related by parity:
\bea
&&{\cal A}_{mmRR}={\cal A}_{ppLL},\ \ \ \ \ \ \ {\cal A}_{mpRR}=-{\cal A}_{pmLL},\nonumber\\
&&{\cal A}_{mmRL}={\cal A}_{ppLR},\ \ \ \ \ \ \ {\cal A}_{mpRL}=-{\cal A}_{pmLR},\nonumber\\
&&{\cal A}_{mmLR}={\cal A}_{ppRL},\ \ \ \ \ \ \ {\cal A}_{mpLR}=-{\cal A}_{pmRL},\nonumber\\
&&{\cal A}_{mmLL}={\cal A}_{ppRR},\ \ \ \ \ \ \ {\cal A}_{mpLL}=-{\cal A}_{pmRR}.
\eea

In the same limit, the magnetic polarizability has an approximation:
\beq
\beta_M=-\frac{4\alpha |G|^2}{3 x M^3}+\beta_{MB}.
\eeq
This may imply that the magnetic polarizability has some relationship with proton-$\Delta^+$ (magnetic) transition.

\section{$\Delta^+\to p+\gamma$ Decay Width}
\label{sec:decayamp}
Aside from proton Compton scattering, another process can be readily accounted for in this unified proton-$\Delta^+$ electromagnetic theory, that is the decay $\Delta^+\to p+\gamma$.

We label the momentum and polarization of $\Delta^+$ as $k_1, \sigma_1$, the produced photon $k_2$,$\lambda_2$ and the proton $k_3$,$\sigma_3$, the matrix elements are:

\beq
{\mathcal A}_{\sigma_1,\lambda_2,\sigma_3}=e \ \ {\bar u_2}(k_3,\sigma_3) \ {\tilde\Gamma}^\mu \ u_4(k_1,\sigma_1)\ \epsilon_\mu^*(k_2, \lambda_2).
\eeq
It is of interest to find the decay width $\Gamma_{3/2}$ and $\Gamma_{1/2}$ for the final state helicity $\frac{3}{2}$ and $\frac{1}{2}$ respectively. 
Evaluating the amplitudes, we obtain after substituting the values of m (real value) and M: 
\bea
\Gamma_{3/2}=&& 0.0047 F_1^2+ 0.056 F_2^2 + 0.001 F_4^2 + 0.001 F_5^2 + 0.032 F_1F_2  + 0.004F_1 F_5  \nonumber\\
&&+0.0139 F_2 F_5 + 0.032 F_1  + 0.1113 F_2 +  0.0139 F_5 + 0.0557,\nonumber\\
\Gamma_{1/2}=&&0.0004 F_1^2 + 0.0120 F_2^2 + 0.0002F_4^2 + 0.0002 F_5^2 + 0.0002 F_6^2 + 0.0058 F_1 F_2 \nonumber\\
&&+ 0.0005  F_1F_5  - 0.0006  F_1F_6  + 0.0037 F_2 F_5-  0.0039 F_2 F_6 - 0.0004 F_5 F_6\nonumber\\
&&+ 0.0043 F_1 + 0.0293 F_2  + 0.0027 F_5  - 0.0028 F_6 + 0.0108.
\eea

If, on the other hand, we let $m=M(1+x)$ then in the limit of small x we find
\bea
&&\Gamma_{3/2}\sim\alpha M x^3 |G|^2,\nonumber\\
&&\Gamma_{1/2}\sim\frac{\alpha M x^3 |G|^2}{3}.
\label{widthapprox}
\eea

Experimentally $\Delta^+\to p+\gamma$ decay amplitudes can be extracted from the $\Delta^+$ peak of proton Compton scattering\cite{Blanpied01}, or from the peak in $\gamma+N\to \pi +N$ \cite{Dugger:2007bt, Ahrens:2004pf}. Then the decay width is estimated from the extracted amplitudes.

\section{Fitting Data}\label{datafit}

We fit the model to the 714 proton Compton scattering datapoints from 8 experiments\cite{Hallin1993,Baranov1974,Zieger92,Hunger97,Blanpied01,MAMI2001,Olmos2001,MacGibbon1995}. Only data points with photon incident energy smaller than 455 MeV are used, in the  so-called first resonance region. In principle, one can also compare the model predictions with polarized measurements, where some data is available \cite{al:1993aa, Wada:1981ab}.

For several reasons, we set $F_4=0$ in our fitting. First, $F_4$ does not enter in the expressions of proton and $\Delta^+$ magnetic moments (\ref{magmoment}). Second, in all fits we attempted, the best fit value of  $F_4$ was nearly exactly zero, and in any case statistically consistent with zero.

The parameters we use to fit are chosen to be $F_1$, $\mu\equiv\frac{\mu_\Delta^+}{\mu_N}=\frac{M}{m}+F_2-\frac{1}{2}F_1$, $G=\frac{1}{4\sqrt{6}}(2F_1+8F_2+F_5+8)$, $F_6$ and the bare polarizabilities $\alpha_B$ and $\beta_B$ in (\ref{poleffLag}). $F_3$ is constrained using proton magnetic moment, see (\ref{magmoment}).

We minimize
$\chi^2=\sum_{i=1}^{714}\frac{((\frac{d\sigma}{d\Omega})_i^{\mbox{\tiny\it calc}}-(\frac{d\sigma}{d\Omega})_i^{\mbox{\tiny\it data}})^2}{\sigma_{\mbox{\tiny (stat)}i}^2+\sigma_{\mbox{\tiny (syst)}i}^2}$.
We do not attempt to rescale the data of each experiment within its own systematical uncertainty to see if it would lead to better consistency between datasets as it was done in  \cite{Baranov:2000na}. The optimal set of parameters is found to be: $F_1=-27.5$, $\mu=14.2$, $G =3.13$, $F_6 = 12.9$, $\alpha_B = 7.5$, $\beta_B = -8.2$ with $\chi^2\sim 6.3\times 10^3$.

We plot the c.m. frame cross section using the above fit parameters, together with data points in Figure.\ref{angleplot} and \ref{energyplot}. Those data measured or recorded in lab frame have been converted to c.m. frame. From Figure.\ref{energyplot} it is seen that the low energy cross section fits badly.

On average for each data point the fit is of $3\sigma$ deviation from the experiment value. The $\Delta^+$ resonance region is fitted well, while the low energy cross section deviates greatly from data points. In fact, the 68 data points with incident photon energy smaller than 140 MeV out of the total 714 data points contribute nearly a third of the total $\chi^2$. Our fit cannot take care of the low energy ($\lae$ 140 MeV) data points and "high" energy ($\gae$ 200 MeV) data points simultaneously. When giving a good fit in the resonance region, where most data points used in this thesis lie in, the predicted cross section at low energy cannot account of the large asymmetry of the cross section data at forward and backward angles. Our fit cross section at low energy is much higher than the data at forward angles and lower at backward angles. Since the polarizabilities are extracted according to low energy expansion of the cross section in (\ref{polarizability}), it is expected that the predicted $\alpha_E+\beta_M$, calculated using (\ref{numpolarizability}) where $F_2$ and $F_5$ are solved from the definitions of $\mu$ and $G$, is smaller than the experimental value, and $\alpha_E-\beta_M$ larger than experimental value. For the above fit values, $\alpha_E+\beta_M=-2.6 (10^{-4}fm^3)$, $\alpha_E-\beta_M=45.1(10^{-4}fm^3)$. 
By contrast, the original experiments have quoted values of $\alpha_E+\beta_M$ at about $14.0 (10^{-4}fm^3)$, and $\alpha_E-\beta_M$ at about $10.0 (10^{-4}fm^3)$ extracted from the same data.

\begin{figure}[htdp]
\setlength{\unitlength}{1cm}
\includegraphics[height=4cm,width=6cm]{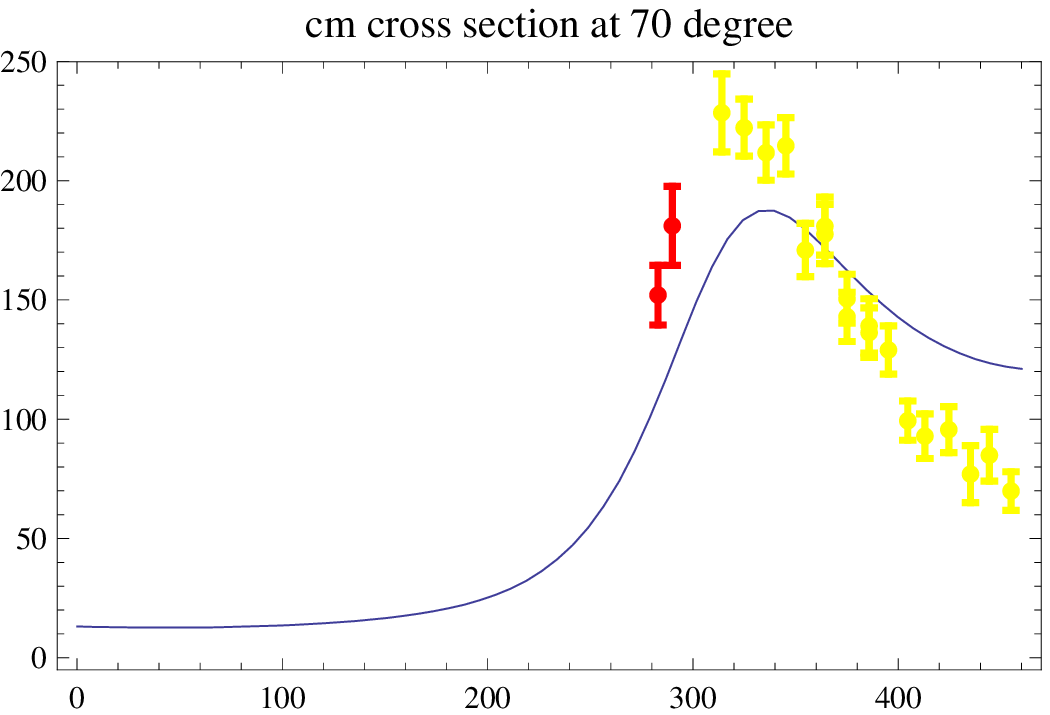}
\includegraphics[height=4cm,width=6cm]{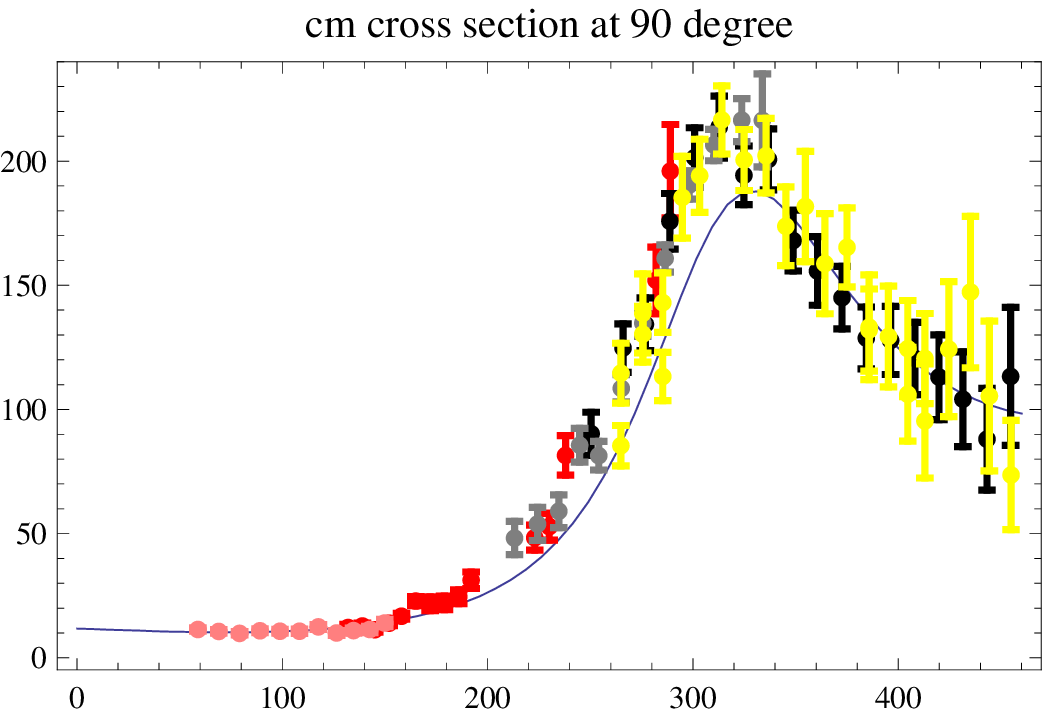}\\
\includegraphics[height=4cm,width=6cm]{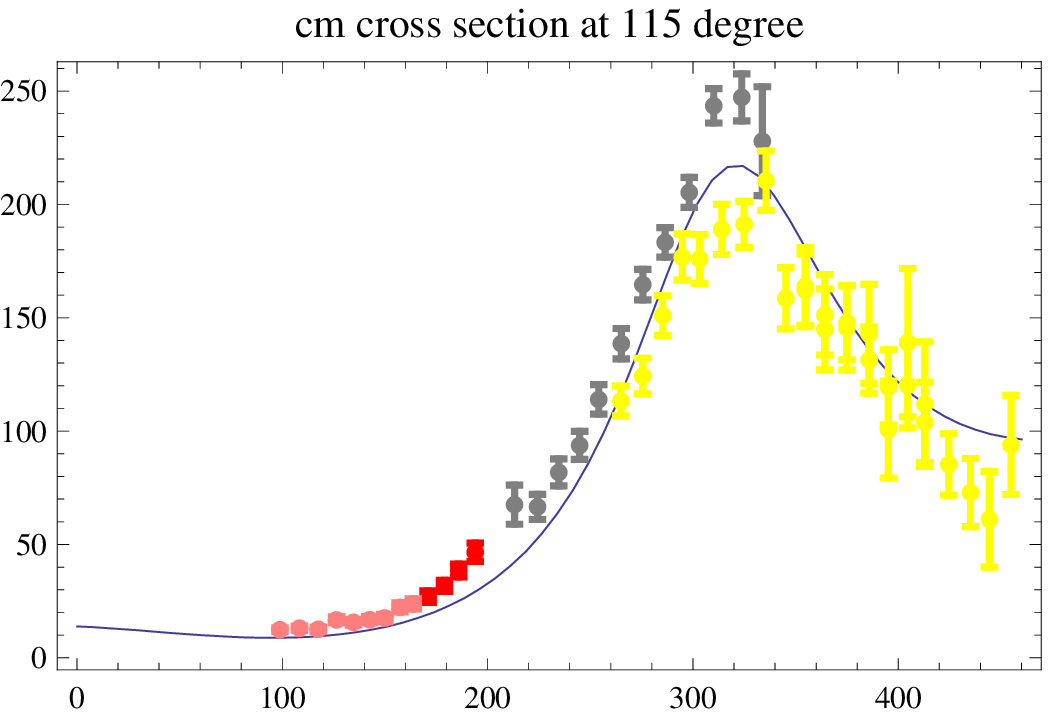}
\includegraphics[height=4cm,width=6cm]{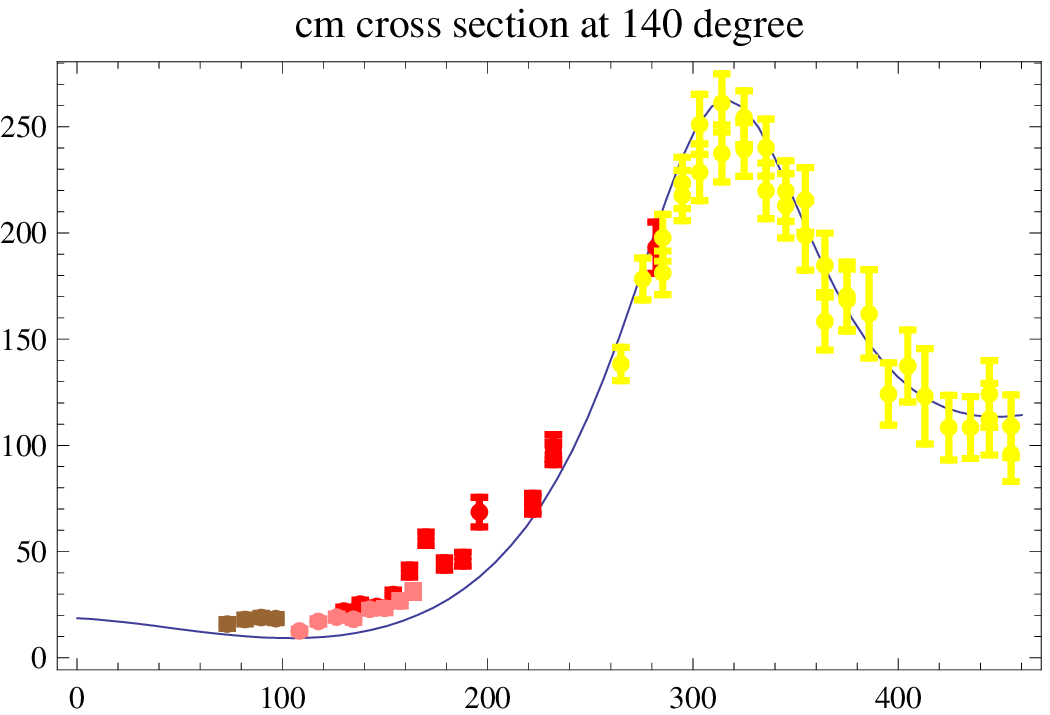}
\caption{Fixed c.m. angle cross section and the data points, where the parameters from the fitting to all the 714 data points are used for the theoretical cross section curve. The x-axis is lab frame photon energy and y-axis is c.m. frame differential cross section in unit of nb. 
For \cite{Hallin1993,Baranov1974,Zieger92,Hunger97,Blanpied01,MAMI2001,Olmos2001,MacGibbon1995}, we use colors: Red, Green, Blue, Black, Gray, Yellow, Pink and Brown respectively. The angles of the data points included may differ from the values claimed by at most 3 degree.}
\label{angleplot}
\end{figure}

\begin{figure}[htdp]
\setlength{\unitlength}{1cm}
\includegraphics[height=4cm,width=6cm]{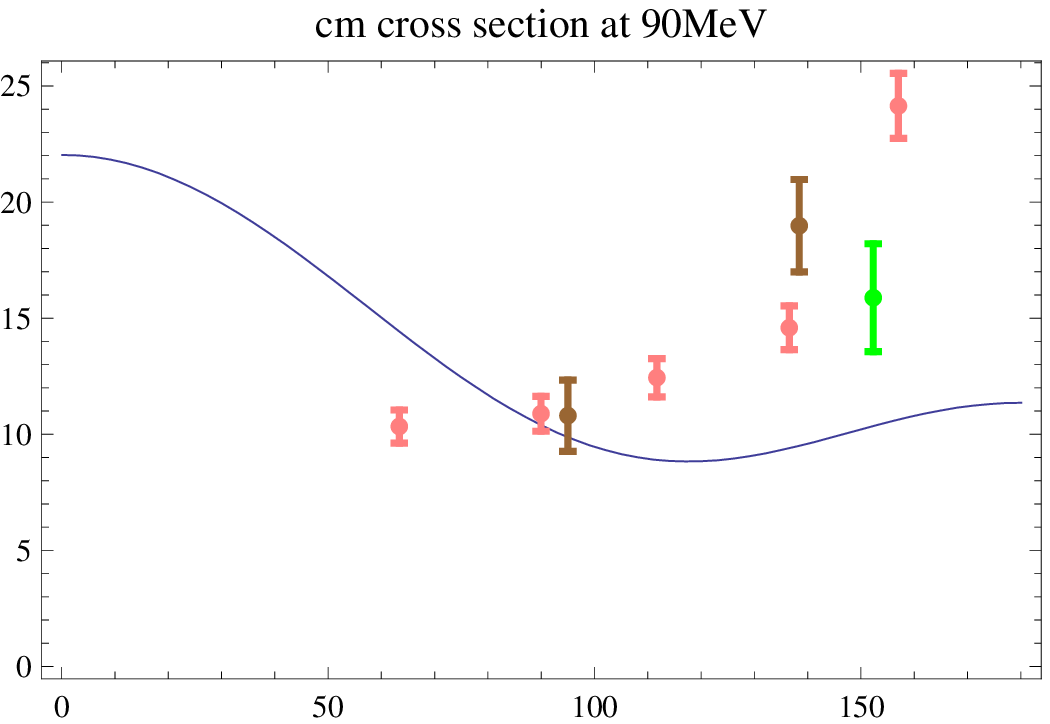}
\includegraphics[height=4cm,width=6cm]{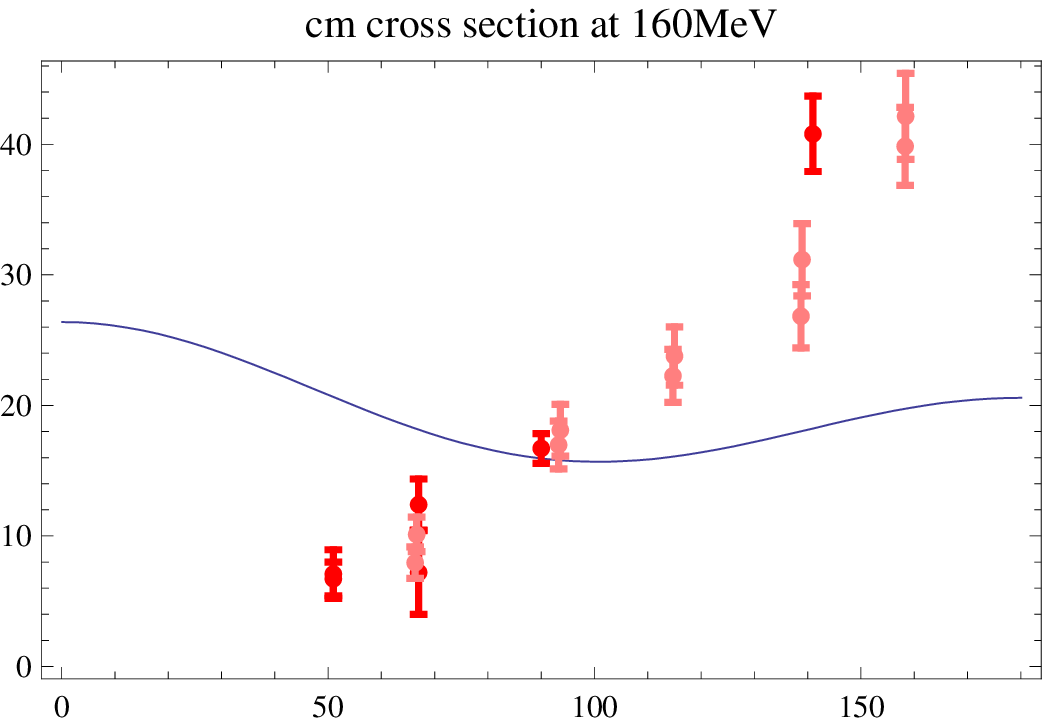}\\
\includegraphics[height=4cm,width=6cm]{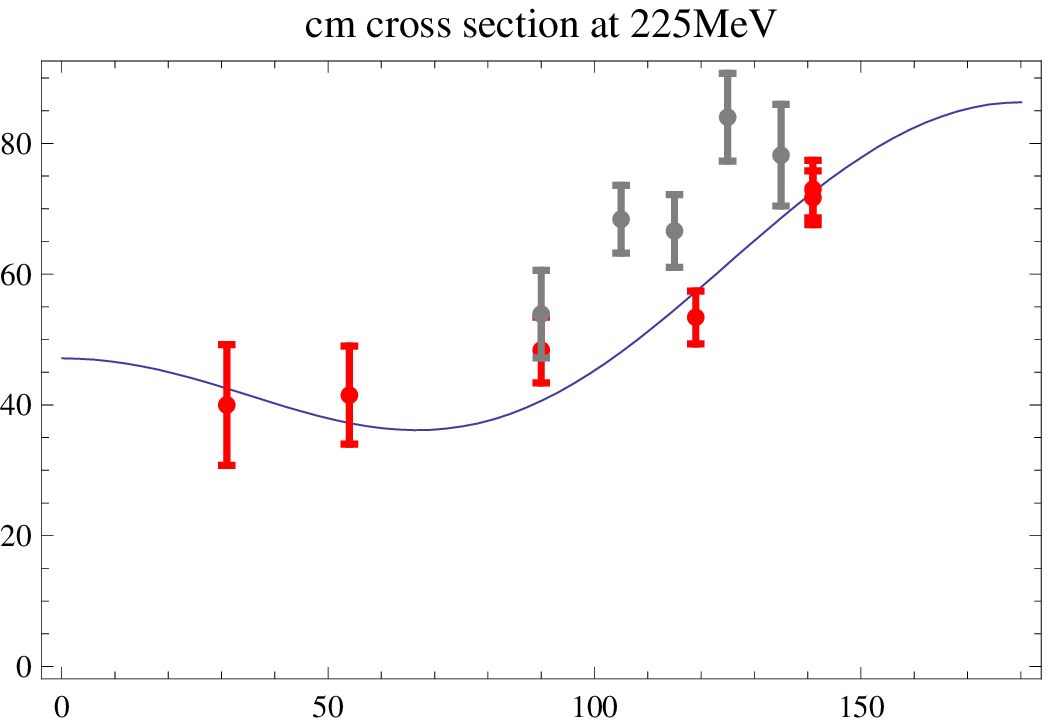}
\includegraphics[height=4cm,width=6cm]{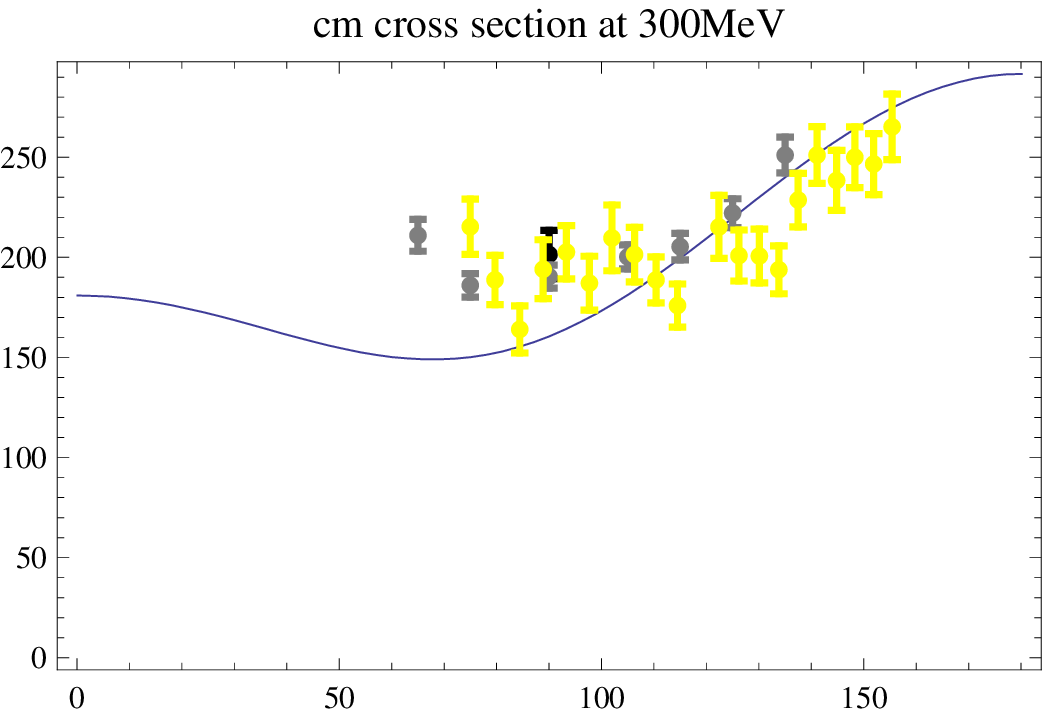}\\
\includegraphics[height=4cm,width=6cm]{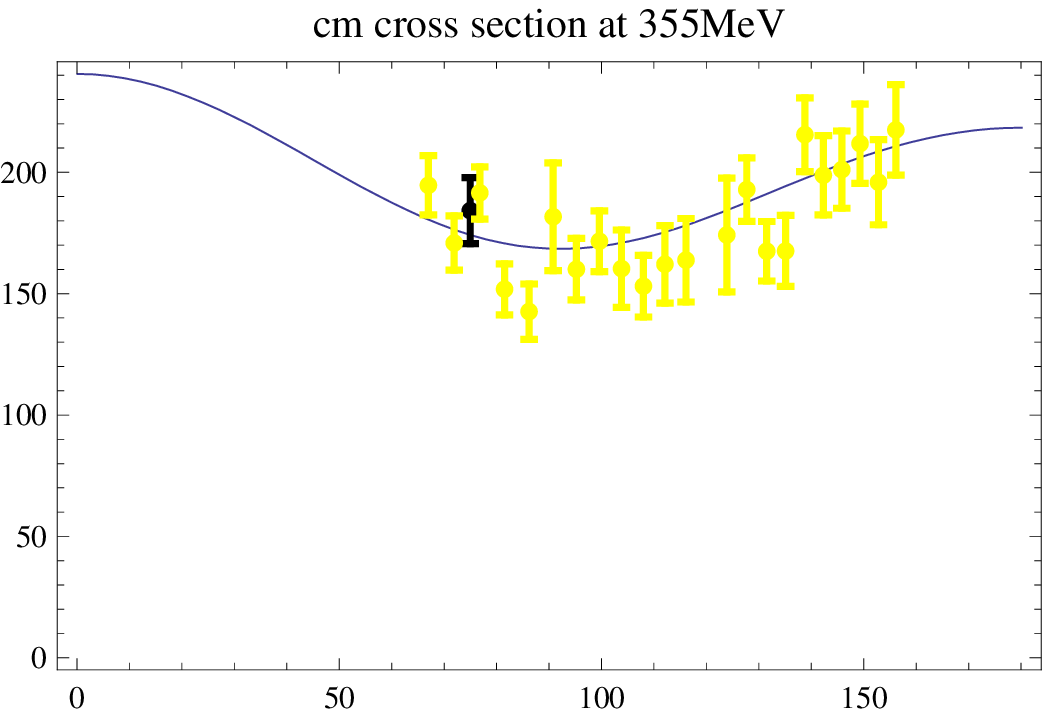}
\includegraphics[height=4cm,width=6cm]{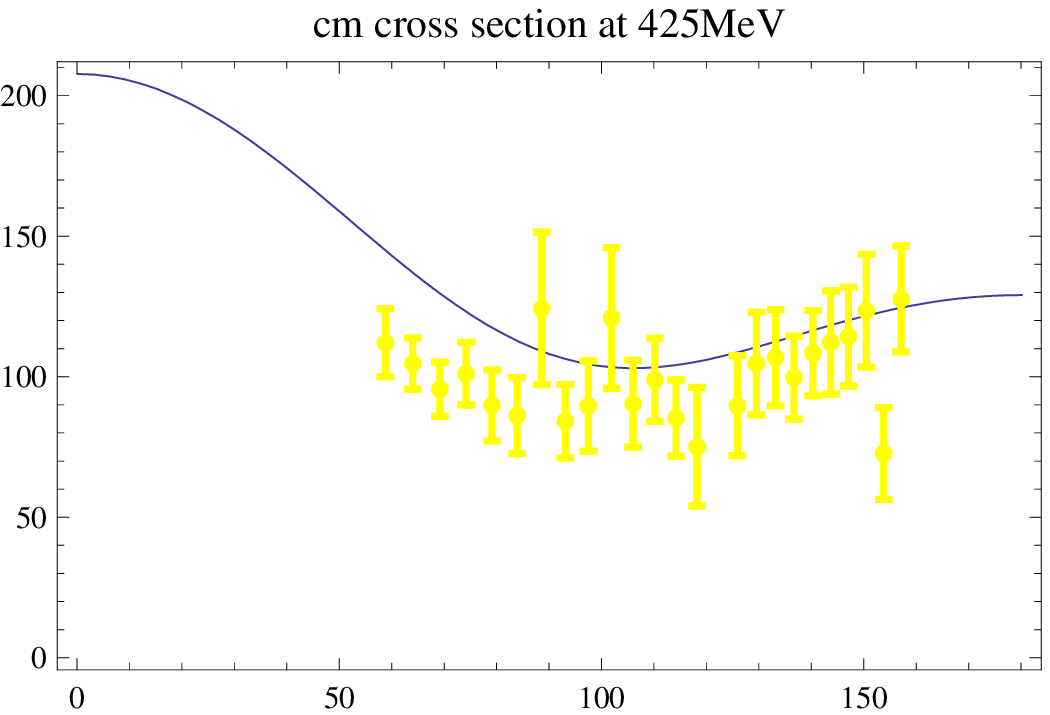}
\caption{Fixed c.m. angle cross section and the data points, where the parameters from the fitting to all the 714 data points are used for the theoretical cross section curve. The x-axis is c.m. frame scattering angle and y-axis is c.m. frame differential cross section in unit of nb. For \cite{Hallin1993,Baranov1974,Zieger92,Hunger97,Blanpied01,MAMI2001,Olmos2001,MacGibbon1995}, we use colors: Red, Green, Blue, Black, Gray, Yellow, Pink and Brown respectively. The incident photon energy of the data points included may differ from the values claimed by at most 4 MeV.}
\label{energyplot}
\end{figure}


\begin{figure}[htdp]
\setlength{\unitlength}{1cm}
\includegraphics[height=14cm,width=16cm]{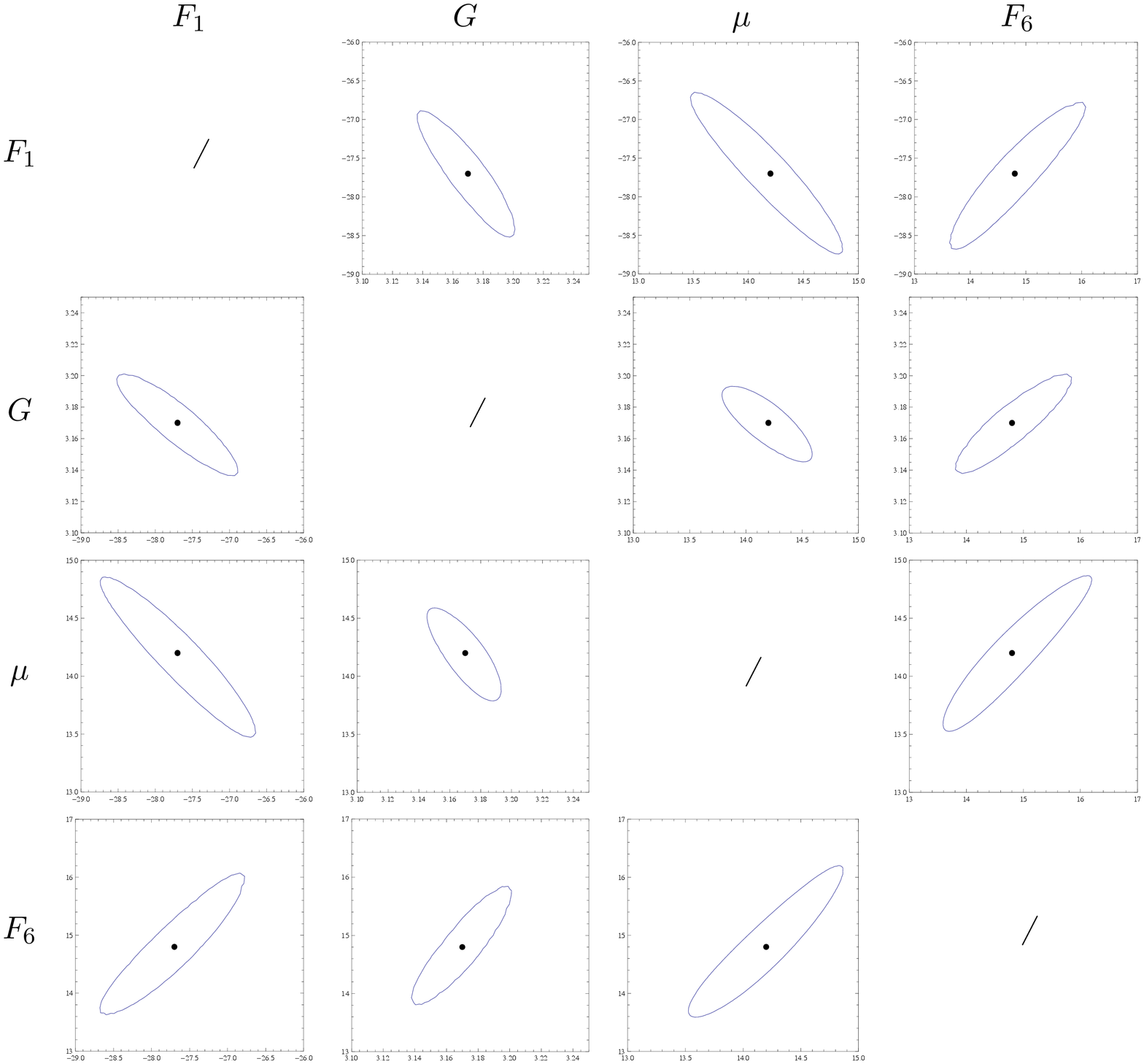}
\caption{Projections of the 6-dimensional 95\% confidence region into planes spanned by several pairs of parameters. Plots in each column share the same x-axis parameter as indicated at the top of each column. Plots in each row share the same y-axis parameter as indicated at the left of each row. Range for $F_1$ is (-29,-26); G: (3.1,3.25); $\mu$: (13,15); $F_6$: (13,17).}
\label{chisqu}
\end{figure}

The challenge is clearly that the low energy part and the $\Delta^+$ resonance range data points are difficult to fit well at the same time. The form factors (and thus the parameters $\mu$ and $G$) are generally functions of $k^2$ where $k$ is photon momentum. For real Compton scattering, $k^2$ is always 0, so the form factors should be constants in this thesis. However in the case of the bare polarizabilities  it is possible that they vary with energy and/or scattering angle \cite{Schumacher2013}, which would make fitting with constant bare polarizabilities unsuccessful.

The strategy we propose to deal with the possible variation of bare polarizabilities is as follows. First, we fit only the peak range data points and fix the form factors(and $\mu$ and $G$) from this fitting. Then we fit the low energy data points varying only the bare polarizabilities. For the peak range, we use only the MAMI(2001) experiment \cite{MAMI2001}, which contains 436 data points with photon incident energy ranging from 260 MeV to 455 MeV. A good fit is achieved at $F_1=-27.7$, $G=3.17$, $\mu=14.2$, $F_6=14.8$, $\alpha_B=2.1$, $\beta_B=-8.1$ with $\chi^2\sim 830$. It is notable that $F_1$, $F_6$, $\mu$ and $G$ have not changed much from the complete fit of all data points, yet $\chi^2$ per data point is much smaller. This may be indicative of the fact that experimental data prior to this latest and more precise measurement may not be consistent with each other. In the past, one of the strategies for dealing with this has been to allow rescaling the cross section data for each experiment within the systematical uncertainty which tends to be large \cite{Baranov:2000na}. 

The inclusion of the sigma channel and/or variation of the mass and width of the sigma meson do not appreciably alter the picture or the goodness of the fit.

In Figure.\ref{chisqu}, we give the contour plots of $\chi^2$ with respect to several pairs of parameters for this fit. It is seen that $G$ is very strictly constrained.

We then fix these values for $F_1$, $F_6$, $\mu$ and $G$ and fit the low energy data points, varying only $\alpha_B$ and  $\beta_B$. We take 68 data points with photon incident energy below 140 MeV and obtain the best fit values of $\alpha_B=-4.6$ and $\beta_B=17.9$ with $\chi^2\sim 194$. See Figure. \ref{lowangleplot} and \ref{lowenergyplot}. At these values, $\alpha_E+\beta_M=11.3\pm 0.9\pm 2.3(95\% C.L.)$ and $\alpha_E-\beta_M=7.8\pm 3.3\pm 2.0(95\% C.L.)$, much closer to values extracted (from the same data) previously. The first error is determined from the MAMI fit, by investigating how the contribution of $F_1$, $G$ and $\mu$ to the polarizabilities varies in the 3-dimensional 95\% confidence region spanned by these three parameters. The second error is from the low energy fit.

\begin{figure}[htdp]
\setlength{\unitlength}{1cm}
\includegraphics[height=4cm,width=5cm]{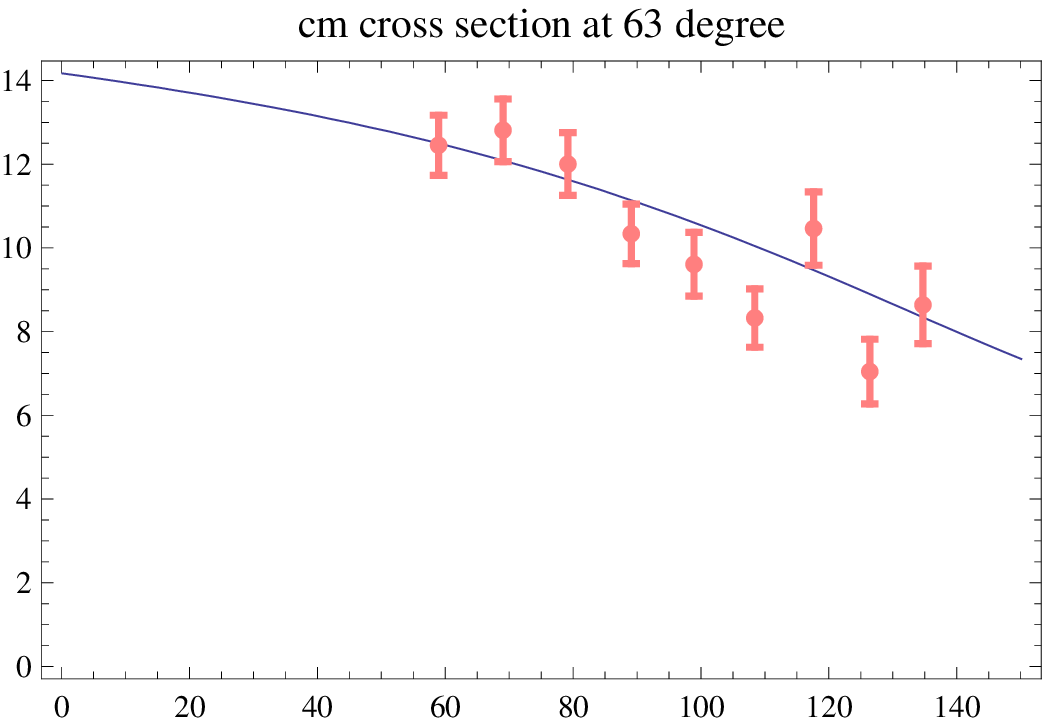}
\includegraphics[height=4cm,width=5cm]{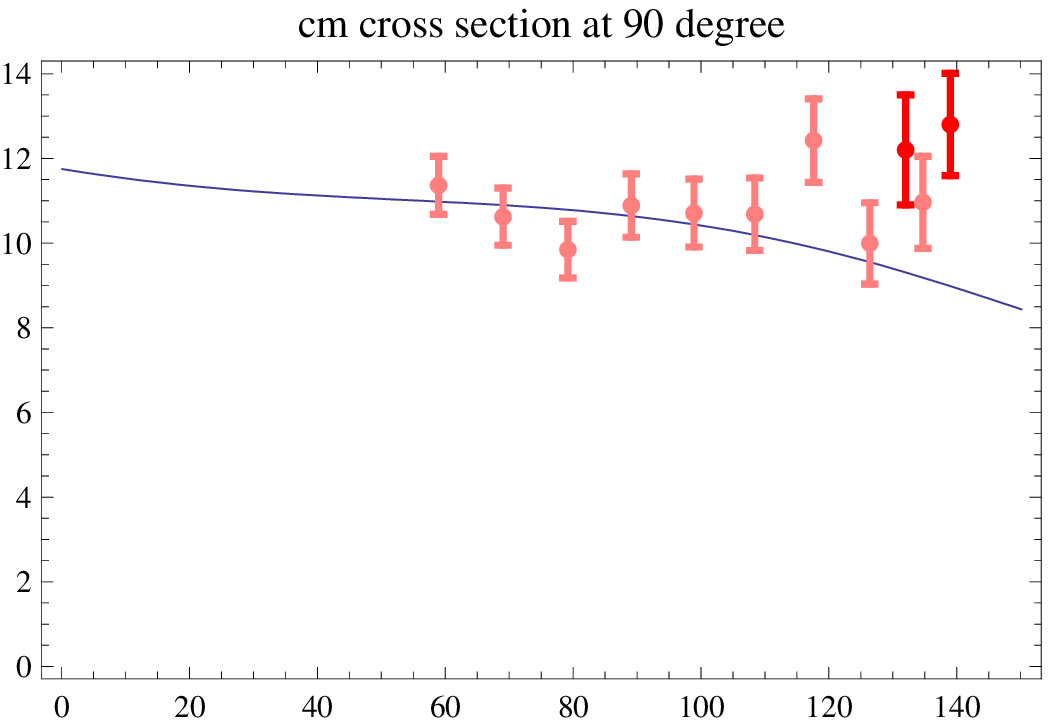}
\includegraphics[height=4cm,width=5cm]{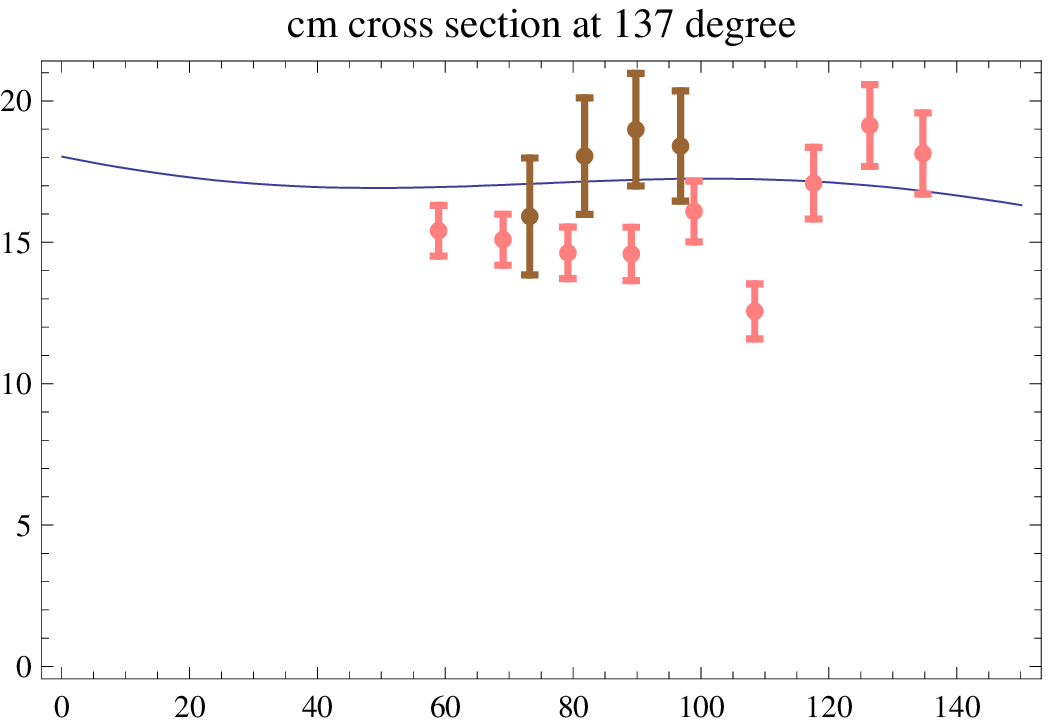}
\caption{Fixed c.m. angle cross section and the data points with photon incident energy below 140 MeV, for the fit parameters of the low energy data points. The x-axis is lab frame photon energy and y-axis is c.m. frame differential cross section in units of nanobarn. The angles of the data points included may differ from the nominal by at most 2.5 degree. 
For \cite{Hallin1993,Baranov1974,Zieger92,Hunger97,Blanpied01,MAMI2001,Olmos2001,MacGibbon1995}, we use colors: Red, Green, Blue, Black, Gray, Yellow, Pink and Brown respectively.}
\label{lowangleplot}
\end{figure}

\begin{figure}[htdp]
\setlength{\unitlength}{1cm}
\includegraphics[height=4cm,width=5cm]{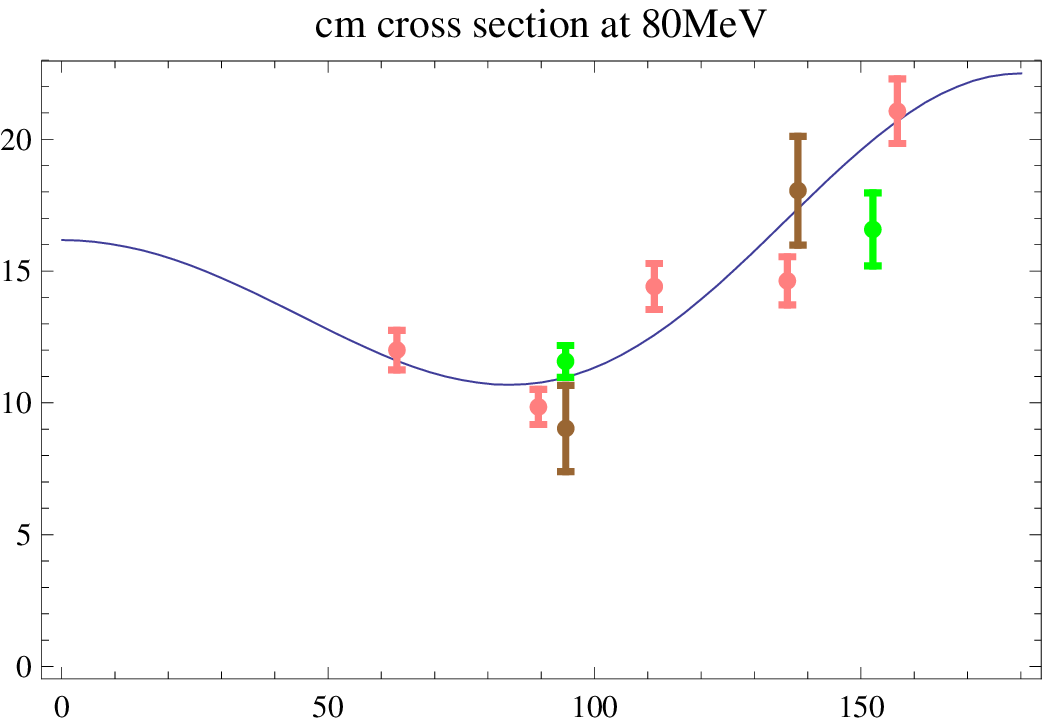}
\includegraphics[height=4cm,width=5cm]{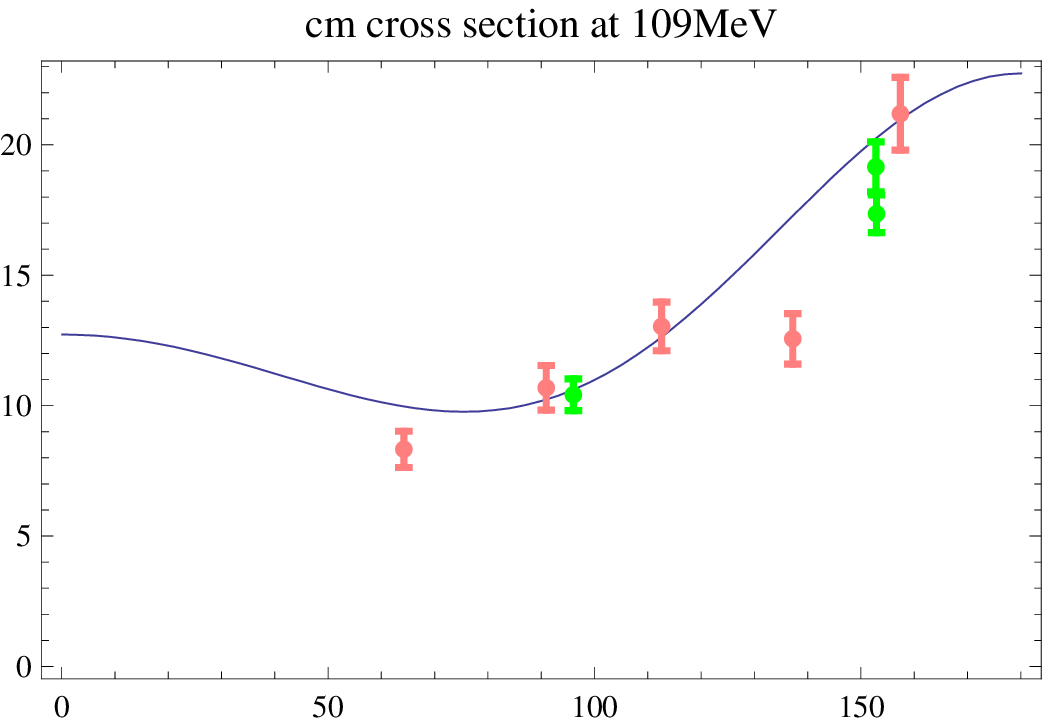}
\includegraphics[height=4cm,width=5cm]{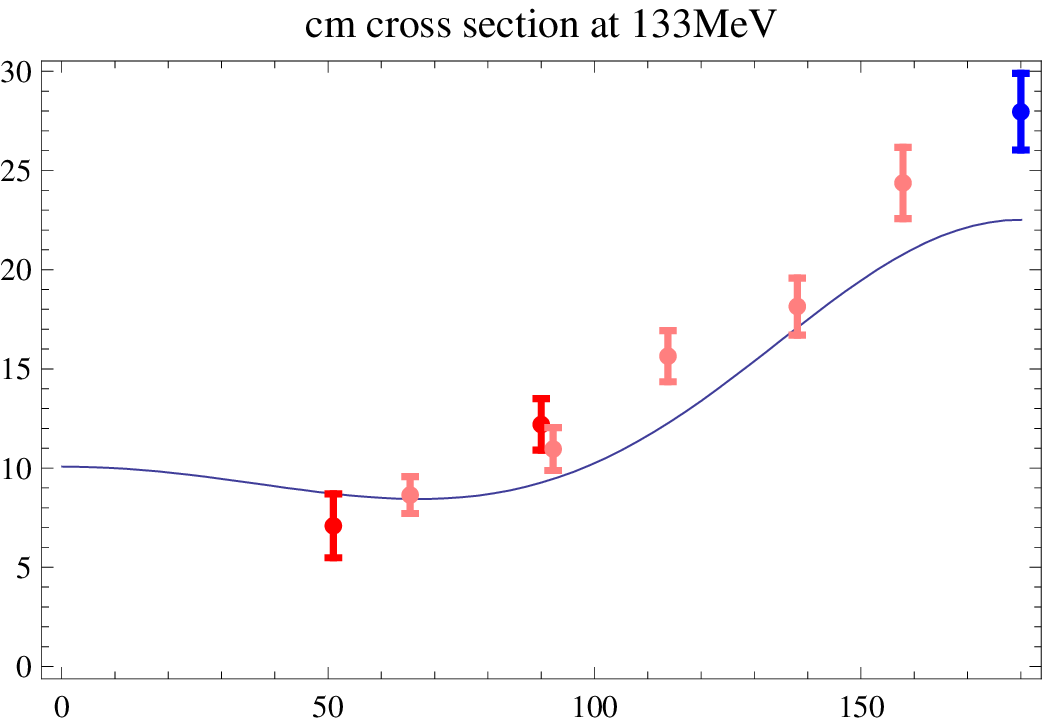}
\caption{Fit cross section and data points with photon incident energy below 140 MeV, for the fit parameters of the low energy data points. The x-axis is c.m. frame scattering angle and y-axis is c.m. frame differential cross section in units of nanobarn. The incident energy of the data points included may differ from the values claimed by at most 2.5 MeV. 
For \cite{Hallin1993,Baranov1974,Zieger92,Hunger97,Blanpied01,MAMI2001,Olmos2001,MacGibbon1995}, we use colors: Red, Green, Blue, Black, Gray, Yellow, Pink and Brown respectively.}
\label{lowenergyplot}
\end{figure}

\begin{table}[htdp]
\begin{center}
$\begin{array}{|c|c|c|c|c|c|c|}
\hline&F_1&G&\mu&F_6&\alpha_B &\beta_B\\ \hline
\mbox{whole fit}&-27.5&3.13&14.2&12.9&7.5&-8.2\\ \hline
1 \sigma &1.4&0.04&0.5&0.9&1.1&0.9\\ \hline
95\% &2.8&0.08&1.1&1.7&2.2&1.6\\ \hline
\hline
\mbox{MAMI fit} &-27.7&3.17&14.2&14.8&2.1&-8.1 \\ \hline
1 \sigma &0.7&0.03&0.4&0.8&0.8&0.5\\ \hline
95\%&1.4&0.05&0.7&1.5&1.5&1.1 \\ \hline
\hline
\end{array}$
\end{center}
\begin{center}
$\begin{array}{|c|c|c|c|}
\hline
&\mbox{value}&\mbox{95\% error from MAMI fit}&\mbox{95\% error from low energy fit}\\ \hline
\alpha_E+\beta_M &11.3&0.9&2.3\\ \hline
\alpha_E-\beta_M & 7.8 &3.3&2.0\\ \hline
\end{array}$
\end{center}
\caption{Our best-fit parameter values and confidence regions}
\end{table}

\section{Discussion and Conclusion}\label{conclu}

 
 Our model incorporates both minimal and non-minimal couplings. The former takes into account that the $\Delta^+$ is charged. The presence of the non-minimal couplings makes the description necessarily complicated, nevertheless the number of free parameters has been kept low. The parameters have clear physical interpretation, namely as the proton and $\Delta^+$ magnetic moments, as well as the strength of the $\gamma N\Delta^+$ magnetic transition (M1) and the two bare polarizabilities. We found that the data near the peak is well fit with our preferred set of parameters in the resonance region, but that the same set of parameters does not well describe the data at low energy. We have dealt with this in manner similar to \cite{McGovern:2012ew}.
 
Although we have been able to extract a value for the $\Delta^+$ magnetic moment, we cannot have high confidence in this value since proton Compton scattering does not probe the $\gamma\Delta^+\Delta^+$ vertex directly. The reason that this is possible at all is that the form factors have definite properties under Lorentz transformations, so that some linear combination of parameters which affects the Compton process also determines the magnetic moment of the $\Delta^+$. Qualitatively, we found that changes in the $\mu_{\Delta^+}$ affect the predicted cross-secion asymmetrically: lower values of the magnetic moment do not drastically change the prediction, but higher values greatly enhance the cross section, both on and off the resonance. Therefore, our result, conservatively stated is that  we exclude values of  the $\Delta^+$ magnetic moment $\mu_{\Delta^+}$ larger than about $14.2$. This should be compared to that extracted in \cite{Kotulla:2002cg} at $2.7^{+1.0}_{-1.3} (stat) \pm1.5 (syst) \pm 3(theor)$, but note that the quoted error is dominated by theoretical model uncertainty.
Our upper bound is also consistent with naive quark model expectations and some model calculations \cite{Linde:1995gr, Kotulla:2003pm, Kotulla:2008zz, Kotulla:2002cg, Kotulla:2002tx}. The more reliable path towards determining the $\Delta^+$ magnetic moment would be to extend the model to include pion form-factors and thus cover the case of $N+\gamma \to N+\gamma+ \pi$ scattering which has also been very well measured by some of the very same experimental groups as the Compton data considered in the thesis.

\chapter{Backgrounds in Gauge Field Scattering Amplitudes}

There are four known fundamental interactions, i.e. strong, weak, electromagnetic and gravity interactions. Standard Model (SM) of particle physics, formulated in 1967, has been a great achievement in particle physics. Three out of the four known fundamental interactions except gravity are unified in SM. Before symmetry breaking, SM has a gauge symmetry group $SU(3)_c\times SU(2)_{EW}\times U(1)_Y$. The massless gauge fields living in the adjoint representation of the gauge group are described by Yang-Mills (YM) theory. The gauge fields for $SU(3)_c$ strong interaction are gluons. The electroweak $SU(2)_{EW}\times U(1)_Y$ part of SM undergoes a symmetry breaking to electromagnetic $U(1)_{EM}$ symmetry in SM as the universe cools down. Higgs mechanism is invoked to break the electroweak gauge symmetry and gives masses to gauge fields and matter fields. The four massless YM gauge fields of $SU(2)_{EW}\times U(1)_Y$ recombine in the symmetry breaking, and three of them--$W^\pm$ and $Z$--become massive after eating three Goldstone bosons which are components of the Higgs field, with the remaining combination being the massless photon. SM has been very successful after its invention. Numerous predictions of SM, like the existence of $W^\pm$ and $Z$ together with their masses and decay channels, were verified precisely in the experiments.

Despite of the successes of SM, there are also many reasons for physicists to research new physics beyond SM. First of all, the key particle of SM, i.e. the Higgs boson, has not been confirmed to exist. ATLAS and CMS groups at Large Hadronic Collider (LHC) in the recent two years have reported their discovery of a Higgs-like particle \cite{CMS:2012nga,Chatrchyan:2012jja,Bhat:2012zj}, but people are far from being sure whether it is the SM Higgs. A second defect of SM is its inability to accommodate gravity, which has been understood little at quantum level. It is believed that eventually all the four known interactions will be unified in a more fundamental theory. A little surprisingly, the most stringent challenges on the microscopic scale physics come from the cosmological scale observations. In the universe there is about six times more dark matter than our known matter. Little is known about the individual dark matter particles, but certainly they are not SM particles. Another mystery of our universe is the baryon asymmetry, i.e. our universe contains more baryons than anti-baryons. To explain this asymmetry, baryon number violating processes are required, which do not exist in SM. There are many other reasons for the need of beyond SM physics: neutrinos in SM are massless while now they are known to be massive through neutrino oscillations; there is strong CP problem in SM--no experimental evidence for CP violation in the strong interaction sector and no reason for CP to be conserved theoretically in this sector of SM--which requires a fine tuning of SM parameters; there is hierarchy problem which concerns the huge discrepancy between the strength of weak force and gravity; there is no reliable theory for the early universe; etc.. Theoretically, several attracting new physics theories have emerged. String theory carries the hope to unify gravity with the other forces, and even to become the final theory of physics. Supersymmetry offers a solution to hierarchy problem, and provides a very good candidate for dark matter, i.e. the lightest supersymmetric particle. Grand Unification Theory allows baryon number violating processes, giving a solution to the cosmological baryon asymmetry mystery. A common problem for them is that they lack experimental supports. Experimentally, LHC--considered as one of the great engineering milestones of mankind--has been built with major purposes to find the Higgs and probe new physics.


On the hadronic colliders, physicists collect huge amount of data of cross sections for various scattering events, which measure the probabilities of these events. Then they compare the data to theoretical predictions. The impacts of new physics at the energy scales of current colliders are promised to be very weak. That is to say, the measured cross sections are very close to the SM predictions. Physicists require 5 sigma deviation of data from the SM background to claim the discovery of new physics. Thus, in order to discover new physics at the colliders, we need on one hand good quality data and on the other hand precise mastery of SM background. SM background can be calculated from the scattering amplitudes, which are determined by the SM Lagrangians. Since gluon scattering constitutes a large fraction of the processes in hadronic colliders, gluon scattering amplitudes--belonging to the scope of YM amplitudes--play important roles in discovering physics beyond SM.



A traditional way to calculate amplitudes is using perturbation theory, drawing Feynman diagrams at each perturbation level and then calculating with Feynman rules. However, this way is very hard for computing YM amplitudes. On one hand, the number of Feynman diagrams increases exponentially with the number of external legs. The numbers of tree level diagrams for four, five, six and seven point YM amplitudes are 4, 25, 220 and 2485 respectively. For seven point scattering at one loop level, the number is 227585. On the other hand, the Feynman rules for YM fields are relatively complicated. The three point and four point YM vertices both have 6 terms. Actually Feynman rules are complicated and inefficient due to the fact that they contain redundant gauge freedom.These two features make the calculation of YM amplitudes very difficult. However, for the collider physics people need precise scattering amplitudes of many gluons and beyond tree level. This is a great task for theoretical physicists.



Fortunately, much progress has been made in the calculation of YM amplitudes. In 1986, there was a breakthrough: Parke and Taylor \cite{Parke:1986gb} conjectured a concise expression for general N-point tree-level Maximal-Helicity-Violating (MHV) gluon amplitudes, in which only two of the gluons are of negative helicity and all others of positive helicity. The formula is proved a few years later \cite{Berends:1987me}. Tree-level on shell YM amplitudes with zero or one negative (positive) helicity vanish, thus MHV amplitudes are the simplest non trivial amplitudes.

It can be dated to the 1980s and 1990s that people use color decomposition \cite{Berends:1987me,Dixon:1996wi,Mangano:1990by} and spinor technique \cite{Dixon:1996wi,Gunion:1985vca,Xu:1986xb,Berends:1981rb} to simplify the calculation of YM amplitudes. The full YM amplitudes contain two aspects of information--gauge group part and dynamic part--which can be separated by color decomposition. At the same time all external lines in the partial amplitudes after color decomposition are ordered, thus the number of Feynman diagrams is greatly reduced compared to full amplitudes. Almost all the later developments of YM amplitudes just deal with the color decomposed amplitudes. Spinor technique is especially useful for massless gauge boson amplitudes. A four-dimensional null vector can be written as a direct product of a two-component spinor and a two-component anti-spinor. Polarization vectors in spinor formalism carry an arbitrary reference spinor or anti-spinor which does not affect the final result due to gauge invariance. The facts that spinor form is very suitable for expressing massless conditions of the external legs and that it can explicitly exploit gauge invariance by fixing the reference (anti-) spinors in the polarization vectors are responsible for the existence of very compact expressions of amplitudes in spinor form.


The past ten years have witnessed many important developments in YM amplitudes. MHV amplitudes in spinor form are holomorphic, i.e. depending only on the spinors but not the anti-spinors corresponding to the external momenta. In 2003, Edward Witten studied this property by Fourier transforming the amplitudes into twistor space \cite{Witten1}. He found that gluon amplitudes are associated with curves in twistor space--especially MHV amplitudes are associated with straight lines. With the insights thus gained, in 2004, Freddy Cachazo, Peter Svrcek and Edward Witten proposed a rule \cite{Cachazo:2004kj} to construct on shell YM amplitudes by taking MHV amplitudes as vertices, known as CSW rule. In this rule, the inputs, i.e. MHV amplitudes, are gauge invariant on-shell amplitudes with compact expressions, and the number of diagrams is less than that of Feynman diagrams. Thus CSW rule has advantages over direct Feynman rule calculations, and over the Berends-Giele recursion relation \cite{Berends:1987me} in the 1980s which directly stems from Feynman rules and deals with off shell amplitudes in the recursion process.


In 2004, R. Britto, F. Cachazo and B. Feng discovered a powerful recursion relation \cite{Britto:2004ap} for tree-level on-shell gluon amplitudes, and in the next year E. Witten gave a nice explanation \cite{Britto:2005fq} for the recursion relation by complexifying the amplitudes, making it known as Britto-Cachazo-Feng-Witten (BCFW) recursion relation, which immediately triggered a hot in the research of YM amplitudes lasting to the current days and has been generalized to many other theories. Due to causality and unitarity, amplitudes are meromorphic functions of momenta viewed as complex variables. In BCFW recursion relation, a pair of momenta is deformed and the amplitudes become meromorphic functions of a single parameter describing the momenta deformation. The deformation can be chosen such that the amplitudes, as meromorphic functions of the deformation parameter, in general only have single poles due to the propagators depending on the deformation parameter. Then the meromorphic amplitudes are determined by these poles together with their residues. The residue at each pole equals to the product of two sub amplitudes, formed by cutting the  corresponding propagator in the whole amplitude. Since the propagator momentum becomes null at pole position, the two sub amplitudes are still on shell amplitudes. Similarly as CSW rule, BCFW recursion relation, as an on-shell recursion relation, is more efficient than Feynman rule calculation or Berends-Giele recursion relation. Compared with CSW rule, BCFW recursion relation is easier to operate and can be applied to many other theories.

Beyond the tree level, there also exist many progresses in one-loop on-shell YM amplitudes. The most famous method is the "unitarity method" \cite{Anastasiou:2006jv,Bern:1994zx,Bern:1994cg,Britto:2005ha,Britto:2004nj,Britto:2004nc,Cachazo:2004dr}. Loop amplitudes can be expanded in a basis of master integrals, with coefficients being rational functions of the momenta. For one-loop amplitudes, there are four types of master integrals, i.e. box, triangle, bubble and tadpole. The last one vanishes for massless field scattering amplitudes like YM amplitudes. By evaluating the master integrals explicitly, the most difficult part of loop level calculations, i.e. integration over loop momentum, can be done once and for all. Viewing the momenta as complex variables, one-loop YM amplitudes have branch cut type singularities, compared to poles for tree-level amplitudes. One loop YM amplitudes are cut-constructible in dimensional regularization, which means that the amplitudes are uniquely determined by their branch cuts and the discontinuity at the cuts, compared to that tree amplitudes can be constructed from their poles and corresponding residues. The discontinuity at the cuts are related to the imaginary parts of the one-loop amplitudes, which by optical theorem, equal to the product of two tree-level amplitudes. Quite similar to the tree-level case that a whole amplitude breaks into two parts by cutting one propagator, for an one-loop amplitude, two propagators on the loop are put on shell and the one-loop amplitude is cut into two on-shell tree-level amplitudes. Different pairs of propagators are cut, depending on different branch cuts, and each master integral has a distinct branch cut.  Finding out the master integrals is one obstacle for this unitarity method to be extended to higher loop levels. Another problem for this method to work at higher loop levels is that one encounters lower level loop amplitudes with external momenta in general D-dimensions--when these external legs come from cutting the loop propagators at higher loop levels. Not many analytic results have been obtained by this method beyond one loop level without maximal supersymmetry.

In the very recent years, there has emerged a new elegant approach to amplitudes \cite{ArkaniHamed:2012nw, ArkaniHamed:2009dg, ArkaniHamed:2009sx, ArkaniHamed:2009vw}, by relating them to a remarkable mathematical structure--positive Grassmannian, initiated by Nima Arkani-Hamed et. al.. Currently most of the efforts in this approach have been devoted to planar $\mathcal N=4$ Super-Yang-Mills (SYM) amplitudes, leaving much to be explored for amplitudes with less supersymmetries or non planar amplitudes. The Grassmannian formalism can manifest all the symmetries, as an improvement over previous approaches to amplitudes. Inherited from previous developments, like the BCFW recursion relation or the unitarity method for loop amplitudes, in this new approach, amplitudes are determined from their singularities and the authors construct the amplitudes with gauge invariant objects--here on shell three point amplitudes. Amplitudes to all loop levels can be constructed with these building blocks, without any off shell line in the diagrams.

At least compared to a decade ago, research in amplitudes is very active and promising these years, leaving many important questions to think about. In the following of this chapter, I will present some progresses already made with a little more details, and end with the motivations of our work.





\section{Color Decomposition}

The Feynman rules for YM three point and four point vertices are:
\begin{eqnarray}
iV^{abc}_{\mu\nu\lambda}=&&i f^{abc}[(k_1-k_2)_\lambda g_{\mu\nu}+(k_2-k_3)_\mu g_{\nu\lambda}+(k_3-k_1)_\nu g_{\mu\lambda}],\nonumber\\
iV^{abcd}_{\mu\nu\lambda\rho}=&& i [f^{abe}f^{cde}(g_{\mu\lambda}g_{\nu\rho}-g_{\nu\lambda}g_{\mu\rho})\nonumber\\
&&+f^{ace}f^{bde}(g_{\mu\nu}g_{\lambda\rho}-g_{\nu\lambda}g_{\mu\rho})\nonumber\\
&&+f^{ade}f^{cbe}(g_{\mu\lambda}g_{\nu\rho}-g_{\rho\lambda}g_{\mu\nu})].\label{YMvertex}
\end{eqnarray}
In (\ref{YMvertex}), $f^{abc}$ is structure constant of the gauge group. YM coupling constant g is omitted throughout the thesis. The convention for momenta here is that positive momenta are out-going. The Feynman rules contain informations of two aspects: one is the dynamic part and the other is gauge group part i.e. structure constants. We can separate the two aspects by color decomposition \cite{Dixon:1996wi,Mangano:1990by} of the amplitudes. Tree level color decomposition is:


\beq
\mathcal A_{\mbox{tot}}(\{p_i,\epsilon_i,a_i\})=\sum_{\sigma\in S_n/Z_n}Tr(T^{a_{\sigma(1)}}T^{a_{\sigma(2)}}\cdots T^{a_{\sigma(n)}})A_n^{\mbox{\tiny tree}}(\sigma(1),\sigma(2),\cdots,\sigma(n)).
\label{treecolordecom}
\eeq
The sum is over all non-cyclic permutations. $\mathcal A$ represents the total amplitude, and A represents partial amplitudes, or color ordered amplitudes. All dynamic informations are contained in the color ordered amplitudes $A_n^{\mbox{\tiny tree}}$. In the decomposition at tree level (\ref{treecolordecom}), all the $n$ group generators are in a single trace. For one loop level color decompositions, there are double trace structures besides the single trace term. Yet there exist general formulas \cite{Bern:1994zx,Bern:1990ux,Bern:1994fz} that relate the partial amplitudes corresponding to double trace terms to those corresponding to the single trace terms. Thus we only need to calculate the single trace terms, also called primitive amplitudes in the literature, for tree and one loop levels. For the rest of the thesis, unless otherwise stated, when we mention YM amplitudes, tree or one loop level, we actually mean the primitive amplitudes.



There are several advantages for dealing with the primitive amplitudes instead of full amplitudes. First, for different gauge groups, the primitive amplitudes are the same, and only the trace parts in (\ref{treecolordecom}) are different. Secondly, all external lines in primitive amplitudes are ordered by color. Thus the number of Feynman diagrams is greatly reduced compared to full amplitudes. The color ordered 5 and 6 point amplitudes for example have 10 and 38 diagrams respectively, compared to 25 and 220 for the full amplitudes. Thirdly, we can use relatively simpler color ordered Feynman rules \cite{Dixon:1996wi} to calculate primitive amplitudes:
\bea
iV_{\mu\nu\lambda}=&&\frac{i}{\sqrt{2}} [(k_1-k_2)_\lambda g_{\mu\nu}+(k_2-k_3)_\mu g_{\nu\lambda}+(k_3-k_1)_\nu g_{\mu\lambda}],\nonumber\\
iV_{\mu\nu\lambda\rho}=&& \frac{i}{2}(2 g_{\mu\lambda}g_{\nu\rho}-g_{\nu\lambda}g_{\mu\rho}-g_{\mu\nu}g_{\lambda\rho}).\label{colorFeynrule}
\eea

\section{Spinor Technique}
For YM amplitudes, spinor formalism \cite{Dixon:1996wi,Gunion:1985vca,Xu:1986xb,Berends:1981rb} is widely used and offers concise expressions in many occasions. In this section we give a short introduction to spinor formalism for massless particles, while for massive particles the spinor formalism is relatively more complicated \cite{Dittmaier:1998nn,Ozeren:2006ft,Rodrigo:2005eu,Schwinn:2005pi,Schwinn:2007ee}. For a null vector $k_\mu$ we can define a spinor $\lambda$ and anti-spinor $\tilde \lambda$ through Dirac equation:
\beq
k_{a {\dot a}}\lambda^a=0,\tilde \lambda^{{\dot a}} k_{a {\dot a}}=0,
\eeq
where
\beq
k_{a{\dot a}}\equiv k_\mu (\sigma^\mu)_{a{\dot a}},\ \sigma^\mu=(1,\vec{\sigma}),\ a,\dot a=1,2
\eeq
with $\vec \sigma\ $ being Pauli matrices. Then the momentum $k$ can be decomposed as:
\beq
k_{a \dot a}=\lmd_a \td \lmd_{\dot a},
\label{spinormassless}
\eeq
and we will abbreviate it as $k=\lmd \td\lmd$. Antisymmetric tensors $\epsilon^{ab}$ and $\epsilon_{ab}$ are used to raise or lower the spinor indices:
\beq
\lambda^a=\epsilon^{ab}\lambda_b,\lambda_a=\epsilon_{ab}\lambda^b,
\eeq
with $\epsilon^{12}=1$ and $\epsilon_{12}=-1$. Similarly for the anti-spinors with dotted indices.

We can define Lorentz invariant inner products for two spinors or anti-spinors as:
\beq
\langle \lambda_i\lambda_j\rangle\equiv\lambda_i^a\lambda_{ja},[\tilde \lambda_i\tilde \lambda_j]\equiv \tilde \lambda_{i \dot a}\tilde \lambda_j^{\dot a}.
\eeq
For two vectors $k_1=\lambda_1 \tilde \lambda_1$ and $k_2=\lambda_2 \tilde \lambda_2$, their inner product in spinor form is:
\beq
k_1\cdot k_2=\frac{1}{2}\langle \lambda_1\lambda_2\rangle [\tilde \lambda_1 \tilde \lambda_2].
\eeq
The inner products are antisymmetric under the interchange of two (anti-)spinors. Another very useful property is Schouten identity:
\beq
\lambda_i\langle \lambda_j \lambda_k \rangle+\lambda_j\langle \lambda_k \lambda_i \rangle+\lambda_k\langle \lambda_i \lambda_j \rangle=0,
\eeq
and similarly for anti-spinors. In the following, the word spinor may also mean anti-spinor when it is not  confusing. For massless on shell YM external leg with momentum $k=\lambda_k\tilde\lambda_k$, the two helicity polarization vectors can be written as:
\beq
\epsilon^+=\frac{\lambda_\mu \tilde \lambda_k}{\sqrt{2}\langle \lambda_\mu\lambda_k\rangle},\epsilon^-=\frac{\lambda_k \tilde \lambda_\mu}{\sqrt{2}[ \tilde \lambda_\mu \tilde \lambda_k]},
\label{polarvec}
\eeq
where $\lambda_\mu$ or $\tilde \lambda_\mu\ $ is an arbitrary reference spinor that is not proportional to $\lambda_k$ or $\tilde \lambda_k$. Different choices of the reference spinors correspond to gauge transformations of YM fields, and the dependences on the reference spinors disappear in the final results of on shell amplitudes.

In short, spinor form can conveniently utilize massless condition as in (\ref{spinormassless}), and gauge redundancy can be explicitly removed by fixing the reference (anti-) spinors in (\ref{polarvec}). These two points can explain why expressions are compact in spinor form.

\section{MHV Amplitudes and CSW Rule}
Although in general YM amplitudes with many legs are very difficult to calculate analytically, there is a class of tree level amplitudes, called MHV amplitudes, that were known at an early time for any number of YM fields scattering and have concise expressions. MHV means that only two of the YM fields are of negative helicity and all others of positive helicity. The formula was first conjectured by Parke and Taylor \cite{Parke:1986gb},\footnote{In this paper the amplitudes were neither color decomposed or in spinor formalism.} and proved a few years later \cite{Berends:1987me}. Denoting the two YM fields with negative helicity as x and y, in the spinor formalism with $k_i=\lambda_i \tilde \lambda_i$, the partial amplitude is:
\beq
A(1^+,2^+,\cdots,x^-,\cdots,y^-,\cdots,n^+)=\frac{\langle \lambda_x \lambda_y\rangle^4}{\Pi_{i=1}^n \langle \lambda_i\lambda_{i+1}\rangle}.
\label{MHV}
\eeq

Freddy Cachazo, Peter Svrcek and Edward Witten proposed a rule \cite{Cachazo:2004kj} to construct on shell YM amplitudes from MHV amplitudes, known as CSW rule. In this rule, the vertices are not the familiar three point and four point Feynman rule vertices, but the MHV vertices with any number of lines attached, of which two are negatively polarized. The propagators are just scalar propagators, and the two ends have opposite helicities. For propagators whose momenta are in general off shell, they use a reference anti spinor to define the spinors for the propagator momenta $q$: $\lambda_{qa}=q_{a\dot a}\eta^{\dot a}$ where the dependence on the arbitrary reference anti spinor $\eta$ cancels out in the final on shell amplitudes. In this rule, the MHV vertices take the on shell form (\ref{MHV}) although off shell propagators are attached to them.


From the facts that each MHV vertex contains two negative helicity YM fields, that each propagator provides one negative helicity YM field, and that the number of propagators is less than that of vertices by 1, it is deduced that for amplitudes with $r$ negative helicity YM fields, the number of MHV vertices is $r-1$. For amplitudes with 0 or 1 negative helicity YM field, there can be no MHV vertices, thus the amplitudes are 0. When there are 2 negative helicity YM fields there is 1 MHV vertex giving the MHV amplitude. These are trivial checks of CSW rule. The first non trivial check would be four point amplitude with 3 negative helicity YM fields. The amplitude is 0, and we can check it by CSW rule. There are two diagrams for this amplitude with two MHV vertices in each diagram, as shown in Figure \ref{CSW4p}.
\begin{figure}[h]
\centerline{\includegraphics[width=14cm,height=6cm]{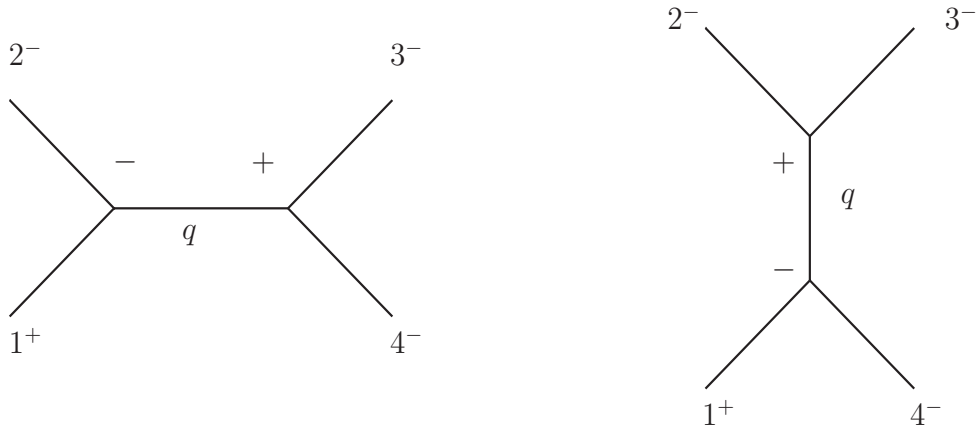}}
\caption{Diagrams to calculate 4 point YM amplitude with 3 external legs of negative helicity in CSW rule.}
\label{CSW4p}
\end{figure}

The momenta of the four external legs in spinor form are $k_i=\lambda_i \tilde \lambda_i$. Using an arbitrary reference anti spinor $\eta^{\dot a}$, which are chosen to be the same for both diagrams, a spinor $\lambda_{qa}=q_{a\dot a}\eta^{\dot a}$ is defined for each propagator with momentum $q$ in the two diagrams. Notate $\tilde \lambda_{i\dot a}\eta^{\dot a}$ as $\phi_i$, the first diagram evaluates as:
\beq
-\frac{\phi_1^3}{\phi_2\phi_3\phi_4}\frac{\langle \lambda_3\lambda_4\rangle}{[\tilde\lambda_2\tilde\lambda_1]},
\eeq
and the second diagram evaluates as:
\beq
-\frac{\phi_1^3}{\phi_2\phi_3\phi_4}\frac{\langle \lambda_3\lambda_2\rangle}{[\tilde\lambda_4\tilde\lambda_1]}.
\eeq
By momentum conservation one has $0=\sum_{i=1}^4 \langle \lambda_3\lambda_i\rangle [\tilde\lambda_i \tilde\lambda_1]=\langle \lambda_3\lambda_2\rangle [\tilde\lambda_2 \tilde\lambda_1]+\langle \lambda_3\lambda_4\rangle [\tilde\lambda_4 \tilde\lambda_1]$. Also applying the antisymmetry property of (anti) spinor inner product, the two diagrams sum up to be 0.

In general, CSW rule is proved in \cite{Risager:2005vk,Britto:2004ap}. Besides its application in calculating tree level YM amplitudes, CSW rule also inspired many developments for loop level amplitudes \cite{Anastasiou:2006jv,Anastasiou:2006gt,0407214,0510253,Britto:2005ha,Britto:2004nj,Britto:2004nc,Bullimore:2010pj,Cachazo:2004dr}.


\section{Berends-Giele Recursion Relation}
An efficient way to calculate YM amplitudes with many legs is through recursion relation. Berends-Giele recursion relation \cite{Berends:1987me} was proposed early and used widely. This recursion relation directly stems from Feynman rule calculation of amplitudes. For a color ordered amplitude with $n$ legs, we can draw all the diagrams with n+1 legs by inserting the (n+1)-th leg into the diagrams with $n$ legs, while maintaining color ordering and involving only three point or four point vertex of YM field. This recursion relation works regardless of whether the legs are on shell or not, thus having wide applications. Assume the (n+1)-th leg has Lorentz index $\mu$, the recursion relation for the current $J^\mu (1,2,\cdots,n)$ is:
\bea
J^\mu (1,2,\cdots,n)=&&\frac{-i}{k_{1,n}^2}\left[\sum_{i=1}^{n-1}V^{\mu\nu\rho}(k_{1,n},k_{1,i},k_{i+1,n}) J_\nu(1,\cdots,i)J_\rho(i+1,\cdots,n)\right.\\
&&\left.+\sum_{j=i+1}^{n-1}\sum_{i=1}^{n-2}V^{\mu\nu\rho\sigma}J_\nu(1,\cdots,i)J_\rho(i+1,\cdots,j) J_\sigma(j+1,\cdots,n)\right],\nb
\eea
where $k_{i,j}=k_i+k_{i+1}+\cdots+k_j$. $V_{\mu\nu\rho}$ and $V_{\mu\nu\rho\sigma}$ are the color ordered three and four point vertices (\ref{colorFeynrule}). The (n+1)-th leg is not contracted with any polarization vector, and the propagator corresponding to the momentum of the (n+1)-th leg is contained in the current $J^\mu$. Berends-Giele recursion relation can be represented by Figure \ref{BGrecursion}.
\begin{figure}[h]
\centerline{\includegraphics[width=14cm,height=5cm]{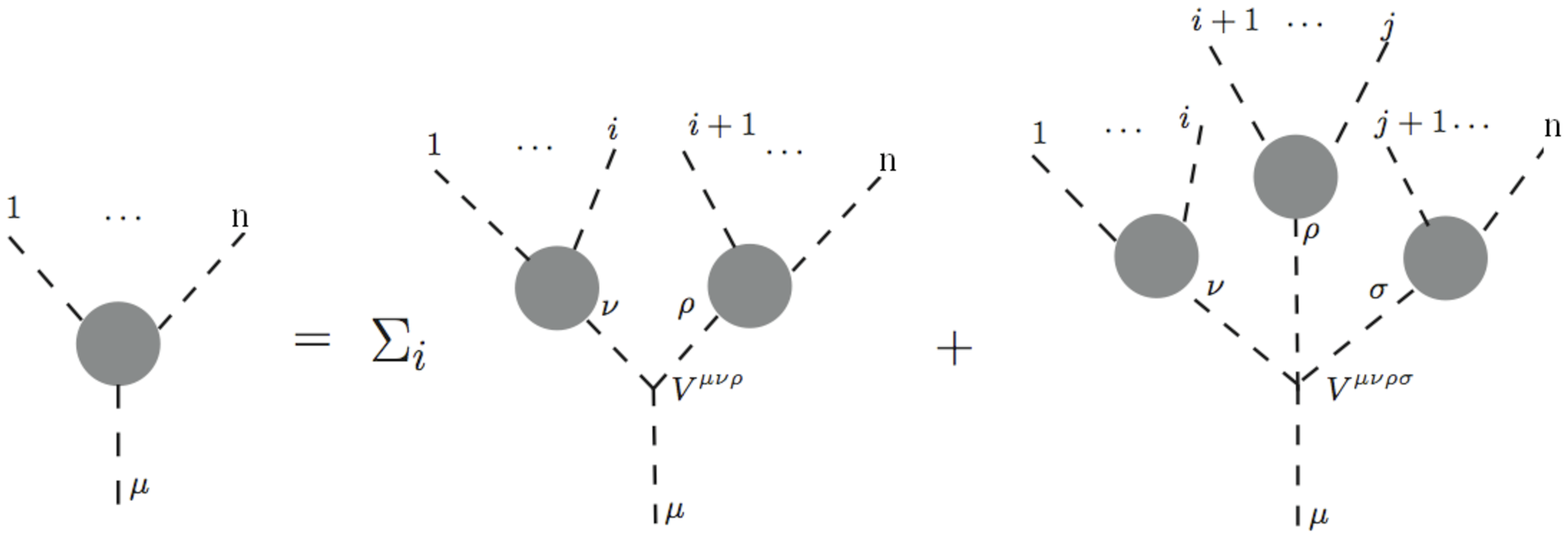}}
\caption{The pictorial representation of Berends-Giele recursion relation. The picture is taken from \cite{Feng:2011np}.}
\label{BGrecursion}
\end{figure}

\section{BCFW Recursion Relation}

Due to causality and unitarity, amplitudes are meromorphic functions of momenta viewed as complex variables. S-matrix program \cite{Olive1964,Chew1966,Eden1966} was proposed to study scattering amplitudes from some general principles like Lorentz invariance, locality, causality, gauge invariance and unitarity etc., especailly in strong interaction theory. This method depends little on the details of the theories. Combined with on shell recursion relations many important conclusions are derived, for example the Britto-Cachazo-Feng-Witten (BCFW) recursion relation \cite{Britto:2004ap,Britto:2005fq}.

BCFW recursion relation is an important progress of recent years in scattering amplitudes. BCFW recursion relation is first developed for tree level on shell YM amplitudes. It results in a hot in the research of amplitudes, and is extensively applied and generalized to many theories \cite{Badger:2005zh,Badger:2005jv,Benincasa:2007qj,Berger:2006sh,Bern:2005hs,Brandhuber:2008pf,Brandhuber:2007up,Cheung:2009dc,Gang:2010gy,Luo:2005rx,Luo:2005my,Park:2006va,Quigley:2005cu,ArkaniHamed:2008gz,Bern:2005hh,Bern:2005ji,Bern:2005cq,Bedford:2005yy,Cachazo:2005ca}. Since amplitudes can be understood as meromorphic functions of momenta, one can deform all the momenta in the complex plane. However, general deformations violate momenta conservation and on shell condition $k^2=0$, and the functions are difficult to deal with due to their multiple variables. Among all deformations, the well known BCFW deformation can leave momenta conservation and on shell condition intact, and have only one variable.

In BCFW deformation, one chooses a pair of momenta $k_i$ and $k_j$ to shift as:
\beq
k_i(z)=k_i+z\eta,\ k_j(z)=k_j-z\eta.
\eeq
This ensures momenta conservation. In order to keep the external legs on shell, conditions are placed on $\eta$:
\beq
\label{qcondition}
\eta^2=\eta\cdot k_i=\eta\cdot k_j=0.
\eeq
In four and higher dimensional spacetime, one can solve $\eta$ from (\ref{qcondition}). For example, for $k_i=\lambda_i \tilde \lambda_i$, $k_j=\lambda_j \tilde \lambda_j$, one can choose $\eta=\lambda_i \tilde \lambda_j$ or $\eta=\lambda_j \tilde \lambda_i$.

After shifting the pair of momenta, the original partial amplitude $A(k_i,k_j,\cdots)$ becomes a meromorphic function $A(z)$ of a single variable z. Many useful mathematic tools of meromorphic functions can then apply to $A(z)$. One relevant conclusion is: a meromorphic function with only pole singularities can be determined by its pole positions and the residues correspondingly.

Propagators result in poles which are in general single poles. By appropriate choices of reference spinors, the polarization vectors (\ref{polarvec}) provide no poles. Doing a contour integral, which encompasses all the poles $\{z_\alpha\}$ from propagators, of the function $\frac{A(z)}{z}$, one gets:
\beq
B=\oint dz\frac{A(z)}{z}\Rightarrow A(z=0)=B-\sum_{z_\alpha}Res(\frac{A(z)}{z})_{z_\alpha}.
\eeq
B is called boundary term. When $A(z)|_{z\to \infty} \to 0$, B is 0. In this case, we only need to calculate the residues at the poles--corresponding to propagators--to determine $A(z)$. Since at each pole, the momentum of the corresponding propagator becomes on shell, the residue at that pole equals the product of two less point on shell amplitudes, which are formed by cutting the corresponding propagator in the whole amplitude. Thus, when the boundary term B is 0, there is the on shell BCFW recursion relation:
\beq
A_n=A(z=0)=\sum_{z_\alpha,h=\pm}A_L(k_i(z_\alpha),k^h(z_\alpha))\frac{1}{k_\alpha^2}A_R(-k^{-h} (z_\alpha),k_j(z_\alpha)).
\label{BCFW}
\eeq
BCFW recursion relation with legs n-1 and $n$ shifted can be represented in Figure \ref{BCFWrep}.

\begin{figure}[h]
\centerline{\includegraphics[width=14cm,height=5cm]{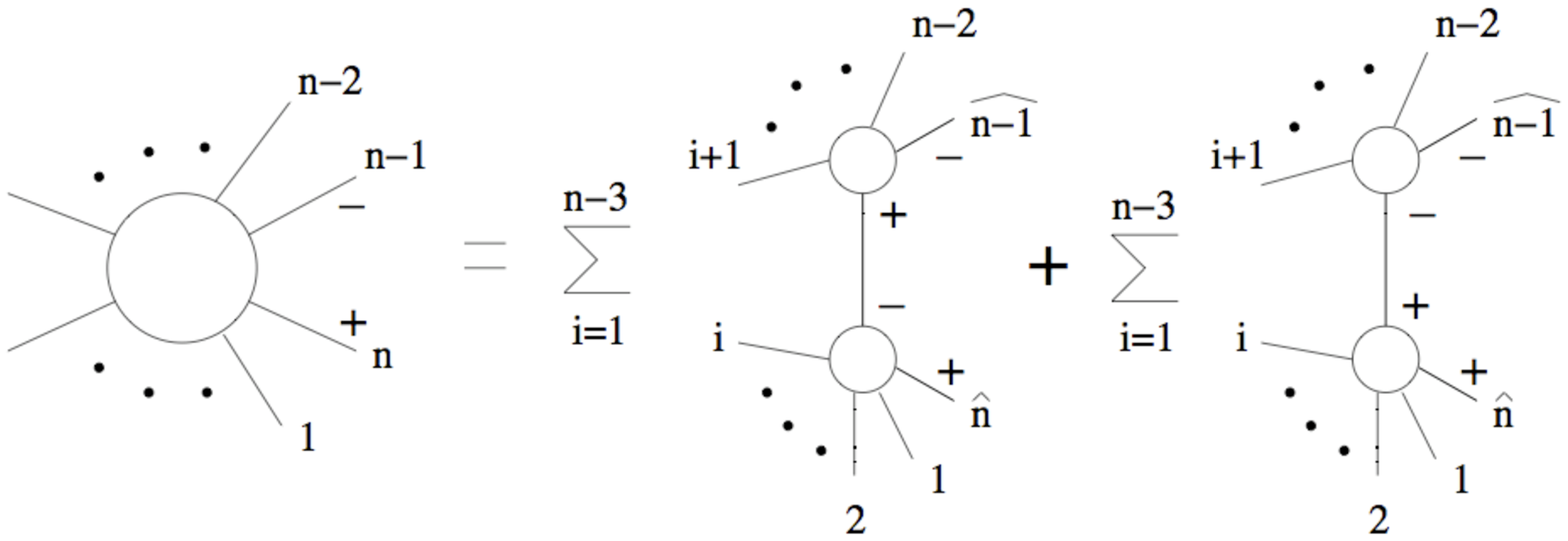}}
\caption{Pictorial representation of BCFW recursion relation with legs n-1 and $n$ shifted. The picture is taken from \cite{Feng:2011np}.}
\label{BCFWrep}
\end{figure}

For on shell YM scattering, it is always possible to choose a pair of momenta to shift such that the boundary term B vanishes \cite{ArkaniHamed:2008yf}, which results in the BCFW recursion relation (\ref{BCFW}). If one shifts a leg $i$ of positive helicity and a leg $j$ of negative helicity, the shift and polarization vectors can be chosen as:
\bea
k_i=\lambda_i \tilde \lambda_i \to \lambda_i \tilde \lambda_i+z\lambda_j \tilde \lambda_i&,&\ \ k_j=\lambda_j \tilde \lambda_j \to \lambda_j \tilde \lambda_j-z\lambda_j \tilde \lambda_i,\nonumber\\
\epsilon_i^+=\frac{\lambda_j \tilde \lambda_i}{\sqrt{2}\langle \lambda_j\lambda_i\rangle}&,&\ \ \epsilon_j^-=\frac{\lambda_j \tilde \lambda_i}{\sqrt{2}[\tilde \lambda_j\tilde\lambda_i]}.
\label{pnshift}
\eea
If one shifts legs $i$ and $j$ both of positive helicity (similarly for negative helicity), the shift and polarization vectors can be chosen as:
\bea
k_i=\lambda_i \tilde \lambda_i \to \lambda_i \tilde \lambda_i+z\lambda_j \tilde \lambda_i&,&\ \ k_j=\lambda_j \tilde \lambda_j \to \lambda_j \tilde \lambda_j-z\lambda_j \tilde \lambda_i,\nonumber\\
\epsilon_i^+=\frac{\lambda_j \tilde \lambda_i}{\sqrt{2}\langle \lambda_j\lambda_i\rangle}&,&\ \ \epsilon_j^+=\frac{\lambda_i (\tilde \lambda_j-z\tilde \lambda_i)}{\sqrt{2}\langle \lambda_i\lambda_j\rangle}.
\label{ppshift}
\eea
Actually, when the legs $i$ and $j$ are not adjacent in the color ordering, $A(z)$ under the shift (\ref{pnshift}) or (\ref{ppshift}) behaves as $\mathcal O(z^{-2})$ when $z$ goes to infinity, compared to the $\mathcal O(z^{-1})$ behavior for shifting legs $i$ and $j$. The behavior of $A(z)$ when $z$ goes to infinity is a non trivial property of the complexified amplitude $A(z)$, since individual Feynman diagram does not have such good scaling properties.

Due to constraints from gauge invariance, on shell amplitudes are in general much simpler than off shell ones. At the same time, for $n$ point amplitudes, there are only $n-3$ propagators which can have poles, compared to the tremendous amount of Feynman diagrams. Thus, BCFW recursion relation is very efficient in calculating on shell amplitudes analytically.


\subsection{BCFW Recursion Relation in QCD with Massive Dirac Fields}

Among the various generalizations of BCFW recursion relation \cite{Badger:2005zh,Badger:2005jv,Benincasa:2007qj,Berger:2006sh,Bern:2005hs,Brandhuber:2008pf,Brandhuber:2007up,Cheung:2009dc,Gang:2010gy,Luo:2005rx,Luo:2005my,Park:2006va,Quigley:2005cu,ArkaniHamed:2008gz,Bern:2005hh,Bern:2005ji,Bern:2005cq,Bedford:2005yy,Cachazo:2005ca, Chen3}, we exemplify its generalization to QCD amplitudes \cite{Chen3}, where there are massive quarks with mass m besides the massless YM fields i.e. gluons. For massive quarks, the momenta can be written as $k=\lambda\tilde \lambda+\beta\tilde\beta$ in spinor form, with $\langle \bt\lmd\rangle=[\td\bt\td\lmd]=m$. There is SU(2) freedom in the choices of $\{\lambda,\beta\}$ and correspondingly of $\{\tilde \lambda,\tilde \beta\}$ which leave $k$ invariant. One can shift two quark legs as:
\bea
k_{\hat q_1}&=&\lambda_{q_1}\tilde\lambda_{q_1}+\beta_{q_1} \tilde\beta_{q_1}+ z \lambda_{q_1}\tilde\beta_{q_1},\nonumber\\
k_{\hat q_2}&=&\lmd_{q_2}\td\lmd_{q_2}+\bt_{q_2}\td\bt_{q_2}-z \lmd_{q_1}\td\bt_{q_1},
\label{quarkshift}
\eea
with the condition:
\beq
\langle\lmd_{q_1},\lmd_{q_2}\rangle[\td\bt_{q_1},\td\lmd_{q_2}]+ \langle\lmd_{q_1},\bt_{q_2}\rangle[\td\bt_{q_1},\td\bt_{q_2}]=0.\label{massivecondition}
\eeq
In this way, the quarks remain on shell. (\ref{massivecondition}) is always possible due to the SU(2) freedom in choosing the spinors for the momentum. Choosing the spin states for the two quark legs as:
\beq
{\lmd_{q_1}\choose \td\bt_{q_1}},  {a\choose \td b}={1\over c_1} {\hat\lmd_{q_2}\choose \td\bt_{q_2}}-{1\over c_2} {\hat\bt_{q_2}\choose -\td\lmd_{q_2}},
\label{spinstates}
\eeq
with $c_1={[\td\bt_{q_2},\td\bt_{q_1}]\over m}$ and $c_2=-{[\td\lmd_{q_2},\td\bt_{q_1}]\over m}$, the amplitude $A(z)$ tends to 0 as $z$ goes to infinity under the shift (\ref{quarkshift}, \ref{massivecondition}), and can be calculated using BCFW technique.

Consider the amplitude with two quark-antiquark pairs as an example. Denote the two quark pairs as $q_1,\bar q_1, q_2, \bar q_2$. For the shift of $q_1-q_2$ as in (\ref{quarkshift}, \ref{massivecondition}), with the spin states (\ref{spinstates}), and taking the spin states of the two anti quarks $\bar q_1$ and $\bar q_2$ as $(\lmd_{\bar q_1},\td\bt_{\bar q_1})$ and $(\lmd_{\bar q_2},\td\bt_{\bar q_2})$ respectively, the amplitude calculated by BCFW technique is:
\beq
A({\hat q_1}^{-1\over 2},{\bar q_1}^{-1\over 2},{\hat q_2}^{z},{\bar q_2}^{-1\over 2})={2\over (k_{q_1}+k_{\bar q_1})^2}\left(\langle\lmd_{\bar q_1}, a\rangle [\td\bt_{q_1},\td\bt_{\bar q_2}]+\langle\lmd_{\bar q_2},\lmd_{q_1}\rangle [\td b,\td\bt_{\bar q_1}]\right).
\label{q1q2shift}
\eeq
The spin state of $q_2$ is correlated with that of $q_1$. In order to remove the correlation, a second shift is needed. Shifting $\bar q_1$ and $q_2$ this time, and choosing the spin states for them as $(\lmd_{\bar q_1}, \td\bt_{\bar q_1})$ and ${a'\choose \td b'}= {-\frac{1}{m}([\td \lmd_{q_2}\td\bt_{\bar q_1}]\lmd_{q_2}+[\td \bt_{q_2}\td\bt_{\bar q_1}]\bt_{q_2})\choose \td\bt_{\bar q_1}}$, the amplitude $A(z)$ can again be calculated by BCFW technique. Let the spin states of $q_1$ and $\bar q_2$ be ${\lmd_{q_1}\choose \td\bt_{q_1}}$, $(\lmd_{\bar q_2},\td\bt_{\bar q_2})$, the same as in the first shift, the amplitude is:
\bea
A({q_1}^{-1\over 2},{\hat{\bar q}_1}^{-1\over 2},{\hat q_2}^{z},{\bar q_2}^{-1\over 2})
={2\over (k_{q_1}+k_{\bar q_1})^2}\left(\langle\lmd_{ q_1}, a'\rangle [\td\bt_{\bar q_1},\td\bt_{\bar q_2}]+\langle\lmd_{\bar q_2},\lmd_{\bar q_1}\rangle [\td b',\td\bt_{q_1}]\right).
\label{barq1q2shift}
\eea
By a combination of (\ref{q1q2shift}) and (\ref{barq1q2shift}), one can obtain the amplitude with spin states ${\lmd_{q_1}\choose \td\bt_{q_1}}$, $(\lmd_{\bar q_1},\td\bt_{\bar q_1})$, ${\lmd_{q_2}\choose \td\bt_{q_2}}$ and $(\lmd_{\bar q_2},\td\bt_{\bar q_2})$:
\bea
A({q_1}^{-1\over 2},{\bar q_1}^{-1\over 2},{ q_2}^{-1\over 2},{\bar q_2}^{-1\over 2})=&&{2\over (k_{q_1}+k_{\bar q_1})^2}\left(\langle\lmd_{\bar q_1}, \lmd_{\bar q_2}\rangle[\td\bt_{q_1},\td \bt_{q_2}]+\langle\lmd_{\bar q_1}, \lmd_{q_2}\rangle [\td\bt_{q_1},\td\bt_{\bar q_2}]\right.\nb\\
&&\left.+\langle\lmd_{\bar q_2},\lmd_{q_1}\rangle [\td \bt_{q_2},\td\bt_{\bar q_1}]+\langle \lmd_{q_2}, \lmd_{q_1}\rangle [\td\bt_{\bar q_2},\td\bt_{\bar q_1}]\right).
\eea
Using little group generators \cite{Chen:2011ve}:
\beq
J^+= \left(\td\lmd\ppt{\td\bt}-\bt\ppt{\lmd} \right),~~  J^-= \left(\td\bt\ppt{\td\lmd}-\lmd\ppt{\bt}\right),
\eeq
one can obtain the amplitudes for other spin states:
\beq
A({q_1}^{\pm{1\over 2}},{\bar q_1}^{\pm{1\over 2}},{ q_2}^{\pm{1\over 2}},{\bar q_2}^{\pm{1\over 2}})=(J^+_{q_1})^{n_{q_1}}(J^+_{\bar q_2})^{n_{\bar q_2}} (J^+_{q_2})^{n_{q_2}} (J^+_{\bar q_1})^{n_{\bar q_1}}A({q_1}^{-1\over 2},{\bar q_1}^{-1\over 2},{ q_2}^{-1\over 2},{\bar q_2}^{-1\over 2}),
\eeq
where $n=1$ and $0$ correspond to helicity $+{1\over 2}$ and $-{1\over 2}$, i.e. spin states $(-\bt_{\bar q},\td\lmd_{\bar q})$ and $(\lmd_{\bar q},\td\bt_{\bar q})$, respectively. 

\subsection{BCFW Recursion Relation for Gauge Field Current}
In many cases off shell gauge field scattering amplitudes are also relevant, for example, as a sub diagram in a QCD scattering process or as building blocks for loop amplitudes. Besides the Berends-Giele recursion relation, BCFW technique is also extended to off shell amplitudes \cite{Feng:2011tw}. For the current $J^\mu(1,2,\cdots,n)$ with $n$ on shell YM legs, one can still choose a pair of on shell legs and do the usual BCFW deformation such that $\hat J^\mu(1,2,\cdots,n)\to 0$ as $z\to \infty$. Since for the current, gauge invariance no longer holds, one has to be careful about the two points: first, the reference spinors in the polarization vectors of the on shell legs should not change throughout recursion calculation; second, one need choose a gauge for the propagators, for which Feynman gauge is often chosen. In the on shell BCFW recursion relation (\ref{BCFW}), two light-like vectors $\epsilon^\pm$ are used as the polarization vectors for the two legs from breaking a propagator. In off shell recursion, due to the loss of gauge invariance, one needs two more light-like vectors $\epsilon^L\propto k$ (propagator momentum) and $\epsilon^T$, such that $\epsilon^{L/T}\cdot \epsilon^\pm=0$. With properly chosen normalizations such that $\epsilon^+\cdot \epsilon^-=1$ and $\epsilon^L\cdot \epsilon^T=1$, in four-dimensional spacetime one has the decomposition for the metric:
\beq
g_{\mu\nu}=\epsilon^+_\mu\epsilon^-_\nu+\epsilon^-_\mu\epsilon^+_\nu+\epsilon^L_\mu\epsilon^T_\nu+\epsilon^T_\mu\epsilon^L_\nu.
\eeq


Choosing legs 1 and $n$ to do the BCFW deformation, one can get:
\bea
J^{\mu}\left(1,2,...,n\right) & = & \sum_{i=2}^{n-1}\sum_{h,\tilde h} \left[A\left(\hat 1,...,i,-\hat k_{1,i}^h\right)\cdot\frac{1}{k_{1,i}^{2}} \cdot J^{\mu}\left(\hat k_{1,i}^{\tilde h},i+1,...,\hat{n}\right)\right.\nb\\
&&+J^{\mu}\left(\hat{1},...,i,-\hat k_{1,i}^h\right)\left. \cdot\frac{1}{k_{i+1,n}^{2}}\cdot A\left(\hat k_{1,i}^{\tilde h},i+1,...,\hat{n}\right)\right].\nb\\ 
&& (h,\tilde h)=(+,-),(-,+),(L,T),(T,L)\label{currentBCFW}
\eea
Using Ward identity, the expression can be simplified a little. For example, by the Ward identity $(\hat k_{1,i}^L)_{\mu}\cdot J^{\mu}\left(i+1,...,\hat{n}\right)=0$, the $(h,\tilde h)=(T,L)$ configuration of the second term in (\ref{currentBCFW}) vanishes. Actually for on shell recursion, both $(h,\tilde h)=(T,L)$ and $(h,\tilde h)=(L,T)$ configurations of both terms in (\ref{currentBCFW}) vanish due to Ward identity, and (\ref{currentBCFW}) reduces to the on shell BCFW recursion (\ref{BCFW}).

\section{Amplitude Relations}\label{amprel}
For $n$ point YM amplitudes, there are $(n-1)!$ independent tree level primitive amplitudes at a first sight, which is the number of all non-cyclic permutations of the $n$ legs. Several important relations of the tree level primitive amplitudes have been discovered.

\begin{itemize}

\item Color-order reversed identity:
\beq 
A_n(1,2,...,n-1,n)=(-)^n A_n(n,n-1,...,2,1).
\label{colorreverse}
\eeq

\item $U(1)$-decoupling identity:
\beq
\sum_{\sigma~ \rm cyclic} A(1,\sigma(2),...,\sigma(n))=0.
\label{U1decouple}
\eeq

\item Kleiss-Kuijf (KK) relations conjectured in \cite{Kleiss:1988ne} and proved in \cite{DelDuca:1999rs}:
\beq
A_n(1,\{\alpha\}, n,\{\beta\}) = (-1)^{n_\beta}\sum_{\sigma\in OP(\{\alpha\},\{\beta^T\})} A_n(1,\sigma, n)~.
\label{KKrel}
\eeq
The order-preserved (OP) sum is over all permutations of the set $\alpha \bigcup \beta^T$ that preserve the relative ordering of legs in set $\alpha$ and the reversed relative ordering of set $\beta$. The number of legs in $\beta$ is $n_\beta$. The $U(1)$-decoupling identity (\ref{U1decouple}) is actually a special case of KK relations (\ref{KKrel}) when there is only one leg in the set $\beta$.
\item Bern-Carrasco-Johansson (BCJ) Relations conjectured in \cite{Bern:2008qj} and proved in \cite{Chen:2011jx,Feng:2010my}:
\beq
A_n(1,2,\{\alpha\},3,\{\beta\})=\sum_{\sigma_i\in POP(\{\alpha,\beta\})}A_n(1,2,3,\sigma_i) {\mathcal F}_i,
\label{BCJrel}
\eeq
where the sum is over all partially ordered permutations (POP) that preserve the ordering of legs in set $\beta$, and ${\mathcal F}_i$ are some functions of the momenta given in \cite{Bern:2008qj}. When there is only one leg in set $\alpha$, the relation is called "Fundamental BCJ Relation" \cite{Feng:2010my}:
\beq
0=s_{21} A_n(1,2,3,\cdots,n-1,n)+\sum_{j=3}^{n-1}(s_{21}+\sum_{t=3}^j s_{2t}) A_n(1,3,4,\cdots,j,2,j+1,\cdots,n),
\label{fundaBCJ}
\eeq
with $s_{ij}=(k_i+k_j)^2$. It can be used to derive the other BCJ relations, with the help of KK relations \cite{Feng:2010my}.
\end{itemize}

Using KK relations (\ref{KKrel}), one can fix the positions of legs 1 and $n$ in the primitive amplitudes, thus reducing the number of independent primitive amplitudes to $(n-2)!$. By BCJ relations, one can fix the positions of legs 1, 2 and 3, and the number of independent primitive amplitudes is reduced to $(n-3)!$.

There are concise proofs \cite{Feng:2010my} of these amplitude relations by induction and BCFW recursion relation. A common spirit of the proofs is that amplitudes are decomposed into two sub amplitudes by BCFW recursion relation, and the induction assumption can be applied to the two sub amplitudes.

In the proof of fundamental BCJ relation (\ref{fundaBCJ}), the fact that for BCFW shifting of a pair of non-adjacent YM legs the complexified amplitude $A(z)\to \mathcal O(z^{-2})$ as $z\to \infty$--instead of simply being vanishing--plays a key role. The main ideas of the proof are contained in the four point example. Shifting the legs 1 and 2, multiplying $\frac{\hat s_{23}(z)}{z}$ with the four point KK relation $A(\hat 1,\hat 2,3,4)+A(\hat 1,3,4,\hat 2)+A(\hat 1,3,\hat 2,4)=0$, and doing a contour integral which includes all the finite poles, one gets:
\beq
\oint \frac{dz}{z} \hat s_{23}(z) [A(\hat 1,\hat 2,3,4)+A(\hat 1,3,4,\hat 2)+A(\hat 1,3,\hat 2,4)]=0.
\eeq
The third term vanishes due to its large $z$ scaling behavior $\mathcal O(z^{-2})$. Performing the integral one gets $s_{23} A(1,2,3,4)+(s_{23}+s_{43}) A(1,3,4,2)=0$. Since $A(1,2,3,4)$ is the same as $A(2,3,4,1)$, and $A(2,4,3,1)=A(1,3,4,2)$ due to color-order reversed identity, one can rewrite it as:
\beq
s_{23} A(2,3,4,1)+(s_{23}+s_{43}) A(2,4,3,1)=0,
\eeq
which is the fundamental four point BCJ relation.

\section{A Latest Important Development: Amplitudes and Positive Grassmannian}

In the very recent years, Nima Arkani-Hamed et. al. have achieved a new understanding of amplitudes \cite{ArkaniHamed:2012nw, ArkaniHamed:2009dg, ArkaniHamed:2009sx, ArkaniHamed:2009vw} by relating them to positive Grassmannian. In this new approach, amplitudes to all loop levels can be constructed with on shell three point amplitudes, without any off shell line in the diagrams. In this section, I provide a preliminary introduction to this new approach. The formulae are all from \cite{ArkaniHamed:2012nw}.

Currently most of the results in this approach have been for planar $\mathcal N=4$ SYM amplitudes. For $\mathcal{N}=4$ SYM, all the helicity states can be grouped into a single Grassmann coherent state labeled by Grassmann (anti-commuting) parameters $\widetilde{\eta}^I$ for $I= 1,\ldots, 4$:
\beq
\left|\widetilde \eta \right> \equiv \left|+1 \right> + \widetilde \eta^I  \left|+\text{$\textstyle\frac{1}{2}$}\right>_I + \frac{1}{2!} \widetilde \eta^I \widetilde \eta^J \left|0 \right>_{IJ} + \frac{1}{3!} \epsilon_{IJKL} \widetilde \eta^I \widetilde \eta^J \widetilde \eta^K \left|-\text{$\textstyle\frac{1}{2}$} \right>^L + \frac{1}{4!} \epsilon_{IJKL} \widetilde \eta^I \widetilde \eta^J \widetilde \eta^K \widetilde \eta^L \left| -1 \right>.
\eeq
For three point amplitudes with the momenta $k_i=\lmd_i \td\lmd_i$, $i=1,2,3$, momentum conservation allows two solutions ie. $\lmd_1\propto\lmd_2\propto\lmd_3$ or $\td\lmd_1\propto\td\lmd_2\propto\td\lmd_3$, which correspond to the two building blocks:
\bea
\mathcal{A}^{(1)}_3&=&\raisebox{-0.8cm}{\includegraphics[width=2cm,height=2cm]{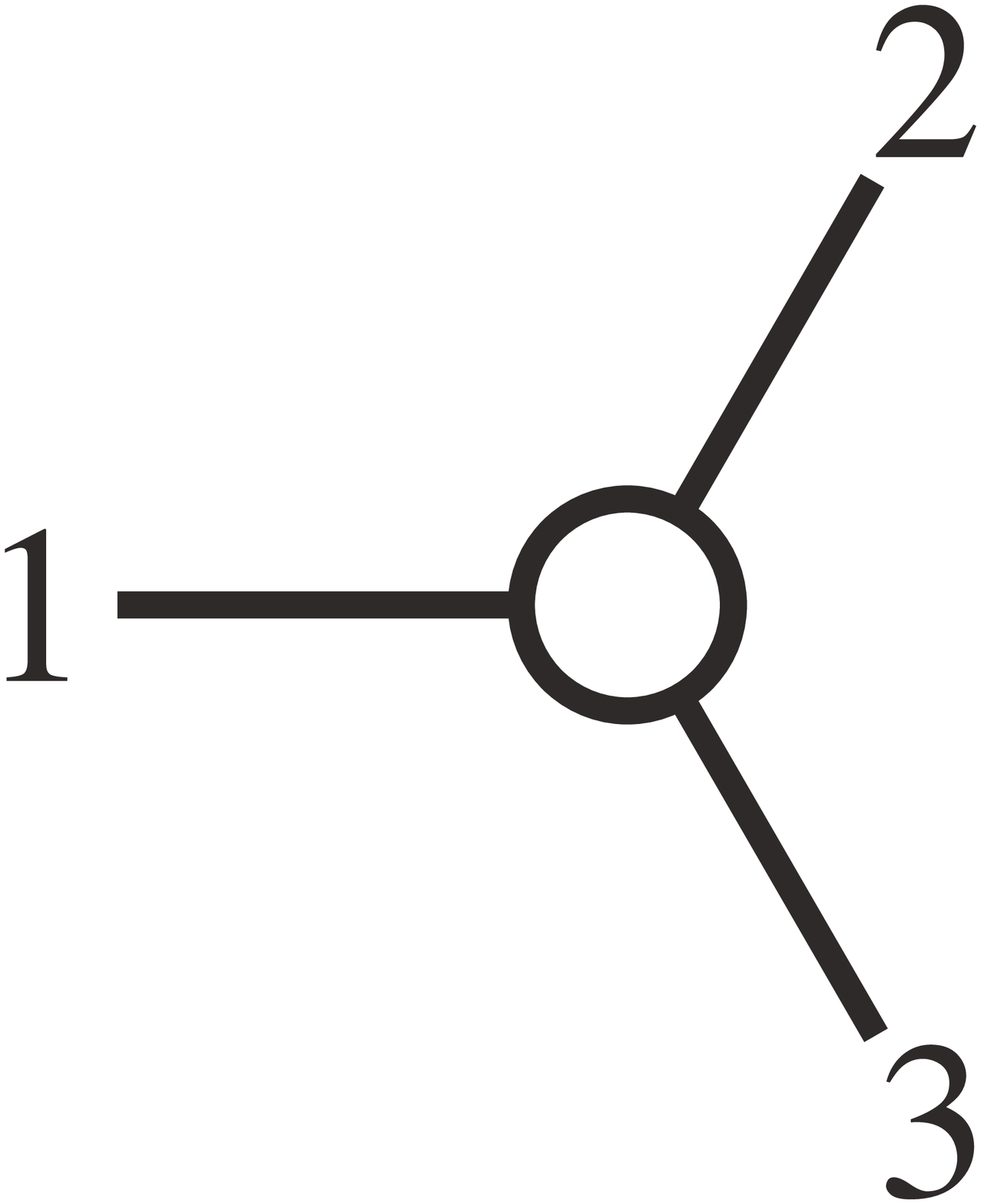}}=\displaystyle\,\frac{\delta^{1\times4}\big([23]\widetilde{\eta}_1+[31]\widetilde{\eta}_2+[12]\widetilde{\eta}_3\big)}{[12][23][31]}\delta^{2\times2}\big(\lambda_1 \widetilde \lambda_1 + \lambda_2 \widetilde \lambda_2 + \lambda_3 \widetilde \lambda_3\big),\nonumber\\
\mathcal{A}^{(2)}_3&=&\raisebox{-0.8cm}{\includegraphics[width=2cm,height=2cm]{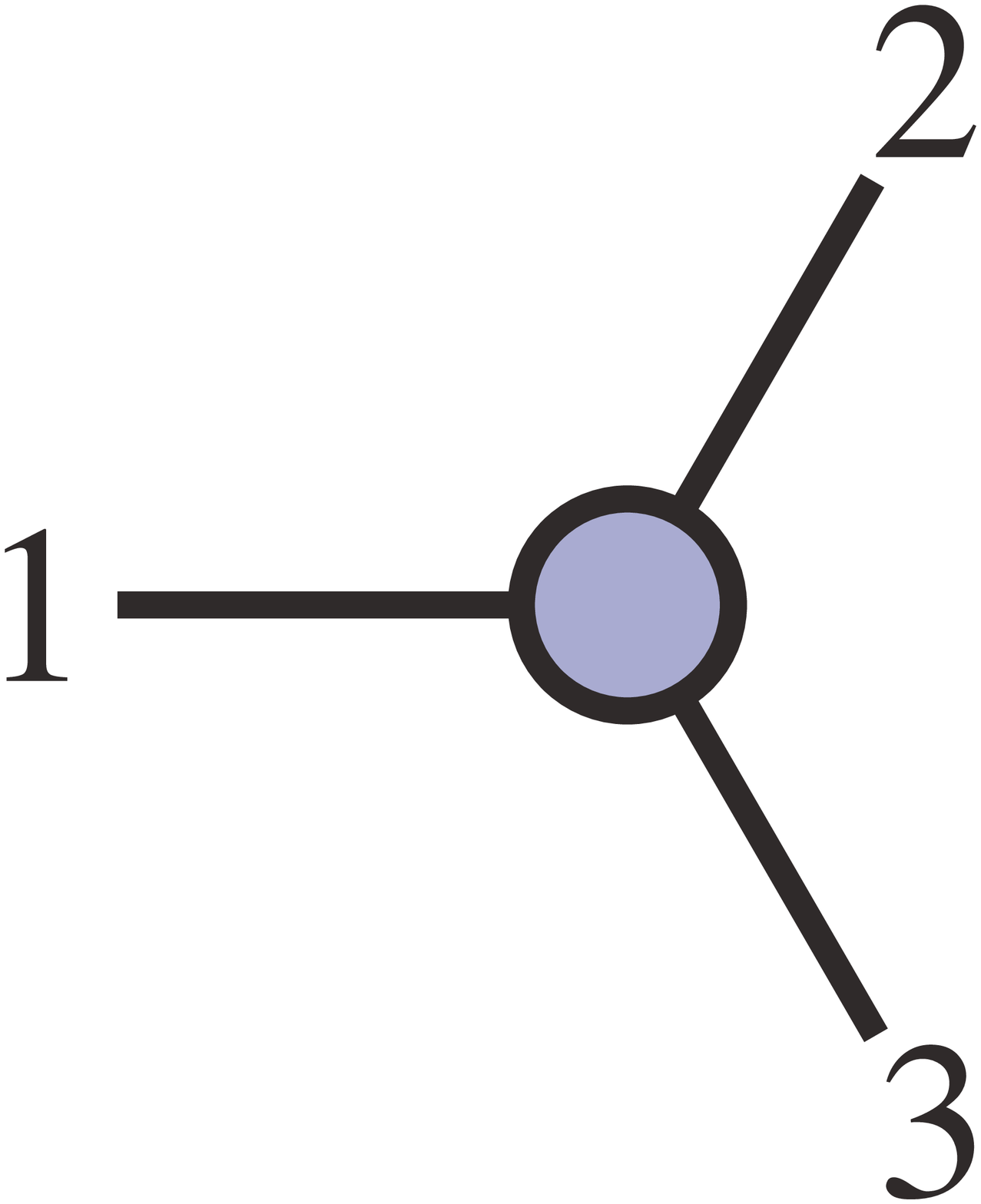}}=\displaystyle \frac{\delta^{2\times4}\big(\lambda_1\widetilde{\eta}_1+\lambda_2\widetilde{\eta}_2+\lambda_3\widetilde{\eta}_3\big)}{\langle12\rangle\langle23\rangle\langle31\rangle} \delta^{2\times2}\big(\lambda_1 \widetilde \lambda_1 + \lambda_2 \widetilde \lambda_2 + \lambda_3 \widetilde \lambda_3\big). 
\label{grassmanblock}
\eea
On shell diagrams for n-point scattering can be built up by joining the above white and black trivalent blocks. There are two operations of the diagrams--''merge" and ''square move"--after which the diagrams are equivalent to the original diagrams. After modulo these equivalent operations, the authors defined reduced diagrams, which they found to correspond to permutations of the external legs. There can appear ''bubble"s in the diagrams which correspond to loops in the familiar treating of amplitudes.

There are $n$ 2-componet spinors $\lambda_i$ and anti-spinors $\td\lmd_i$ to describe the kinematics for n-point scattering. All the $\lambda_i$'s can be collectively denoted by a $(2\!\times\!n)$-matrix:
\beq
\lambda\equiv\left(\begin{array}{@{}cccc@{}}\lambda_1^1&\lambda_2^1&\cdots&\lambda^1_n\\\lambda_1^2&\lambda_2^2&\cdots&\lambda^2_n\end{array}\right)\raisebox{-.75pt}{\text{{\Large$\Leftrightarrow$}}}\left(\begin{array}{@{}cccc@{}}\lambda_1&\lambda_2&\cdots&\lambda_n\end{array}\right).
\eeq
This $\lambda$ can be viewed as a $2$-plane in $n$ dimensions---an element of $G(2,n)$ as realized in \cite{ArkaniHamed:2009dn}. Similarly $\widetilde \lambda$ is defined and is an element of $G(2,n)$.

In Grassmannian space, one can linearize the momentum conservation condition $\sum\limits_i \lambda_i \td\lmd_i=0$ by introducing auxiliary planes. Then the two building blocks in (\ref{grassmanblock}) can be represented as:
\bea
\mathcal{A}^{(1)}_{3}&=&\!\int\!\!\frac{d^{1\times3}W}{\mathrm{vol}(GL(1))}\frac{\delta^{1\times4}\big(W\!\cdot\!\widetilde{\eta}\big)}{(1)(2)(3)} \delta^{1\times2}\big(W\!\cdot\!\widetilde{\lambda}\big) \delta^{2\times2}\big(\lambda\!\cdot\!W^\perp\!\big),\\
\mathcal{A}^{(2)}_{3}&=&\!\int\!\!\frac{d^{2\times3}B}{\mathrm{vol}(GL(2))}\frac{\delta^{2\times4}\big(B\!\cdot\!\widetilde{\eta}\big)}{(12)(23)(31)} \delta^{2\times2}\big(B\!\cdot\!\widetilde{\lambda}\big) \delta^{2\times1}\big(\lambda\!\cdot\!B^\perp\!\big),
\eea
where $W\in G(1,3)$ is auxiliary 1-plane in 3-dimensions; $B\in G(2,3)$ is auxiliary 2-plane in 3-dimensions; $W^\perp$ and $B^\perp$ are the orthogonal planes for $W$ and $B$; $W$ being given as a $(1\times 3)$-matrix $W\equiv(w_1,w_2,w_3)$ then $(a)\equiv(w_a)$; $B$ being $(2\times 3)$-matrix then $(ab)$ is the determinant of the matrix formed by the a-th and b-th columns in $B$. In short, each white trivalent block carries a measure:
\beq
d\Omega_w\equiv\frac{d^{1\times3}W}{\mathrm{vol}{(GL(1))}}\frac{1}{(a)(b)(c)},
\label{measure1}
\eeq
and each black trivalent block:
\beq
d\Omega_b\equiv\frac{d^{2\times3}B}{\mathrm{vol}{(GL(2))}}\frac{1}{(a\,b)(b\,c)(c\,a)}.
\label{measure2}
\eeq
The integration measure for a whole on shell diagram is just the product of the above two types of measures corresponding to each trivalent building block in the diagram. With these auxiliary planes $W$ and $B$, the momentum conservation constraints on $\lambda$ and $\td\lmd$ are decoupled such that $\lmd$ and $\td\lmd$ appear linearly in the $\delta$-functions. This makes the integration over internal lines--reminding that they are all on shell lines--actually trivial.


For $\mathcal{N}<4$ SYM, there appear Jaccobi factors in the measures (\ref{measure1}) and (\ref{measure2}), which is the major apparent difference from the case with maximal supersymmetry. Results for $\mathcal{N}<4$ SYM have been rare compared to $\mathcal{N}=4$ SYM.

\section{Motivations and Sketches of Our Work in Gauge Field Amplitudes}
From the previous sections, we have seen that the behavior of an amplitude under large deformations of momenta is a very important property of the amplitude. First, with vanishing behavior at large deformations, BCFW recursion relation can be constructed. Second, as in the proof of BCJ relations in Section \ref{amprel}, $\mathcal O(z^{-2})$ scaling behavior of the amplitude $A(z)$ with a pair of non-adjacent legs deformed is crucial for BCJ relations to hold in YM theory. Thirdly, were there better scaling behavior than $\mathcal O(z^{-2})$ at large deformations, there would be more powerful amplitude relations, which could further decrease the number of independent primitive amplitudes. Thus, studying the large deformation behavior (boundary behavior) of amplitudes are quite meaningful.

The boundary behavior of on shell YM amplitudes are studied in \cite{ArkaniHamed:2008yf}. Besides on shell amplitudes, off shell amplitudes are also very important. First, off shell YM amplitudes like gluon amplitudes are often encountered as sub diagrams in a complete QCD process. Second, due to confinement, gluon scattering is in principle off shell. Third, off shell tree amplitudes are useful for building up loop level amplitudes. In \cite{Nima1} and \cite{Boels}, the authors analyzed the boundary behavior in Arkani-Hamed-Kaplan (AHK) gauge \cite{Nima1} with legs not shifted being off shell. We intended to analyze the boundary behavior in Feynman gauge with all legs off shell, to uncover why amplitudes with non-adjacent legs deformed have better scaling properties than those with adjacent legs deformed, to construct off shell recursion relation and to obtain off shell amplitude relations. These are a sketch of our joint work \cite{Zhang:2013cha} with my advisor Dr. Gang Chen that will be described in Chapter \ref{boubehav}.

In a previous work by Gang Chen \cite{Chen1}, he developed a new recursion relation, induced from complexified Ward identity, to calculate boundary terms for off shell YM amplitudes. Then we together generalized the method to obtain recursion relations for the whole off shell amplitudes both at tree and one loop levels \cite{Chen2}. These are described in Chapter \ref{Wardidenrec}.

\chapter{Boundary Behaviors for General Off-shell Amplitudes in Yang-Mills Theory}\label{boubehav}

\section{Introduction}

Recent years, BCFW recursion relation \cite{Britto:2004nj,Britto:2004nc,Britto:2004ap} has been widely used in various quantum field theories. At tree level, the amplitudes in pure Yang-Mills theory are rational functions of external momenta and external polarization vectors in spinor form 
\cite{Parke:1986gb,Xu:1986xb,Berends:1987me,Kosower,Dixon:1996wi,Witten1}. According to this, BCFW recursion relation was proposed and developed in \cite{Britto:2004nj,Britto:2004nc,Britto:2004ap}, and then proved in \cite{Britto:2005fq} using the pole structures of tree level on shell amplitudes.  Besides the progresses on on-shell amplitudes,  off-shell amplitudes are also studied using BCFW or other methods \cite{Feng:2011tw, Chen1, Chen3, Britto, Chen2, Berends:1987me}.  Although off-shell amplitudes are gauge dependent and usually complicated, they are of great importance in the phenomenological calculations. Moreover, off-shell amplitudes emerge in the construction of on-shell loop level amplitudes.   Hence it is also valuable to get recursion relations for general off-shell amplitudes.

BCFW recursion relation works very well when the amplitudes vanish at large BCFW shift limit. Hence the boundary behavior of the amplitudes is very important for building up BCFW recursion relation.  Furthermore, improved boundary behavior also implies new amplitude relations like BCJ relations \cite{Bern:2008qj,Boels,Feng:2010my}.  For tree and loop level Yang-Mills amplitudes, the boundary behavior was analyzed in \cite{Nima1, Boels} in AHK gauge for both adjacent and non-adjacent BCFW shifts. Hence a natural question is whether it is possible to analyze the boundary behavior in usual Feynman gauge, and why essentially non-adjacent BCFW shifts result in improved boundary behavior in Feynman gauge compared with adjacent BCFW shifts. Furthermore, according to the boundary behavior, can we build up the recursion relation correspondingly for general off-shell amplitudes?

In this chapter, we first describe the procedure to obtain general off-shell amplitudes recursively in \sref{Sec:Off-shell} using BCFW technique and the technique in \cite{Chen3}. The procedure bases on the boundary behavior of amplitudes in Feynman gauge, which is proved in the following sections. We use this technique to calculate off shell amplitudes and analyze off shell amplitude relations. In \sref{Sec:Reduce} we prove that the boundary behavior of amplitudes can be analyzed using reduced vertexes, which are defined in the section. Using the conclusion of this section, we directly obtain the boundary behavior for adjacent shifts. In \sref{Sec:Non-Adj} we analyze the behavior of the amplitudes for non-adjacent shifts. We find that permutation sum greatly improves the boundary behavior for non-adjacent shifts compared to adjacent shifts. 


 
\section{Recursion Relation for General Off-shell Amplitudes}\label{Sec:Off-shell}
Throughout this chapter, we will use $k_l$ and $k_r$ for the pair of momenta to be shifted, with indices $\mu$ and $\nu$. The momenta shift is \cite{Schwinn:2007ee,Badger:2005zh,Badger:2005jv}:
\beq
\hat k_l=k_l+z \eta~~~~~~\hat k_r=k_r-z\eta,
\label{momshift4}
\eeq
with 
\beq\label{conditionEta4}
\eta^2=k_l\cdot \eta=k_r\cdot \eta=0. 
\eeq 
Since we need to shift two off-shell legs for general off-shell amplitudes in Yang-Mills theory, we do not require the momenta of the two shifted legs, ie. $k_l$ and $k_r$, to be on-shell. Other un-shifted legs are also in general off shell. There are no propagators for the un-contracted external legs. Let two arbitrary vectors $\epsilon_{l\ \mu}$ and $\epsilon_{r\ \nu}$ couple to the two shifted legs, the amplitude component is $\mathcal{M}^{\mu\nu} \epsilon_{l\ \mu} \epsilon_{r\ \nu}$. The indices of other external legs are suppressed.

To get all the components of $\mathcal{M}^{\mu\nu}$, we need to know the amplitudes $\mathcal{M}^{\mu\nu} \epsilon_{l\ \mu} \epsilon_{r\ \nu}$ for $4\times 4$ independent pairs of $\epsilon_{l\ \mu}$ and $\epsilon_{r\ \nu}$ in four dimensional field theory.  According to \cite{Chen1,Chen2}, when one of the shifted legs is contracted with its momentum, the amplitude is reduced to less point amplitudes according to the cancellation details of Ward identity in Feynman gauge. For example with color ordered amplitude $\mathcal M(k_1,k_2,\cdots,k_{N+1})^{\mu_1\mu_2\cdots\mu_{N+1}}$, we derive:

{\begin{small}
\bea
&&k_{\mu_{N+1}} \mathcal{M}(k_1,k_2,\cdots,k_{N+1})^{\mu_1\mu_2\cdots\mu_{N+1}}\label{longitudinalcomp}\\
=&&\frac{i}{\sqrt{2}}\frac{g_{\rho\sigma}}{(k_1+k_{N+1})^2}\mathcal{M}(k_1,-k_1)^{\mu_1\sigma} \mathcal{M}(k_2,k_3,\cdots,k_{N},-K_{2,N})^{\mu_2\mu_3\cdots\mu_N \rho}\nonumber\\
&&-\frac{i}{\sqrt{2}}\frac{g_{\rho\sigma}}{(k_N+k_{N+1})^2}\mathcal{M}(k_N,-k_N)^{\mu_N\sigma} \mathcal{M}(k_1,k_2,\cdots,k_{N-1},-K_{1,N-1})^{\mu_1\mu_2\cdots\mu_{N-1} \rho}\nonumber\\
&&+\sum_{j=1}^{N-1}\frac{i}{\sqrt{2}K_{1,j}^2 K_{j+1,N}^2} [K_{1,j\ \rho}\mathcal{M}(k_1,k_2,\cdots,k_j,-K_{1,j})^{\mu_1\mu_2\cdots\mu_j \rho}] \nonumber\\
&&\ \ \ \ \ \ \ \ \ \ \ \ \ \ \ \ \ \ \ \ \ \ \ \ \ \ \ \ \cdot[k_{N+1\ \sigma} \mathcal{M}(k_{j+1},k_{j+2},\cdots,k_N,-K_{j+1,N})^{\mu_{j+1}\mu_{j+2}\cdots\mu_N \sigma}]\nonumber\\
&&-\sum_{j=1}^{N-1}\frac{i}{\sqrt{2}K_{1,j}^2 K_{j+1,N}^2}[ k_{N+1\ \rho}\mathcal{M}(k_1,k_2,\cdots,k_j,-K_{1,j})^{\mu_1\mu_2\cdots\mu_j \rho} ]\nonumber\\
&&\ \ \ \ \ \ \ \ \ \ \ \ \ \ \ \ \ \ \ \ \ \ \ \ \ \ \ \ \cdot[K_{j+1,N\ \sigma} \mathcal{M}(k_{j+1},k_{j+2},\cdots,k_N,-K_{j+1,N})^{\mu_{j+1}\mu_{j+2}\cdots\mu_N \sigma}].\nonumber
\eea
\end{small}}
In the above we have reduced $k_{\mu_{N+1}} \mathcal{M}(k_1,k_2,\cdots,k_{N+1})^{\mu_1\mu_2\cdots\mu_{N+1}}$ to less point amplitudes. $K_{1,j}=k_1+k_2+\cdots+k_j$ and $K_{j+1,N}=k_{j+1}+\cdots+k_N$. The indices for the amplitudes are in the same order as the momenta in the brackets of the amplitudes. We define $\mathcal{M}(k_1,-k_1)^{\mu_1 \rho}=i k_1^2 g^{\mu_1\rho}$ and $\mathcal{M}(k_N,-k_N)^{\mu_N \sigma}=i k_N^2 g^{\mu_N \sigma}$.

Hence to build up BCFW recursion relation for general off-shell amplitudes, we only need to consider other three transverse components of the external vectors coupling to the shifted legs.

For convenience, the momenta can be written in spinor form \cite{Dittmaier:1998nn}:
\be
k=\left\{ \begin{array}{cl}
\lmd\td\lmd+\bt\td\bt & ~~~~~\textrm{if $k$ is time-like}\\
\lmd\td\lmd-\bt\td\bt & ~~~~~\textrm{if $k$ is space-like.}\\
\lmd\td\lmd & ~~~~~\textrm{if $k$ is light-like}
\end{array}\right.
\ee
Here we exemplify the cases with time-like or light-like $k_l$ and $k_r$, and the case with either space-like $k_l$ or $k_r$ is similar.

We first consider the case with both $k_l$ and $k_r$ off shell. We write $k_l$ as $k_l= \lmd_l\td\lmd_l+\bt_l\td\bt_l$ \cite{Chalmers}.  As analyzed in \cite{Chen3}, since there is $SU(2)$ freedom for choosing the  spinors of $k_l$, we can choose them such that $(\lmd_l\td\bt_l)\cdot k_r=(\bt_l\td\lmd_l)\cdot k_r=0$. At the same time we can set the spinors for $k_r$ to be either $k_r= \lmd_l\td\lmd_r+\bt_r\td\bt_l$ or $k_r= \lmd'_r\td\lmd_l+\bt_l\td\bt'_r$. Hence we have two choices for the shifting momentum $\eta$ as $\eta=\lmd_l\td\bt_l$ or $\eta'=\bt_l\td\lmd_l$, which satisfy the condition \eref{conditionEta4}.

First for $\eta=\lmd_l\td\bt_l$, the external vectors are written as 
 \be
\begin{array}{cc}
\epsilon_l\in\left( \begin{array}{cl}
 \epsilon_l^-&= \lmd_l\td\bt_l\\
 \epsilon_l^+&= \bt_l\td\lmd_l\\
  \epsilon_l^{\perp}&= \lmd_l\td\lmd_l-\bt_l\td\bt_l,
\end{array}\right) &~~~~~\epsilon_r\in\left( \begin{array}{cl}
 \epsilon_r^-&= \lmd_l\td\bt_l\\
 \epsilon_r^+&= \bt_r\td\lmd_r\\
  \epsilon_r^{\perp}&= \lmd_l\td\lmd_r-\bt_r\td\bt_l-z\lmd_l\td\bt_l
\end{array}\right).
\ea
\label{epsilonset1}
\ee
Under the momenta shift \eref{momshift4}, we have
\bea\label{spinorShift}
\td\lmd_l&\rightarrow& \hat{\td\lmd}_l=\td\lmd_l+z\td\bt_l \nb\\
\bt_r&\rightarrow& \hat\bt_r=\bt_r-z\lmd_l.
\eea
In $\epsilon_r^\perp$, we add the term $-z\lmd_l\td\bt_l$, such that after the momenta shift \eref{spinorShift}, $\hat\epsilon_r^\perp$ is independent of $z$ and still $\hat k_r\cdot\hat\epsilon_r^\perp=0$.



Then for $\eta'=\bt_l\td\lmd_l$, we just replace $\epsilon_r$ with $\epsilon'_r$ which is defined as following:
\be
\epsilon'_r\in\left( \begin{array}{cl}
 \epsilon_r^{'-}&= \lmd'_r\td\bt'_r\\
 \epsilon_r^{'+}&= \bt_l\td\lmd_l\\
  \epsilon_r^{'\perp}&= \lmd'_r\td\lmd_l-\bt_l\td\bt'_r-z\bt_l\td\lmd_l
\end{array}\right).
\label{epsilonset2}
\ee
Under the momenta shift, we have 
\bea\label{spinorShift2}
\lmd_l&\rightarrow& \hat{\lmd}_l=\lmd_l+z\bt_l \nb\\
\td\bt'_r&\rightarrow& \hat{\td\bt}'_r=\td\bt'_r-z\td\lmd_l.
\eea
If one of the legs is on shell and the other is off-shell, without loss of generality, we set $l$-leg to be on-shell and $r$-leg to be off-shell. Writing $k_l$ as $\lmd_l\td\lmd_l$ and using the little group transformation of $k_r$, the momentum of $r$-leg can be written as $k_r=\lmd_l\td\lmd_r+\bt_r\td\bt'_r=\lmd'_r\td\lmd_l+\bt{''}_r\td\bt{'''}_r$. Correspondingly, one of the shifting momentum is $\eta=\lmd_l \td\bt'_r$ and the other is $\eta'=\bt^{''}_r\td\lmd_l$. When the shifting momentum is $\eta$, the external vectors  are written as 
 \be
\begin{array}{cc}
\epsilon_l\in\left( \begin{array}{cl}
 \epsilon_l^-&= {\lmd_l\td\bt'_r\over [\td\lmd_l,\td\bt'_r]}\\
 \epsilon_l^+&={\bt_l\td\lmd_l\over \la\bt_l,\lmd_l\ra}\\
\end{array}\right) &~~~~~\epsilon_r\in\left( \begin{array}{cl}
 \epsilon_r^-&= \lmd_l\td\bt'_r\\
 \epsilon_r^+&= \bt_r\td\lmd_r\\
  \epsilon_r^{\perp}&=\lmd_l\td\lmd_r-\bt_r\td\bt'_r-z \lmd_l\td\bt'_r
\end{array}\right) .
\ea
\ee
Under the momenta shift, the spinors transform as 
\bea\label{spinorShift3}
\td\lmd_l&\rightarrow& \hat{\td\lmd}_l=\td\lmd_l+z\td\bt'_r \nb\\
\bt_r&\rightarrow& \hat\bt_r=\bt_r-z\lmd_l.
\eea
When the shifting momentum is $\eta'$, then the external vectors can be written as 
\be
\begin{array}{cc}
\epsilon_l\in\left( \begin{array}{cl}
 \epsilon_l^-&= {\lmd_l\td\bt'_r\over [\td\lmd_l,\td\bt'_r]}\\
 \epsilon_l^+&={\bt''_r\td\lmd_l\over \la\bt''_r,\lmd_l\ra}\\
\end{array}\right) &~~~~~\epsilon'_r\in\left( \begin{array}{cl}
 \epsilon_r^{'-}&= \lmd'_r\td\bt{'''}_r\\
 \epsilon_r^{'+}&= \bt{''}_r\td\lmd_l\\
  \epsilon_r^{'\perp}&=\lmd'_r\td\lmd_l-\bt{''}_r\td\bt{'''}_r-z\bt{''}_r\td\lmd_l
\end{array}\right).
\ea
\ee
Correspondingly, the spinors transform as 
\bea\label{spinorShift4}
\lmd_l&\rightarrow& \hat{\lmd}_l=\lmd_l+z\bt^{''}_r \nb\\
\td\bt^{'''}_r&\rightarrow& \hat{\td\bt}^{'''}_r=\td\bt^{'''}_r-z\td\lmd_l.
\eea

The case with both shifted lines on-shell has been discussed in \cite{Boels}.

To use BCFW recursion relation for the full amplitudes, we need to analyze the boundary behavior of the amplitudes with shifted momenta. For all the cases discussed above, the following conditions hold:
\be\label{conditionVector}
\hat k_l\cdot \hat\epsilon_l=\hat k_r \cdot \hat\epsilon_r=0.
\ee

As will be proved in the following sections, under the conditions \eref{conditionEta4} and \eref{conditionVector}, we have  
\be\label{behaviorOff}
\mathcal{\hat M}^{\mu\nu}=\left\{ \begin{array}{cl}
z  A_1 g^{\mu\nu}+ A_0 g^{\mu\nu}+ B^{\mu\nu}+\mathcal{O}({1\over z}) &~~~\text{for adjacent shift}\\
 A'_0 g^{\mu\nu}+\mathcal{O}({1\over z}) &~~~\text{for non-adjacent shift}
\end{array}\right. .
\ee
In \eref{behaviorOff}, all the un-shifted and shifted external legs can be off-shell. 


According to \eref{behaviorOff}, we can get the large $z$ scaling behavior for general off-shell amplitudes $\mathcal{M}^{\mu\nu}$ $\epsilon_{l\ \mu}$ $\epsilon_{r\ \nu}$ for all the BCFW shifts above:
\begin{itemize}
\item Both $k_l$ and $k_r$ off-shell with shifting momentum: $\eta=\lmd_l\td\bt_l$
\bea\label{behaviorTable1}
\left.\begin{array}{c|c|c|c}
 &\epsilon^-_r &\epsilon^+_r&\epsilon_r^{\perp}\\ \hline 
\epsilon^-_l & z^{-1}&z^2&z^0 \\ \hline 
 \epsilon^+_l&z^2 &z^3&z^2 \\ \hline 
\epsilon_l^{\perp}  &z &z^3&z^2 \\ 
\end{array}\right. &~~~~~~\left.\begin{array}{c|c|c|c}
 &\epsilon^-_r &\epsilon^+_r&\epsilon_r^{\perp}\\ \hline 
\epsilon^-_l & z^{-1}&z&z^{-1} \\ \hline 
 \epsilon^+_l&z &z^2&z\\ \hline 
\epsilon_l^{\perp}  &z^{0} &z^2&z \\ 
\end{array}\right. \nb \\
\text{Adjacent}&~~\text{Non-adjacent }
\eea
\item Both $k_l$ and $k_r$ off-shell with shifting momentum $\eta'=\bt_l \td\lmd_l$
\bea\label{behaviorTable2}
\left.\begin{array}{c|c|c|c}
 &\epsilon^-_r &\epsilon^+_r&\epsilon_r^{\perp}\\ \hline 
\epsilon^-_l & z^{3}&z^2&z^2 \\ \hline 
 \epsilon^+_l&z^2 &z^{-1}&z^0 \\ \hline 
\epsilon_l^{\perp}  &z^3 &z&z^2 \\ 
\end{array}\right. &~~~~~~\left.\begin{array}{c|c|c|c}
 &\epsilon^-_r &\epsilon^+_r&\epsilon_r^{\perp}\\ \hline 
\epsilon^-_l & z^{2}&z&z \\ \hline 
 \epsilon^+_l&z &z^{-1}&z^{-1} \\ \hline 
\epsilon_l^{\perp}  &z^2 &z^{0}&z \\ 
\end{array}\right. \nb \\
\text{Adjacent}&~~\text{Non-adjacent }
\eea
\item $k_l$ on-shell and $k_r$ off-shell with shifting momentum $\eta=\lmd_l \td\bt'_r$
\bea\label{behaviorTable3}
\left.\begin{array}{c|c|c|c}
 &\epsilon^-_r &\epsilon^+_r&\epsilon_r^{\perp}\\ \hline 
\epsilon^-_l & z^{-1}&z^2&z^0 \\ \hline 
 \epsilon^+_l&z^2 &z^3&z^2 \\ 
\end{array}\right. &~~~~~~\left.\begin{array}{c|c|c|c}
 &\epsilon^-_r &\epsilon^+_r&\epsilon_r^{\perp}\\ \hline 
\epsilon^-_l & z^{-1}&z&z^{-1} \\ \hline 
 \epsilon^+_l&z &z^2&z \\ 
\end{array}\right. \nb \\
\text{Adjacent}&~~\text{Non-adjacent }
\eea
\item $k_l$ on-shell and $k_r$ off-shell with shifting momentum $\eta'=\bt^{''}_r\td\lmd_l$
\bea\label{behaviorTable4}
\left.\begin{array}{c|c|c|c}
 &\epsilon^-_r &\epsilon^+_r&\epsilon_r^{\perp}\\ \hline 
\epsilon^-_l & z^{3}&z^2&z^2 \\ \hline 
 \epsilon^+_l&z^2 &z^{-1}&z^0 \\ 
\end{array}\right. &~~~~~~\left.\begin{array}{c|c|c|c}
 &\epsilon^-_r &\epsilon^+_r&\epsilon_r^{\perp}\\ \hline 
\epsilon^-_l & z^{2}&z&z \\ \hline 
 \epsilon^+_l&z &z^{-1}&z^{-1} \\ 
\end{array}\right. \nb \\
\text{Adjacent}&~~\text{Non-adjacent }
\eea
\end{itemize}

According to the little group property and the analysis in \cite{Chen3}, and using essentially the same procedures therein, we can construct the BCFW recursion relation for off shell amplitudes. We exemplify the procedure in the case that all external legs are off shell and show how it is reduced to less point amplitudes.

We choose a specific $r$-leg, and two non adjacent $l$-legs, ie. $l_1$ and $l_2$. Then we can do two shifts: $l_1$ and $r$ legs, or $l_2$ and $r$ legs. When we shift $l_1$ and $r$ legs, we shift them as in table \ref{behaviorTable1}, and we choose the vectors coupling to $l_1$ as $\epsilon_{l_1}^-=\eta_1=\lambda_{l_1}\tilde \beta_{l_1}$. At the same time we couple to $l_2$ a vector $\epsilon_{l_2}^-=\eta_2=\lambda_{l_2}\tilde \beta_{l_2}$. For choices of $\epsilon_{r(1)}^-$ and $\epsilon_{r(1)}^\perp$ on $r$ leg, the two amplitudes:
\beq
\mathcal{M}_{\mu_r\mu_{l_1}\mu_{l_2}} \epsilon_{l_1}^{-\,\mu_{l_1}} \epsilon_{l_2}^{-\,\mu_{l_2}} \epsilon_{r(1)}^{-\,\mu_r} \ \ \mbox{and}\ \ \mathcal{M}_{\mu_r\mu_{l_1}\mu_{l_2}} \epsilon_{l_1}^{-\,\mu_{l_1}} \epsilon_{l_2}^{-\,\mu_{l_2}} \epsilon_{r(1)}^{\perp\,\mu_r}
\label{decoherence1}
\eeq
are of $\mathcal{O}(z^{-1})$, and can be reduced to less point amplitudes using BCFW technique. The subscript ''$(1)$'' in $\epsilon_{r(1)}^-$ or $\epsilon_{r(1)}^\perp$ means that it is for $l_1-r$ shifting. For the same reason when we shift $l_2$ and $r$-legs, we also obtain two amplitudes:
\beq
\mathcal{M}_{\mu_r\mu_{l_1}\mu_{l_2}} \epsilon_{l_1}^{-\,\mu_{l_1}} \epsilon_{l_2}^{-\,\mu_{l_2}} \epsilon_{r(2)}^{-\,\mu_r} \ \ \mbox{and}\ \ \mathcal{M}_{\mu_r\mu_{l_1}\mu_{l_2}} \epsilon_{l_1}^{-\,\mu_{l_1}} \epsilon_{l_2}^{-\,\mu_{l_2}} \epsilon_{r(2)}^{\perp\,\mu_r}
\label{decoherence2}
\eeq
that are of $\mathcal{O}(z^{-1})$, and can be reduced to less point amplitudes using BCFW technique. In the four amplitudes of \eref{decoherence1} and \eref{decoherence2}, the vectors $\epsilon_r$ coupling to $r$-leg are correlated with the vectors coupling to $l_1$ or $l_2$, thus we cannot act on $l_1$ or $l_2$ with their little group generators to obtain other components of the amplitudes. However, in four dimensional spacetime, from the four amplitudes\footnote{Actually three amplitudes are enough since the vectors coupling to the r-leg are all transverse to $k_r$.} we can solve out $\mathcal{M}_{\mu_r\mu_{l_1}\mu_{l_2}} \epsilon_{l_1}^{-\,\mu_{l_1}} \epsilon_{l_2}^{-\,\mu_{l_2}} \epsilon_r^{-\,\mu_r}$, where $\epsilon_r^{-\,\mu_r}$ is independent of the vectors $\epsilon_{l_1}$ and $\epsilon_{l_2}$. Then we can act on $\mathcal{M}_{\mu_r\mu_{l_1}\mu_{l_2}} \epsilon_{l_1}^{-\,\mu_{l_1}} \epsilon_{l_2}^{-\,\mu_{l_2}} \epsilon_r^{-\,\mu_r}$ with the little group generators \cite{Chen:2011ve} for $l_1$, $l_2$ and $r$ legs, and get all of $\mathcal{M}_{\mu_r\mu_{l_1}\mu_{l_2}} \epsilon_{l_1}^{i\,\mu_{l_1}} \epsilon_{l_2}^{j\,\mu_{l_2}} \epsilon_r^{k\,\mu_r}$ with $i,j,k\in\{-,\perp,+\}$. Together with the longitudinal components which have been reduced to less point amplitudes in \eref{longitudinalcomp}, we have built up a BCFW recursion relation for general off shell amplitudes.

Several supplements for the above procedure. First, if for some special cases, \eref{decoherence1} and \eref{decoherence2} cannot determine $\mathcal{M}_{\mu_r\mu_{l_1}\mu_{l_2}} \epsilon_{l_1}^{-\,\mu_{l_1}} \epsilon_{l_2}^{-\,\mu_{l_2}} \epsilon_r^{-\,\mu_r}$, we can replace either $l_1-r$ shift or $l_2-r$ shift as in Table \ref{behaviorTable2}. Secondly, when one of the shifted legs is on shell, we can get the $\epsilon^-$ and $\epsilon^+$ components on this on shell line using the above procedure, and the momentum component from \eref{longitudinalcomp}. These components are sufficient for an on shell leg. Thirdly, in the above procedure, we required $l_1$ and $l_2$ both non-adjacent to $r$ leg. Actually for the procedure to work, we only need three amplitudes which can be reduced by BCFW technique, with the fourth amplitude from \eref{longitudinalcomp}. From Table \ref{behaviorTable1} or \ref{behaviorTable2}, we can see that a non-adjacent shift plus an adjacent shift is already enough for the procedure to work, which means that our procedure works from 4 point level.

In conclusion, with the proper boundary behavior to be discussed in the following sections, and using the little group techniques in \cite{Chen3}, BCFW recursion relation can be generalized to calculate general tree level off shell amplitudes. 

\subsection{Applications}

First we show a concrete and simple example of our method of constructing off shell tree level amplitudes, ie. constructing a four point amplitude from three point amplitudes. If one external leg takes its longitudinal component, ie. contracted with its momentum, then the amplitude can be decomposed as in \eref{longitudinalcomp}. We will focus on the transverse components of the amplitude. For that, we will first compute $\mathcal{M}(k_1,k_2,k_3,k_4)^{\mu_1\mu_2\mu_3\mu_4}(\lambda_1\tilde \beta_1)_{\mu_1}$ $(\lambda_2\tilde \beta_2)_{\mu_2} (\lambda_3\tilde \beta_3)_{\mu_3}(\lambda_4\tilde \beta_4)_{\mu_4}$. The time-like or space-like momenta are $k_i=\lambda_i\tilde \lambda_i\pm\beta_i\tilde \beta_i$. Other transverse components of the amplitude can be obtained by little group generators \cite{Chen3}: $\mathcal R^+_i=\tilde \lambda_i \frac{\partial}{\partial \tilde \beta_i}-\beta_i\frac{\partial}{\partial \lambda_i}$, $\mathcal R^-_i=\tilde \beta_i \frac{\partial}{\partial \tilde \lambda_i}-\lambda_i\frac{\partial}{\partial \beta_i}$. We will shift $k_1-k_4$, $k_2-k_4$ and $k_3-k_4$ respectively, so according to our previous discussions we choose $\lambda_i\tilde \beta_i$ to be normal to $k_4$ for $i=1,2,3$. We notate $\epsilon_i=\lambda_i\tilde \beta_i$ for $i=1,2,3,4$.

We first do $k_1-k_4$ shift with $\eta_1=\lambda_1\tilde \beta_1=\epsilon_1$. We define $\epsilon_+=\lambda_1\tilde \beta_2$, $\epsilon_-=\lambda_2\tilde \beta_1$. The metric can be decomposed as $g^{\mu\nu}=\sum\limits_h \epsilon_h^{\mu}\epsilon_{\bar h}^{\nu}=\frac{\epsilon_1^\mu \epsilon_2^\nu+\epsilon_2^\mu \epsilon_1^\nu-\epsilon_+^{\mu} \epsilon_-^{\nu}-\epsilon_-^{\mu} \epsilon_+^{\nu}}{\epsilon_1\cdot \epsilon_2}$, with $\epsilon_h=\{\frac{\epsilon_1}{\sqrt{|\epsilon_1\cdot\epsilon_2|}}$, $\frac{\epsilon_2}{\sqrt{|\epsilon_1\cdot\epsilon_2|}}$, $\frac{\epsilon_+}{\sqrt{|\epsilon_1\cdot\epsilon_2|}}$, $-\frac{\epsilon_-}{\sqrt{|\epsilon_1\cdot\epsilon_2|}}\}$, $\epsilon_{\bar h}=\{\frac{\epsilon_2}{\sqrt{|\epsilon_1\cdot\epsilon_2|}},\frac{\epsilon_1}{\sqrt{|\epsilon_1\cdot\epsilon_2|}}$, $-\frac{\epsilon_-}{\sqrt{|\epsilon_1\cdot\epsilon_2|}}$, $\frac{\epsilon_+}{\sqrt{|\epsilon_1\cdot\epsilon_2|}}\}$. We have explained that $\mathcal{M}(\hat k_1,k_2,k_3,\hat k_4)^{\mu_1\mu_2\mu_3\mu_4}$ $(\lambda_1\tilde \beta_1)_{\mu_1}$ $(\lambda_2\tilde \beta_2)_{\mu_2}$ $(\lambda_3\tilde \beta_3)_{\mu_3}$ $(\lambda_1\tilde \beta_1)_{\mu_4}$ is of $\mathcal{O}(z^{-1})$ and can be decomposed using BCFW technique. For this shift, there is only one pole term, and with a little simplifications we can get at $z=0$:
\bea
&&\mathcal{M}(k_1,k_2,k_3,k_4)^{\mu_1\mu_2\mu_3\mu_4}(\lambda_1\tilde \beta_1)_{\mu_1}(\lambda_2\tilde \beta_2)_{\mu_2}(\lambda_3\tilde \beta_3)_{\mu_3}(\lambda_1\tilde \beta_1)_{\mu_4}\label{14shift}\\
=&&\frac{-i}{(k_1+k_2)^2}\sum_h A^h_L(z_{14}) A^{\bar h}_R(z_{14})\nonumber\\
=&&\frac{i}{2(k_1+k_2)^2}(k_1-k_2)\cdot \epsilon_+\ (2k_3\cdot\epsilon_1\  \epsilon_-\cdot \epsilon_3-(2k_3+k_1+k_2)\cdot\epsilon_-\  \epsilon_1\cdot\epsilon_3)\nonumber\\
&&+\frac{i}{2(k_1+k_2)^2}(k_1-k_2)\cdot \epsilon_-\ (2k_3\cdot\epsilon_1\  \epsilon_+\cdot \epsilon_3-(2k_3+k_1+k_2)\cdot\epsilon_+\  \epsilon_1\cdot\epsilon_3)\nonumber\\
&&-\frac{i}{(k_1+k_2)^2}k_2\cdot\epsilon_1\ (k_3\cdot\epsilon_1 \ \epsilon_2\cdot\epsilon_3-k_3\cdot\epsilon_2\  \epsilon_1\cdot\epsilon_3+z_{14}\eta_1\cdot\epsilon_3\ \epsilon_1\cdot\epsilon_2).\nonumber
\eea
We use the letter A to represent the amplitudes where each leg is contracted with its "polarization vector", to differentiate from the tensor amplitude $\mathcal M$. The pole position $z_{14}=-\frac{(k_1+k_2)^2}{2k_2\cdot \eta_1}$. 



Similarly for $k_3-k_4$ shifting with $\eta_3=\lambda_3\tilde\beta_3=\epsilon_3$ and $z_{34}=-\frac{(k_3+k_2)^2}{2k_2\cdot \eta_3}$, we choose $\tilde\epsilon_h=\{\frac{\epsilon_2}{\sqrt{|\epsilon_2\cdot\epsilon_3|}},\frac{\epsilon_3}{\sqrt{|\epsilon_2\cdot\epsilon_3|}},\frac{\tilde\epsilon_+}{\sqrt{|\epsilon_2\cdot\epsilon_3|}},-\frac{\tilde\epsilon_-}{\sqrt{|\epsilon_2\cdot\epsilon_3|}}\}$, $\tilde\epsilon_{\bar h}=\{\frac{\epsilon_3}{\sqrt{|\epsilon_2\cdot\epsilon_3|}},\frac{\epsilon_2}{\sqrt{|\epsilon_2\cdot\epsilon_3|}},-\frac{\tilde\epsilon_-}{\sqrt{|\epsilon_2\cdot\epsilon_3|}},\frac{\tilde\epsilon_+}{\sqrt{|\epsilon_2\cdot\epsilon_3|}}\}$, with $\tilde\epsilon_+=\lambda_2\tilde \beta_3$ and $\tilde\epsilon_-=\lambda_3\tilde \beta_2$. The metric $g^{\mu\nu}=\sum\limits_h \tilde\epsilon_h^{\mu}\tilde\epsilon_{\bar h}^{\nu}=\frac{\epsilon_2^\mu \epsilon_3^\nu+\epsilon_3^\mu \epsilon_2^\nu-\tilde \epsilon_+^{\mu} \tilde \epsilon_-^{\nu}-\tilde \epsilon_-^{\mu}\tilde \epsilon_+^{\nu}}{\epsilon_2\cdot \epsilon_3}$. We can get at z=0:
\bea
&&\mathcal{M}(k_1,k_2,k_3,k_4)^{\mu_1\mu_2\mu_3\mu_4}(\lambda_1\tilde \beta_1)_{\mu_1}(\lambda_2\tilde \beta_2)_{\mu_2}(\lambda_3\tilde \beta_3)_{\mu_3}(\lambda_3\tilde \beta_3)_{\mu_4}\label{34shift}\\
=&&\frac{-i}{(k_2+k_3)^2}\sum_h A^h_L(z_{34}) A^{\bar h}_R(z_{34})\nonumber\\
=&&-\frac{i}{2(k_2+k_3)^2}(k_2-k_3)\cdot \tilde\epsilon_+\ (2k_1\cdot\epsilon_3\  \tilde\epsilon_-\cdot \epsilon_1-(2k_1+k_2+k_3)\cdot\tilde\epsilon_-\  \epsilon_1\cdot\epsilon_3)\nonumber\\
&&-\frac{i}{2(k_2+k_3)^2}(k_2-k_3)\cdot \tilde\epsilon_-\ (2k_1\cdot\epsilon_3\  \tilde\epsilon_+\cdot \epsilon_1-(2k_1+k_2+k_3)\cdot\tilde\epsilon_+\  \epsilon_1\cdot\epsilon_3)\nonumber\\
&&-\frac{i}{(k_2+k_3)^2}k_2\cdot\epsilon_3\ (k_1\cdot\epsilon_3 \ \epsilon_1\cdot\epsilon_2-k_1\cdot\epsilon_2\  \epsilon_1\cdot\epsilon_3+z_{34}\eta_3\cdot\epsilon_1\ \epsilon_2\cdot\epsilon_3).\nonumber
\eea

For $k_2-k_4$ shifting with $\eta_2=\lambda_2\tilde\beta_2=\epsilon_2$, there are two pole contributions from s and u channels, with $z_{24}^{(s)}=-\frac{(k_1+k_2)^2}{2k_1\cdot \eta_2}$ and $z_{24}^{(u)}=-\frac{(k_2+k_3)^2}{2k_3\cdot \eta_2}$, and we can get at $z=0$:
\bea
&&\mathcal{M}(k_1,k_2,k_3,k_4)^{\mu_1\mu_2\mu_3\mu_4}(\lambda_1\tilde \beta_1)_{\mu_1}(\lambda_2\tilde \beta_2)_{\mu_2}(\lambda_3\tilde \beta_3)_{\mu_3}(\lambda_2\tilde \beta_2)_{\mu_4}\label{24shift}\\
=&&\frac{-i}{(k_1+k_2)^2}\sum_h A^h_L(z_{24}^{(s)}) A^{\bar h}_R(z_{24}^{(s)})+\frac{-i}{(k_2+k_3)^2}\sum_h A^h_L(z_{24}^{(u)}) A^{\bar h}_R(z_{24}^{(u)})\nonumber\\
=&&\frac{i}{2(k_1+k_2)^2}(k_1-k_2)\cdot \epsilon_+\ (2k_3\cdot\epsilon_2\  \epsilon_-\cdot \epsilon_3-(2k_3+k_1+k_2)\cdot\epsilon_-\  \epsilon_2\cdot\epsilon_3)\nonumber\\
&&+\frac{i}{2(k_1+k_2)^2}(k_1-k_2)\cdot \epsilon_-\ (2k_3\cdot\epsilon_2\  \epsilon_+\cdot \epsilon_3-(2k_3+k_1+k_2)\cdot\epsilon_+\  \epsilon_2\cdot\epsilon_3)\nonumber\\
&&-\frac{i}{(k_1+k_2)^2}k_1\cdot\epsilon_2\ (k_3\cdot\epsilon_1 \ \epsilon_2\cdot\epsilon_3-k_3\cdot\epsilon_2\  \epsilon_1\cdot\epsilon_3-z_{24}^{(s)}\eta_2\cdot\epsilon_3\ \epsilon_1\cdot\epsilon_2)\nonumber\\
&&-\frac{i}{2(k_2+k_3)^2}(k_2-k_3)\cdot \tilde\epsilon_+\ (2k_1\cdot\epsilon_2\  \tilde\epsilon_-\cdot \epsilon_1-(2k_1+k_2+k_3)\cdot\tilde\epsilon_-\  \epsilon_2\cdot\epsilon_3)\nonumber\\
&&-\frac{i}{2(k_2+k_3)^2}(k_2-k_3)\cdot \tilde\epsilon_-\ (2k_1\cdot\epsilon_2\  \tilde\epsilon_+\cdot \epsilon_1-(2k_1+k_2+k_3)\cdot\tilde\epsilon_+\  \epsilon_2\cdot\epsilon_3)\nonumber\\
&&-\frac{i}{(k_2+k_3)^2}k_3\cdot\epsilon_2\ (k_1\cdot\epsilon_3 \ \epsilon_1\cdot\epsilon_2-k_1\cdot\epsilon_2\  \epsilon_1\cdot\epsilon_3-z_{24}^{(u)}\eta_2\cdot\epsilon_1\ \epsilon_2\cdot\epsilon_3)\nonumber
\eea

Since $\epsilon_i$ have all been chosen to be normal to $k_4$ for $i=1,2,3,4$, we can express $\epsilon_4$ linearly in $\epsilon_{\{1,2,3\}}$:
\bea
&&\epsilon_4=c_1 \epsilon_1+c_2 \epsilon_2+c_3 \epsilon_3,\nonumber\\
&&c_1=\frac{\langle \lambda_2\lambda_4\rangle [\tilde \beta_3 \tilde \beta_4]}{\langle \lambda_1\lambda_2\rangle [\tilde \beta_1 \tilde \beta_3]}, c_2=-\frac{\langle \lambda_1\lambda_4\rangle [\tilde \beta_3 \tilde \beta_4]}{\langle \lambda_1\lambda_2\rangle [\tilde \beta_2 \tilde \beta_3]}, c_3=\frac{\langle \lambda_1\lambda_4\rangle [\tilde \beta_2 \tilde \beta_4]}{\langle \lambda_1\lambda_3\rangle [\tilde \beta_2 \tilde \beta_3]}.
\eea
Then we can obtain $\mathcal{M}(k_1,k_2,k_3,k_4)^{\mu_1\mu_2\mu_3\mu_4}(\lambda_1\tilde \beta_1)_{\mu_1}(\lambda_2\tilde \beta_2)_{\mu_2}(\lambda_3\tilde \beta_3)_{\mu_3}(\lambda_4\tilde \beta_4)_{\mu_4}$ from a linear combination of \eref{14shift}, \eref{24shift} and \eref{34shift} with coefficients $c_1, c_2, c_3$, and we have the final expression:
\bea
&&\mathcal{M}(k_1,k_2,k_3,k_4)^{\mu_1\mu_2\mu_3\mu_4}(\lambda_1\tilde \beta_1)_{\mu_1}(\lambda_2\tilde \beta_2)_{\mu_2}(\lambda_3\tilde \beta_3)_{\mu_3}(\lambda_4\tilde \beta_4)_{\mu_4}\label{offshell4p}\\
=&&\frac{i}{2(k_1+k_2)^2}(k_1-k_2)\cdot \epsilon_+\ (2k_3\cdot\epsilon_4\  \epsilon_-\cdot \epsilon_3-(2k_3+k_1+k_2)\cdot\epsilon_-\  \epsilon_4\cdot\epsilon_3)\nonumber\\
&&+\frac{i}{2(k_1+k_2)^2}(k_1-k_2)\cdot \epsilon_-\ (2k_3\cdot\epsilon_4\  \epsilon_+\cdot \epsilon_3-(2k_3+k_1+k_2)\cdot\epsilon_+\  \epsilon_4\cdot\epsilon_3)\nonumber\\
&&-\frac{i}{2(k_2+k_3)^2}(k_2-k_3)\cdot \tilde\epsilon_+\ (2k_1\cdot\epsilon_4\  \tilde\epsilon_-\cdot \epsilon_1-(2k_1+k_2+k_3)\cdot\tilde\epsilon_-\  \epsilon_4\cdot\epsilon_3)\nonumber\\
&&-\frac{i}{2(k_2+k_3)^2}(k_2-k_3)\cdot \tilde\epsilon_-\ (2k_1\cdot\epsilon_4\  \tilde\epsilon_+\cdot \epsilon_1-(2k_1+k_2+k_3)\cdot\tilde\epsilon_+\  \epsilon_4\cdot\epsilon_3)\nonumber\\
&&+\frac{i}{(k_1+k_2)^2}k_3\cdot\epsilon_4\ (k_3\cdot\epsilon_1 \ \epsilon_2\cdot\epsilon_3-k_3\cdot\epsilon_2\  \epsilon_1\cdot\epsilon_3)\nonumber\\
&&+\frac{i}{(k_2+k_3)^2}k_1\cdot\epsilon_4\ (k_1\cdot\epsilon_3 \ \epsilon_1\cdot\epsilon_2-k_1\cdot\epsilon_2\  \epsilon_1\cdot\epsilon_3)\nonumber\\
&&+\frac{i}{2}(2\epsilon_1\cdot\epsilon_3 \ \epsilon_2\cdot\epsilon_4-\epsilon_1\cdot\epsilon_2 \ \epsilon_3\cdot\epsilon_4-\epsilon_1\cdot\epsilon_4 \ \epsilon_2\cdot\epsilon_3).\nonumber
\eea

Besides the direct application of our method in calculating off shell amplitudes, the method can be used to formally analyze some amplitude relations. For example, tree level on shell KK relations \cite{Kleiss:1988ne} can be easily proved by BCFW technique \cite{Feng:2010my}. The key point is that by BCFW, the amplitude is decomposed to two less point amplitudes and for each sub amplitude the induction can be directly applied. We shall show that by our method, the on shell version of proof can be easily extended to the proof of off shell tree level KK relations, which have recently been proved using Berends-Giele recursion relation \cite{Berends:1987me} in \cite{Fu:2012uy}. Instead of proving the KK relation directly, we shall prove the generalized U(1) decoupling property, which is equivalent to KK relations \cite{Ma:2011um}. The generalized U(1) decoupling is expressed as:
\beqs
\sum\limits_{OP(\{\alpha\}\cup \{\beta\})}{\mathcal M}(1,OP(\{\alpha\}\cup \{\beta\}))=0,
\eeqs
where OP means preserving the ordering of legs in set $\alpha$ and set $\beta$.

We use induction to prove this identity, and it is easily verified that the identity holds for three and four point off shell tree level amplitudes.

First we consider the longitudinal components, ie. some legs are contracted with their momenta. In this case, the amplitude can be decomposed using \eref{longitudinalcomp}. Since each term of \eref{longitudinalcomp} is a product of two less point amplitudes, induction is easily realized. If the leg 1 is contracted with $k_1$, then $\{\alpha_i,\beta_j\}$ are divided into two groups, say $G_L$ and $G_R$. If in $G_L$ or $G_R$, there are legs from the sets $\{\alpha_i\}$ and $\{\beta_j\}$ simultaneously, the identity holds by induction. The rest case is when $G_L$ and $G_R$ contain only legs in $\{\alpha_i\}$ or $\{\beta_j\}$. The two contributions from ${\mathcal M}(1,\{\alpha\},\{\beta\})$ and ${\mathcal M}(1,\{\beta\},\{\alpha\})$ are obviously opposite to each other. Thus the induction works when leg 1 is contracted with $k_1$. If the leg $\alpha_{i_0}$ (similarly for $\beta_{j_0}$) is contracted with $k_{\alpha_{i_0}}$, other legs are again divided into two groups $G_L$ and $G_R$. Assume the leg 1 is in $G_L$. The amplitude is $\mathcal M(1\in G_L, \alpha_{i_0}, G_R)$. If $G_R$ contains legs from the sets $\{\alpha_i\}$ and $\{\beta_j\}$ simultaneously, using induction we can obtain the generalized U(1) decoupling identity. If $G_R$ only contains legs in $\{\beta_j\}$, notated as $\{\beta_{j_R}\}$, these terms cancel those with the ordering of $\alpha_{i_0}$ and $\{\beta_{j_R}\}$ interchanged, ie. $\mathcal M(1\in G_L, G_R=\{\beta_{j_R}\}, \alpha_{i_0})$. Finally, if $G_R$ only contains legs in $\{\alpha_i\}$, notated as $\{\alpha_{i_R}\}$, by viewing $\alpha_{i_0}$ and $\{\alpha_{i_R}\}$ together as one leg $\alpha_{\tilde i_0}$, the situation is reduced to a less point amplitude and induction can be applied.

Then we investigate the transverse components. As done in our previous example, we choose one leg  and three other legs to shift respectively. The common leg can be chosen as leg 1, with no difference for choosing other legs. All other legs are contracted with $\epsilon_i$ normal to both $k_i$ and $k_1$. For the shift $k_{s_1}-k_1$ with $\eta_{s_1}=\epsilon_{s_1}$, regardless of being adjacent or non-adjacent, we know that $\hat{\mathcal M}^{\mu_1\mu_{i_1}\mu_{i_2}\cdots\mu_{i_n}}(\eta_{s_1})_{\mu_1}(\epsilon_{i_1})_{\mu_{i_1}}(\epsilon_{i_2})_{\mu_{i_2}}\cdots(\epsilon_{i_n})_{\mu_{i_n}}$ are of $\mathcal O(z^{-1})$ and can be decomposed using BCFW technique. Then the proof of on shell KK relations in \cite{Feng:2010my} can be directly used here to show that at $z=0$ $\sum\limits_{OP(\{\alpha\}\cup \{\beta\})}{\mathcal M}(1,OP(\{\alpha\}\cup \{\beta\}))^{\mu_1\mu_{i_1}\mu_{i_2}\cdots\mu_{i_n}}$ $(\eta_{s_1})_{\mu_1}$ $(\epsilon_{i_1})_{\mu_{i_1}}$ $(\epsilon_{i_2})_{\mu_{i_2}}$ $\cdots$ $(\epsilon_{i_n})_{\mu_{i_n}}=0$. Similarly for the second and third shifts we have $\sum\limits_{OP(\{\alpha\}\cup \{\beta\})}{\mathcal M}(1,OP(\{\alpha\}\cup \{\beta\}))^{\mu_1\mu_{i_1} \mu_{i_2}\cdots\mu_{i_n}}$ $(\eta_{s_2})_{\mu_1}$ $(\epsilon_{i_1})_{\mu_{i_1}}$ $(\epsilon_{i_2})_{\mu_{i_2}}$ $\cdots(\epsilon_{i_n})_{\mu_{i_n}}=0$ and $\sum\limits_{OP(\{\alpha\}\cup \{\beta\})}{\mathcal M}(1,OP(\{\alpha\}\cup \{\beta\}))^{\mu_1\mu_{i_1}\mu_{i_2}\cdots\mu_{i_n}}$ $(\eta_{s_3})_{\mu_1}$ $(\epsilon_{i_1})_{\mu_{i_1}}$ $(\epsilon_{i_2})_{\mu_{i_2}}$ $\cdots(\epsilon_{i_n})_{\mu_{i_n}}=0$. By a linear combination we have $\sum\limits_{OP(\{\alpha\}\cup \{\beta\})}{\mathcal M}(1,OP(\{\alpha\}\cup \{\beta\}))^{\mu_1\mu_{i_1}\mu_{i_2}\cdots\mu_{i_n}}$ $(\epsilon_1)_{\mu_1}$ $(\epsilon_{i_1})_{\mu_{i_1}}$ $(\epsilon_{i_2})_{\mu_{i_2}}\cdots (\epsilon_{i_n})_{\mu_{i_n}}=0$, and by little group generators we see that all the transverse components of the amplitudes satisfy the generalized U(1) decoupling identity.

Thus we have proven off shell KK relations based on the on shell proof and our generalized BCFW technique of constructing off shell tree level amplitudes.

Based on the $\mathcal{O}(z^{-2})$ behavior of on shell gauge field amplitudes under non adjacent shifts, the BCJ relations \cite{Bern:2008qj} are proved in \cite{Feng:2010my}. As discussed in our \sref{order1overz}, we have $\mathcal{M}(\hat k_1,k_2,\hat k_4,k_3)^{\mu_1\mu_2\mu_4\mu_3}$ $(\lambda_1\tilde \beta_1)_{\mu_1}$ $(\lambda_1\tilde \beta_1)_{\mu_4}\to \mathcal{O}(z^{-2})$ at large $z$ for $k_1-k_4$ shifting. Then with the same arguments as in \cite{Feng:2010my}, we have the "BCJ-like" relations, as done in \sref{amprel} for some components of four point off shell amplitudes:
$$(s_{12}\mathcal{M}(k_1,k_2,k_3,k_4)^{\mu_1\mu_2\mu_3\mu_4}-s_{13}\mathcal{M}(k_1,k_3,k_2,k_4)^{\mu_1\mu_3\mu_2\mu_4})(\lambda_1\tilde \beta_1)_{\mu_1}(\lambda_1\tilde \beta_1)_{\mu_4}=0.$$
Similarly we have by $k_2-k_4$ shifting:
$$(s_{21}\mathcal{M}(k_2,k_1,k_3,k_4)^{\mu_2\mu_1\mu_3\mu_4}-s_{23}\mathcal{M}(k_2,k_3,k_1,k_4)^{\mu_2\mu_3\mu_1\mu_4})(\lambda_2\tilde \beta_2)_{\mu_2}(\lambda_2\tilde \beta_2)_{\mu_4}=0,$$
and some other similar expressions.

For more point amplitudes, we do not have the $\mathcal{O}(z^{-2})$ behavior for non adjacent shifts in general. Although by careful study of the $\mathcal{O}(z^{-1})$ terms it is possible to prove BCJ relations in some theories, like the color scalar theory \cite{Du:2011js}, we have not been able to obtain similar BCJ relations for general more point off shell YM amplitudes. We leave this direction to future work.

\section{Amplitudes with Reduced Vertexes}\label{Sec:Reduce}

In this section we are going to introduce some reduced vertexes, and prove that amplitudes constructed from the reduced vertexes have the same boundary behavior as those constructed from ordinary vertexes.


We first clarify some conventions for the rest of this chapter. If we draw the complex momentum line from left to right, other external legs besides the shifted pair would be either above or below this complex line.  For a given shift, the set of external legs above (or below) the complex line is fixed together with their order, however the legs above the complex line and those below it can have all possible relative positions. To further specify the vertexes, we sort the vertexes as in Figure \ref{vertexclass}.
\begin{figure}[htb]
\centering
\includegraphics{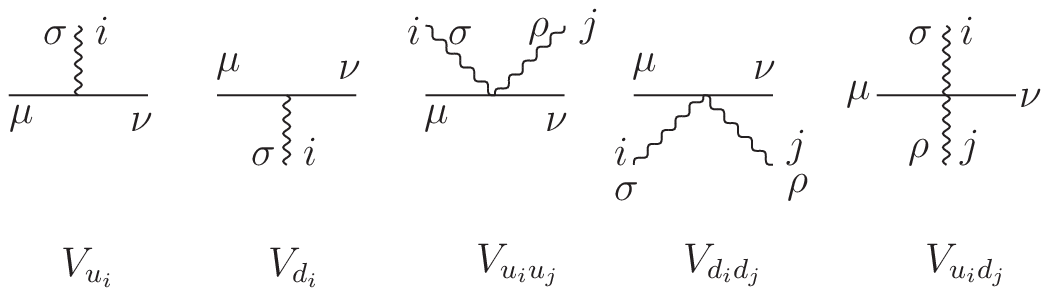}
\caption{A classification of the vertexes. The horizontal line is the complex line and the photon lines are external legs besides the shifted pair. Letters $i$ and $j$ are the numbering of the external legs, and $\mu$, $\nu$, $\sigma$, $\rho$ are the indices.}
\label{vertexclass}
\end{figure}

For a three-point vertex with legs 1, 2 and 3 in anti-clockwise order, we write it in the following form:
\bea\label{newV3c4}
V_{\mu_1\mu_2\mu_3}&\equiv&S_{\mu_1\mu_2\mu_3}+ R_{\mu_1\mu_2\mu_3}+M_{\mu_1\mu_2\mu_3},
\eea
where  
\bea\label{newV3s4}
S_{\mu_1\mu_2\mu_3}&=&\frac{i}{\sqrt 2}\left(g_{\mu_1\mu_2}(k_1-k_2)_{\mu_3}\right) \nb\\
R_{\mu_1\mu_2\mu_3}&=&\frac{i}{\sqrt 2}\left(-2g_{\mu_2\mu_3}(k_3)_{\mu_1}+2g_{\mu_3\mu_1}(k_3)_{\mu_2}\right) \nb\\
M_{\mu_1\mu_2\mu_3}&=&\frac{i}{\sqrt 2}\left(-g_{\mu_2\mu_3}(k_1)_{\mu_1}+g_{\mu_3\mu_1}(k_2)_{\mu_2}\right). 
\eea
In this manner, $k_3$ is in a special role and we will choose the appropriate one as $k_3$ in specific situations. When the legs 1 and 2 are on the complex line and 3 is an external leg, we further divide the M term into $M^L$ and $M^R$ as represented in Figure \ref{Msymbol}.
\begin{figure}[]
\centering
\includegraphics[height=4cm,width=14cm]{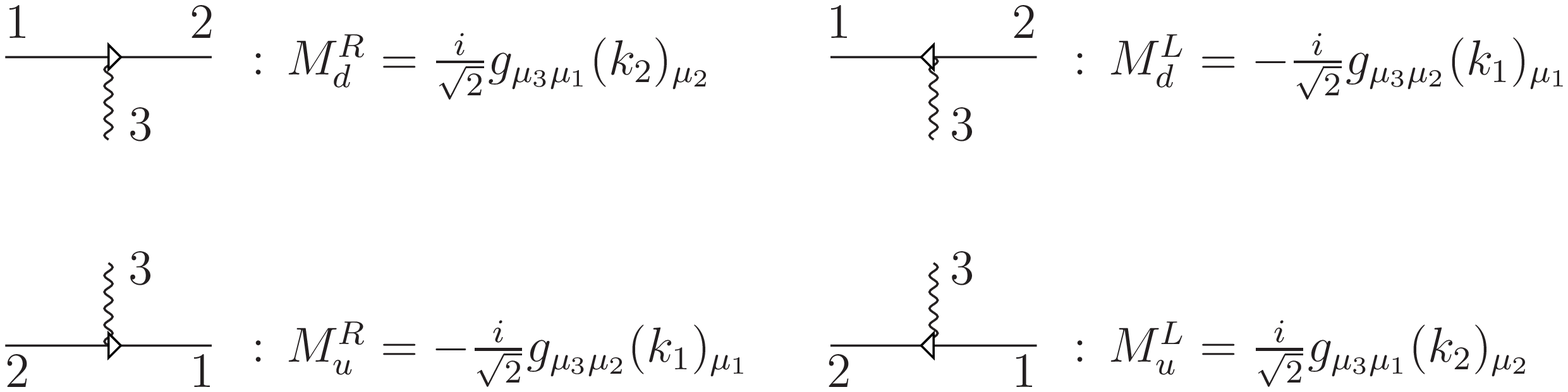}
\caption{The symbols and meanings of $M_{u/d}^{L/R}$.}
\label{Msymbol}
\end{figure}

Contracting a three point vertex $V_{\mu_1\mu_2\mu_3}$ with $k_3^{\mu_3}$, we get:
\beq
k_3^{\mu_3} \cdot V_{\mu_1\mu_2\mu_3}=\frac{i}{\sqrt{2}}g_{\mu_1\mu_2} k_2^2-\frac{i}{\sqrt{2}}g_{\mu_1\mu_2} k_1^2-\frac{i}{\sqrt{2}}k_{2\ \mu_2}k_2{}_{\ \mu_1}+\frac{i}{\sqrt{2}}k_{1\ \mu_1}k_1{}_{\ \mu_2},
\label{kdotV4}
\eeq
and we represent these terms by the symbols in Figure \ref{vertexnotation4}.
\begin{figure}[htb]
\centering
\includegraphics[height=4cm,width=10cm]{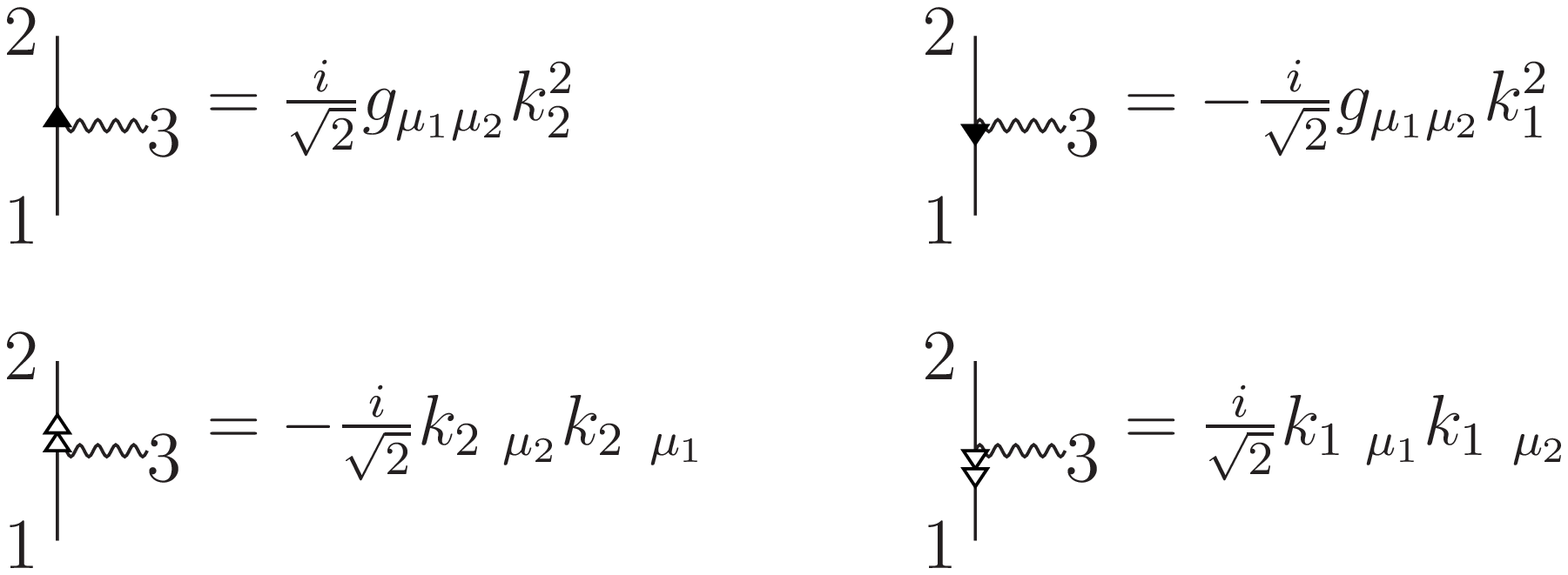}
\caption{Notations for \eref{kdotV4}.}
\label{vertexnotation4}
\end{figure}

In the following of this chapter, the method of induction is assumed. For example, when we discuss the $\mathcal{O}(z^1)$, $\mathcal{O}(z^0)$ and $\mathcal{O}(z^{-1})$ behavior of N point amplitudes, we only need to consider the diagrams with all the external legs attaching the complex line. When some of these external legs form vertexes outside the complex line, the conclusions for less point amplitudes apply to these diagrams as we do not require the external legs to be on shell.


\subsection{Reduced Vertexes}
The central conclusion of this subsection is that the boundary behavior of amplitudes with BCFW momenta shift \eref{momshift4} under the conditions \eref{conditionEta4} and \eref{conditionVector} can be obtained by using the reduced vertexes as following:
\bea
{\bar V}_{u/d}&=&S_{u/d}+R_{u/d},\nonumber\\
{\bar V}_{u_i u_j}&=&\frac{i}{2}(2 g^{\nu\sigma}g^{\mu\rho}-2g^{\mu\sigma}g^{\nu\rho}-g^{\sigma\rho}g^{\mu\nu}),\nonumber\\
{\bar V}_{d_i d_j}&=&\frac{i}{2}(2g^{\nu\sigma}g^{\mu\rho}-2g^{\mu\sigma}g^{\nu\rho}-g^{\sigma\rho}g^{\mu\nu}),\nonumber\\
{\bar V}_{u_i d_j}&=&ig^{\sigma\rho}g^{\mu\nu}.
\label{ReduV}
\eea
The meanings of the vertex names, the external legs and their indices refer to Figure \ref{vertexclass}, and the meanings of S term and R term in the first line refer to \eref{newV3s4} with the external leg playing the role of leg 3 and legs 1 and 2 on the complex line.

We first prove some useful lemmas. First, for a tree level tensor current $\mathcal{M}_{12\cdots N}^{\mu_1\mu_2\cdots \mu_N}$, we shift $k_i$ and $k_j$: $\hat k_i\to k_i+z\eta$ and $\hat k_j\to k_j-z\eta$ with $\eta^2=0$ and $k_i\cdot \eta=0$. We couple $\hat \epsilon_i$ to the $\hat k_i$ leg with $\hat k_i \cdot \hat \epsilon_i=0$. If $\hat \epsilon_i\sim \mathcal{O}(z^{n_i})$, naive power counting gives ${\hat k}_{j\ \mu_j} \hat{\mathcal{M}}_{12\cdots N}^{\mu_1\mu_2\cdots \mu_N} \hat \epsilon_{i\ \mu_i}\sim \mathcal{O}(z^{2+n_i})$. However, we have:

\begin{lemma}{Generalized Ward Identity 1}
\\${\hat k}_{j\ \mu_j} \hat{\mathcal{M}}_{12\cdots N}^{\mu_1\mu_2\cdots \mu_N} \hat \epsilon_{i\ \mu_i}\sim \mathcal{O}(z^{n_i})$, for $\hat k_i\to k_i+z\eta$ and $\hat k_j\to k_j-z\eta$ with $\eta^2=0$, $k_i\cdot \eta=0$ and $\hat k_i \cdot \hat \epsilon_i=0$.
\label{GWI1}
\end{lemma}
\begin{proof}
The proof can be done by induction, similar to the proof of actual tree-level Ward identity in \cite{Chen1,Chen2}. For three point tensor currents this Lemma can be verified directly. Assume it holds for no more than N point tensor currents and we assume $j=N+1$. We construct an (N+1) point tensor current by inserting $(N+1)$-th leg into an N-point one.

When $(N+1)$-th leg is inserted into a propagator or external leg to form a three vertex $V_j$, we use the notations in Figure \ref{vertexnotation4} to decompose $k_j\cdot V_j$. Among the four terms, the first line two terms, ie. solid triangle terms, plus the terms from $k_j$ inserted to a three point vertex in the N point diagram cancel as in Figure \ref{treecancel2c4}.
\begin{figure}[]
\centering
\includegraphics{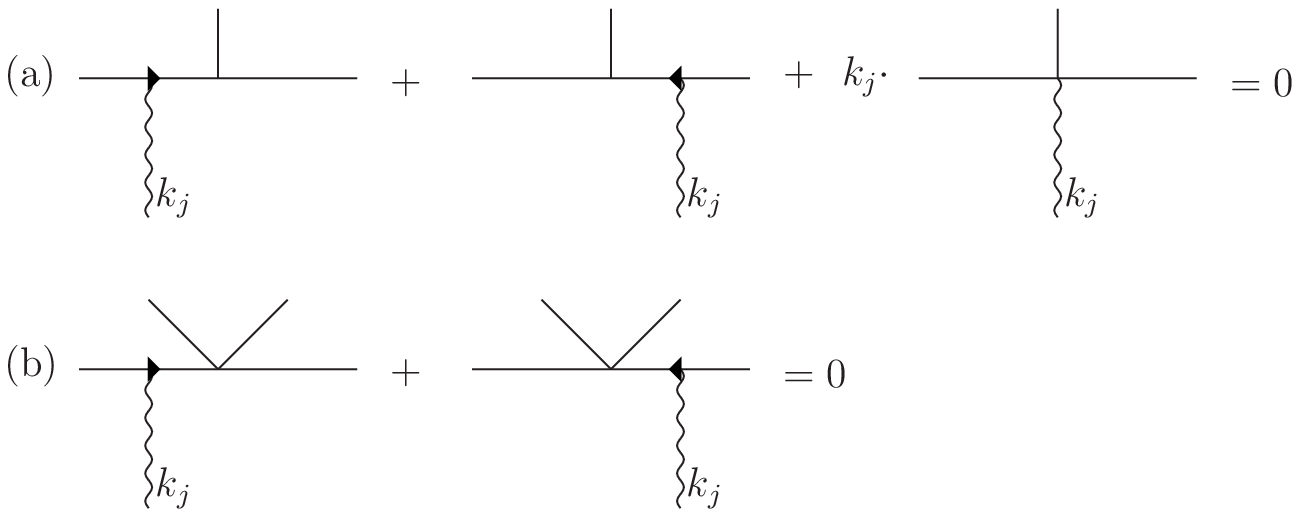}
\caption{Groups of terms that cancel. These terms cancel solely due to the vertex, without any on shell conditions on the legs. Only the $k_j$ leg as in \lref{GWI1} are specially represented by photon line. This Figure is from \cite{Chen2}.} 
\label{treecancel2c4}
\end{figure}

Then the remaining terms are the second line double hollow triangle terms in Figure \ref{vertexnotation4} when $(N+1)$-th leg is inserted to a propagator or external leg in the original N point diagram. Then by direct power counting or the use of the induction assumption, it is seen that the order of $z$ are decreased by at least 2. Thus, we have proven that for N+1 point amplitude, the order of $z$ for ${\hat k}_{j\ \mu_j} \hat{\mathcal {M}}_{12\cdots N}^{\mu_1\mu_2\cdots \mu_N} \hat \epsilon_{i\ \mu_i}$ are decreased by at least 2 compared to naive power counting, finishing the proof for \lref{GWI1}. \endofproof
\end{proof}

\begin{lemma}{Generalized Ward Identity 2}
\\${\hat k}_{j\ \mu_j} {\hat k}_{i\ \mu_i} \mathcal{M}_{12\cdots N}^{\mu_1\mu_2\cdots \mu_N}\sim \mathcal{O}(z^1)$, for a shift: $\hat k_i\to k_i+z\eta$ and $\hat k_j\to k_j-z\eta$ with $\eta^2=0$.
\label{GWI2}
\end{lemma}
In this Lemma, no on shell condition is placed on leg $i$ or $j$. By naive power counting, ${\hat k}_{j\ \mu_j} {\hat k}_{i\ \mu_i} \mathcal{M}_{12\cdots N}^{\mu_1\mu_2\cdots \mu_N} \sim \mathcal{O}(z^3)$, yet actually decreased by 2 orders of z. This Lemma can also be proved by induction with the same procedure as the proof for the previous Lemma.

With the above two Lemmas, we are ready to prove our central conclusion \tref{redamplitude} of this section.

For each diagram the vertexes in it are \{$V_{u_i}$, $V_{d_j}$, $V_{u_i u_{i+1}}$, $V_{d_j d_{j+1}}$, $V_{u_i d_j}$\}, determined by the different orderings of the external legs. We denote this diagram as $\mathcal{M}^{\mu\nu}(\{u_i,d_j,$ $u_i u_{i+1},$ $d_j d_{j+1},$ $u_i d_j\})$ where $\mu$ and $\nu$ are the indices of the shifted legs. In the rest of the chapter, and also for \eref{behaviorOff}, when we talk about $\hat{\mathcal{M}}^{\mu\nu}(\{u_i,d_j, u_i u_{i+1}, d_j d_{j+1}, u_i d_j\})$ with $\mu$ and $\nu$ indices not contracted, we will always assume it contracted with $\hat \epsilon^\mu_l\sim \mathcal{O}(z^{n_l})$ and $\hat \epsilon^\nu_r\sim \mathcal{O}(z^{n_r})$, which satisfy $\hat k_l \cdot \hat \epsilon_l=0$ and $\hat k_r \cdot \hat \epsilon_r=0$, and we will not write $\hat \epsilon^\mu_l$ and $\hat \epsilon^\nu_r$, and suppress $n_l+n_r$ in the order $z$ analysis of the amplitudes.

\begin{theorem} 
For the shift of a pair of momenta $\hat k_l^\mu=k_l^\mu+z \eta^\mu$ and $\hat k_r^\nu=k_r^\nu-z\eta^\nu$, the amplitude at large $z$ has the property:
\beq
\hat{\mathcal{M}}^{\mu\nu}(\{u_i,d_j, u_i u_{i+1}, d_j d_{j+1}, u_i d_j\})=\hat{\mathcal{M}}^{\mu\nu}(\{\overline{u_i},\overline{d_j}, \overline{u_i u_{i+1}}, \overline{d_j d_{j+1}}, \overline{u_i d_j}\})+\mathcal{O}(z^{-1}).
\eeq
\label{redamplitude}
\end{theorem}

$(\{\overline{u_i},\overline{d_j}, \overline{u_i u_{i+1}}, \overline{d_j d_{j+1}}, \overline{u_i d_j}\})$ means that the vertexes are the reduced vertexes $(\{\bar V_{u_i},\bar V_{d_j},$ $\bar V_{u_i u_{i+1}},$ $\bar V_{d_j d_{j+1}},$ $\bar V_{u_i d_j}\})$, see \eref{ReduV}. The highest possible scaling behavior for $\hat{\mathcal{M}}^{\mu\nu}(\{u_i, d_j,$ $u_i u_{i+1},$ $d_j d_{j+1},$ $u_i d_j\})$ is $\mathcal{O}(z^1)$, and this theorem says that the first two orders can be calculated using the reduced vertexes.\footnote{Following the proof of this theorem, it can be seen that for on shell amplitudes, the reduced vertexes give exactly the same amplitudes as original vertexes.}

\begin{proof}
{\bf Step 1.} We notate a diagram by the positions of the vertexes from left to right on the complex line. Using $V_{u/d}=\bar V_{u/d}+M_{u/d}^L+M_{u/d}^R$, we have:
\bea
&&{\bar V}_{(u/d)_1} {\bar V}_{(u/d)_2}\cdots{\bar V}_{(u/d)_n}\\
=&&(V_{(u/d)_1}-M_{(u/d)_1}^L-M_{(u/d)_1}^R)(V_{(u/d)_2}-M_{(u/d)_2}^L-M_{(u/d)_2}^R)\cdots(V_{(u/d)_n}-M_{(u/d)_n}^L-M_{(u/d)_n}^R),\nonumber
\eea
and by expanding it we get:
\bea
&&V_{(u/d)_1} V_{(u/d)_2}\cdots V_{(u/d)_n}\nonumber\\
=&&{\bar V}_{(u/d)_1} {\bar V}_{(u/d)_2}\cdots{\bar V}_{(u/d)_n}\nonumber\\
+&&\sum_i (-1)^{i-1} \mbox{(i vertexes are replaced with their M term components}).
\label{expansion}
\eea
For diagrams containing four point vertexes, we only re-express the three point vertexes therein without any change to four point vertexes at this step, and then do the similar expansion as in \eref{expansion}.

{\bf Step 2.} In this step, we prove that for each term in \eref{expansion}, in order to contribute at $\mathcal{O}(z^1)$ and $\mathcal{O}(z^0)$, the last M factor in the term should be $M_{u/d}^L$, and for the same reason the first M factor should be $M_{u/d}^R$. This is clearly shown in (a) in Figure \ref{Mtermcancel}. 

\begin{figure}[]
\includegraphics[height=12cm,width=12cm]{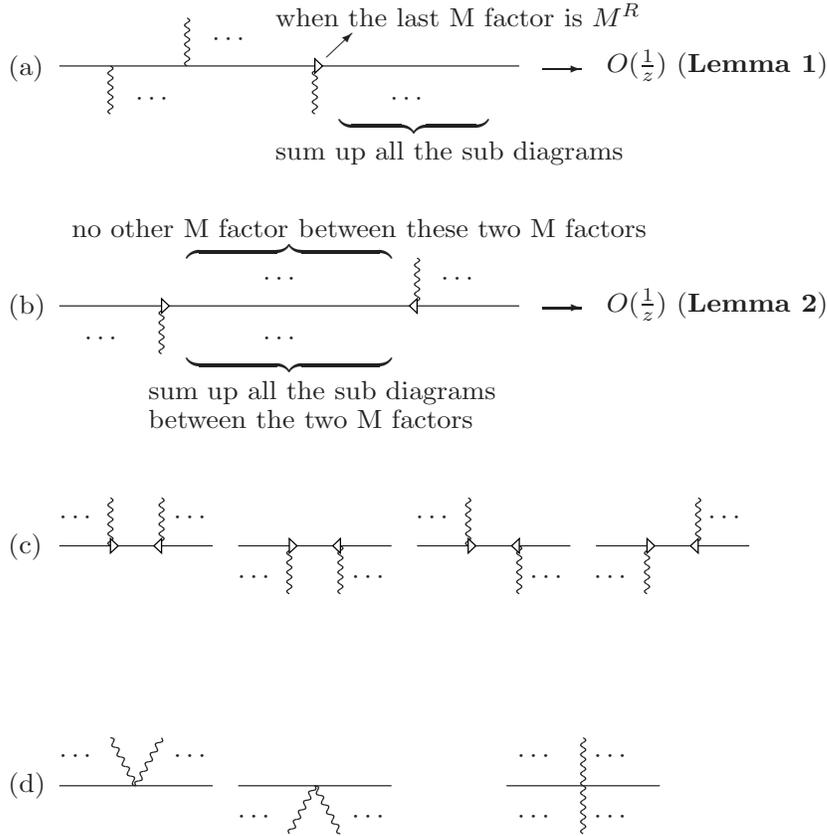}~~
\caption{In (a), naive power counting tells us that any diagram is $\mathcal{O}(z^{1})$, yet the sum turns out to be $\mathcal{O}(z^{-1})$ since \lref{GWI1} tells us that for this sum the actual $z$ dependence is at least lowered by 2 orders compared to naive power counting. The same manner works for the case when the last M factor is $M_u^R$ instead of $M_d^R$, and also works for the analysis of the first M factor in the terms of \eref{expansion}. In (b) \lref{GWI2} works when and only when there are some vertexes between the two M factors. When the two M factors are next to each other, the contributions are shown in (c), which escape \lref{GWI2} and may contribute to $\mathcal{O}(z^0)$. The vertexes besides these two M terms are summed up to be reduced three point vertexes as explained in {\bf Step 3}. (d) gives the corresponding terms with four point vertexes which add up with (c) to replace the four point vertexes with the reduced ones.}
\label{Mtermcancel}
\end{figure}

{\bf Step 3.} From step 2 we know that for contributions at $\mathcal{O}(z^{1})$ and $\mathcal{O}(z^0)$, for all the terms containing at least one M factor in \eref{expansion}, we only need to consider the terms where the last M factor is $M_{u/d}^L$ and the first M factor is $M_{u/d}^R$. For such terms, there clearly exists a pair of M factors $M_{u/d}^R$ and $M_{u/d}^L$, where $M_{u/d}^R$ is on the left of $M_{u/d}^L$ and there are no other M factors between them. This is represented in (b) of Figure \ref{Mtermcancel}. Due to \lref{GWI2} these terms do not contribute to $\mathcal{O}(z^{1})$ and $\mathcal{O}(z^0)$, except the special terms represented in (c) of Figure \ref{Mtermcancel}, where the $M_{u/d}^R$ and $M_{u/d}^L$ are next to each other. For the terms in (c), since the product of the two M terms decrease the order of $z$ by 1, there can be no other four point vertexes at the two sides of the two M terms, in order to contribute at $\mathcal{O}(z^{1})$ and $\mathcal{O}(z^0)$. The terms in (c) add up to be:
\beq
{\bar V}_{(u/d)_1} {\bar V}_{(u/d)_2}\cdots{\bar V}_{(u/d)_i} M_{(u/d)_{i+1}}^R M_{(u/d)_{i+2}}^L{\bar V}_{(u/d)_{i+3}}\cdots{\bar V}_{(u/d)_n},
\eeq
which means that on the two sides of the two M terms all the vertexes are the reduced three point vertexes \eref{ReduV}.

{\bf Step 4.} In the first 3 steps, we have analyzed the terms in \eref{expansion} with at least one M factor, which are reduced to the terms in (c) of Figure \ref{Mtermcancel}. The other terms in \eref{expansion} are either all comprised of reduced three point vertexes, or of reduced three point vertexes plus one and only one four point vertex. The latter case is given in (d) of Figure \ref{Mtermcancel}. (c) and (d) sum up to replace the four point vertex with the reduced one. Thus, we have shown that at $\mathcal{O}(z^{1})$ and $\mathcal{O}(z^0)$, all the terms in \eref{expansion} are reduced either to a product of reduced three point vertexes, or a product of reduced three point vertexes and one reduced four point vertex.\endofproof
\end{proof}

\subsection{Application}\label{theorem1Application}
As a simple application of \tref{redamplitude}, we can directly obtain the large-z scaling behavior for amplitudes with adjacent BCFW shifts. 

For $\mathcal{M}^{\mu\nu}(\{u_i,$ $d_j,$ $u_i u_{i+1},$ $d_j d_{j+1},$ $u_i d_j\})$, we denote the product of all the vertexes in it as $\mathcal{N}^{\mu\nu}(\{u_i,$ $d_j,$ $u_i u_{i+1},$ $d_j d_{j+1},$ $u_i d_j\})$, and the product of all the propagators in the complex line in it as $\mathcal{C}(\{u_i,$ $d_j,$ $u_i u_{i+1},$ $d_j d_{j+1},$ $u_i d_j\})$.  Here and following, we usually suppress $(\{u_i,$ $d_j,$ $u_i u_{i+1},$ $d_j d_{j+1},$ $u_i d_j\})$ for convenience. 
Then  the amplitude is written as 
\bea\label{infyTerm}
\mathcal{M}^{\mu\nu}=\sum_{\mathcal{D}}{\mathcal{N}^{\mu\nu}\over \mathcal{C}} , 
\eea
where the sum is over all the Feynman diagrams.
 
The amplitude can be expanded as $\mathcal{M}^{\mu\nu}=\mathcal{M}_1^{\mu\nu} z + \mathcal{M}_0^{\mu\nu}+\mathcal{M}^{\mu\nu}_{-1} {1\over z}+ \mathcal{O}({1\over z^2})$ in the large $z$ limit. We need to discuss the large-z scaling behavior for some types of Feynman diagrams. For convenience we denote the types of Feynman diagrams as following: $\mathcal{D}_I$ denotes the diagrams where all vertexes in the complex line are reduced three point vertexes. $\mathcal{D}_{II}$ denotes the diagrams where the complex line contains only one reduced four point vertex which is not $\bar V_{u_id_j}$ and other vertexes are reduced three point vertexes, while in $\mathcal{D'}_{II}$ the four point vertex is $\bar V_{u_id_j}$. In $\mathcal{D}_{III}$, there are two reduced four point vertexes in the complex line neither of which is $\bar V_{u_id_j}$ and other vertexes are reduced three point vertexes, while in $\mathcal{D'}_{III}$ at least one of the four point vertexes is $\bar V_{u_id_j}$. For $\mathcal{M}_1^{\mu\nu}$ and $\mathcal{M}_0^{\mu\nu}$, we only need to take $\mathcal{D}_I$, $\mathcal{D}_{II}$ and $\mathcal{D'}_{II}$ into consideration.

The contribution to the amplitudes from each kind of Feynman diagrams can be expanded respectively as:
\bea
{\mathcal{N}_{h_I}^{\mu\nu} z^{h_I}+ \mathcal{N}_{{h_I}-1}^{\mu\nu} z^{{h_I}-1}+\mathcal{N}_{{h_I}-2}^{\mu\nu}z^{{h_I}-2}\cdots\over \mathcal{C}_{{h_I}-1} z^{{h_I}-1} +\mathcal{C}_{{h_I}-2} z^{{h_I}-2}+\mathcal{C}_{{h_I}-3} z^{{h_I}-3}\cdots} ~&&\text{\ \ for} ~\mathcal{D}_I \\
{\mathcal{N}_{h_{II}}^{\mu\nu} z^{h_{II}}+ \mathcal{N}_{{h_{II}}-1}^{\mu\nu} z^{{h_{II}}-1}+\mathcal{N}_{{h_{II}}-2}^{\mu\nu}z^{{h_{II}}-2}\cdots\over \mathcal{C}_{{h_{II}}} z^{h_{II}} +\mathcal{C}_{h_{II}-1} z^{h_{II}-1}+\mathcal{C}_{h_{II}-2} z^{h_{II}-2}\cdots} ~&&\text{\ \ for}~ \mathcal{D}_{II} \\
{\mathcal{N}_{h_{II'}}^{\mu\nu} z^{h_{II'}}+ \mathcal{N}_{{h_{II'}}-1}^{\mu\nu} z^{{h_{II'}}-1}+\mathcal{N}_{{h_{II'}}-2}^{\mu\nu}z^{{h_{II'}}-2}\cdots\over \mathcal{C}_{{h_{II'}}} z^{h_{II'}} +\mathcal{C}_{h_{II'}-1} z^{h_{II'}-1}+\mathcal{C}_{h_{II'}-2} z^{h_{II'}-2}\cdots} ~&&\text{\ \ for}~ \mathcal{D}_{II'} \\
{\mathcal{N}_{h_{III}}^{\mu\nu} z^{h_{III}}+ \mathcal{N}_{h_{III}-1}^{\mu\nu} z^{h_{III}-1}+\mathcal{N}_{h_{III}-2}^{\mu\nu}z^{h_{III}-2}\cdots\over \mathcal{C}_{h_{III}+1} z^{h_{III}+1} +\mathcal{C}_{h_{III}} z^{h_{III}}+\mathcal{C}_{{h_{III}}-1} z^{{h_{III}}-1}\cdots} ~&&\text{\ \ for}~ \mathcal{D}_{III}\\
{\mathcal{N}_{h_{III'}}^{\mu\nu} z^{h_{III'}}+ \mathcal{N}_{h_{III'}-1}^{\mu\nu} z^{h_{III'}-1}+\mathcal{N}_{h_{III'}-2}^{\mu\nu}z^{h_{III'}-2}\cdots\over \mathcal{C}_{h_{III'}+1} z^{h_{III'}+1} +\mathcal{C}_{h_{III'}} z^{h_{III'}}+\mathcal{C}_{{h_{III'}}-1} z^{{h_{III'}}-1}\cdots} ~&&\text{\ \ for}~ \mathcal{D}_{III'}
\eea 
where we use $h_I$ to denote the highest $z$-order of $\mathcal{N}^{\mu\nu}$ for $\mathcal D_I$ type of Feynman diagrams, and similarly $h_{II}$ and $h_{III}$.

Then we can write 
\bea\label{LargeExap}
\mathcal{M}_1^{\mu\nu}&=&\mathcal{\bar M}_1^{\mu\nu}\nb\\
\mathcal{M}_0^{\mu\nu}&=&\mathcal{\bar M}_0^{\mu\nu}+\sum_{\mathcal{D}_{II'}} {\mathcal{N}_{h_{II'}}^{\mu\nu}\over \mathcal{C}_{h_{II'}}}-\sum_{\mathcal{D}_I} {\mathcal{C}_{h_I-2} \mathcal{N}_{h_I}^{\mu\nu}\over \mathcal{C}_{h_I-1}^2 }\nb\\
\mathcal{M}_{-1}^{\mu\nu}&=&\mathcal{\bar M}_{-1}^{\mu\nu}-\sum_{\mathcal{D}_1}{\mathcal{C}_{h_I-2} \mathcal{N}_{h_I-1}^{\mu\nu}\over \mathcal{C}_{h_I-1}^2}+\sum_{\mathcal{D}_I} {(\mathcal{C}_{h_I-2}^2-\mathcal{C}_{h_I-1}\mathcal{C}_{h_I-3})\mathcal{N}_{h_I}^{\mu\nu}\over \mathcal{C}_{h_I-1}^3}-\sum_{\mathcal{D}_{II}}{\mathcal{C}_{h_{II}-1} \mathcal{N}_{h_{II}}^{\mu\nu}\over \mathcal{C}_{h_{II}}^2}\nb\\
&&-\sum_{\mathcal{D}_{II'}}{\mathcal{C}_{h_{II'}-1} \mathcal{N}_{h_{II'}}^{\mu\nu}\over \mathcal{C}_{h_{II'}}^2}+\sum_{\mathcal{D}_{II'}} {\mathcal{N}_{h_{II'}-1}^{\mu\nu}\over \mathcal{C}_{h_{II'}}}+\sum_{\mathcal{D}_{III'}}
{\mathcal{N}_{h_{III'}}^{\mu\nu}\over \mathcal{C}_{h_{III'}+1}}+\mathcal{M}_{-1(M)}^{\mu\nu},
\eea
with
\bea\label{barTerm}
\mathcal{\bar M}_1^{\mu\nu}&=&\sum_{\mathcal{D}_I} {\mathcal{N}_{h_I}^{\mu\nu}\over \mathcal{C}_{h_I-1}}\nb\\
\mathcal{\bar M}_0^{\mu\nu}&=&\sum_{\mathcal{D}_1} {\mathcal{N}_{h_I-1}^{\mu\nu}\over \mathcal{C}_{h_I-1}}+\sum_{\mathcal{D}_{II}} {\mathcal{N}_{h_{II}}^{\mu\nu}\over \mathcal{C}_{h_{II}}}\nb\\
\mathcal{\bar M}_{-1}^{\mu\nu}&=&\sum_{\mathcal{D}_I} {\mathcal{N}_{h_I-2}^{\mu\nu}\over \mathcal{C}_{h_I-1}}+\sum_{\mathcal{D}_{II}} {\mathcal{N}_{h_{II}-1}^{\mu\nu}\over \mathcal{C}_{h_{II}}}+\sum_{\mathcal{D}_{III}} {\mathcal{N}_{h_{III}}^{\mu\nu}\over \mathcal{C}_{h_{III}+1}}.
\eea
In \eref{LargeExap} the last term for $\mathcal{M}_{-1}^{\mu\nu}$, ie $\mathcal{M}_{-1(M)}^{\mu\nu}$, is the contribution from M terms of the three point vertexes,  which is represented by the diagrams (a) and (b) in Figure \ref{Mtermcancel}. This term will be discussed in \sref{order1overz}. In \eref{LargeExap} and \eref{barTerm}, the summations are over ordered product $OP\{\alpha_{u_N} \bigcup \alpha_{d_M}\}$ \cite{Boels}, where $\alpha_{u_N}$ is the ordered subsets of $N$ up-legs $\{u_1,u_2, \cdots, u_N\}$ and $\alpha_{d_M}$ is the ordered subsets of $M$ down-legs $\{d_1,d_2, \cdots, d_M\}$. The ordered product  is the set of all permutations which leave the order of $\alpha_{u_N}$ and $\alpha_{d_M}$ invariant. For example, we have
\be
\sum_{\mathcal{D}_I} {\mathcal{N}_{h_I}^{\mu\nu}\over \mathcal{C}_{h_I-1}}\equiv\sum_{OP\{\alpha_{u_N} \bigcup \alpha_{d_M}\}} {\mathcal{N}_{h_I}^{\mu\nu}\over \mathcal{C}_{h_I-1}}.
\ee

Using \tref{redamplitude}, we can classify the terms that contribute to $\mathcal{M}_1^{\mu\nu}$ and $\mathcal{M}_0^{\mu\nu}$ into the following groups:
\begin{enumerate}
\item $\mathcal{D}_I$ with all the reduced three point vertexes  taking their S term components.
\item $\mathcal{D}_I$ with only one of the reduced three point vertex taking its R term part.
\item $\mathcal{D}_{II}$ with all the reduced three point vertexes  taking their S term components.
\item $\mathcal{D}_{II'}$ with all the reduced three point vertexes  taking their S term components.
\end{enumerate}
For the meaning of R and S terms in the reduced three point vertexes, refer to \eref{newV3s4} and \eref{ReduV}, with the external legs playing the role of leg 3 therein.

Case 1 is manifestly proportional to $g^{\mu\nu}$ and contributes to $\mathcal{N}_{h_I}^{\mu\nu}$ and $\mathcal{N}_{h_I-1}^{\mu\nu}$ in \eref{LargeExap} and \eref{barTerm}; Case 2 and Case 3 contribute to $\mathcal{N}_{h_I-1}^{\mu\nu}$ and $\mathcal{N}_{h_{II}}^{\mu\nu}$ respectively, and are manifestly antisymmetric in $\mu$ ad $\nu$; Case 4, which contributes to $\mathcal{N}_{h_{II'}}^{\mu\nu}$, is manifestly proportional to $g^{\mu\nu}$. Thus according to \eref{LargeExap}, an immediate conclusion is made that, for adjacent or non-adjacent BCFW shifts, $\mathcal{M}_1^{\mu\nu}$ is proportional to $g^{\mu\nu}$, and $\mathcal{M}_0^{\mu\nu}$  is in the form of $A g^{\mu\nu}+B^{\mu\nu}$  with $B^{\mu\nu}$ antisymmetric in $\mu$ and $\nu$. In the next section, we will see how non-adjacent shifts imply improved boundary behavior compared with adjacent shifts.

\section{Amplitudes with Non-adjacent BCFW Shifts}\label{Sec:Non-Adj}
We first show a property which is special for non-adjacent BCFW shifts. Such property is very useful in analyzing each summation in the right hand side of \eref{barTerm}. Furthermore, it is this property that results in better boundary behavior for amplitudes under non adjacent shifts.   



\subsection{Permutation Sums}
In this subsection, we discuss $\sum_{\mathcal{D}_I} {\mathcal{N}_{h_I}^{\mu\nu}\over \mathcal{C}_{h_I-1}}$ in detail. The conclusions also hold for other similar summations in \eref{barTerm}, ie. $\sum_{\mathcal{D}_{II}} {\mathcal{N}_{h_{II}}^{\mu\nu}\over \mathcal{C}_{h_{II}}}$ and $\sum_{\mathcal{D}_{III}} {\mathcal{N}_{h_{III}}^{\mu\nu}\over \mathcal{C}_{h_{III}+1}}$. We use $k_{l,u_i}$ to denote for $k_l+k_{u_1}+k_{u_{u_1}}+\cdots+k_{u_i}$ and $k_{d_j,u_i}$ for $k_{d_j}+k_{d_{j-1}}+\cdots +k_{d_1}+k_l+k_{u_1}+k_{u_2}+\cdots+k_{u_j}$.  As a warm-up exercise, we investigate an example with N legs above and 1 leg below the complex line, see Figure \ref{example1}.
\begin{figure}[]
\centering
\includegraphics{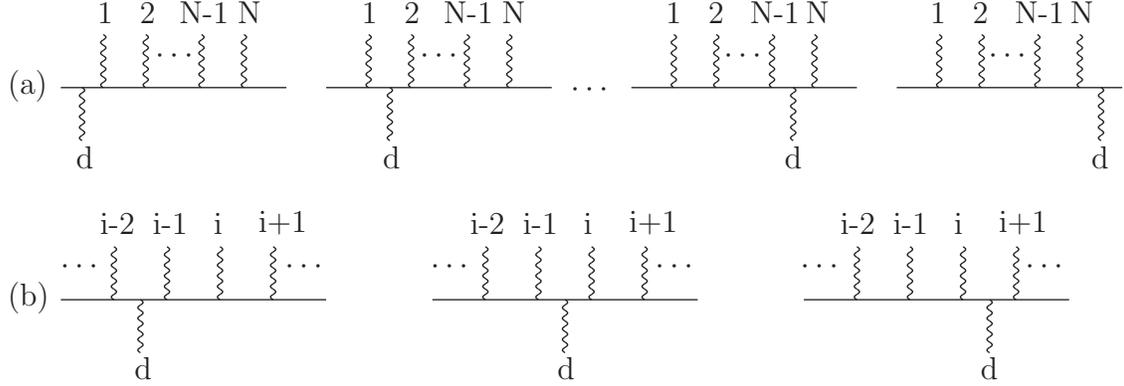}
\caption{When there are N legs above and 1 leg below the complex line, (a) shows all the diagrams with all external legs attaching the complex line and contributing at $\mathcal{O}(z^1)$. (b) contains three diagrams out of (a) for analysis.}
\label{example1}
\end{figure}
We first investigate the highest $z$ order terms of the products of the propagators for the three diagrams as in (b) of Figure \ref{example1}. For convenience, we will omit the $-i$ factors in the propagators in the following. Since there is only one leg "d" below the complex line, this "d" can be viewed as "$d_1$". For the three diagrams of (b) in Figure \ref{example1}, $\mathcal{C}_{h_I-1} (\{u_i,d_j, u_i u_{i+1}, d_j d_{j+1}, u_i d_j\})$ are:


\bea\label{b1}
&&\frac{1}{2 z k_{l,u_1}\cdot \eta}\frac{1}{2 z k_{l,u_2}\cdot \eta}\cdots \frac{1}{2 z k_{l,u_{i-2}}\cdot \eta}\frac{1}{2 z k_{d,u_{i-2}}\cdot \eta}\frac{1}{2 z k_{d,u_{i-1}}\cdot \eta}\cdots\frac{1}{2 z k_{d,u_{N-1}}\cdot \eta}\\
=&&\frac{1}{2z k_d\cdot \eta}\frac{1}{2 z k_{l,u_1}\cdot \eta}\cdots \frac{1}{2 z k_{l,u_{i-3}}\cdot \eta}(\frac{1}{2 z k_{l,u_{i-2}}\cdot \eta}-\frac{1}{2 z k_{d,u_{i-2}}\cdot \eta})\frac{1}{2 z k_{d, u_{i-1}}\cdot \eta}\cdots\frac{1}{2 z k_{d, u_{N-1}}\cdot \eta}.\nonumber
\eea
\bea\label{b2}
&&\frac{1}{2 z k_{l,u_1}\cdot \eta}\frac{1}{2 z k_{l,u_2}\cdot \eta}\cdots \frac{1}{2 z k_{l,u_{i-1}}\cdot \eta}\frac{1}{2 z k_{d,u_{i-1}}\cdot \eta}\frac{1}{2 z k_{d,u_{i}}\cdot \eta}\cdots\frac{1}{2 z k_{d, u_{N-1}}\cdot \eta}\\
=&&\frac{1}{2z k_d\cdot \eta}\frac{1}{2 z k_{l,u_1}\cdot \eta}\cdots \frac{1}{2 z k_{l,u_{i-2}}\cdot \eta}(\frac{1}{2 z k_{l,u_{i-1}}\cdot \eta}-\frac{1}{2 z k_{d,u_{i-1}}\cdot \eta})\frac{1}{2 z k_{d,u_{i}}\cdot \eta}\cdots\frac{1}{2 z k_{d,u_{N-1}}\cdot \eta}.\nonumber
\eea
\bea\label{b3}
&&\frac{1}{2 z k_{l,u_1}\cdot \eta}\frac{1}{2 z k_{l,u_2}\cdot \eta}\cdots \frac{1}{2 z k_{l,u_i}\cdot \eta}\frac{1}{2 z k_{d, u_{i}}\cdot \eta}\frac{1}{2 z k_{d, u_{i+1}}\cdot \eta}\cdots\frac{1}{2 z k_{d, u_{N-1}}\cdot \eta}\\
=&&\frac{1}{2z k_d\cdot \eta}\frac{1}{2 z k_{l,u_1}\cdot \eta}\cdots \frac{1}{2 z k_{l,u_{i-1}}\cdot \eta}(\frac{1}{2 z k_{l,u_{i}}\cdot \eta}-\frac{1}{2 z k_{d,u_{i}}\cdot \eta})\frac{1}{2 z k_{d, u_{i+1}}\cdot \eta}\cdots\frac{1}{2 z k_{d, u_{N-1}}\cdot \eta}.\nonumber
\eea
It is observed that the first term in \eref{b2} cancels the second term in \eref{b3} and the first term in \eref{b1} cancels the second term in \eref{b2}. This manner of cancellation happens for each two successive diagrams in (a) of Figure \ref{example1}, and it is found that the sum of the highest order z terms of the products of propagators in all diagrams in (a) of Figure \ref{example1} turns out to be 0. When including the numerator, ie. the product of the vertexes $\mathcal{N}^{\mu\nu}$, the summation  of  equations such as  \eref{b1}, \eref{b2} and \eref{b3} for all the diagrams in (a) of Figure \ref{example1} is  just
\bea\label{Rec1}
\sum_{\mathcal{D}_I}{\mathcal{N}^{\mu\nu}\over \mathcal{C}_{h_I-1}}&\equiv&\sum_{OP\{\alpha_{u_N} \bigcup d\}}{\mathcal{N}^{\mu\nu}\over \mathcal{C}_{h_I-1}}\nb\\
&=&\sum_{i=1}^N{1\over 2 z k_{l,u_1}\cdot \eta}{1\over 2 z k_{l,u_2}\cdot \eta}\cdots {1\over 2 z k_{l,u_{i-1}}\cdot \eta}{1\over 2 z k_{d}\cdot \eta}{1\over 2 z k_{d,u_{i}}\cdot \eta}\cdots{1\over 2 z k_{d,u_N}\cdot \eta}\nb\\
&\times&(\mathcal{N}^{\mu\nu}(\cdots d, u_i, \cdots)-\mathcal{N}^{\mu\nu}(\cdots u_i, d, \cdots)).
\eea

For general non-adjacent BCFW shifts with $N$ up-legs and $M$ down-legs.  We can prove that the summation in \eref{barTerm} can be recombined into the summation of terms  like \eref{Rec1}.
\begin{theorem}\label{lemResum}
\bea\label{Resum}
&&\sum_{\mathcal{D}_I}{\mathcal{N}^{\mu\nu}(\{u_i,d_j\})\over \mathcal{C}_{h_I-1}(\{u_i,d_j\})}\equiv \sum_{OP\{\alpha_{u_N} \bigcup \alpha_{d_M}\}}{\mathcal{N}^{\mu\nu}(\{u_i,d_j\})\over \mathcal{C}_{h_I-1}(\{u_i,d_j\})}\nb \\
&=&\sum_{j, i=1}^{M,N}\sum_{OP\{\alpha_{u_{i-1}}\atop \bigcup \alpha_{d_{j-1}}\}}\sum_{{OP\{\alpha_{(u_{i+1},u_M)}\atop \alpha_{(d_{j+1}, d_{M})}\} }}\nb \\
&&{1\over 2 z k_{l,u_1}\cdot \eta}{1\over 2 z k_{d_1,u_1}\cdot \eta}\cdots {1\over 2 z k_{d_{j-1},u_{i-1}}\cdot \eta}{1\over 2 z (k_{d_1}+\cdots+k_{d_M})\cdot \eta}{1\over 2 z k_{d_{j},u_i}\cdot \eta}\cdots{1\over 2 z k_{d_{M-1},u_N}\cdot \eta}\nb\\ \nb\\ 
&\times&(\mathcal{N}^{\mu\nu}(\cdots, d_j, u_i, \cdots)-\mathcal{N}^{\mu\nu}(\cdots, u_i, d_j, \cdots)).\nb\\
\eea
In the last line of \eref{Resum}, only the order of nearby up-leg and down-leg pair, ie. $u_i$ and $d_j$ is inter-changed.  In the original form in large z limit only one of the propagators in $\mathcal{M}^{\mu\nu}(\cdots, d_j, u_i, \cdots)$ and $\mathcal{M}^{\mu\nu}(\cdots, u_i, d_j, \cdots)$ is different which is the propagator between $u_i$  and $d_j$. In the recombined summation, this different propagator is replaced with ${1\over 2 z (k_{d_1}+\cdots+k_{d_M})\cdot \eta}$ while other propagators are not changed. Similar equations hold for the other summations in \eref{barTerm}. For example, for $\sum_{\mathcal{D}_{II}} {\mathcal{N}_{h_{II}}^{\mu\nu}\over \mathcal{C}_{h_{II}}}$, say the four point vertex in $\mathcal{D}_{II}$ is $\bar V_{d_j d_j+1}$, we simplify identify this four point vertex as a three point vertex at the corresponding position in the corresponding $\mathcal D_I$ type of diagram. That is to say, we define $d'_i=d_i$ for $i<j$, $d'_i=d_j d_{j+1}$ for $i=j$, $d'_i=d_{i+1}$ for $i>j$, $k_{d'_i}=k_{d_i}$ for $i<j$, $k_{d'_i}=k_{d_j}+k_{d_{j+1}}$ for $i=j$, $k_{d'_i}=k_{d_{i+1}}$ for $i>j$, and replace the $\{d_i\}$ in \eref{Resum} with $\{d'_i\}$. We do not repeat for other summations in \eref{barTerm}.
\end{theorem}

\begin{proof}
To prove this, we only need to prove  that each term with a fixed order of up and down type legs in the left hand side of \eref{Resum} is equal to the sum of terms in the right hand side with the same order in $\mathcal{N}^{\mu\nu}$. This can be done recursively. First we assume that, for each ordering of legs in $\mathcal{N}^{\mu\nu}$, the summation  of the right hand side of \eref{Resum} with $N-1$ up-legs and $M-1$ down-legs is  
\bea
{\mathcal{N}^{\mu\nu}(\cdots u_{N-1})\over\mathcal{\bar C}_{h_I-1}(\cdots u_{N-1})}
\eea
when the most right side leg is $u_{N-1}$, with
$$\frac{1}{\mathcal{\bar C}_{h_I-1}(\cdots u_{N-1})}={1\over \mathcal{C}_{h_I-1}(\cdots u_{N-1})}=
{1\over 2 z k_{u_1}\cdot \eta}{1\over 2 z (k_{u_1}+k_{d_1})\cdot \eta}\cdots  \cdots{1\over 2 z k_{d_{M-1} u_{N-2}}\cdot \eta}.
$$
Similarly for the case with the most right side leg being $d_{M-1}$, the summation is
\bea
{\mathcal{N}^{\mu\nu}(\cdots d_{M-1})\over \mathcal{\bar C}_{h_I-1}(\cdots d_{M-1})}
\eea
with ${1\over \mathcal{\bar C}_{h_I-1}(\cdots d_{M-1})}={1\over \mathcal{C}_{h_I-1}(\cdots d_{M-1})}{-2 z (k_{u_1}+\cdots+k_{u_{N-1}})\cdot \eta\over 2 z (k_{d_1}+\cdots+k_{d_{M-1}})\cdot \eta}$ and $${1\over\mathcal{C}_{h_I-1}(\cdots d_{M-1})}={1\over 2 z k_{u_1}\cdot \eta}{1\over 2 z (k_{u_1}+k_{d_1})\cdot \eta}\cdots  \cdots{1\over 2 z k_{d_{M-2} u_{N-1}}\cdot \eta}. $$
Then if we attach leg $u_N$ to the complex line following the sequence  $(\cdots u_{N-1})$, we can get 
\bea
{\mathcal{\bar C}_{h_I-1}(\cdots u_{N-1} u_{N})}={\mathcal{C}_{h_I-1}(\cdots u_{N-1} u_{N})}.
\eea

If we attach $u_N$ to the complex line following the sequence $(\cdots d_{M-1})$, we can obtain 
\bea
\mathcal{\bar C}_{h_I-1}(\cdots d_{M-1} u_{N})&=&\mathcal{\bar C}_{h_I-1}(\cdots d_{M-1}){1\over2 z (k_{u_{N-1}}+\cdots+k_{d_{M-1}})\cdot \eta}\nb\\ &&+\mathcal{C}_{h_I-1}(\cdots d_{M-1}){1\over 2 z (k_{d_1}+\cdots+k_{d_{M-1}})\cdot\eta}\nb\\
&=&\mathcal{C}_{h_I-1}(\cdots d_{M-1}u_N).
\eea
Here there is one additional contribution from changing the order of  $d_{M-1}$ and $u_N$ in the right hand side of \eref{Resum}.

Similarly, if we attach the leg $d_M$ to the complex line following the sequence $(\cdots d_{M-1})$, we can get 
\bea
{1\over\mathcal{\bar{C}}_{h_I-1}(\cdots d_{M-1} d_{M})}&=&{1\over\mathcal{\bar{C}}_{h_I-1}(\cdots d_{M-1})}{1\over2 z (k_{u_{N-1}}+\cdots+k_{d_{M-1}})\cdot \eta}\times {2 z (k_{d_1}+\cdots+k_{d_{M-1}})\cdot\eta\over 2 z (k_{d_1}+\cdots+k_{d_{M}})\cdot\eta}\nb\\ &=&{1\over \mathcal{C}_{h_I-1}(\cdots d_{M-1}d_{M})}{-2 z (k_{u_1}+\cdots+k_{u_{N-1}})\cdot \eta\over 2 z (k_{d_1}+\cdots+k_{d_{M}})\cdot \eta}.
\eea
And if attaching  the line $d_M$ to the complex line following the sequence $(\cdots u_{N-1})$,  we can get 
\bea
{1\over \mathcal{\bar{C}}_{h_I-1}(\cdots  u_{N-1} d_{M})}&=&{1\over \mathcal{\bar C}_{h_I-1}(\cdots u_{N-1})}{2 z (k_{d_1}+\cdots+k_{d_{M-1}})\cdot\eta\over 2 z (k_{d_1}+\cdots+k_{d_{M}})\cdot\eta}{1\over2 z (k_{u_{N-1}}+\cdots+k_{d_{M-1}})\cdot \eta}\nb\\ &&+{1\over \mathcal{C}_{h_I-1}(\cdots u_{N-1})}{-1\over 2 z (k_{d_1}+\cdots+k_{d_{M}})\cdot\eta}\nb\\
&=&{1\over \mathcal{\bar C}_{h_I-1}(\cdots u_{N-1}d_M)}{-2 z (k_{u_1}+\cdots+k_{u_{N-1}})\cdot \eta\over 2 z (k_{d_1}+\cdots+k_{d_{M}})\cdot \eta}.
\eea

Thus for $N$ up legs and $M$ down legs, we get:
\bea
\mathcal{\bar C}_{h_I-1}(\cdots u_{N})&=&\mathcal{C}_{h_I-1}(\cdots u_{N}) \nb\\
{1\over  \mathcal{\bar{C}}_{h_I-1}(\cdots  d_{M})}&=&{1\over \mathcal{C}_{h_I-1}(\cdots d_M)}{-2 z (k_{u_1}+\cdots+k_{u_{N}})\cdot \eta\over 2 z (k_{d_1}+\cdots+k_{d_{M}})\cdot \eta}.
\eea
With momenta conservation and the shift condition \eref{conditionEta4} it is easy to see 
\bea
{-2 z (k_{u_1}+\cdots+k_{u_{N}})\cdot \eta\over 2 z (k_{d_1}+\cdots+k_{d_{M}})\cdot \eta}=1.
\eea
By induction, the equation \eref{Resum}, ie. \tref{lemResum}, has been proved. \endofproof
\end{proof}
\begin{corollary}\label{corResum}
When the $\mathcal{N}^{\mu\nu}(\{u_i,d_j\})$ are independent of the relative orders of the external legs, we have 
\be
\sum_{\mathcal{D}_I}{1\over \mathcal{C}_{h_I-1}}=0.
\ee
Such equations hold also for the other cases in \eref{barTerm}. For example, $\sum_{\mathcal{D}_{II}}{1\over \mathcal{C}_{h_{II}}}=0$ and $\sum_{\mathcal{D}_{III}}{1\over \mathcal{C}_{h_{III}+1}}=0$.
\end{corollary}

\subsection{$\mathcal{O}(z^1)$ Behavior of the Amplitudes in the Large $z$ Limit under Non-adjacent BCFW Shifts}

To obtain the $\mathcal{O}(z^1)$ behavior of the amplitude $\mathcal{M}^{\mu\nu}$, we only need the case 1 in \sref{theorem1Application}, that is $\mathcal{D}_I$ with all the reduced three point vertexes taking their S term components. Furthermore we only need to keep the terms with highest order of $z$ in all the vertexes and propagators, ie. $\mathcal{N}_{h_I}^{\mu\nu}$ and $\mathcal{C}^{h_I-1}$. The $z$ order of S term $S_{u/d}$ does not depend on its position on the complex line. As a result, all diagrams of type $\mathcal{D}_I$ have the same $\mathcal{N}_{h_I}^{\mu\nu}\propto g^{\mu\nu}$ and using \cref{corResum} we obtain:
\beq
\mathcal{M}_1^{\mu\nu}=\sum_{\mathcal{D}_I}\frac{\mathcal{N}_{h_I}^{\mu\nu}}{\mathcal{C}^{h_I-1}}=\mathcal{N}_{h_I}^{\mu\nu}\sum_{\mathcal{D}_I}\frac{1}{\mathcal{C}^{h_I-1}}=0.
\label{orderz}
\eeq
In conclusion, $\mathcal{O}(z^1)$ of $\mathcal{M}^{\mu\nu}$ for non-adjacent shifts vanish.

\subsection{$\mathcal{O}(z^0)$ Behavior of the Amplitudes in the Large $z$ Limit under Non-adjacent BCFW Shifts}\label{orderz0}
In this subsection, we are going to show that: for non-adjacent shifts, 
\be\label{Mz0}
\mathcal{M}_0^{\mu\nu}\propto g^{\mu\nu}.
\ee

Using \eref{LargeExap} and \eref{barTerm}, we can classify the terms that contribute to $\mathcal{M}_0^{\mu\nu}$ into the following groups:
\begin{itemize}
\item $\sum_{\mathcal{D}_I} {c_{h_I-2} \mathcal{N}_{h_I}^{\mu\nu}\over c_{h_I-1}^2 }\propto g^{\mu\nu}$, since $\mathcal{N}_{h_I}^{\mu\nu}$ is proportional to $g^{\mu\nu}$ in diagrams $\mathcal{D}_I$. 
\item $\sum_{\mathcal{D}_{II'}} {\mathcal{N}_{h_{II'}}^{\mu\nu}\over c_{h_{II'}}}\propto g^{\mu\nu}$.  In $\mathcal{D}_{II'}$,  there is one reduced four point vertex ${\bar V}_{u_i d_j}$ in the complex line. And all the others are reduced three point vertexes with only their $S$ term components. According to the forms of ${\bar V}_{u_i d_j}$ and $S$ term, it is seen $\mathcal{N}_{h_{II'}}^{\mu\nu}\propto g^{\mu\nu}.$
\item $\sum_{\mathcal{D}_{II}} {(\mathcal{N}_{II})_{h_{II}}^{\mu\nu}\over (c_{II})_{h_{II}}}=0$, using \cref{corResum}, essentially the same as in \eref{orderz}.
\item $\sum_{\mathcal{D}_I} {\mathcal{N}_{h_I-1}^{\mu\nu}\over c_{h_I-1}}\propto g^{\mu\nu}.$ $\mathcal{D}_I$ are the diagrams comprised all of reduced three point vertexes. There are two contributions to this summation. One contribution is when only one of the reduced three point vertexes takes its R term part and other vertexes take their S components. Without loss of generality, we assume the vertex with the leg $u_i$ takes its R part. All these diagrams have the same $\mathcal{N}_{h_I-1}^{\mu\nu}$. According to \cref{corResum}, the sum of all these diagrams contribute 0 to $\sum_{\mathcal{D}_I} {\mathcal{N}_{h_I-1}^{\mu\nu}\over c_{h_I-1}}$. The other contribution is when all the reduced three point vertexes take their $S$ term components. This contribution is obviously proportional to $g^{\mu\nu}$.
\end{itemize}
Thus we have proven that for non-adjacent shifts, $\mathcal{M}_0^{\mu\nu}$ is proportional to $g^{\mu\nu}$.

\subsection{$\mathcal{O}(z^{-1})$ Behavior of the Amplitudes in the Large $z$ Limit under Non-adjacent BCFW Shifts}\label{order1overz}
The previous two sub sections do not depend on whether the external legs are on-shell or off-shell. In this sub section, we discuss $\mathcal{M}_{-1}^{\mu\nu}$ in the two cases when the external lines are all on-shell and when some of them are off-shell.

When all external lines are on shell, the "generalized Ward identities" in \lref{GWI1} and \lref{GWI2} become the real Ward identities where the expressions are exactly zero. Thus the last term for $\mathcal{M}_{-1}^{\mu\nu}$ in \eref{LargeExap}, ie. $\mathcal{M}_{-1(M)}^{\mu\nu}$, is 0. By the similar arguments as in the last sub section,  it is easy to see that each other term except $\mathcal{\bar M}_{-1}^{\mu\nu}$ in the third equation of \eref{LargeExap} is in the form of $A g^{\mu\nu}+B^{\mu\nu}$ with $B^{\mu\nu}$ antisymmetric in $\mu$ and $\nu$. We are going to concentrate on terms that contribute to $\mathcal{\bar M}_{-1}^{\mu\nu}$ in \eref{barTerm}:
\begin{itemize}

\item $\sum_{\mathcal{D}_I} {\mathcal{N}_{h_I-2}^{\mu\nu}\over c_{h_I-1}}\propto A g^{\mu\nu}+B^{\mu\nu}$. In $\mathcal{D}_I$, all the vertexes in the complex line are the three point vertexes $\bar V_{u/d}$. We can classify them into the following groups:  \\ \textbf{\textcircled{\bf{a}}} When $\bar V_{u/d}$ all take their $S$-term components or only one of them takes its $R$ term part, such contributions are obviously of form $A g^{\mu\nu}+B^{\mu\nu}$. \\ \textbf{\textcircled{\bf{b}}} When  the two vertexes with R parts are all above (or below) the complex line, for example $R_{u_i}$ and $R_{u_j}$, and others are all taking S terms. Furthermore since each R term decreases order of $z$ by 1 compared to S term, to contribute to the next to next order of the product of the vertexes ie. $\mathcal{N}_{h_I-2}^{\mu\nu}$, each S term of other vertexes should take its highest z order term and contributes the same to $\mathcal{N}_{h_I-2}^{\mu\nu}$ regardless of its position on the complex line. Thus $\mathcal{N}_{h_I-2}^{\mu\nu}$ are the same for all these diagrams. Same to \eref{orderz}, using \cref{corResum}, these terms contribute 0 to $\mathcal{\bar M}_{-1}^{\mu\nu}$. \\ \textbf{\textcircled{\bf{c}}} When the two vertexes with R parts are $R_{u_i}$ and $R_{d_j}$, with indices $\mu_{u_i}$ and $\mu_{d_j}$, other vertexes are all taking S components. $R_{u_i}$ and $R_{d_j}$ are also independent of their positions on the complex line. Thus as for the calculation of $\mathcal{N}_{h_I-2}^{\mu\nu}$, we can regard $S_{u_i'}$ and $S_{d_j'}$ as commuting, $S_{u_i'}$ and $R_{d_j}$ commuting, and $R_{u_i}$ and $S_{d_j'}$ commuting. Applying \tref{lemResum}, we can see that the only non-vanishing terms are from:

\begin{small}
\begin{equation*}
\mathcal{N}_{h_I-2}^{\mu\nu}(\cdots, d_j, u_i, \cdots)-\mathcal{N}_{h_I-2}^{\mu\nu}(\cdots, u_i, d_j, \cdots)\propto (R_{d_j})^\mu{}_{\rho}{}^{\mu_{d_j}}(R_{u_i})^{\nu\rho}{}^{\mu_{u_i}}-(R_{u_i})^{\rho\mu}{}^{\mu_{u_i}}(R_{d_j})_\rho{}^\nu{}^{\mu_{d_j}},
\end{equation*}
\end{small}

which is antisymmetric in $\mu$ and $\nu$, invoking that $R$ term is antisymmetric in its first two indices, referring to \eref{newV3s4}.


\item $\sum_{\mathcal{D}_{II}} {\mathcal{N}_{h_{II}-1}^{\mu\nu}\over c_{h_{II}}}\propto A g^{\mu\nu}+B^{\mu\nu}$. In $\sum_{\mathcal{D}_{II}}$, the diagrams are comprised of one reduced four point vertex, which is not ${\bar V}_{u_i d_j}$, and the rest vertexes are reduced three point vertexes. In the definition of the reduced vertexes \eref{ReduV}, we call the last term of $\bar{V}_{u_i u_j}$ or $\bar{V}_{d_i d_j}$ as metric term and the first two terms as antisymmetric term. For three point reduced vertex, we also call S term as the metric term and R term as the antisymmetric term. The discussion is parallel to the case above: \\ \textbf{\textcircled{\bf{a}}}   Only one or none of the reduced vertexes takes its anti-symmetric part. The contribution is of form $A g^{\mu\nu}+B^{\mu\nu}$. \\ \textbf{\textcircled{\bf{b}}} The four point vertex and one three point vertex take their antisymmetric terms. When they are both above (or below) the complex line, the contribution to $\sum_{\mathcal{D}_{II}} {\mathcal{N}_{h_{II}-1}^{\mu\nu}\over c_{h_{II}}}$ is 0.  \\ \textbf{\textcircled{\bf{c}}} The four point vertex and one three point vertex take their antisymmetric terms and they are on the opposite sides of the complex line. The contribution is antisymmetric in $\mu$ and $\nu$.

\item $\sum_{\mathcal{D}_{III}} {\mathcal{N}_{h_{III}}^{\mu\nu}\over c_{h_{III}+1}}\propto A g^{\mu\nu}+B^{\mu\nu}. $ In $\mathcal{D}_{III}$, the diagrams are comprised of two reduced four point vertexes, neither of which is ${\bar V}_{u_i d_j}$, and the other reduced three point vertexes all take their S term parts. The discussion is again parallel to the cases above: \\ \textbf{\textcircled{\bf{a}}}   Only one or none of the reduced four point vertexes takes its anti-symmetric part. The contribution is of form $A g^{\mu\nu}+B^{\mu\nu}$. \\ \textbf{\textcircled{\bf{b}}} The two reduced four point vertexes both take their anti-symmetric parts and are both above (or below) the complex line. It contributes 0 to $\sum_{\mathcal{D}_{III}} {\mathcal{N}_{h_{III}}^{\mu\nu}\over c_{h_{III}+1}}$. \\ \textbf{\textcircled{\bf{c}}} The two reduced four point vertexes take their antisymmetric parts and are on the opposite sides of the complex line. The contribution is antisymmetric in $\mu$ and $\nu$.

\end{itemize}

Above all, when all the external legs are on shell, for non-adjacent shifts, $\mathcal{O}(z^{-1})$ of $\mathcal{M}^{\mu\nu}$, ie. $\mathcal{M}_{-1}^{\mu\nu}$, is in form of a metric term plus a term antisymmetric in $\mu$ and $\nu$.

Now we discuss the case when some external legs are off-shell. The additional contribution is from the last term $\mathcal{M}_{-1(M)}^{\mu\nu}$ in \eref{LargeExap}, which is from the diagrams (a) and (b) of Figure \ref{Mtermcancel}. We analyze how the diagrams contribute to $\mathcal{M}_{-1}^{\mu\nu}$. Take the diagram (a) for example, with the last $M^R$ factor to be $M_{d_i}^R$ (same analysis for $M_{u_{i}}^R$). Assume the next vertex is $V_{u_{j}}$ (same analysis for $V_{d_{j}}$, and if the next vertex is four point vertex see below). Then $M_{d_i}^R\  V_{u_{j}}$ can be decomposed according to \eref{kdotV4} and Figure \ref{vertexnotation4}, see Figure \ref{MRdotV}.

\begin{figure}[htb]
\centering
\includegraphics{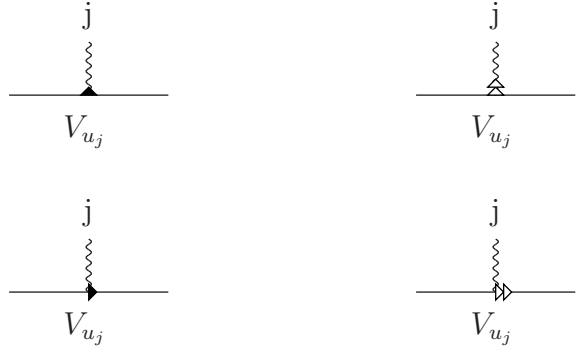}
\caption{Decomposition of $M_{d_i}^R\  V_{u_j}$ in the notations of Figure \ref{vertexnotation4}. The horizontal line is the complex line, and photon line represents external leg.}
\label{MRdotV}
\end{figure}

Among the four terms in Figure \ref{MRdotV}, the first line two terms combined is in the form
\beq
k_j^2 g^{\mu_j\delta}-k_j^{\mu_j} k_j^\delta,
\label{kkterm}
\eeq
where $\delta$ is some index we do not care here. The first term in the second line of Figure \ref{MRdotV} need not be considered since they will cancel in group (together with the terms when the next vertex is a four point vertex) in the manner of Figure \ref{treecancel2c4}. In this cancellation, diagrams with some vertexes outside the complex line is involved, but it does not affect the property of our conclusion, once we apply less point results to these diagrams. The second term in the second line of Figure \ref{MRdotV} acts on the next vertex on the complex line, and can be analyzed in the same steps as in this paragraph. Only when the vertex being acted on is the last vertex on the complex line, the second line two terms of Figure \ref{MRdotV} should be retained, which sum up to equal $k_r^2 g^{\nu\delta}-k_r^{\nu} k_r^\delta$, also in the form of \eref{kkterm}. (b) of Figure \ref{Mtermcancel} is similarly analyzed, and results in terms in the form of \eref{kkterm}. \eref{kkterm} is 0 when $k_j$ is on shell and only receives contributions from off shell external legs. Thus we can make the conclusion that the additional contributions to $\mathcal{M}_{-1}^{\mu\nu}$ from off shell external legs are:
\beq
\sum_{\mbox{\tiny off shell\ }j} (k_j^2 g^{\mu_j\delta}-k_j^{\mu_j} k_j^\delta)\cdots,
\label{kkform}
\eeq
where the sum is over each off shell external leg.

Direct calculation shows that \eref{kkform} is antisymmetric in $\mu$ and $\nu$ when there is only 1 leg above and 1 leg below the complex line, and not antisymmetric for 5 point amplitudes, unlike to be antisymmetric for more point amplitudes.

In conclusion, for non adjacent BCFW shifts of on shell tree amplitudes, $\mathcal{O}(z^{-1})$ of $\mathcal{M}^{\mu\nu}$ is in form of a metric term plus a term antisymmetric in $\mu$ and $\nu$; for amplitudes with off shell legs, $\mathcal{O}(z^{-1})$ has additional contributions from the off shell legs in the form of \eref{kkform}, which manifestly vanish when the legs become on shell. We guess that when gluing off shell tree amplitudes into on shell loop level amplitudes, terms in \eref{kkform} may cancel the contribution from ghost loops, which deserves further investigation.

\subsection{Example}
Using our results in \sref{Sec:Reduce} and \sref{Sec:Non-Adj}, besides proving the large $z$ behavior as above, we can also easily calculate the boundary terms of the shifted amplitudes. Two main simplifying features of our method of calculating boundary terms are: the tensor structures of the contributing terms are very simple due to the reduced vertexes \eref{ReduV} and \tref{redamplitude}; the number of contributing Feynman diagrams are reduced due to permutation sum \tref{lemResum}. We will show the calculation of boundary terms for non-adjacently shifted 4 and 5 point amplitudes.

According to \eref{LargeExap} and \eref{barTerm} and \sref{orderz0}, $\mathcal{M}_0^{\mu\nu}=\sum_{\mathcal{D}_1} {\mathcal{N}_{h_I-1}^{\mu\nu}\over \mathcal{C}_{h_I-1}}-\sum_{\mathcal{D}_I} {\mathcal{C}_{h_I-2} \mathcal{N}_{h_I}^{\mu\nu}\over \mathcal{C}_{h_I-1}^2 }+\sum_{\mathcal{D}_{II'}} {\mathcal{N}_{h_{II'}}^{\mu\nu}\over \mathcal{C}_{h_{II'}}}$. As discussed in \sref{orderz0}, for $\mathcal{N}_{h_I-1}^{\mu\nu}$ in the first term, all the three point vertexes only need to take their S parts.

First we calculate the boundary term of four point amplitude $\hat{\mathcal M}^{\mu_1\mu_2\mu_3\mu_4}$, with $k_1$ and $k_3$ shifted as in \eref{momshift4} and \eref{conditionEta4}: $k_1\to k_1+z\eta$, $k_3\to k_3-z\eta$, $k_1\cdot \eta=0$, $k_3\cdot \eta=0$. As explained in the paragraph above \tref{redamplitude}, we assume $\hat{\mathcal M}^{\mu_1\mu_2\mu_3\mu_4}$ to be contracted with $\hat \epsilon^{\mu_1}_1$ and $\hat \epsilon^{\mu_3}_3$, which satisfy $\hat k_1 \cdot \hat \epsilon_1=0$ and $\hat k_3 \cdot \hat \epsilon_3=0$. The diagrams are in Figure \ref{fourpoint}.

\begin{figure}[htb]
\centering
\includegraphics[height=4cm,width=12cm]{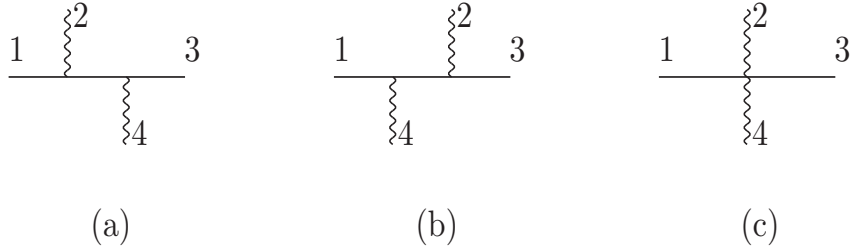}
\caption{The Feynman diagrams for four point amplitude. Legs $k_1$ and $k_3$ are shifted.}
\label{fourpoint}
\end{figure}

For $\sum_{\mathcal{D}_1} {\mathcal{N}_{h_I-1}^{\mu\nu}\over \mathcal{C}_{h_I-1}}$, (a) and (b) in Figure \ref{fourpoint} contribute. For the propagators, we only need the highest order in z. And for all the three point vertexes, we take their S parts, and extract the next leading order term in $z$ in the product of the vertexes. The result can be quickly written out. Especially, there is no trouble from tensor index contractions. $\sum_{\mathcal{D}_1} {\mathcal{N}_{h_I-1}^{\mu\nu}\over \mathcal{C}_{h_I-1}}=\frac{ig^{\mu_1\mu_3}(k_4^{\mu_2}\eta^{\mu_4}-k_2^{\mu_4}\eta^{\mu_2})}{k_2\cdot \eta}$. In it, $k_2\cdot\eta=-k_4\cdot \eta$ has been used, which is due to momentum conservation and \eref{conditionEta4}.

For $-\sum_{\mathcal{D}_I} {\mathcal{C}_{h_I-2} \mathcal{N}_{h_I}^{\mu\nu}\over \mathcal{C}_{h_I-1}^2 }$, (a) and (b) in Figure \ref{fourpoint} contribute. For the propagators, we need the leading and next leading order term. And for all the three point vertexes, we only need the leading order of their S parts. The result can again be quickly written out. $-\sum_{\mathcal{D}_I} {\mathcal{C}_{h_I-2} \mathcal{N}_{h_I}^{\mu\nu}\over \mathcal{C}_{h_I-1}^2 }=\frac{i g^{\mu_1\mu_3}\eta^{\mu_2}\eta^{\mu_4}(k_2^2+k_4^2-2k_1\cdot k_3)}{2(k_2\cdot\eta)^2}$.

For $\sum_{\mathcal{D}_{II'}} {\mathcal{N}_{h_{II'}}^{\mu\nu}\over \mathcal{C}_{h_{II'}}}$, only (c) in Figure \ref{fourpoint} contributes, and the vertex is very simple, refer to \eref{ReduV}. $\sum_{\mathcal{D}_{II'}} {\mathcal{N}_{h_{II'}}^{\mu\nu}\over \mathcal{C}_{h_{II'}}}=i g^{\mu_1\mu_3}g^{\mu_2\mu_4}$.

In sum, for four point amplitude with $k_1$ and $k_3$ shifted, the boundary term is:
\bea
\mathcal{M}_0^{\mu\nu}=&&\sum_{\mathcal{D}_1} {\mathcal{N}_{h_I-1}^{\mu\nu}\over \mathcal{C}_{h_I-1}}-\sum_{\mathcal{D}_I} {\mathcal{C}_{h_I-2} \mathcal{N}_{h_I}^{\mu\nu}\over \mathcal{C}_{h_I-1}^2 }+\sum_{\mathcal{D}_{II'}} {\mathcal{N}_{h_{II'}}^{\mu\nu}\over \mathcal{C}_{h_{II'}}}\nonumber\\
=&&i(\frac{k_4^{\mu_2}\eta^{\mu_4}-k_2^{\mu_4}\eta^{\mu_2}}{k_2\cdot \eta}+\frac{\eta^{\mu_2}\eta^{\mu_4}(k_2^2+k_4^2-2k_1\cdot k_3)}{2(k_2\cdot\eta)^2}+g^{\mu_2\mu_4})g^{\mu_1\mu_3}.
\eea

Then for five point amplitude $\hat{\mathcal M}^{\mu_1\mu_2\mu_3\mu_4\mu_5}$, with $k_1$ and $k_3$ shifted: $k_1\to k_1+z\eta$, $k_3\to k_3-z\eta$, $k_1\cdot \eta=0$, $k_3\cdot \eta=0$, the diagrams are in Figure \ref{fivepoint}.

\begin{figure}[htb]
\centering
\includegraphics[height=8cm,width=15cm]{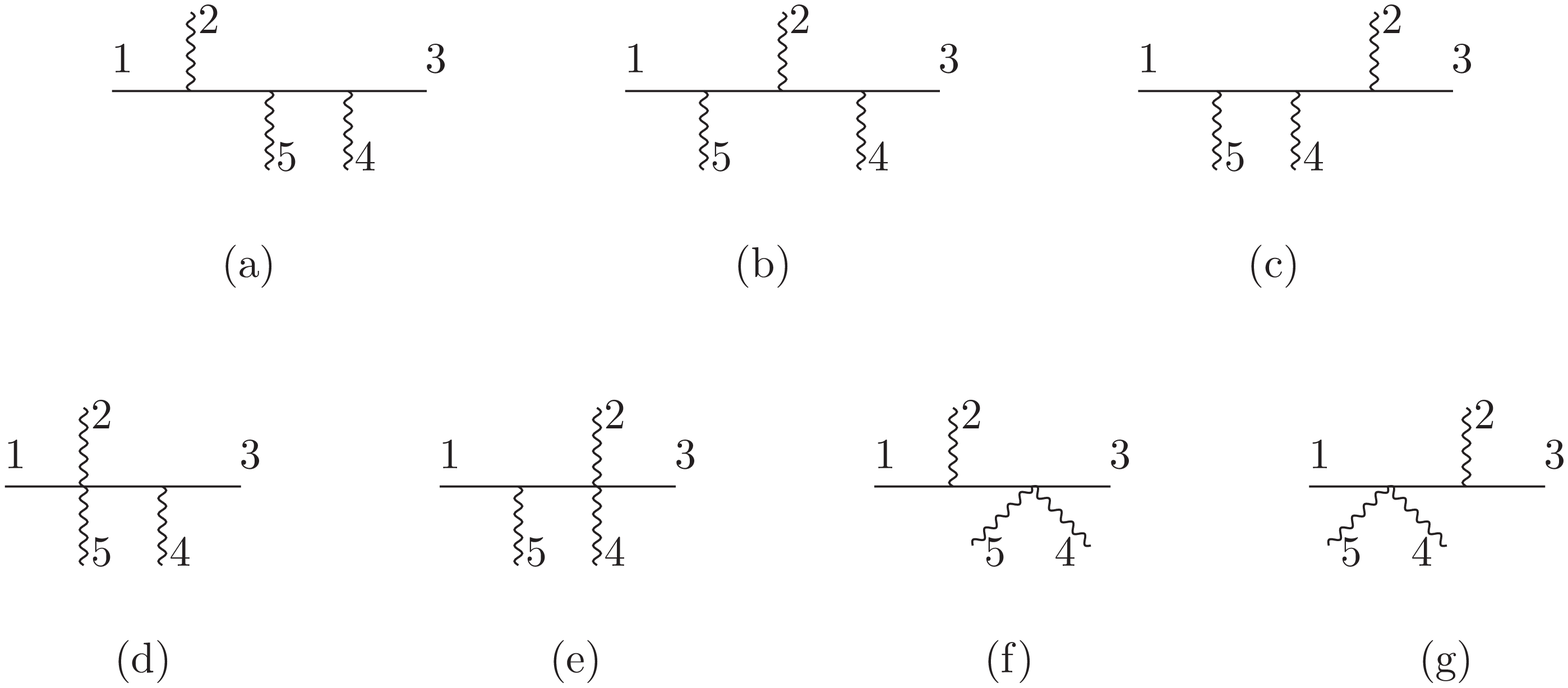}
\caption{The Feynman diagrams for five point amplitude. Legs $k_1$ and $k_3$ are shifted.}
\label{fivepoint}
\end{figure}

From \tref{lemResum} and discussions in \sref{orderz0}, the sum of (f) and (g) in Figure \ref{fivepoint} does not contribute to $\mathcal{M}_0^{\mu\nu}$. Again we can easily obtain each term in $\mathcal{M}_0^{\mu\nu}=\sum_{\mathcal{D}_1} {\mathcal{N}_{h_I-1}^{\mu\nu}\over \mathcal{C}_{h_I-1}}-\sum_{\mathcal{D}_I} {\mathcal{C}_{h_I-2} \mathcal{N}_{h_I}^{\mu\nu}\over \mathcal{C}_{h_I-1}^2 }+\sum_{\mathcal{D}_{II'}} {\mathcal{N}_{h_{II'}}^{\mu\nu}\over \mathcal{C}_{h_{II'}}}$:

\bea
\sum_{\mathcal{D}_1} {\mathcal{N}_{h_I-1}^{\mu\nu}\over \mathcal{C}_{h_I-1}}&=&\frac{i g^{\mu_1\mu_3}}{\sqrt{2}k_2\cdot \eta }(\frac{\eta^{\mu_2}\eta^{\mu_4} k_2^{\mu_5}-\eta^{\mu_4}\eta^{\mu_5} k_5^{\mu_2}}{k_4\cdot \eta}+\frac{\eta^{\mu_4}\eta^{\mu_5} k_4^{\mu_2}-\eta^{\mu_2}\eta^{\mu_5} k_2^{\mu_4}}{k_5\cdot \eta}),\nonumber\\
-\sum_{\mathcal{D}_I} {\mathcal{C}_{h_I-2} \mathcal{N}_{h_I}^{\mu\nu}\over \mathcal{C}_{h_I-1}^2 }&=&\frac{i g^{\mu_1\mu_3}\eta^{\mu_2}\eta^{\mu_4}\eta^{\mu_5}}{2\sqrt{2}(k_2\cdot\eta)^2 k_4\cdot \eta k_5\cdot \eta}(k_2\cdot\eta(k_1+k_5)^2-k_2\cdot\eta(k_3+k_4)^2\nonumber\\
&&\ \ -k_5\cdot\eta(k_1+k_2)^2+k_4\cdot\eta(k_2+k_3)^2),\nonumber\\
\sum_{\mathcal{D}_{II'}} {\mathcal{N}_{h_{II'}}^{\mu\nu}\over \mathcal{C}_{h_{II'}}}&=&\frac{i g^{\mu_1\mu_3}}{\sqrt{2}}(\frac{\eta^{\mu_5}g^{\mu_2\mu_4}}{k_5\cdot\eta}-\frac{\eta^{\mu_4}g^{\mu_2\mu_5}}{k_4\cdot\eta}).
\eea

From the above two simple examples we can see the convenience of our method in calculating boundary terms. First, we have no trouble from complicated tensor contractions, compared to calculating amplitudes in usual way. Only slight amount of algebraic simplification of the expressions is needed. Second, although some individual diagrams contribute to the boundary term, the sum of some such diagrams do not contribute by using \tref{lemResum}, like (f) and (g) in Figure \ref{fivepoint}. This further lighten the burden of calculating boundary terms.


\section{Conclusion}

In this chapter, we have carefully analyzed the boundary behavior of pure Yang-Mills amplitudes under adjacent and non adjacent BCFW shifts in Feynman gauge. We introduced reduced vertexes for Yang-Mills fields, proved that these reduced vertexes are equivalent to the original vertexes, as for the study of the first two orders of boundary behavior, which greatly simplifies our analysis of boundary behavior. Boundary behavior for adjacent shifts is readily obtained using reduced vertexes. Then we find that the boundary behavior for non-adjacent shifts is much better than those of adjacent shifts. Comparing to adjacent shifts, non adjacent shifts allow us to permute the external legs while retaining color ordering. We proved a theorem about permutation sum, which plays key roles in our analysis of non-adjacent boundary behavior, and the theorem is the essential reason for the improvement of boundary behavior for non adjacent shifts compared to adjacent shifts. The conclusions are, $\mathcal{O}(z^{1})$ of $\mathcal{M}^{\mu\nu}$ is proportional to metric $g^{\mu\nu}$ for adjacent shifts, and vanishes for non adjacent shifts; $\mathcal{O}(z^{0})$ of $\mathcal{M}^{\mu\nu}$ is metric term plus antisymmetric term for adjacent shifts, and is proportional to $g^{\mu\nu}$ for non adjacent shifts. Based on the boundary behavior, we find that it is possible to generalize BCFW recursion relation to calculate general tree level off shell amplitudes, with the aid of previous papers \cite{Chen1,Chen2,Chen3}. The procedure is described in the second section, before we discuss boundary behavior.

We proved that boundary behavior at $\mathcal{O}(z^1)$ and $\mathcal{O}(z^0)$ does not depend on whether the external legs are on shell or not. We also analyzed the $\mathcal{O}(z^{-1})$ behavior for non adjacent shifts. When all the external legs are on shell, $\mathcal{O}(z^{-1})$ of $\mathcal{M}^{\mu\nu}$ is metric term plus antisymmetric term. When some external legs are off shell, we also give the general form of the contribution to $\mathcal{O}(z^{-1})$ from each off shell leg, which manifestly vanishes when the leg becomes on shell. For on shell loop level amplitudes, the loop lines can be dealt with as off shell legs here and has the contribution to $\mathcal{O}(z^{-1})$ in the form we have obtained, which seems very likely to cancel the ghost loop contributions, resulting in some good $\mathcal{O}(z^{-1})$ behavior for loop level non adjacently shifted on shell amplitudes. This deserves our further investigation.

Our conclusions on boundary behavior in Feynman gauge are consistent with those in AHK gauge in \cite{Boels,Nima1}. Our work has two major advantages. First, we can present a procedure to calculate general tree level off shell amplitudes using BCFW technique and the technique in \cite{Chen3}. And the second is related to our permutation sum theorem, ie. \tref{lemResum}. This theorem tells us why the amplitudes with non-adjacent BCFW shifts have improved boundary behavior. Actually, in \cite{Boels} there are several important conjectures about the relationship between the improved boundary behavior and the general permutation sums, which are partially proved in \cite{Du:2011se}. Hopefully, some generalization of our theorem here will be helpful for the proof of these conjectures. This will be left for further work.

\chapter{Ward Identity Implies Recursion Relation at Tree and Loop Levels }\label{Wardidenrec}

\allowdisplaybreaks{}

\section{Introduction}


At tree level, the amplitudes of pure Yang-Mills fields can be written as rational functions of external momenta and polarization vectors in spinor form \cite{Parke:1986gb,Xu:1986xb,Berends:1987me,Kosower,Dixon:1996wi,Witten1}. Such rational functions can be analyzed in detail in algebra system. According to this, BCFW recursion relation was proposed and developed in \cite{Britto:2004nj,Britto:2004nc,Britto:2004ap}, and then proved in \cite{Britto:2005fq} using the pole structure of the tree level on shell amplitudes. This has been an exiting progress on the amplitudes in pure Yang-Mills theory. For theory with massive fields \cite{Badger:2005zh,Ozeren:2006ft,Schwinn,Chen:2011ve,Chen3}, the amplitudes are also rational functions of external momenta and polarization vectors in spinor form.


At loop level, although the whole amplitudes are no longer rational functions in general, they can be decomposed into some basic scalar integrals with coefficients being rational functions of external spinors \cite{Bern:1994zx,BernD2}. The coefficient structures are studied in depth in \cite{Dixon4,Bern:2005hh,Bern1}. On the other hand, the integrands of the amplitudes are rational functions of the external spinors and integral momenta. For the N=4 planar super Yang-Mills theory, \cite{Nima}  gives an explicit recursive formula for the all-loop integrands of scattering amplitudes.

The amplitudes in gauge theory are constrained by gauge symmetry. This leads to Ward identity which constrains the amplitudes at all loop levels. Inspired by the BCFW momenta shift, Gang Chen considered the Ward identity for tree level amplitudes with complexified momenta for a pair of external legs, and then obtained a recursion relation for the boundary terms using BCFW technique in a recent article \cite{Chen1}. However, in \cite{Chen1}, the author chose a particular momenta shift such that the external states of the complexified legs are independent of the complex parameter $z$. Then a natural question is how to obtain recursion relations for other possible momenta shifts. Furthermore, is it possible to obtain the full amplitudes from the Ward identity,  and to extend the technique to one loop amplitudes? In this chapter, we will give positive answers to all these questions. 

In Section \ref{TreeRec}, we first give the proof of Ward identity at tree level using Feynman rules directly, and then derive the recursion relation for off shell tree level amplitudes, where the cancellation details in the proof of Ward identity helps to simplify the recursion relations. Section \ref{LoopWard} is parallel to Section \ref{TreeRec}. We first extend the proof of Ward identity to one loop level and then derive the recursion relation for one loop off shell amplitudes. Our technique does not rely on the on-shell momenta shifts. Also, in our calculation using the recursion relation, four point vertexes are not used explicitly. We calculate three and four point one loop off shell amplitudes as examples in Section \ref{example}.

\section{Ward Identity and Implied Recursion Relation at Tree Level}\label{TreeRec}
In \cite{Chen1}, Gang Chen directly proved complexified Ward identity for pure Yang-Mills amplitudes at tree level, and then used it to deduce a recursion relation for the boundary terms of the tree level complexified amplitudes. Here we generalize the method to deduce a recursion relation for tree level amplitudes with one external off shell leg. This section will serve as a basis for our generalization to one loop level in the next section. We will call the external off shell leg $\Le$ with momentum $\ke^\mu$, and the corresponding off shell amplitudes $A_\mu$, with the propagator corresponding to $\ke$ stripped.


\subsection{Proof of Ward Identity at Tree Level}
Although done in a previous paper \cite{Chen1}, we briefly summerize some key points in the direct proof of tree level Ward identity, since these points are useful for deriving tree level recursion relation and also will be part of the proof at one loop level in the next section.

The amplitude is complexified by shifting the momenta of a pair of external legs. We choose $\Le$ and one on shell leg $L_s$ with momentum $k_s=\lmd_s\td\lmd_s$, and the shift is:
\beq
k_s \to k_s-z \eta,\ \ \ \ \ \ \ \ \ \ \ \ \ \ \ke\to \ke+z\eta,
\label{momshift}
\eeq
where $z$ is the complexifing parameter and $\eta$ should satisfy $\eta^2=0$ and $k_s \cdot \eta=0$. Unlike in \cite{Chen1} in which one should choose a particular momenta shift such that the external states of the complexified legs are independent of the complex parameter $z$, here we do not need this requirement and have more freedom in choosing $\eta$.


The color ordered Feynman rules of the gauge field are as in \cite{Dixon:1996wi}, also given in (\ref{colorFeynrule}), with outgoing momenta. We also write the Feynman rules for ghost fields here in Figure \ref{ghostFeynrule}, which will be used in the next section. 

\begin{figure}[htb]
\centering
\includegraphics[height=3.5cm,width=15cm]{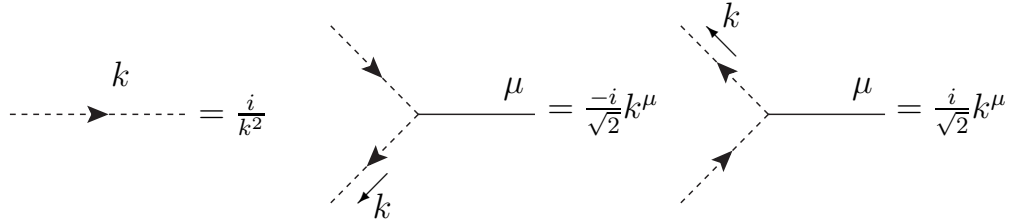}
\caption{Ghost field color ordered Feynman rules. Dashed line for ghost field and solid line for gauge field.}
\label{ghostFeynrule}
\end{figure}

For a three point vertex with legs 1, 2 and $\Le$ in anti-clockwise order, we write it in the following form:\footnote{At the time the works in Chapter \ref{boubehav} and \ref{Wardidenrec} were done, a little different conventions were used, thus we give the conventions again. Chapter \ref{boubehav} and \ref{Wardidenrec} can be read independently. Note that (\ref{kedotV}) and (\ref{kdotV4}) are actually different but equal.}
\bea\label{newV3}
V_{\mu_1\mu_2\mu}&\equiv&S_{\mu_1\mu_2\mu}+ R_{\mu_1\mu_2\mu}+M_{\mu_1\mu_2\mu},
\eea
where  
\bea\label{newV3s}
S_{\mu_1\mu_2\mu}&=&\frac{i}{\sqrt 2}\left(\eta_{\mu_1\mu_2}(k_1-k_2)_{\mu}\right) \nb\\
R_{\mu_1\mu_2\mu}&=&\frac{i}{\sqrt 2}\left(-2\eta_{\mu_2\mu}(\ke)_{\mu_1}+2\eta_{\mu\mu_1}(\ke)_{\mu_2}\right) \nb\\
M_{\mu_1\mu_2\mu}&=&\frac{i}{\sqrt 2}\left(-\eta_{\mu_2\mu}(k_1)_{\mu_1}+\eta_{\mu\mu_1}(k_2)_{\mu_2}\right). 
\eea
We will refer to these terms as S, R and M parts of the vertex. Contracting this vertex with $\ke$, we get:
\begin{equation}
\ke^\mu \cdot V_{\mu_1\mu_2\mu}=\frac{i}{\sqrt{2}}\eta_{\mu_1\mu_2} k_2^2-\frac{i}{\sqrt{2}}\eta_{\mu_1\mu_2} k_1^2+\frac{i}{\sqrt{2}}k_{2\ \mu_2}k_3{}_{\ \mu_1}-\frac{i}{\sqrt{2}}k_{1\ \mu_1}k_3{}_{\ \mu_2},
\label{kedotV}
\end{equation}
and we represent these terms by the symbols in Figure \ref{vertexnotation}. These terms are frequently used throughout the chapter, and we will call the terms in the first line of Figure \ref{vertexnotation} as solid triangle terms, and the second line terms as hollow triangle terms.
\begin{figure}[htb]
\centering
\includegraphics[height=4cm,width=12cm]{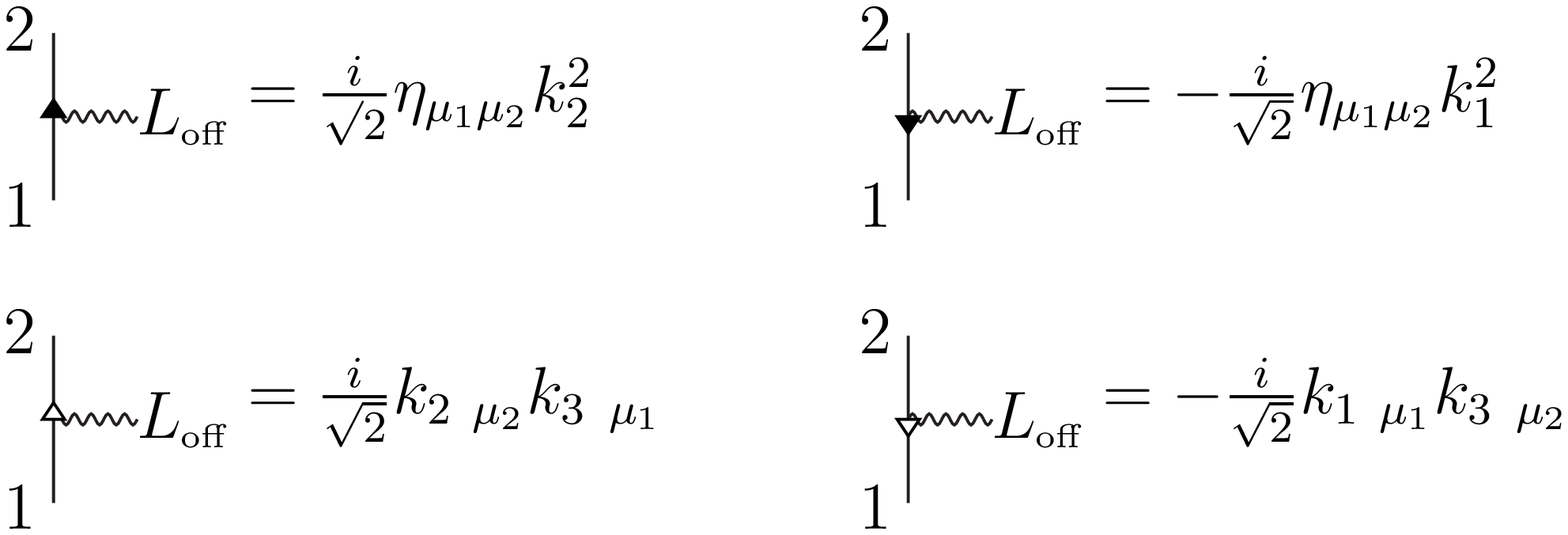}
\caption{Notations for \eref{kedotV}. We specialize $\Le$ using photon line.}
\label{vertexnotation}
\end{figure}

Then a proof of tree level Ward identity can be shown in several steps. Assume it holds for N-point and less than N-point amplitudes (for example, 3-point case can be immediately checked), we will show how it holds for (N+1)-point amplitudes. We choose $\Le$ as the (N+1)-th leg. We can construct an (N+1)-point color ordered diagram from an N-point one by inserting $\Le$ to an N-point diagram between leg 1 and leg N.

First, when $\Le$ is inserted to a propagator or leg 1 or leg N, we denote the vertex as $V_{\mbox{\tiny off}}$, and contract it with $\ke$, the following two hollow triangle terms in Figure \ref{treecancel1} vanish due to less-point Ward identities or the on-shell conditions of leg 1 or N. The meaning of the symbols are in Figure \ref{vertexnotation}.

\begin{figure}[htb]
\centering
\includegraphics[height=2.5cm,width=10cm]{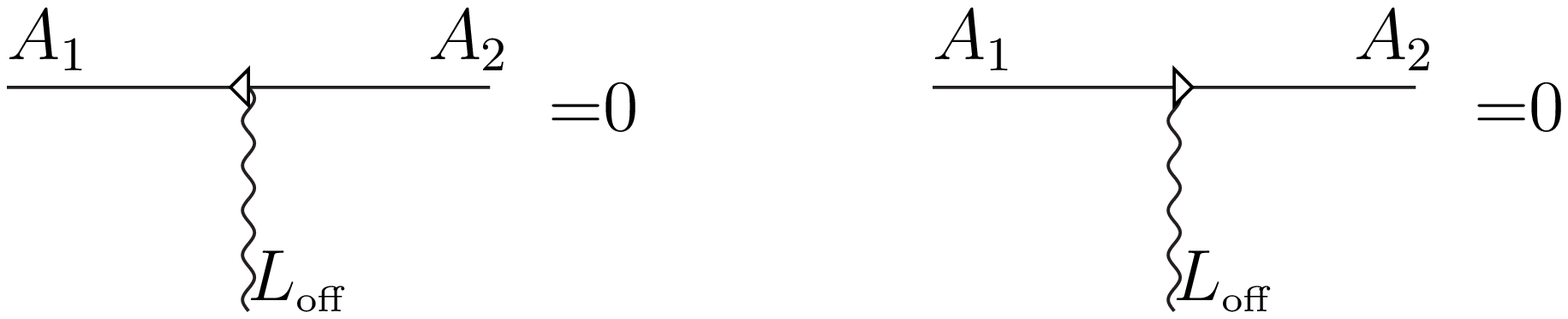}
\caption{When $\Le$ is inserted to a propagator or leg 1 or leg N, these terms vanish due to less point Ward identity or the on-shell conditions of leg 1 or N. $A_1$ and $A_2$ are sub amplitudes.}
\label{treecancel1}
\end{figure}

Second, $\Le$ is inserted to a three-point vertex in the N-point diagram. These terms and the remaining terms from the above case--ie. solid triangle terms--can be re-combined as in Figure \ref{treecancel2} to cancel each other.
\begin{figure}[htb]
\centering
\includegraphics[height=10cm,width=15cm]{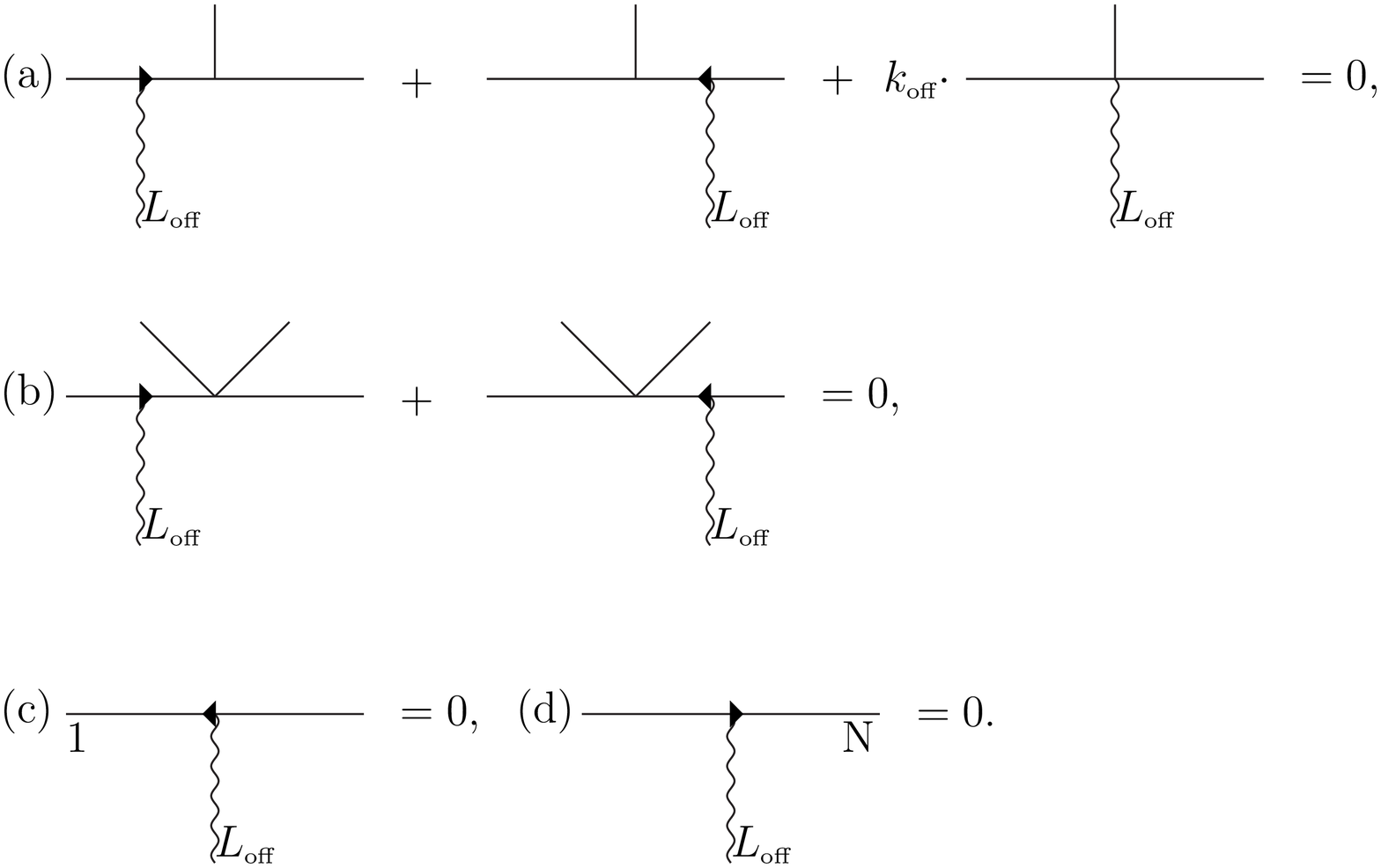}
\caption{A group of diagrams cancel. In (a) and (b), the cancellation is solely due to the vertex structures, not dependent on whether the legs are on shell or off shell. (c) and (d) are due to on shell conditions for leg 1 and leg N: $k_1^2=0$ and $k_N^2$=0.} 
\label{treecancel2}
\end{figure}

Figure \ref{treecancel1} and Figure \ref{treecancel2} constitute the proof of Ward identity at tree level.

\subsection{Recursion Relation for Tree Level Off Shell Amplitudes}\label{treerec}
As discussed in \cite{Chen1}, from the complexified Ward identity ${{\hat k}_{\mbox{\tiny off}}}^\mu \cdot {\hat A}_\mu=0$, by a derivative over $z$ we get:  
\be\label{recz}
\hat A_\mu \eta^\mu |_{z\to 0}=-{d\hat A_\mu\over dz}{\hat k}_{\mbox{\tiny off}}^\mu |_{z\to 0}.
\ee
The symbol $\hat{}$ represents that the quantity is complexified, ie. depends on the shift parameter z. Here $\ke$ is shifted as in \eref{momshift}: ${\hat k}_{\mbox{\tiny off}}=\ke+z\eta$. Our destination is to calculate $A_\mu$, and we will realize it by calculating the right hand side of \eref{recz}.


We name the vertex which contains $\Le$ as $V_{\mbox{\tiny off}}$. At tree level, we have the following three cases: 
\begin{enumerate}
\item the derivative acts on a propagator;
\item the derivative acts on a three point vertex which does not contain $\Le$;
\item the derivative acts on a three point vertex $V_{\mbox{\tiny off}}$.
\end{enumerate}
In the first and second cases, when $V_{\mbox{\tiny off}}$ is a three point vertex, we write $\ke^\mu \cdot \Ve{}_{\ \mu}$ as in Figure \ref{vertexnotation}, and take out the hollow triangle terms. These terms, together with the terms from the third case where the derivative acts on the M part of $\Ve{}_{\ \mu}$, add up to be 0 due to Ward identity.


From above we know that in the first and second cases, we only need the solid triangle terms for $\ke^\mu \cdot \Ve{}_{\ \mu}$ as represented in Figure \ref{vertexnotation}, when $\Ve$ is a three point vertex; in the third case, $\frac{d}{dz}$ only need to act on the S and R parts of $\Ve{}_{\ \mu}$ as written in \eref{newV3s}. The first two cases can be further simplified. Due to (a) and (b) in Figure \ref{treecancel2}, the terms relevant for the first two cases are reduced to those with $\ke$ neighboring to the three point vertex or the propagator to be differentiated, as depicted in Figure \ref{Tree1}.

\begin{figure}[htb]
\centering
\includegraphics[height=13cm,width=10cm]{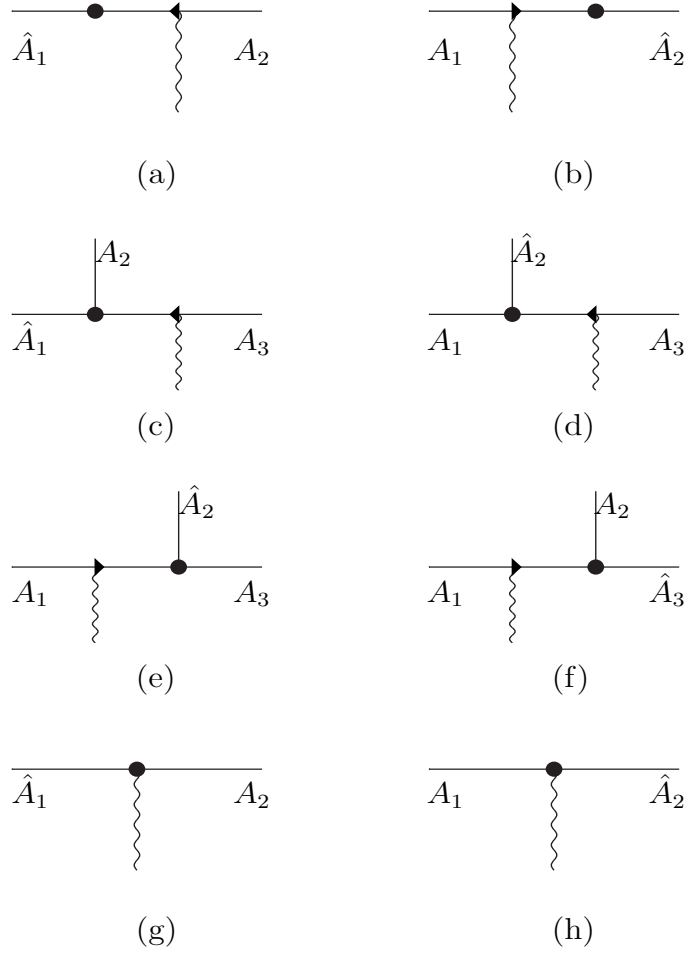}
\caption{Terms to be calculated for tree-level off-shell amplitudes. Here and following, the dark solid circle symbol $\bullet$ denotes where we act $d\over dz$. We shift $\Le$ and some other on shell leg $L_s$. $\{A_i\}$ denote some sub-diagrams with less external states. $\hat A_k$ with hat includes $L_s$. In different diagrams, the same $A_k$ symbols do not mean the same sub amplitudes. They sum over all allowed sub amplitudes.}
\label{Tree1}
\end{figure}


Thus, for the first case, the diagrams are (a) and (b) in Figure \ref{Tree1}. The contributions from (a) and (b) to $-{d\hat A_\mu\over dz}{\hat k}_{\mbox{\tiny off}}^\mu |_{z\to 0}$ are:
\bea
&&\mbox{for (a)}\ \ \frac{-i\sqrt{2}}{k_{A_1}^2 k_{A_2}^2}k_{A_1}\cdot \eta\  A_1\cdot A_2,\nonumber\\
&&\mbox{for (b)}\ \ \frac{i\sqrt{2}}{k_{A_1}^2 k_{A_2}^2}k_{A_2}\cdot \eta\  A_1\cdot A_2.
\label{preexptree1}
\eea
As noted in Figure \ref{Tree1}, \{$A_i$\} are some less point amplitudes; $k_{A_i}$ is the total momentum of the external legs contained in the sub amplitude $A_i$. If some $A_i$ just contains one external leg $L_m$, we define this $A_i$ to be $i k_m^2 \epsilon_m$.


The second case corresponds to (c) (d) (e) and (f) in Figure \ref{Tree1}, and the contributions to $-{d\hat A_\mu\over dz}{\hat k}_{\mbox{\tiny off}}^\mu |_{z\to 0}$ are:
\bea
&&\mbox{for (c)}\ \ \frac{1}{2 k_{A_1}^2 k_{A_2}^2 k_{A_3}^2}(A_3\cdot \eta\  A_1\cdot A_2+A_1\cdot \eta\  A_2\cdot A_3-2A_2\cdot \eta\  A_1\cdot A_3),\nonumber\\
&&\mbox{for (d)}\ \ \frac{1}{2 k_{A_1}^2 k_{A_2}^2 k_{A_3}^2}(-A_3\cdot \eta\  A_1\cdot A_2+2A_1\cdot \eta\  A_2\cdot A_3-A_2\cdot \eta\  A_1\cdot A_3),\nonumber\\
&&\mbox{for (e)}\ \ \frac{-1}{2 k_{A_1}^2 k_{A_2}^2 k_{A_3}^2}(-2A_3\cdot \eta\  A_1\cdot A_2+A_1\cdot \eta\  A_2\cdot A_3+A_2\cdot \eta\  A_1\cdot A_3),\nonumber\\
&&\mbox{for (f)}\ \ \frac{-1}{2 k_{A_1}^2 k_{A_2}^2 k_{A_3}^2}(-A_1\cdot \eta\  A_2\cdot A_3+2A_2\cdot \eta\  A_1\cdot A_3-A_3\cdot \eta\  A_1\cdot A_2).
\label{preexptree2}
\eea

And the third case corresponds to (g) and (h) in Figure \ref{Tree1}, whose contributions are:
\bea
&&\mbox{for (g)}\ \ \frac{-i}{\sqrt{2} k_{A_1}^2 k_{A_2}^2}(\ke\cdot \eta \ A_1\cdot A_2+2A_1\cdot \eta\  \ke\cdot A_2-2A_2\cdot \eta\  \ke\cdot A_1),\nonumber\\
&&\mbox{for (h)}\ \ \frac{-i}{\sqrt{2} k_{A_1}^2 k_{A_2}^2}(-\ke\cdot \eta\  A_1\cdot A_2+2A_1\cdot \eta\  \ke\cdot A_2-2A_2\cdot \eta\  \ke\cdot A_1).
\label{preexptree3}
\eea
As explained before \eref{preexptree1}, in this case $\frac{d}{dz}$ only need to act on the S and R parts of $\Ve{}_{\ \mu}$ as written in \eref{newV3s}.

It can be observed that, \eref{preexptree2} for $L_s$ contained in $A_1$ or $A_2$ or $A_3$, the expressions are the same. In the case when $L_s$ is contained in $A_2$ we should sum (d) and (e) in Figure \ref{Tree1} to see that the expression is the same as when $L_s$ is contained in $A_1$ or $A_3$. The common expression is:
\beq
\frac{1}{2 k_{A_1}^2 k_{A_2}^2 k_{A_3}^2}(A_3\cdot \eta\  A_1\cdot A_2+A_1\cdot \eta\  A_2\cdot A_3-2A_2\cdot \eta\  A_1\cdot A_3).
\label{finaltree1}
\eeq
\eref{preexptree1} and \eref{preexptree3} summed up also give a common expression, regardless of whether $L_s$ is contained in $A_1$ or $A_2$:
\beq
\frac{-i}{\sqrt{2} k_{A_1}^2 k_{A_2}^2}\left(\ (k_{A_1}-k_{A_2})\cdot \eta\  A_1\cdot A_2+2A_1\cdot \eta\  \ke\cdot A_2-2A_2\cdot \eta\  \ke\cdot A_1\right).
\label{finaltree2}
\eeq
The final tree level result for $A_\mu \eta^\mu$ is the sum of \eref{finaltree1} and \eref{finaltree2}. There are three choices for $\eta$, ie. $\epsilon_s^+$, $\epsilon_s^-$ and $k_s$, satisfying $\eta^2=0$ and $k_s\cdot\eta=0$. In four dimensional spacetime, the last component of $A_\mu$ can then be determined by the Ward identity $\ke^\mu A_\mu=0$.


Compare to Berends-Giele recursion relation \cite{Berends:1987me}, it is seen that \eref{finaltree1} corresponds to $\ke$ contained in a four point vertex, and \eref{finaltree2} is equivalent to the contribution when $\ke$ is contained in a three point vertex. Thus our method at tree level is equivalent to Berends-Giele recursion relation. This on one hand supports the correctness of our method, and on the other hand a little undermines the value of our method at tree level. There are also other recursion relations for off shell tree level amplitudes, eg. \cite{Feng:2011tw}. Yet we are going to extend our method to one loop level, where the situation is much more complicated and our method is new. 


\section{Ward Identity and Implied Recursion Relation at 1-loop Level}\label{LoopWard}

In this section we are going to extend our method to 1-loop level. We will show how complexified Ward identity holds at 1-loop level and then we deduce the corresponding recursive calculation of 1-loop off shell amplitudes. Using our method, we will calculate three and four point 1-loop off shell amplitudes as examples. In our calculation we use FDH scheme \cite{Bern:1991aq}, in which only the loop momentum is continued to dimensionality different from 4.

We first explain some subtleties at loop level. First, after momentum shifting, some lines on the loop carry complex momenta. This brings ambiguities to the meaning of the loop integral and prevents us from translating the loop momentum $l\to l+k$ or flip it $l\to -l$. However, according to \eref{recz}, what we need for our technique is the derivative of the integral at the value $z\to 0$. And it is easy to prove that:
\begin{eqnarray}
&&\int d^D l \frac{d}{dz}f(l^\mu,{\hat k}^\mu)|_{z\to 0}=\int d^D l \frac{d}{dz}f(-l^\mu,{\hat k}^\mu)|_{z\to 0},\\
&&\int d^D l \frac{d}{dz}f(l^\mu,{\hat k}^\mu)|_{z\to 0}=\int d^D l \frac{d}{dz}f(l^\mu+{\hat k}'^\mu,{\hat k}^\mu)|_{z\to 0}.
\label{complexintegrand}
\end{eqnarray}
Thus for our technique, we can translate or flip the loop momentum even when the integrand is complexified.


Second, some attention should be paid to color orderings and symmetry factors. At tree level there is only one color ordering contributing to the primitive part of the color ordered amplitudes. At one loop level, most diagrams also only have one color ordering. However, for gauge field loop diagrams, there are three kinds of diagrams having two color orderings. Those are diagrams with two vertexes on the loop: two three-point vertexes; two four-point vertexes; a three-point vertex and a four-point vertex. For the first and second cases, the contributions from the two color orderings are the same at integrand level. For the third case, the contributions from the two color orderings at integrand level differ by a translation and flip of the loop momentum, and due to \eref{complexintegrand} the two orderings contribute the same in our method after integration. In a word, these three kinds of diagrams have a factor of 2 from possible color orderings. At the same time, these three kinds of diagrams have symmetry factor $\frac{1}{2}$, just canceling the doubling from color orderings. For ghost loop diagrams, those with two vertexes on the loop also have a doubling from two color orderings, while there is only either clockwise or anti clockwise ghost loop when there are only two vertexes on the ghost loop. We replace the doubling from color orderings by drawing both clockwise and anti-clockwise ghost loop diagrams which are actually equal when there are only two vertexes on the ghost loop.

Finally, as our convention for the loop momentum for all our loop diagrams, we specify the loop momentum in the following way. For each external leg $L_i$, when we want to make a path from it to the loop, there is one definite vertex $V$ on the loop first encountered in the path, then we say the external leg $L_i$ is associated with this loop vertex $V$. We find the vertex with which $\Le$ is associated and call it $V_0$. Assume all the legs associated with $V_0$ in color ordering are $L_j, L_{j+1}, \cdots, L_N, \Le, L_1, \cdots, L_i$, then we assign the momentum of the first loop propagator on the counter clockwise side of $V_0$ as $l-k_1-\cdots-k_i$, with the loop momentum flowing in counter clockwise direction. Loop momentum $l$ is to be integrated.

\subsection{Proof of Ward Identity at 1-Loop Level}\label{loopWI}

In this section we use $A^l$ for 1-loop amplitudes, and $A^t$ for tree level amplitudes.

Two point and three point 1-loop Ward identity is easy to verify directly. Similar to the proof at tree level, we use induction, assume Ward identity holds for N and less than N point one loop amplitudes, and construct an (N+1) point diagram from an N point one by inserting $\ke$ in different places. We denote the vertex with $\ke$ as $V_{\mbox{\tiny off}}$ and when $V_{\mbox{\tiny off}}$ is a three point vertex, we decompose $\ke \cdot V_{\mbox{\tiny off}}$ as in Figure \ref{vertexnotation}.

{\bf Case 1.} When $\ke$ is linked to a propagator (including gauge field loop propagator) or external leg of the N point diagram, the solid triangle terms from $\ke \cdot V_{\mbox{\tiny off}}$ mostly cancel the terms with $\ke$ in a four point vertex, in the manner of Figure \ref{treecancel2}. Only the terms in Figure \ref{loopremain1} remain. 
\begin{figure}[htb]
\centering
\includegraphics[height=3.8cm,width=13cm]{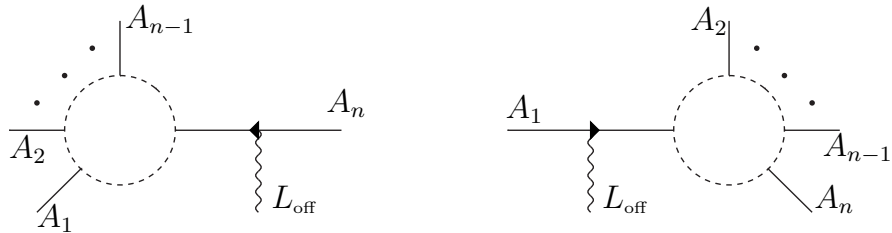}
\caption{The remaining terms in the first case that does not cancel in the manner of Figure \ref{treecancel2}. The loop is ghost loop and has two directions.}
\label{loopremain1}
\end{figure}

{\bf Case 2.} We need to consider the hollow triangle terms from $\ke \cdot V_{\mbox{\tiny off}}$ remaining from the above case, and we divide them into two sub cases:

\ \ {\bf \small Sub Case 1.} When $V_{\mbox{\tiny off}}$ is not on the loop, these terms vanish due to Ward identity for less point amplitudes in the induction assumption, similar to the tree level counterpart Figure \ref{treecancel1}.

\ \ {\bf \small Sub Case 2.} The remaining sub case is that $V_{\mbox{\tiny off}}$ is on the gauge field loop. We analyze one of the hollow triangle terms in Figure \ref{goontheloop}. The Figure has considered all the possible cases with the first right side vertex to be three or four point, and different types of second vertex relevant. When the first right side vertex is a three point vertex, acting on it with one of the factor in the hollow triangle term, we can again decompose it as in Figure \ref{vertexnotation} into solid and hollow triangle terms. (a) and (b) in Figure \ref{goontheloop} are in fact the same diagrams as in Figure \ref{treecancel2}. (c) vanishes due to tree level Ward identity, and (d) is due to on shell condition for external legs besides $\Le$. Then the type of term in (e) of Figure \ref{goontheloop} remains, which is a hollow triangle term staying on the loop, and it will act on the next vertex on the loop, repeating the same processes as in (a)-(d) of Figure \ref{goontheloop}, until it meets the final vertex on the loop. For this sub case, the remaining diagrams are in Figure \ref{loopremain2}.

\begin{figure}[htb]
\centering
\includegraphics[height=15cm,width=14cm]{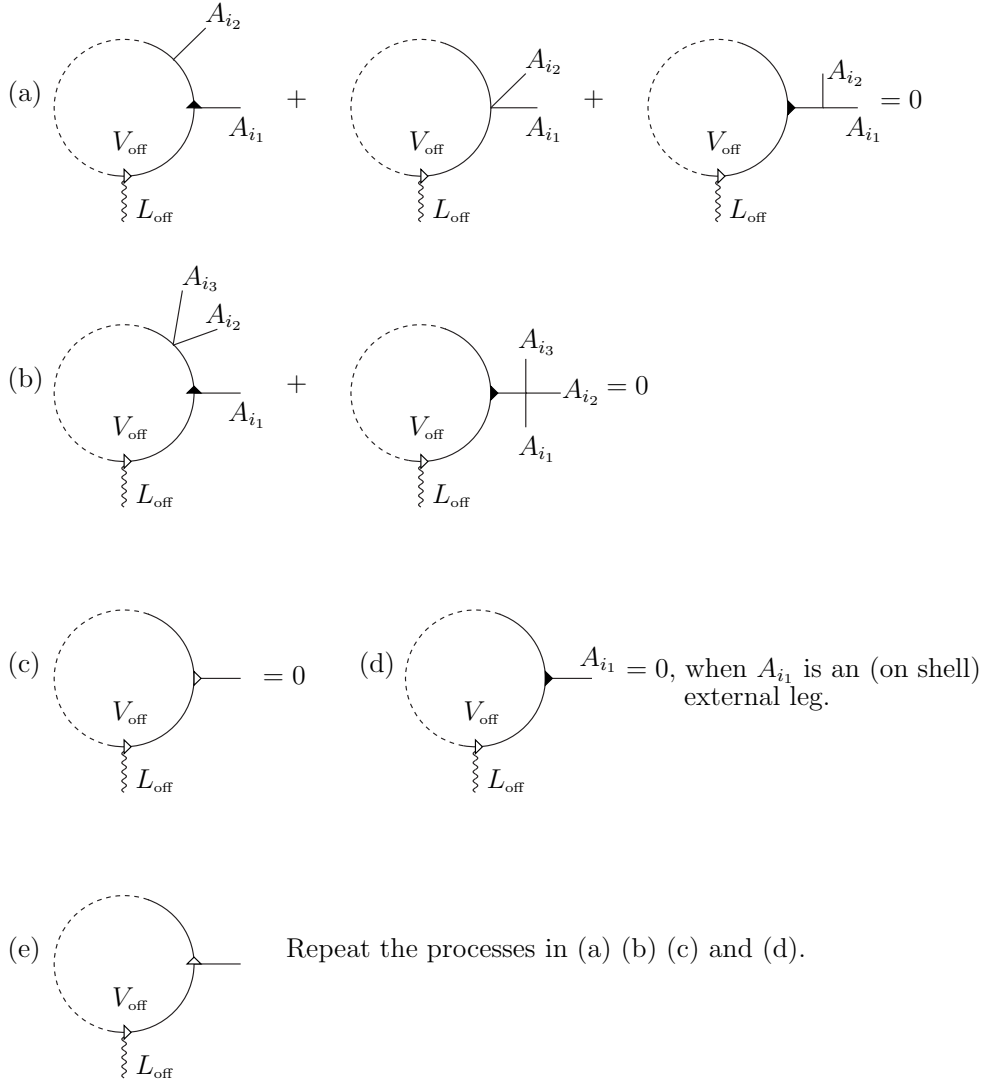}
\caption{Analysis of the action of the hollow triangle terms in {\bf\small Sub Case 2}. The dashed line is not ghost field, but just part of the loop diagram not relevant.}
\label{goontheloop}
\end{figure}

\begin{figure}[htb]
\centering
\includegraphics[height=8cm,width=12cm]{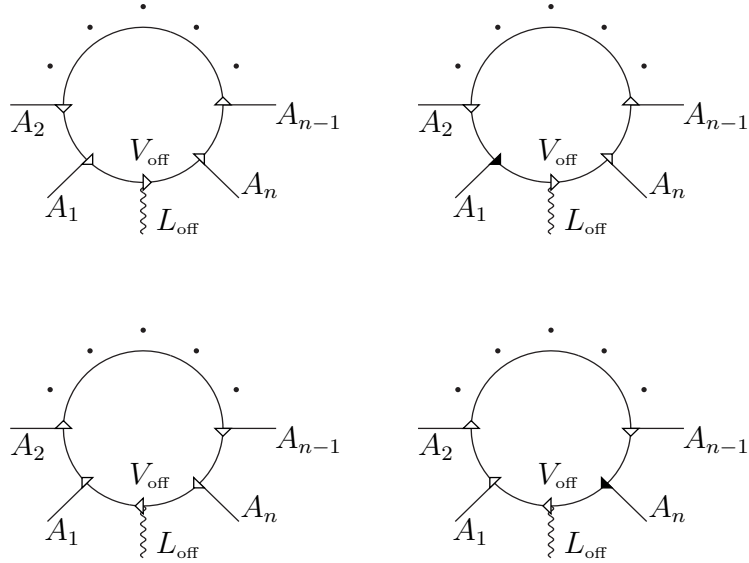}
\caption{The terms from {\bf \small Sub Case 2}. Except the hollow triangle terms at $V_{\mbox{\tiny off}}$, other hollow and solid triangle terms on the loop are induced from the hollow triangle term at the previous loop vertex, as described in the text of {\bf\small Sub Case 2.} and Figure \ref{goontheloop}.}
\label{loopremain2}
\end{figure}


{\bf Case 3.} The remaining case: $\ke$ is linked to a ghost propagator of the N point diagram, as in Figure \ref{loopremain3}.

\begin{figure}[htb]
\centering
\includegraphics[height=4cm,width=6cm]{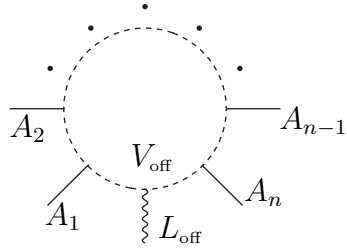}
\caption{Diagram for {\bf Case 3.} The ghost loop can be in two directions.}
\label{loopremain3}
\end{figure}

By direct and simple calculations, the terms from Figure \ref{loopremain1}, Figure \ref{loopremain2} and Figure \ref{loopremain3}, with same set of sub amplitudes $A_i$, add up to be 0. Combine {\bf Case 1, 2, 3}, we have proven that Ward identity holds at $N+1$ point one loop level. Thus by induction we have proven Ward identity at one loop level using Feynman rules in a direct way.


\subsection{Recursion Relation for Loop Level Off Shell Amplitudes}\label{LoopRec}
Similar to the tree level off shell amplitudes, we can use $\hat A_\mu \eta^\mu |_{z\to 0}=-{d\hat A_\mu\over dz}{\hat k}_{\mbox{\tiny off}}^\mu |_{z\to 0}$ to calculate one loop level off shell amplitudes. The experience at tree level, and the details of how Ward identity holds at one loop level discussed in the last subsection, help us to simplify our discussion and calculation of one loop level off shell amplitudes.


When the derivative acts on a gauge field propagator or a vertex which is not on the loop, we can use the expressions derived in Section \ref{treerec} directly, ie. \eref{finaltree1} and \eref{finaltree2}, for the contributions to $-{d\hat A_\mu\over dz}{\hat k}_{\mbox{\tiny off}}^\mu |_{z\to 0}$:
\bea
&&\frac{1}{2 k_{A_1}^2 k_{A_2}^2 k_{A_3}^2}(A^{l/t}_3\cdot \eta\  A^{l/t}_1\cdot A^{l/t}_2+A^{l/t}_1\cdot \eta\  A^{l/t}_2\cdot A^{l/t}_3-2A^{l/t}_2\cdot \eta\  A^{l/t}_1\cdot A^{l/t}_3),\label{finalloop1}
\\
&&\frac{-i}{\sqrt{2} k_{A_1}^2 k_{A_2}^2}\left(\ (k_{A_1}-k_{A_2})\cdot \eta\  A^{l/t}_1\cdot A^{l/t}_2+2A^{l/t}_1\cdot \eta\  \ke\cdot A^{l/t}_2-2A^{l/t}_2\cdot \eta\  \ke\cdot A^{l/t}_1\right).\nonumber
\eea
In \eref{finalloop1}, we allocate the on shell external legs into $\{A^{l/t}_i\}$ in color ordering, with one and only one $A_i^{l/t}$ being one loop level. As in tree level, in each expression we should sum over all allowed allocations of the on shell external legs into $\{A^{l/t}_i\}$.

When the derivative acts on a gauge field loop propagator or a loop vertex, these are shown in Figure \ref{loop1}. For the same reasons as discussed in tree level recursion calculation, in (a) to (f), we only need to consider $\Le$ next to the propagator or vertex differentiated and only need the solid triangle term. In (g), we only differentiate the S and R terms of the vertex. The contributions from M part of the vertex in (g) will be dealt with in the following and are grouped into \eref{finalloop3}. In Figure \ref{loop1}, we encounter tree level two leg off shell amplitudes $A^t_{\sigma\rho}$. This quantity can also be calculated recursively using our method, but in this chapter we will not discuss it, and will use Feynman rules to calculate it in our example. Those $A^t$ without sub indices are tree level one leg off shell amplitudes, with the off shell leg index suppressed. (a) is 0 due to our convention for the loop momentum, described in the paragraph before Section \ref{loopWI}.

Regardless of whether the other shifted leg $L_s$ is among $\{L_1,L_2,\cdots,L_j\}$ or among $\{L_{j+1},\cdots$, $L_N\}$, the contributions to the integrand of $-{d\hat A_\mu\over dz}{\hat k}_{\mbox{\tiny off}}^\mu |_{z\to 0}$ from Figure \ref{loop1} are (we use $K_{m,n}$ to represent for $k_m+k_{m+1}+\cdots+k_n$):
\bea
&&(a)\ \ \ \ \ \ \ \,:0\label{finalloop2}\\
&&(b)+(g): \frac{-i}{\sqrt{2} l^2 (l+\ke)^2}(\ (2 l+\ke)\cdot \eta\  A^t_{\sigma\rho}(1,2,\cdots,N) \ g^{\sigma\rho}\nonumber\\
&&\ \ \ \ \ \ \ \ \ \ \ \ \ +2\eta^\sigma\  A^t_{\sigma\rho}(1,2,\cdots,N) \ k^\rho-2\eta^\rho\  A^t_{\sigma\rho}(1,2,\cdots,N) \ k^\sigma)\nonumber\\
&&(d)+(e): \frac{1}{2 l^2(l-K_{1,j})^2 K_{j+1,N}^2}(A^t(j+1,\cdots,N)\cdot \eta\  A^t_{\sigma\rho}(1,2,\cdots,j) \ g^{\sigma\rho}\nonumber\\
&&\ \ \ \ \ \ \ \ \ \ \ \ \ -2A^{t\ \sigma}(j+1,\cdots,N) \ \eta^\rho \ A^t_{\sigma\rho}(1,2,\cdots,j) +A^{t\ \rho}(j+1,\cdots,N)\ \eta^\sigma\  A^t_{\sigma\rho}(1,2,\cdots,j)\ )\nonumber\\
&&(c)+(f):\frac{1}{2 (l+\ke)^2(l-K_{1,j})^2 K_{1,j}^2}(A^t(1,\cdots,j)\cdot \eta\  A^t_{\sigma\rho}(j+1,\cdots,N)\  g^{\sigma\rho}\nonumber\\
&&\ \ \ \ \ \ \ \ \ \ \ \ \ -2A^{t\ \rho}(1,\cdots,j) \ \eta^\sigma\  A^t_{\sigma\rho}(j+1,\cdots,N)  +A^{t\ \sigma}(1,\cdots,j)\ \eta^\rho\  A^t_{\sigma\rho}(j+1,\cdots,N)\ )\nonumber
\eea

\begin{figure}[htb]
\centering
\includegraphics[height=15cm,width=11cm]{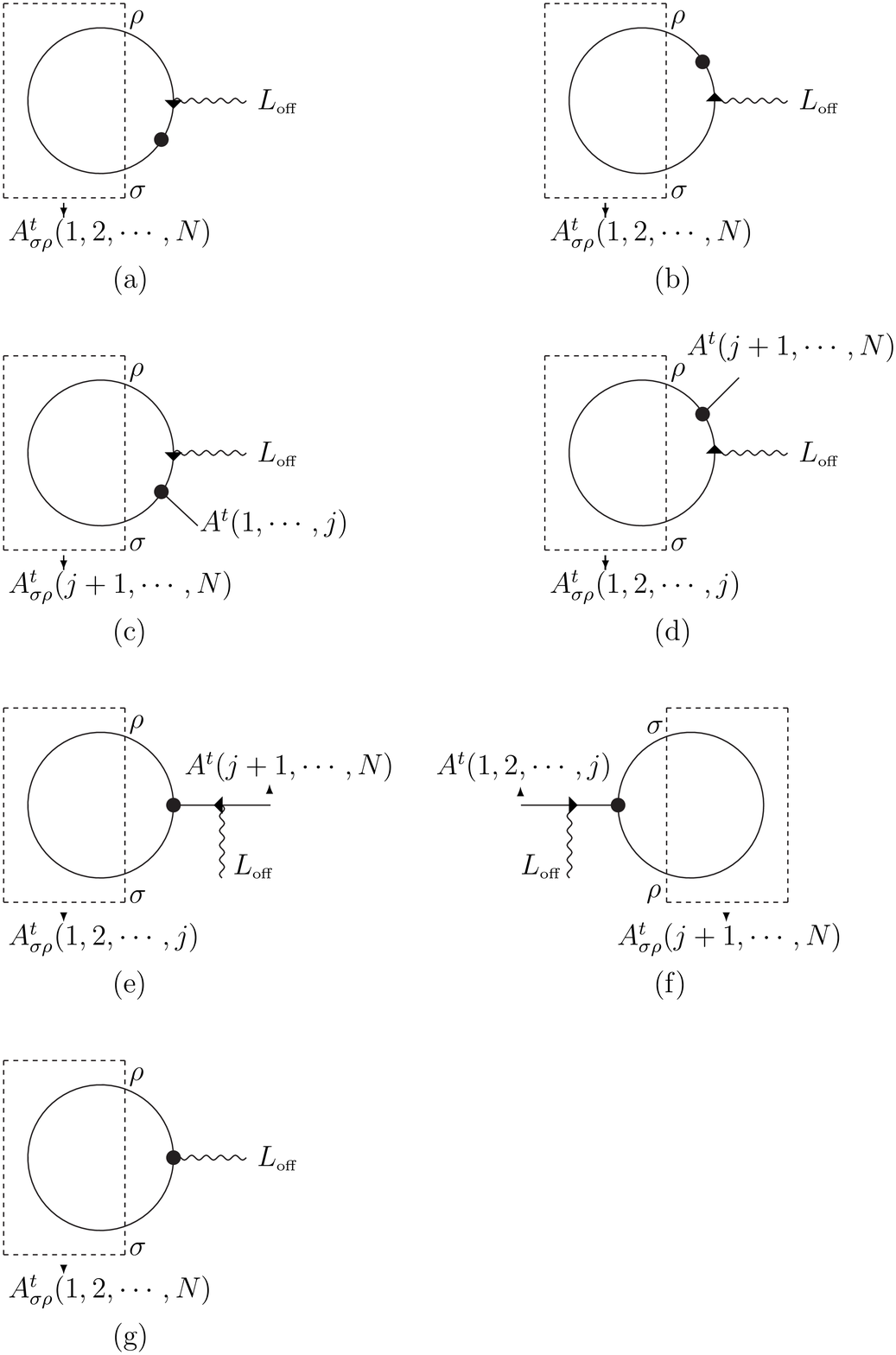}
\caption{Diagrams with derivative acting on the propagator or vertex on the loop, which cannot directly apply the tree level results.}
\label{loop1}
\end{figure}

The final contributions to $-{d\hat A^l_\mu\over dz} {{\hat k}_{\mbox{\tiny off}}}^\mu |_{z\to 0}$ come from the derivatives in the diagrams of Figure \ref{loopremain1}, Figure \ref{loopremain2} and Figure \ref{loopremain3}. Denoting the diagrams as $D_i$, since $\sum D_i\cdot \ke=0$ from the last subsection, we have $-\sum {d\hat D_{i\ \mu}\over dz}{\hat k}_{\mbox{\tiny off}}^\mu |_{z\to 0}=\sum  D_{i\ \mu}\  \eta^\mu$. This is like an opposite operation compared to the method in the current chapter, to deal with the set of diagrams $D_i$, but it simplifies the local calculation: eg. the diagrams in Figure \ref{loopremain1} turn out to be not contributing. We use $K_{A_{m,n}}$ to represent for $k_{A_m}+k_{A_{m+1}}+\cdots+k_{A_n}$, with $k_{A_i}$ the total momentum of the external legs contained in the sub amplitude $A_i$. The total momentum conservation is then $K_{A_{1,n}}+\ke=0$. The contributions to the integrand of $-{d\hat A^l_\mu\over dz} {{\hat k}_{\mbox{\tiny off}}}^\mu |_{z\to 0}$ from derivatives in Figure \ref{loopremain1}, Figure \ref{loopremain2} and Figure \ref{loopremain3} are:
\bea
&&\frac{(-i)^n\ (l-K_{A_{1,1}})\cdot A_2\ (l-K_{A_{1,2}})\cdot A_3\ \cdots (l-K_{A_{1,n-2}})\cdot A_{n-1}}{(\sqrt{2})^{n+1}k_{A_1}^2\cdots k_{A_n}^2(l-K_{A_{1,1}})^2\cdots(l-K_{A_{1,n-1}})^2}\nonumber\\
&&(-\frac{2 l\cdot A_1\ (l-K_{A_{1,n-1}})\cdot A_n\ (2l+\ke)\cdot \eta}{l^2(l+\ke)^2}\nonumber\\
&&+\frac{l\cdot A_1\ A_n \cdot \eta}{l^2}+\frac{(l-K_{A_{1,n-1}})\cdot A_n \ A_1\cdot \eta}{(l+\ke)^2}).
\label{finalloop3}
\eea
This expression is well defined when $n\ge 2$. Especially when $n=2$, one should multiply the pre-factor with each term in the bracket to see that it is well defined. When $n=1$, the last two terms in the bracket vanish.

\eref{finalloop1}, \eref{finalloop2} and \eref{finalloop3} constitute our expressions for recursively calculating one loop off shell amplitudes. In each expression, eg. in the above one \eref{finalloop3}, we should sum over all the allowed different allocations of the on shell external legs into $A_1, \cdots, A_n$. This summation is not written explicitly in the expressions. As in tree level, summing over \eref{finalloop1}, \eref{finalloop2} and \eref{finalloop3}, we get $A^l_\mu \eta^\mu$ for three choices of $\eta$, ie. $\epsilon_s^+$, $\epsilon_s^-$ and $k_s$, satisfying $\eta^2=0$ and $k_s\cdot\eta=0$. In four dimensional spacetime, the last component of $A_\mu^l$ can then be determined by the Ward identity $\ke^\mu A_\mu^l=0$.


\subsection{Examples of 1-loop Off Shell Amplitudes}\label{example}
As an application and verification of our method, we have computed three and four point one loop amplitudes with one off shell leg. ie. $A^l_\mu(k_1,k_2)$ and $A^l_\mu(k_1,k_2,k_3)$, by summing up the contributions from \eref{finalloop1}, \eref{finalloop2} and \eref{finalloop3}. We use the integral reduction method in \cite{Ellis} to reduce the integrals to scalar integrals. We use the following notations for the scalar integrations:
\beas
&&B0[1,3]=\int \frac{d^D l}{(2\pi)^D}\frac{1}{l^2(l-k_1-k_2)^2},\ \ B0[1,4]=\int \frac{d^D l}{(2\pi)^D}\frac{1}{l^2(l-k_1-k_2-k_3)^2},\\
&&B0[2,4]=\int \frac{d^D l}{(2\pi)^D}\frac{1}{l^2(l-k_2-k_3)^2},\ \ C0[1,2,3]=\int \frac{d^D l}{(2\pi)^D}\frac{1}{l^2(l-k_1)^2(l-k_1-k_2)^2},\\
&&C0[1,2,4]=\int \frac{d^D l}{(2\pi)^D}\frac{1}{l^2(l-k_1)^2(l-k_1-k_2-k_3)^2},\\
&&C0[1,3,4]=\int \frac{d^D l}{(2\pi)^D}\frac{1}{l^2(l-k_1-k_2)^2(l-k_1-k_2-k_3)^2},\\
&&C0[2,3,4]=\int \frac{d^D l}{(2\pi)^D}\frac{1}{l^2(l-k_2)^2(l-k_2-k_3)^2},\\
&&D0[1,2,3,4]=\int \frac{d^D l}{(2\pi)^D}\frac{1}{l^2(l-k_1)^2(l-k_1-k_2)^2(l-k_1-k_2-k_3)^2}.
\eeas
Other scalar integrations are not needed in this chapter. The evaluation of the scalar integrals see \cite{BernD2}.

We start from the two point loop amplitude with both indices not contracted:
\begin{equation}
A^l_{\mu\nu}(k)=\frac{2-3D}{2(1-D)}(k^2 g_{\mu\nu}-k_\mu k_\nu) \int \frac{d^D l}{(2\pi)^D}\frac{1}{l^2(l-k)^2}.
\label{twopoint1loop}
\end{equation}

Then we can calculate three point one loop off shell amplitude using our method:
\bea
&&A^l_\mu(k_1,k_2)=\frac{1}{2\sqrt{2}}\left[k_1\cdot \epsilon_2 \ \epsilon_{1\ \mu}-k_2\cdot \epsilon_1\  \epsilon_{2\ \mu}-\frac{2D-5}{D-1}\epsilon_1\cdot \epsilon_2\  (k_1-k_2)_\mu\right.\nonumber\\
&&\ \ \ \ \ \ \ \ \ \ \ \ \ \ \ \ \ \ \ \ \ \ \ \ \ \left.+\frac{D-4}{D-1}\frac{k_1\cdot \epsilon_2\  k_2 \cdot \epsilon_1}{k_1\cdot k_2}(k_1-k_2)_\mu\right] B0[1,3]\label{threepoint1loop}\\
&&\ \ \ \ \ \ \ \ \ \ \ \ \ \ \ +\frac{1}{2\sqrt{2}}k_1\cdot k_2[-3 \epsilon_1\cdot \epsilon_2 \ (k_1-k_2)_\mu+4 k_1\cdot \epsilon_2 \ \epsilon_{1\ \mu}-4 k_2 \cdot \epsilon_1 \ \epsilon_{2\ \mu}]C0[1,2,3].\nonumber
\eea

At four point, the length of the expressions grow very quickly, and we will only give $A^l_\mu(1^+,2^+,3^+)$. Instead of giving this expression directly, we will give $A^l_\mu(1^+,2^+,3^+) \epsilon_1^\mu$, $A^l_\mu(1^+,2^+,3^+) \epsilon_3^\mu$ and $A^l_\mu(1^+,$ $2^+,$ $3^+)$ $k_1^\mu$. Together with $A^l_\mu(1^+,2^+,3^+) (k_1+k_2+k_3)^\mu=0$, the expressions are enough to determine all the 4 components of $A^l_\mu(1^+,2^+,3^+)$. We choose the spinor representations for $k_{1,2,3}$ and $\epsilon_{1,2,3}$ to be:\footnote{A factor $\frac{1}{\sqrt{2}}$ for the polarization vectors were not considered at the time this calculation was done for convenience.}
\begin{equation}
k_1=\lambda_1 \tilde \lambda_1,\ k_2=\lambda_2 \tilde \lambda_2,\ k_3=\lambda_3 \tilde \lambda_3,\ \epsilon_1=\frac{\lambda_\nu \tilde \lambda_1}{\langle\lambda_\nu \lambda_1\rangle},\ \epsilon_2=\frac{\lambda_\nu \tilde \lambda_2}{\langle\lambda_\nu \lambda_2\rangle},\ \epsilon_3=\frac{\lambda_\nu \tilde \lambda_3}{\langle\lambda_\nu \lambda_3\rangle},
\end{equation}
with $\lambda_\nu$ an arbitrary but fixed reference spinor. We will use $\langle\nu 1\rangle$ to stand for $\langle\lambda_\nu \lambda_1\rangle$ and similarly others. We use $(A^l_\mu(1^+,2^+,3^+) \epsilon_1^\mu)_{D0[1,2,3,4]}$ to denote for the coefficient of $D0[1,2,3,4]$ in $A^l_\mu(1^+,2^+,3^+) \epsilon_1^\mu$, and similarly for others. We give the coefficients at $D=4$. The off shell leg makes the expressions much more complicated than those with all on shell legs. On one hand, when all legs are on shell, since the amplitudes are gauge invariant, we can choose some specific reference spinor, while in the off shell case we should keep the reference spinor $\lambda_\nu$ arbitrary. On the other hand, there are many terms in the following expressions which are 0 when all legs are on shell. For example, the first coefficient below would be 0 due to $(\langle 12 \rangle [12]+\langle 13 \rangle [13]+\langle 23 \rangle [23])=0$ when all legs were on shell.


Then for $A^l_\mu(1^+,2^+,3^+) \epsilon_1^\mu$:

\begin{scriptsize}
\beas
&&(A^l_\mu(1^+,2^+,3^+) \epsilon_1^\mu)_{D0[1,2,3,4]}\\
&&=-\frac{(\langle 13 \rangle \langle 2\nu \rangle-2 \langle 12 \rangle \langle 3\nu \rangle) [12] (\langle 2\nu \rangle [12]+\langle 3\nu \rangle [13]) [23] (\langle 12 \rangle [12]+\langle 13 \rangle [13]+\langle 23 \rangle [23])}{64 \langle 13 \rangle^2 \langle 1\nu \rangle \langle 2\nu \rangle},\\
&&\\
&&(A^l_\mu(1^+,2^+,3^+) \epsilon_1^\mu)_{C0[1,2,4]}\\
&&=-\frac{(\langle 2\nu \rangle [12]+\langle 3\nu \rangle [13])}{64 \langle 12 \rangle \langle 13 \rangle^2 \langle 1\nu \rangle^2 \langle 23 \rangle \langle 2\nu \rangle \langle 3\nu \rangle (\langle 12 \rangle [12]+\langle 13 \rangle [13])} (\langle 1\nu \rangle^2 \langle 23 \rangle^2 (\langle 12 \rangle \langle 13 \rangle \langle 2\nu \rangle [12]^2+2 \langle 12 \rangle \langle 12 \rangle \langle 3\nu \rangle [12]^2\\
&&+4 \langle 12 \rangle \langle 13 \rangle \langle 3\nu \rangle [13] [12]+\langle 13 \rangle^2 \langle 3\nu \rangle [13]^2) [23]-\langle 1\nu \rangle (\langle 12 \rangle [12]+\langle 13 \rangle [13]) (2 \langle 3\nu \rangle^2 [12]^2 \langle 12 \rangle^3+\langle 13 \rangle \langle 3\nu \rangle [12] (4 \langle 3\nu \rangle [13]\\
&&-\langle 2\nu \rangle [12]) \langle 12 \rangle^2-2 \langle 13 \rangle^2 \langle 2\nu \rangle [12] (\langle 2\nu \rangle [12]+3 \langle 3\nu \rangle [13]) \langle 12 \rangle-\langle 13 \rangle^3 \langle 2\nu \rangle \langle 3\nu \rangle [13]^2)),\\
&&\\
&&(A^l_\mu(1^+,2^+,3^+) \epsilon_1^\mu)_{C0[2,3,4]}\\
&&=-\frac{(\langle 2\nu \rangle [12]+\langle 3\nu \rangle [13]) [23] }{64 \langle 12 \rangle \langle 13 \rangle^2 \langle 1\nu \rangle^2 \langle 2\nu \rangle}(2 \langle 3\nu \rangle (\langle 1\nu \rangle [12]-\langle 3\nu \rangle [23]) \langle 12 \rangle^2+\langle 13 \rangle (-\langle 1\nu \rangle \langle 2\nu \rangle [12]+2 \langle 1\nu \rangle \langle 3\nu \rangle [13]\\
&&+3 \langle 2\nu \rangle \langle 3\nu \rangle [23]) \langle 12 \rangle+\langle 13 \rangle^2 \langle 2\nu \rangle (\langle 1\nu \rangle [13]+\langle 2\nu \rangle [23])),\\
&&\\
&&(A^l_\mu(1^+,2^+,3^+) \epsilon_1^\mu)_{C0[1,3,4]}\\
&&=\frac{(\langle 2\nu \rangle [12]+\langle 3\nu \rangle [13])}{64 \langle 12 \rangle \langle 13 \rangle^2 \langle 1\nu \rangle^2 \langle 23 \rangle \langle 2\nu \rangle (\langle 13 \rangle [13]+\langle 23 \rangle [23])} (-2 \langle 3\nu \rangle^2 [12] [23] (2 \langle 13 \rangle [13]+\langle 23 \rangle [23]) \langle 12 \rangle^3+\langle 3\nu \rangle (-2 \langle 23 \rangle^2 \langle 3\nu \rangle [23]^3\\
&&+3 \langle 13 \rangle \langle 23 \rangle (\langle 2\nu \rangle [12]-2 \langle 3\nu \rangle [13]) [23]^2+\langle 13 \rangle^2 [13] (\langle 1\nu \rangle [12] [13]-4 \langle 3\nu \rangle [23] [13]+4 \langle 2\nu \rangle [12] [23])) \langle 12 \rangle^2+\\
&&\langle 13 \rangle \langle 2\nu \rangle \langle 3\nu \rangle [23] (\langle 13 \rangle^2 [13]^2+5 \langle 13 \rangle \langle 23 \rangle [23] [13]+3 \langle 23 \rangle^2 [23]^2) \langle 12 \rangle+\langle 13 \rangle^2 \langle 2\nu \rangle (\langle 23 \rangle^2 \langle 2\nu \rangle [23]^3\\
&&+3 \langle 13 \rangle \langle 23 \rangle \langle 2\nu \rangle [13] [23]^2+\langle 13 \rangle^2 [13]^2 (\langle 1\nu \rangle [13]+3 \langle 2\nu \rangle [23]))),\\
&&\\
&&(A^l_\mu(1^+,2^+,3^+) \epsilon_1^\mu)_{C0[1,2,3]}\\
&&=\frac{[12] (\langle 2\nu \rangle [12]+\langle 3\nu \rangle [13])}{64 \langle 13 \rangle^2 \langle 1\nu \rangle^2 \langle 23 \rangle \langle 2\nu \rangle \langle 3\nu \rangle} (\langle 2\nu \rangle (2 \langle 1\nu \rangle \langle 2\nu \rangle [12]+3 \langle 1\nu \rangle \langle 3\nu \rangle [13]+\langle 2\nu \rangle \langle 3\nu \rangle [23]) \langle 13 \rangle^2\\
&&+\langle 12 \rangle \langle 3\nu \rangle (\langle 1\nu \rangle \langle 2\nu \rangle [12]-2 \langle 1\nu \rangle \langle 3\nu \rangle [13]-3 \langle 2\nu \rangle \langle 3\nu \rangle [23]) \langle 13 \rangle+2 \langle 12 \rangle^2 \langle 3\nu \rangle^2 (\langle 3\nu \rangle [23]-\langle 1\nu \rangle [12])),\\
&&\\
&&(A^l_\mu(1^+,2^+,3^+) \epsilon_1^\mu)_{B0[2,4]}\\
&&=-\frac{(\langle 2\nu \rangle [12]+\langle 3\nu \rangle [13])}{64 \langle 12 \rangle \langle 13 \rangle \langle 1\nu \rangle^2 \langle 23 \rangle \langle 2\nu \rangle \langle 3\nu \rangle (\langle 12 \rangle [12]+\langle 13 \rangle [13])^2} (4 \langle 3\nu \rangle^2 [12] [13] (\langle 1\nu \rangle [12]-\langle 3\nu \rangle [23]) \langle 12 \rangle^3\\
&&+\langle 13 \rangle (\langle 1\nu \rangle [12] (2 \langle 2\nu \rangle^2 [12]^2-6 \langle 2\nu \rangle \langle 3\nu \rangle [13] [12]+5 \langle 3\nu \rangle^2 [13]^2)+\langle 3\nu \rangle (-4 \langle 2\nu \rangle^2 [12]^2+7 \langle 2\nu \rangle \langle 3\nu \rangle [13] [12]\\
&&-2 \langle 3\nu \rangle^2 [13]^2) [23]) \langle 12 \rangle^2+\langle 13 \rangle^2 (\langle 1\nu \rangle [13] (7 \langle 2\nu \rangle^2 [12]^2-6 \langle 2\nu \rangle \langle 3\nu \rangle [13] [12]+\langle 3\nu \rangle^2 [13]^2)+\langle 2\nu \rangle (3 \langle 2\nu \rangle^2 [12]^2\\
&&-10 \langle 2\nu \rangle \langle 3\nu \rangle [13] [12]+\langle 3\nu \rangle^2 [13]^2) [23]) \langle 12 \rangle+5 \langle 13 \rangle^3 \langle 2\nu \rangle^2 [12] [13] (\langle 1\nu \rangle [13]+\langle 2\nu \rangle [23])),\\
&&\\
&&(A^l_\mu(1^+,2^+,3^+) \epsilon_1^\mu)_{B0[1,4]}\\
&&=\frac{(\langle 2\nu \rangle [12]+\langle 3\nu \rangle [13])}{64 \langle 12 \rangle \langle 13 \rangle \langle 1\nu \rangle^2 \langle 23 \rangle \langle 2\nu \rangle \langle 3\nu \rangle (\langle 12 \rangle [12]+\langle 13 \rangle [13])^2 (\langle 13 \rangle [13]+\langle 23 \rangle [23])^3} (-2 \langle 3\nu \rangle^3 [12]^3 [13] [23] (3 \langle 13 \rangle [13]\\
&&+2 \langle 23 \rangle [23]) \langle 12 \rangle^5+\langle 3\nu \rangle^2 [12]^2 [13] (-8 \langle 23 \rangle^2 \langle 3\nu \rangle [23]^3+\langle 13 \rangle \langle 23 \rangle (5 \langle 2\nu \rangle [12]-28 \langle 3\nu \rangle [13]) [23]^2+2 \langle 13 \rangle^2 [13] (\langle 1\nu \rangle [12] [13]\\
&&-14 \langle 3\nu \rangle [23] [13]+5 \langle 2\nu \rangle [12] [23])) \langle 12 \rangle^4-\langle 3\nu \rangle [12] (4 \langle 23 \rangle^3 \langle 3\nu \rangle^2 [13] [23]^4+4 \langle 13 \rangle \langle 23 \rangle^2 (\langle 2\nu \rangle^2 [12]^2-3 \langle 2\nu \rangle \langle 3\nu \rangle [13] [12]\\
&&+6 \langle 3\nu \rangle^2 [13]^2) [23]^3+2 \langle 13 \rangle^2 \langle 23 \rangle [13] (6 \langle 2\nu \rangle [12]-11 \langle 3\nu \rangle [13]) (\langle 2\nu \rangle [12]-2 \langle 3\nu \rangle [13]) [23]^2+\langle 13 \rangle^3 [13]^2 (\langle 1\nu \rangle [12] [13] (3 \langle 2\nu \rangle [12]\\
&&-8 \langle 3\nu \rangle [13])+2 (6 \langle 2\nu \rangle^2 [12]^2-28 \langle 2\nu \rangle \langle 3\nu \rangle [13] [12]+15 \langle 3\nu \rangle^2 [13]^2) [23])) \langle 12 \rangle^3+\langle 13 \rangle (\langle 23 \rangle^3 \langle 3\nu \rangle (-4 \langle 2\nu \rangle^2 [12]^2\\
&&+7 \langle 2\nu \rangle \langle 3\nu \rangle [13] [12]-2 \langle 3\nu \rangle^2 [13]^2) [23]^4+\langle 13 \rangle \langle 23 \rangle^2 (\langle 2\nu \rangle^3 [12]^3-32 \langle 2\nu \rangle^2 \langle 3\nu \rangle [13] [12]^2+33 \langle 2\nu \rangle \langle 3\nu \rangle^2 [13]^2 [12]\\
&&-10 \langle 3\nu \rangle^3 [13]^3) [23]^3+\langle 13 \rangle^2 \langle 23 \rangle [13] (3 \langle 2\nu \rangle^3 [12]^3-72 \langle 2\nu \rangle^2 \langle 3\nu \rangle [13] [12]^2+62 \langle 2\nu \rangle \langle 3\nu \rangle^2 [13]^2 [12]-14 \langle 3\nu \rangle^3 [13]^3) [23]^2\\
&&+\langle 13 \rangle^3 [13]^2 (\langle 1\nu \rangle [12] [13] (\langle 2\nu \rangle^2 [12]^2-14 \langle 2\nu \rangle \langle 3\nu \rangle [13] [12]+6 \langle 3\nu \rangle^2 [13]^2)+(3 \langle 2\nu \rangle^3 [12]^3-64 \langle 2\nu \rangle^2 \langle 3\nu \rangle [13] [12]^2\\
&&+53 \langle 2\nu \rangle \langle 3\nu \rangle^2 [13]^2 [12]-6 \langle 3\nu \rangle^3 [13]^3) [23])) \langle 12 \rangle^2+\langle 13 \rangle^2 \langle 2\nu \rangle (\langle 23 \rangle^3 (3 \langle 2\nu \rangle^2 [12]^2-10 \langle 2\nu \rangle \langle 3\nu \rangle [13] [12]+\langle 3\nu \rangle^2 [13]^2) [23]^4\\
&&+\langle 13 \rangle \langle 23 \rangle^2 [13] (\langle 3\nu \rangle [13]-15 \langle 2\nu \rangle [12]) (3 \langle 3\nu \rangle [13]-\langle 2\nu \rangle [12]) [23]^3+3 \langle 13 \rangle^2 \langle 23 \rangle [13]^2 (9 \langle 2\nu \rangle^2 [12]^2-26 \langle 2\nu \rangle \langle 3\nu \rangle [13] [12]\\
&&+\langle 3\nu \rangle^2 [13]^2) [23]^2+\langle 13 \rangle^3 [13]^3 (\langle 1\nu \rangle [12] [13] (6 \langle 2\nu \rangle [12]-11 \langle 3\nu \rangle [13])+(21 \langle 2\nu \rangle^2 [12]^2-58 \langle 2\nu \rangle \langle 3\nu \rangle [13] [12]\\
&&+\langle 3\nu \rangle^2 [13]^2) [23])) \langle 12 \rangle+5 \langle 13 \rangle^3 \langle 2\nu \rangle^2 [12] [13] (\langle 23 \rangle^3 \langle 2\nu \rangle [23]^4+4 \langle 13 \rangle \langle 23 \rangle^2 \langle 2\nu \rangle [13] [23]^3+6 \langle 13 \rangle^2 \langle 23 \rangle \langle 2\nu \rangle [13]^2 [23]^2\\
&&+\langle 13 \rangle^3 [13]^3 (\langle 1\nu \rangle [13]+4 \langle 2\nu \rangle [23]))),\\
&&\\
&&(A^l_\mu(1^+,2^+,3^+) \epsilon_1^\mu)_{B0[1,3]}\\
&&=-\frac{(\langle 2\nu \rangle [12]+\langle 3\nu \rangle [13])}{64 \langle 12 \rangle \langle 13 \rangle \langle 1\nu \rangle^2 \langle 23 \rangle \langle 2\nu \rangle \langle 3\nu \rangle (\langle 13 \rangle [13]+\langle 23 \rangle [23])^3} (\langle 2\nu \rangle [13]^2 (\langle 1\nu \rangle [13] (\langle 2\nu \rangle [12]-\langle 3\nu \rangle [13])+\langle 2\nu \rangle (3 \langle 2\nu \rangle [12]\\
&&-5 \langle 3\nu \rangle [13]) [23]) \langle 13 \rangle^4+[13] (3 \langle 23 \rangle \langle 2\nu \rangle^2 (\langle 2\nu \rangle [12]-3 \langle 3\nu \rangle [13]) [23]^2+\langle 12 \rangle \langle 3\nu \rangle [13] (\langle 1\nu \rangle [13] (\langle 3\nu \rangle [13]-3 \langle 2\nu \rangle [12])\\
&&+6 \langle 2\nu \rangle (\langle 3\nu \rangle [13]-2 \langle 2\nu \rangle [12]) [23])) \langle 13 \rangle^3+(\langle 23 \rangle^2 \langle 2\nu \rangle^2 (\langle 2\nu \rangle [12]-7 \langle 3\nu \rangle [13]) [23]^3+6 \langle 12 \rangle \langle 23 \rangle \langle 2\nu \rangle \langle 3\nu \rangle [13] (\langle 3\nu \rangle [13]\\
&&-2 \langle 2\nu \rangle [12]) [23]^2+2 \langle 12 \rangle^2 \langle 3\nu \rangle^2 [13]^2 (\langle 1\nu \rangle [12] [13]-3 \langle 3\nu \rangle [23] [13]+5 \langle 2\nu \rangle [12] [23])) \langle 13 \rangle^2\\
&&-\langle 3\nu \rangle [23] (6 \langle 3\nu \rangle^2 [12] [13]^2 \langle 12 \rangle^3+\langle 23 \rangle \langle 3\nu \rangle [13] (9 \langle 3\nu \rangle [13]-5 \langle 2\nu \rangle [12]) [23] \langle 12 \rangle^2+2 \langle 23 \rangle^2 \langle 2\nu \rangle (2 \langle 2\nu \rangle [12]\\
&&-\langle 3\nu \rangle [13]) [23]^2 \langle 12 \rangle+2 \langle 23 \rangle^3 \langle 2\nu \rangle^2 [23]^3) \langle 13 \rangle-4 \langle 12 \rangle^2 \langle 23 \rangle \langle 3\nu \rangle^3 [13] [23]^2 (\langle 12 \rangle [12]+\langle 23 \rangle [23])).
\eeas
\end{scriptsize}

For $A^l_\mu(1^+,2^+,3^+) \epsilon_3^\mu$:

\begin{scriptsize}
\beas
&&(A^l_\mu(1^+,2^+,3^+) \epsilon_3^\mu)_{D0[1,2,3,4]}\\
&&=\frac{(\langle 13 \rangle \langle 2\nu \rangle-2 \langle 12 \rangle \langle 3\nu \rangle) [12] [23] (\langle 1\nu \rangle [13]+\langle 2\nu \rangle [23]) (\langle 12 \rangle [12]+\langle 13 \rangle [13]+\langle 23 \rangle [23])}{64 \langle 13 \rangle^2 \langle 2\nu \rangle \langle 3\nu \rangle},\\
&&\\
&&(A^l_\mu(1^+,2^+,3^+) \epsilon_3^\mu)_{C0[1,2,4]}\\
&&=\frac{(\langle 1\nu \rangle [13]+\langle 2\nu \rangle [23])}{64 \langle 12 \rangle \langle 13 \rangle^2 \langle 1\nu \rangle \langle 23 \rangle \langle 2\nu \rangle \langle 3\nu \rangle^2 (\langle 12 \rangle [12]+\langle 13 \rangle [13])} (\langle 1\nu \rangle^2 \langle 23 \rangle^2 (\langle 12 \rangle (\langle 13 \rangle \langle 2\nu \rangle+2 \langle 12 \rangle \langle 3\nu \rangle) [12]^2\\
&&+4 \langle 12 \rangle \langle 13 \rangle \langle 3\nu \rangle [13] [12]+\langle 13 \rangle^2 \langle 3\nu \rangle [13]^2) [23]-\langle 1\nu \rangle (\langle 12 \rangle [12]+\langle 13 \rangle [13]) (2 \langle 3\nu \rangle^2 [12]^2 \langle 12 \rangle^3\\
&&+\langle 13 \rangle \langle 3\nu \rangle [12] (4 \langle 3\nu \rangle [13]-\langle 2\nu \rangle [12]) \langle 12 \rangle^2-2 \langle 13 \rangle^2 \langle 2\nu \rangle [12] (\langle 2\nu \rangle [12]+3 \langle 3\nu \rangle [13]) \langle 12 \rangle-\langle 13 \rangle^3 \langle 2\nu \rangle \langle 3\nu \rangle [13]^2)),\\
&&\\
&&(A^l_\mu(1^+,2^+,3^+) \epsilon_3^\mu)_{C0[2,3,4]}\\
&&=\frac{[23] (\langle 1\nu \rangle [13]+\langle 2\nu \rangle [23])}{64 \langle 12 \rangle \langle 13 \rangle^2 \langle 1\nu \rangle \langle 2\nu \rangle \langle 3\nu \rangle} (2 \langle 3\nu \rangle (\langle 1\nu \rangle [12]-\langle 3\nu \rangle [23]) \langle 12 \rangle^2+\langle 13 \rangle (-\langle 1\nu \rangle \langle 2\nu \rangle [12]+2 \langle 1\nu \rangle \langle 3\nu \rangle [13]\\
&&+3 \langle 2\nu \rangle \langle 3\nu \rangle [23]) \langle 12 \rangle+\langle 13 \rangle^2 \langle 2\nu \rangle (\langle 1\nu \rangle [13]+\langle 2\nu \rangle [23])) ,\\
&&\\
&&(A^l_\mu(1^+,2^+,3^+) \epsilon_3^\mu)_{C0[1,3,4]}\\
&&=-\frac{1}{64 \langle 12 \rangle \langle 13 \rangle^2 \langle 1\nu \rangle \langle 23 \rangle \langle 2\nu \rangle \langle 3\nu \rangle (\langle 13 \rangle [13]+\langle 23 \rangle [23])}(2 \langle 3\nu \rangle^3 [12] [13] [23]^2 \langle 12 \rangle^4+\langle 3\nu \rangle^2 [23] (2 \langle 23 \rangle (\langle 3\nu \rangle [13]\\
&&-\langle 2\nu \rangle [12]) [23]^2+\langle 13 \rangle [13] (-4 \langle 1\nu \rangle [12] [13]+6 \langle 3\nu \rangle [23] [13]-9 \langle 2\nu \rangle [12] [23])) \langle 12 \rangle^3+\langle 3\nu \rangle (-2 \langle 23 \rangle^2 \langle 2\nu \rangle \langle 3\nu \rangle [23]^4\\
&&+\langle 13 \rangle \langle 23 \rangle \langle 2\nu \rangle (3 \langle 2\nu \rangle [12]-11 \langle 3\nu \rangle [13]) [23]^3+\langle 13 \rangle^2 [13] (\langle 1\nu \rangle^2 [12] [13]^2+\langle 1\nu \rangle (5 \langle 2\nu \rangle [12]-4 \langle 3\nu \rangle [13]) [23] [13]\\
&&+\langle 2\nu \rangle (7 \langle 2\nu \rangle [12]-15 \langle 3\nu \rangle [13]) [23]^2)) \langle 12 \rangle^2+\langle 13 \rangle \langle 2\nu \rangle \langle 3\nu \rangle [23] (3 \langle 23 \rangle^2 \langle 2\nu \rangle [23]^3+7 \langle 13 \rangle \langle 23 \rangle \langle 2\nu \rangle [13] [23]^2\\
&&+\langle 13 \rangle^2 [13]^2 (\langle 1\nu \rangle [13]+3 \langle 2\nu \rangle [23])) \langle 12 \rangle+\langle 13 \rangle^2 \langle 2\nu \rangle (\langle 23 \rangle^2 \langle 2\nu \rangle^2 [23]^4+4 \langle 13 \rangle \langle 23 \rangle \langle 2\nu \rangle^2 [13] [23]^3\\
&&+\langle 13 \rangle^2 [13]^2 (\langle 1\nu \rangle^2 [13]^2+4 \langle 1\nu \rangle \langle 2\nu \rangle [23] [13]+6 \langle 2\nu \rangle^2 [23]^2))) ,\\
&&\\
&&(A^l_\mu(1^+,2^+,3^+) \epsilon_3^\mu)_{C0[1,2,3]}\\
&&=-\frac{[12] (\langle 1\nu \rangle [13]+\langle 2\nu \rangle [23])}{64 \langle 13 \rangle^2 \langle 1\nu \rangle \langle 23 \rangle \langle 2\nu \rangle \langle 3\nu \rangle^2}(\langle 2\nu \rangle (2 \langle 1\nu \rangle \langle 2\nu \rangle [12]+3 \langle 1\nu \rangle \langle 3\nu \rangle [13]+\langle 2\nu \rangle \langle 3\nu \rangle [23]) \langle 13 \rangle^2+\langle 12 \rangle \langle 3\nu \rangle (\langle 1\nu \rangle \langle 2\nu \rangle [12]\\
&&-2 \langle 1\nu \rangle \langle 3\nu \rangle [13]-3 \langle 2\nu \rangle \langle 3\nu \rangle [23]) \langle 13 \rangle+2 \langle 12 \rangle^2 \langle 3\nu \rangle^2 (\langle 3\nu \rangle [23]-\langle 1\nu \rangle [12])) ,\\
&&\\
&&(A^l_\mu(1^+,2^+,3^+) \epsilon_3^\mu)_{B0[2,4]}\\
&&=\frac{(\langle 1\nu \rangle [13]+\langle 2\nu \rangle [23])}{64 \langle 12 \rangle \langle 13 \rangle \langle 1\nu \rangle \langle 23 \rangle \langle 2\nu \rangle \langle 3\nu \rangle^2 (\langle 12 \rangle [12]+\langle 13 \rangle [13])^2} (4 \langle 3\nu \rangle^2 [12] [13] (\langle 1\nu \rangle [12]-\langle 3\nu \rangle [23]) \langle 12 \rangle^3\\
&&+\langle 13 \rangle (\langle 1\nu \rangle [12] (2 \langle 2\nu \rangle^2 [12]^2-6 \langle 2\nu \rangle \langle 3\nu \rangle [13] [12]+5 \langle 3\nu \rangle^2 [13]^2)+\langle 3\nu \rangle (-4 \langle 2\nu \rangle^2 [12]^2+7 \langle 2\nu \rangle \langle 3\nu \rangle [13] [12]\\
&&-2 \langle 3\nu \rangle^2 [13]^2) [23]) \langle 12 \rangle^2+\langle 13 \rangle^2 (\langle 1\nu \rangle [13] (7 \langle 2\nu \rangle^2 [12]^2-6 \langle 2\nu \rangle \langle 3\nu \rangle [13] [12]+\langle 3\nu \rangle^2 [13]^2)+\langle 2\nu \rangle (3 \langle 2\nu \rangle^2 [12]^2\\
&&-10 \langle 2\nu \rangle \langle 3\nu \rangle [13] [12]+\langle 3\nu \rangle^2 [13]^2) [23]) \langle 12 \rangle+5 \langle 13 \rangle^3 \langle 2\nu \rangle^2 [12] [13] (\langle 1\nu \rangle [13]+\langle 2\nu \rangle [23])) ,\\
&&\\
&&(A^l_\mu(1^+,2^+,3^+) \epsilon_3^\mu)_{B0[1,4]}\\
&&=-\frac{1}{64 \langle 12 \rangle \langle 13 \rangle \langle 1\nu \rangle \langle 23 \rangle \langle 2\nu \rangle \langle 3\nu \rangle^2 (\langle 12 \rangle [12]+\langle 13 \rangle [13])^2 (\langle 13 \rangle [13]+\langle 23 \rangle [23])^2}(-4 \langle 3\nu \rangle^3 [12]^2 [13] [23] (\langle 1\nu \rangle [12] [13]\\
&&-2 \langle 3\nu \rangle [23] [13]+\langle 2\nu \rangle [12] [23]) \langle 12 \rangle^5+\langle 3\nu \rangle^2 [12] [13] (4 \langle 23 \rangle \langle 3\nu \rangle (\langle 3\nu \rangle [13]-2 \langle 2\nu \rangle [12]) [23]^3+\langle 13 \rangle (2 \langle 1\nu \rangle^2 [12]^2 [13]^2\\
&&+\langle 1\nu \rangle [12] (7 \langle 2\nu \rangle [12]-20 \langle 3\nu \rangle [13]) [23] [13]+(9 \langle 2\nu \rangle^2 [12]^2-40 \langle 2\nu \rangle \langle 3\nu \rangle [13] [12]+20 \langle 3\nu \rangle^2 [13]^2) [23]^2)) \langle 12 \rangle^4\\
&&+\langle 3\nu \rangle (-4 \langle 23 \rangle^2 \langle 2\nu \rangle \langle 3\nu \rangle^2 [12] [13] [23]^4+\langle 13 \rangle \langle 23 \rangle (-4 \langle 2\nu \rangle^3 [12]^3+16 \langle 2\nu \rangle^2 \langle 3\nu \rangle [13] [12]^2-31 \langle 2\nu \rangle \langle 3\nu \rangle^2 [13]^2 [12]\\
&&+2 \langle 3\nu \rangle^3 [13]^3) [23]^3+\langle 13 \rangle^2 [13] (\langle 1\nu \rangle^2 [12]^2 (8 \langle 3\nu \rangle [13]-3 \langle 2\nu \rangle [12]) [13]^2+\langle 1\nu \rangle [12] (-11 \langle 2\nu \rangle^2 [12]^2+42 \langle 2\nu \rangle \langle 3\nu \rangle [13] [12]\\
&&-24 \langle 3\nu \rangle^2 [13]^2) [23] [13]+(-13 \langle 2\nu \rangle^3 [12]^3+74 \langle 2\nu \rangle^2 \langle 3\nu \rangle [13] [12]^2-70 \langle 2\nu \rangle \langle 3\nu \rangle^2 [13]^2 [12]+8 \langle 3\nu \rangle^3 [13]^3) [23]^2)) \langle 12 \rangle^3\\
&&+\langle 13 \rangle (\langle 23 \rangle^2 \langle 2\nu \rangle \langle 3\nu \rangle (-4 \langle 2\nu \rangle^2 [12]^2+7 \langle 2\nu \rangle \langle 3\nu \rangle [13] [12]-2 \langle 3\nu \rangle^2 [13]^2) [23]^4+\langle 13 \rangle \langle 23 \rangle \langle 2\nu \rangle (\langle 2\nu \rangle^3 [12]^3\\
&&-35 \langle 2\nu \rangle^2 \langle 3\nu \rangle [13] [12]^2+43 \langle 2\nu \rangle \langle 3\nu \rangle^2 [13]^2 [12]-11 \langle 3\nu \rangle^3 [13]^3) [23]^3+\langle 13 \rangle^2 [13] (\langle 1\nu \rangle^2 [12] (\langle 2\nu \rangle^2 [12]^2\\
&&-14 \langle 2\nu \rangle \langle 3\nu \rangle [13] [12]+6 \langle 3\nu \rangle^2 [13]^2) [13]^2+\langle 1\nu \rangle (3 \langle 2\nu \rangle^3 [12]^3-58 \langle 2\nu \rangle^2 \langle 3\nu \rangle [13] [12]^2+42 \langle 2\nu \rangle \langle 3\nu \rangle^2 [13]^2 [12]\\
&&-6 \langle 3\nu \rangle^3 [13]^3) [23] [13]+\langle 2\nu \rangle (3 \langle 2\nu \rangle^3 [12]^3-84 \langle 2\nu \rangle^2 \langle 3\nu \rangle [13] [12]^2+98 \langle 2\nu \rangle \langle 3\nu \rangle^2 [13]^2 [12]-16 \langle 3\nu \rangle^3 [13]^3) [23]^2)) \langle 12 \rangle^2\\
&&+\langle 13 \rangle^2 \langle 2\nu \rangle (\langle 23 \rangle^2 \langle 2\nu \rangle (3 \langle 2\nu \rangle^2 [12]^2-10 \langle 2\nu \rangle \langle 3\nu \rangle [13] [12]+\langle 3\nu \rangle^2 [13]^2) [23]^4+3 \langle 13 \rangle \langle 23 \rangle \langle 2\nu \rangle [13] (5 \langle 2\nu \rangle^2 [12]^2\\
&&-17 \langle 2\nu \rangle \langle 3\nu \rangle [13] [12]+\langle 3\nu \rangle^2 [13]^2) [23]^3+\langle 13 \rangle^2 [13]^2 (\langle 1\nu \rangle^2 [12] (6 \langle 2\nu \rangle [12]-11 \langle 3\nu \rangle [13]) [13]^2+\langle 1\nu \rangle (21 \langle 2\nu \rangle^2 [12]^2\\
&&-53 \langle 2\nu \rangle \langle 3\nu \rangle [13] [12]+\langle 3\nu \rangle^2 [13]^2) [23] [13]+3 \langle 2\nu \rangle (9 \langle 2\nu \rangle^2 [12]^2-31 \langle 2\nu \rangle \langle 3\nu \rangle [13] [12]+\langle 3\nu \rangle^2 [13]^2) [23]^2)) \langle 12 \rangle\\
&&+5 \langle 13 \rangle^3 \langle 2\nu \rangle^2 [12] [13] (\langle 23 \rangle^2 \langle 2\nu \rangle^2 [23]^4+4 \langle 13 \rangle \langle 23 \rangle \langle 2\nu \rangle^2 [13] [23]^3+\langle 13 \rangle^2 [13]^2 (\langle 1\nu \rangle^2 [13]^2\\
&&+4 \langle 1\nu \rangle \langle 2\nu \rangle [23] [13]+6 \langle 2\nu \rangle^2 [23]^2))) ,\\
&&\\
&&(A^l_\mu(1^+,2^+,3^+) \epsilon_3^\mu)_{B0[1,3]}\\
&&=\frac{1}{64 \langle 12 \rangle \langle 13 \rangle \langle 1\nu \rangle \langle 23 \rangle \langle 2\nu \rangle \langle 3\nu \rangle^2 (\langle 13 \rangle [13]+\langle 23 \rangle [23])^2}(\langle 2\nu \rangle [13] (\langle 1\nu \rangle^2 (\langle 2\nu \rangle [12]-\langle 3\nu \rangle [13]) [13]^2+\langle 1\nu \rangle \langle 2\nu \rangle (3 \langle 2\nu \rangle [12]\\
&&-5 \langle 3\nu \rangle [13]) [23] [13]+3 \langle 2\nu \rangle^2 (\langle 2\nu \rangle [12]-3 \langle 3\nu \rangle [13]) [23]^2) \langle 13 \rangle^3+(\langle 23 \rangle \langle 2\nu \rangle^3 (\langle 2\nu \rangle [12]-7 \langle 3\nu \rangle [13]) [23]^3\\
&&+\langle 12 \rangle \langle 3\nu \rangle [13] (\langle 1\nu \rangle^2 (\langle 3\nu \rangle [13]-3 \langle 2\nu \rangle [12]) [13]^2+\langle 1\nu \rangle \langle 2\nu \rangle (5 \langle 3\nu \rangle [13]-11 \langle 2\nu \rangle [12]) [23] [13]+\langle 2\nu \rangle^2 (11 \langle 3\nu \rangle [13]\\
&&-13 \langle 2\nu \rangle [12]) [23]^2)) \langle 13 \rangle^2+\langle 3\nu \rangle (-2 \langle 23 \rangle^2 \langle 2\nu \rangle^3 [23]^4+4 \langle 12 \rangle \langle 23 \rangle \langle 2\nu \rangle^2 (\langle 3\nu \rangle [13]-\langle 2\nu \rangle [12]) [23]^3\\
&&+\langle 12 \rangle^2 \langle 3\nu \rangle [13] (2 \langle 1\nu \rangle^2 [12] [13]^2+\langle 1\nu \rangle (7 \langle 2\nu \rangle [12]-5 \langle 3\nu \rangle [13]) [23] [13]+\langle 2\nu \rangle (9 \langle 2\nu \rangle [12]-11 \langle 3\nu \rangle [13]) [23]^2)) \langle 13 \rangle\\
&&-4 \langle 12 \rangle^2 \langle 3\nu \rangle^3 [13] [23] (\langle 23 \rangle \langle 2\nu \rangle [23]^2+\langle 12 \rangle (\langle 1\nu \rangle [12] [13]-\langle 3\nu \rangle [23] [13]+\langle 2\nu \rangle [12] [23]))) .
\eeas
\end{scriptsize}

For $A^l_\mu(1^+,2^+,3^+) k_1^\mu$:

\begin{scriptsize}
\beas
&&(A^l_\mu(1^+,2^+,3^+) k_1^\mu)_{D0[1,2,3,4]}\\
&&=\frac{[12]}{64 \langle 13 \rangle^2 \langle 1\nu \rangle \langle 2\nu \rangle \langle 3\nu \rangle} (\langle 2\nu \rangle (4 \langle 2\nu \rangle [12]+5 \langle 3\nu \rangle [13]) \langle 13 \rangle^2+3 \langle 12 \rangle \langle 2\nu \rangle \langle 3\nu \rangle [12] \langle 13 \rangle\\
&&-2 \langle 12 \rangle^2 \langle 3\nu \rangle^2 [12]) [23] (\langle 12 \rangle \langle 1\nu \rangle [12]+\langle 13 \rangle \langle 1\nu \rangle [13]+\langle 13 \rangle \langle 2\nu \rangle [23]-\langle 12 \rangle \langle 3\nu \rangle [23]),\\
&&\\
&&(A^l_\mu(1^+,2^+,3^+) k_1^\mu)_{C0[1,3,4]}\\
&&=\frac{1}{64 \langle 12 \rangle \langle 13 \rangle^2 \langle 1\nu \rangle \langle 23 \rangle \langle 2\nu \rangle \langle 3\nu \rangle}(2 \langle 3\nu \rangle^2 [12]^2 (\langle 3\nu \rangle [23]-\langle 1\nu \rangle [12]) \langle 12 \rangle^4+\langle 13 \rangle \langle 3\nu \rangle [12] (\langle 1\nu \rangle [12] (3 \langle 2\nu \rangle [12]-4 \langle 3\nu \rangle [13])\\
&&+\langle 3\nu \rangle (2 \langle 3\nu \rangle [13]-5 \langle 2\nu \rangle [12]) [23]) \langle 12 \rangle^3+\langle 13 \rangle^2 ([13] (\langle 3\nu \rangle [13]-8 \langle 2\nu \rangle [12]) [23] \langle 3\nu \rangle^2+\langle 1\nu \rangle [12] (2 \langle 2\nu \rangle [12]\\
&&-\langle 3\nu \rangle [13]) (\langle 2\nu \rangle [12]+4 \langle 3\nu \rangle [13])) \langle 12 \rangle^2+\langle 13 \rangle^3 (\langle 1\nu \rangle [13] (6 \langle 2\nu \rangle^2 [12]^2+9 \langle 2\nu \rangle \langle 3\nu \rangle [13] [12]-2 \langle 3\nu \rangle^2 [13]^2)+\langle 2\nu \rangle (3 \langle 2\nu \rangle^2 [12]^2\\
&&+2 \langle 2\nu \rangle \langle 3\nu \rangle [13] [12]-6 \langle 3\nu \rangle^2 [13]^2) [23]) \langle 12 \rangle+\langle 13 \rangle^4 \langle 2\nu \rangle [13] (4 \langle 2\nu \rangle [12]+5 \langle 3\nu \rangle [13]) (\langle 1\nu \rangle [13]+\langle 2\nu \rangle [23])),\\
&&\\
&&(A^l_\mu(1^+,2^+,3^+) k_1^\mu)_{C0[2,3,4]}\\
&&=\frac{[23]}{64 \langle 12 \rangle \langle 13 \rangle^2 \langle 1\nu \rangle \langle 2\nu \rangle \langle 3\nu \rangle}
 (2 \langle 3\nu \rangle^2 [12] (\langle 1\nu \rangle [12]-\langle 3\nu \rangle [23]) \langle 12 \rangle^3+\langle 13 \rangle \langle 3\nu \rangle [12] (-3 \langle 1\nu \rangle \langle 2\nu \rangle [12]+2 \langle 1\nu \rangle \langle 3\nu \rangle [13]\\
&&+5 \langle 2\nu \rangle \langle 3\nu \rangle [23]) \langle 12 \rangle^2+\langle 13 \rangle^2 \langle 2\nu \rangle (2 \langle 1\nu \rangle [12] (2 \langle 2\nu \rangle [12]+\langle 3\nu \rangle [13])-\langle 3\nu \rangle (9 \langle 2\nu \rangle [12]+7 \langle 3\nu \rangle [13]) [23]) \langle 12 \rangle\\
&&+\langle 13 \rangle^3 \langle 2\nu \rangle (4 \langle 2\nu \rangle [12]+5 \langle 3\nu \rangle [13]) (\langle 1\nu \rangle [13]+\langle 2\nu \rangle [23])),\\
&&\\
&&(A^l_\mu(1^+,2^+,3^+) k_1^\mu)_{C0[1,3,4]}\\
&&=\frac{1}{64 \langle 12 \rangle \langle 13 \rangle^2 \langle 1\nu \rangle \langle 23 \rangle \langle 2\nu \rangle \langle 3\nu \rangle (\langle 13 \rangle [13]+\langle 23 \rangle [23])}(2 \langle 3\nu \rangle^3 [12]^2 [23] (2 \langle 13 \rangle [13]+\langle 23 \rangle [23]) \langle 12 \rangle^4\\
&&+\langle 3\nu \rangle^2 [12] (2 \langle 23 \rangle^2 \langle 3\nu \rangle [23]^3+\langle 13 \rangle \langle 23 \rangle (6 \langle 3\nu \rangle [13]-5 \langle 2\nu \rangle [12]) [23]^2-\langle 13 \rangle^2 [13] (\langle 1\nu \rangle [12] [13]-2 \langle 3\nu \rangle [23] [13]\\
&&+10 \langle 2\nu \rangle [12] [23])) \langle 12 \rangle^3-\langle 13 \rangle \langle 3\nu \rangle (5 \langle 23 \rangle^2 \langle 2\nu \rangle \langle 3\nu \rangle [12] [23]^3+2 \langle 13 \rangle \langle 23 \rangle \langle 3\nu \rangle [13] (9 \langle 2\nu \rangle [12]+2 \langle 3\nu \rangle [13]) [23]^2\\
&&+\langle 13 \rangle^2 [13] ((2 \langle 2\nu \rangle^2 [12]^2+17 \langle 2\nu \rangle \langle 3\nu \rangle [13] [12]+6 \langle 3\nu \rangle^2 [13]^2) [23]-\langle 1\nu \rangle [12] [13] (2 \langle 2\nu \rangle [12]\\
&&+\langle 3\nu \rangle [13]))) \langle 12 \rangle^2+\langle 13 \rangle^2 (\langle 23 \rangle^2 \langle 2\nu \rangle \langle 3\nu \rangle (\langle 2\nu \rangle [12]-\langle 3\nu \rangle [13]) [23]^3+\langle 13 \rangle \langle 23 \rangle \langle 2\nu \rangle (4 \langle 2\nu \rangle^2 [12]^2+3 \langle 2\nu \rangle \langle 3\nu \rangle [13] [12]\\
&&-3 \langle 3\nu \rangle^2 [13]^2) [23]^2+\langle 13 \rangle^2 [13] (\langle 1\nu \rangle [13] (4 \langle 2\nu \rangle^2 [12]^2+5 \langle 2\nu \rangle \langle 3\nu \rangle [13] [12]+2 \langle 3\nu \rangle^2 [13]^2)+\langle 2\nu \rangle (8 \langle 2\nu \rangle^2 [12]^2\\
&&+3 \langle 2\nu \rangle \langle 3\nu \rangle [13] [12]-3 \langle 3\nu \rangle^2 [13]^2) [23])) \langle 12 \rangle+\langle 13 \rangle^3 \langle 2\nu \rangle (4 \langle 2\nu \rangle [12]+3 \langle 3\nu \rangle [13]) (\langle 23 \rangle^2 \langle 2\nu \rangle [23]^3\\
&&+3 \langle 13 \rangle \langle 23 \rangle \langle 2\nu \rangle [13] [23]^2+\langle 13 \rangle^2 [13]^2 (\langle 1\nu \rangle [13]+3 \langle 2\nu \rangle [23]))),\\
&&\\
&&(A^l_\mu(1^+,2^+,3^+) k_1^\mu)_{C0[1,2,3]}\\
&&=\frac{[12]}{64 \langle 13 \rangle^2 \langle 1\nu \rangle \langle 23 \rangle \langle 2\nu \rangle \langle 3\nu \rangle} (\langle 1\nu \rangle (\langle 12 \rangle [12]+\langle 13 \rangle [13]) (\langle 2\nu \rangle (6 \langle 2\nu \rangle [12]+5 \langle 3\nu \rangle [13]) \langle 13 \rangle^2-3 \langle 12 \rangle \langle 2\nu \rangle \langle 3\nu \rangle [12] \langle 13 \rangle\\
&&+2 \langle 12 \rangle^2 \langle 3\nu \rangle^2 [12])+(\langle 13 \rangle \langle 2\nu \rangle-\langle 12 \rangle \langle 3\nu \rangle) (\langle 2\nu \rangle (4 \langle 2\nu \rangle [12]+3 \langle 3\nu \rangle [13]) \langle 13 \rangle^2-3 \langle 12 \rangle \langle 2\nu \rangle \langle 3\nu \rangle [12] \langle 13 \rangle+2 \langle 12 \rangle^2 \langle 3\nu \rangle^2 [12]) [23]),\\
&&\\
&&(A^l_\mu(1^+,2^+,3^+) k_1^\mu)_{B0[2,4]}\\
&&=\frac{1}{64 \langle 12 \rangle \langle 13 \rangle \langle 1\nu \rangle \langle 23 \rangle \langle 2\nu \rangle \langle 3\nu \rangle (\langle 12 \rangle [12]+\langle 13 \rangle [13])}(\langle 1\nu \rangle (\langle 12 \rangle [12]+\langle 13 \rangle [13]) (5 \langle 13 \rangle^2 [12] [13] \langle 2\nu \rangle^2+4 \langle 12 \rangle^2 \langle 3\nu \rangle^2 [12] [13]\\
&&+\langle 12 \rangle \langle 13 \rangle (2 \langle 2\nu \rangle^2 [12]^2-6 \langle 2\nu \rangle \langle 3\nu \rangle [13] [12]+\langle 3\nu \rangle^2 [13]^2))+(5 \langle 13 \rangle^3 [12] [13] \langle 2\nu \rangle^3+13 \langle 12 \rangle^2 \langle 13 \rangle \langle 3\nu \rangle^2 [12] [13] \langle 2\nu \rangle\\
&&+\langle 12 \rangle \langle 13 \rangle^2 (\langle 2\nu \rangle^2 [12]^2-12 \langle 2\nu \rangle \langle 3\nu \rangle [13] [12]+\langle 3\nu \rangle^2 [13]^2) \langle 2\nu \rangle-4 \langle 12 \rangle^3 \langle 3\nu \rangle^3 [12] [13]) [23]),\\
&&\\
&&(A^l_\mu(1^+,2^+,3^+) k_1^\mu)_{B0[1,4]}\\
&&=\frac{1}{64 \langle 12 \rangle \langle 13 \rangle \langle 1\nu \rangle \langle 23 \rangle \langle 2\nu \rangle \langle 3\nu \rangle (\langle 12 \rangle [12]+\langle 13 \rangle [13]) (\langle 13 \rangle [13]+\langle 23 \rangle [23])^3}(2 \langle 3\nu \rangle^3 [12]^3 [13] [23] (3 \langle 13 \rangle [13]+2 \langle 23 \rangle [23]) \langle 12 \rangle^5\\
&&+\langle 3\nu \rangle^2 [12]^2 [13] (8 \langle 23 \rangle^2 \langle 3\nu \rangle [23]^3-11 \langle 13 \rangle \langle 23 \rangle (\langle 2\nu \rangle [12]-2 \langle 3\nu \rangle [13]) [23]^2-2 \langle 13 \rangle^2 [13] (\langle 1\nu \rangle [12] [13]-10 \langle 3\nu \rangle [23] [13]\\
&&+9 \langle 2\nu \rangle [12] [23])) \langle 12 \rangle^4+\langle 3\nu \rangle [12] (4 \langle 23 \rangle^3 \langle 3\nu \rangle^2 [13] [23]^4-4 \langle 13 \rangle \langle 23 \rangle^2 (\langle 2\nu \rangle^2 [12]^2+7 \langle 2\nu \rangle \langle 3\nu \rangle [13] [12]-4 \langle 3\nu \rangle^2 [13]^2) [23]^3\\
&&+2 \langle 13 \rangle^2 \langle 23 \rangle [13] (-2 \langle 2\nu \rangle^2 [12]^2-42 \langle 2\nu \rangle \langle 3\nu \rangle [13] [12]+11 \langle 3\nu \rangle^2 [13]^2) [23]^2+\langle 13 \rangle^3 [13]^2 (\langle 1\nu \rangle [12] [13] (5 \langle 2\nu \rangle [12]\\
&&-6 \langle 3\nu \rangle [13])+2 (2 \langle 2\nu \rangle^2 [12]^2-40 \langle 2\nu \rangle \langle 3\nu \rangle [13] [12]+7 \langle 3\nu \rangle^2 [13]^2) [23])) \langle 12 \rangle^3+\langle 13 \rangle (-\langle 23 \rangle^3 \langle 2\nu \rangle \langle 3\nu \rangle [12] (4 \langle 2\nu \rangle [12]\\
&&+17 \langle 3\nu \rangle [13]) [23]^4+\langle 13 \rangle \langle 23 \rangle^2 (\langle 2\nu \rangle^3 [12]^3-2 \langle 2\nu \rangle^2 \langle 3\nu \rangle [13] [12]^2-83 \langle 2\nu \rangle \langle 3\nu \rangle^2 [13]^2 [12]-4 \langle 3\nu \rangle^3 [13]^3) [23]^3\\
&&+\langle 13 \rangle^2 \langle 23 \rangle [13] (3 \langle 2\nu \rangle^3 [12]^3+26 \langle 2\nu \rangle^2 \langle 3\nu \rangle [13] [12]^2-136 \langle 2\nu \rangle \langle 3\nu \rangle^2 [13]^2 [12]-8 \langle 3\nu \rangle^3 [13]^3) [23]^2\\
&&+\langle 13 \rangle^3 [13]^2 (\langle 1\nu \rangle [12] [13] (\langle 2\nu \rangle^2 [12]^2+18 \langle 2\nu \rangle \langle 3\nu \rangle [13] [12]-4 \langle 3\nu \rangle^2 [13]^2)+(3 \langle 2\nu \rangle^3 [12]^3+42 \langle 2\nu \rangle^2 \langle 3\nu \rangle [13] [12]^2\\
&&-91 \langle 2\nu \rangle \langle 3\nu \rangle^2 [13]^2 [12]-4 \langle 3\nu \rangle^3 [13]^3) [23])) \langle 12 \rangle^2+\langle 13 \rangle^2 \langle 2\nu \rangle (\langle 23 \rangle^3 (3 \langle 2\nu \rangle^2 [12]^2+12 \langle 2\nu \rangle \langle 3\nu \rangle [13] [12]-5 \langle 3\nu \rangle^2 [13]^2) [23]^4\\
&&+\langle 13 \rangle \langle 23 \rangle^2 [13] (3 \langle 2\nu \rangle [12]-\langle 3\nu \rangle [13]) (3 \langle 2\nu \rangle [12]+19 \langle 3\nu \rangle [13]) [23]^3+9 \langle 13 \rangle^2 \langle 23 \rangle [13]^2 (\langle 2\nu \rangle^2 [12]^2+10 \langle 2\nu \rangle \langle 3\nu \rangle [13] [12]\\
&&-3 \langle 3\nu \rangle^2 [13]^2) [23]^2+\langle 13 \rangle^3 [13]^3 (17 \langle 1\nu \rangle \langle 3\nu \rangle [12] [13]^2+(3 \langle 2\nu \rangle^2 [12]^2+66 \langle 2\nu \rangle \langle 3\nu \rangle [13] [12]-17 \langle 3\nu \rangle^2 [13]^2) [23])) \langle 12 \rangle\\
&&-\langle 13 \rangle^3 \langle 2\nu \rangle [13] (\langle 2\nu \rangle [12]-4 \langle 3\nu \rangle [13]) (\langle 23 \rangle^3 \langle 2\nu \rangle [23]^4+4 \langle 13 \rangle \langle 23 \rangle^2 \langle 2\nu \rangle [13] [23]^3+6 \langle 13 \rangle^2 \langle 23 \rangle \langle 2\nu \rangle [13]^2 [23]^2\\
&&+\langle 13 \rangle^3 [13]^3 (\langle 1\nu \rangle [13]+4 \langle 2\nu \rangle [23]))),\\
&&\\
&&(A^l_\mu(1^+,2^+,3^+) k_1^\mu)_{B0[1,3]}\\
&&=-\frac{1}{64 \langle 12 \rangle \langle 13 \rangle \langle 1\nu \rangle \langle 23 \rangle \langle 2\nu \rangle \langle 3\nu \rangle (\langle 13 \rangle [13]+\langle 23 \rangle [23])^3}(2 \langle 3\nu \rangle^3 [12]^2 [13] [23] (3 \langle 13 \rangle [13]+2 \langle 23 \rangle [23]) \langle 12 \rangle^4\\
&&+\langle 3\nu \rangle^2 [12] [13] (4 \langle 23 \rangle^2 \langle 3\nu \rangle [23]^3+\langle 13 \rangle \langle 23 \rangle (7 \langle 3\nu \rangle [13]-11 \langle 2\nu \rangle [12]) [23]^2-2 \langle 13 \rangle^2 [13] (\langle 1\nu \rangle [12] [13]-2 \langle 3\nu \rangle [23] [13]\\
&&+9 \langle 2\nu \rangle [12] [23])) \langle 12 \rangle^3\\
&&-\langle 13 \rangle \langle 3\nu \rangle [13] (2 \langle 23 \rangle^2 \langle 3\nu \rangle (7 \langle 2\nu \rangle [12]+2 \langle 3\nu \rangle [13]) [23]^3+\langle 13 \rangle \langle 23 \rangle (-8 \langle 2\nu \rangle^2 [12]^2+33 \langle 2\nu \rangle \langle 3\nu \rangle [13] [12]+9 \langle 3\nu \rangle^2 [13]^2) [23]^2\\
&&+\langle 13 \rangle^2 [13] (\langle 1\nu \rangle [12] [13] (\langle 3\nu \rangle [13]-5 \langle 2\nu \rangle [12])+2 (-8 \langle 2\nu \rangle^2 [12]^2+12 \langle 2\nu \rangle \langle 3\nu \rangle [13] [12]+3 \langle 3\nu \rangle^2 [13]^2) [23])) \langle 12 \rangle^2\\
&&+\langle 13 \rangle (2 \langle 23 \rangle^3 \langle 2\nu \rangle^2 \langle 3\nu \rangle [12] [23]^4+\langle 13 \rangle \langle 23 \rangle^2 \langle 2\nu \rangle (-3 \langle 2\nu \rangle^2 [12]^2+13 \langle 2\nu \rangle \langle 3\nu \rangle [13] [12]+2 \langle 3\nu \rangle^2 [13]^2) [23]^3\\
&&+3 \langle 13 \rangle^2 \langle 23 \rangle \langle 2\nu \rangle [13] (-3 \langle 2\nu \rangle^2 [12]^2+9 \langle 2\nu \rangle \langle 3\nu \rangle [13] [12]+2 \langle 3\nu \rangle^2 [13]^2) [23]^2+\langle 13 \rangle^3 [13]^2 (\langle 1\nu \rangle [13] (-3 \langle 2\nu \rangle^2 [12]^2\\
&&+4 \langle 2\nu \rangle \langle 3\nu \rangle [13] [12]+\langle 3\nu \rangle^2 [13]^2)+\langle 2\nu \rangle (-9 \langle 2\nu \rangle^2 [12]^2+23 \langle 2\nu \rangle \langle 3\nu \rangle [13] [12]+6 \langle 3\nu \rangle^2 [13]^2) [23])) \langle 12 \rangle\\
&&-\langle 13 \rangle^2 \langle 2\nu \rangle (2 \langle 23 \rangle^3 \langle 2\nu \rangle (2 \langle 2\nu \rangle [12]+\langle 3\nu \rangle [13]) [23]^4+\langle 13 \rangle \langle 23 \rangle^2 \langle 2\nu \rangle [13] (15 \langle 2\nu \rangle [12]+7 \langle 3\nu \rangle [13]) [23]^3\\
&&+3 \langle 13 \rangle^2 \langle 23 \rangle \langle 2\nu \rangle [13]^2 (7 \langle 2\nu \rangle [12]+3 \langle 3\nu \rangle [13]) [23]^2+\langle 13 \rangle^3 [13]^3 (\langle 1\nu \rangle [13] (3 \langle 2\nu \rangle [12]+\langle 3\nu \rangle [13])\\
&&+\langle 2\nu \rangle (13 \langle 2\nu \rangle [12]+5 \langle 3\nu \rangle [13]) [23]))).
\eeas
\end{scriptsize}

We checked our four point amplitudes using the known simple results of $A^l(1^+,2^+,3^+,4^+)$ and $A^l(1^+,2^+,3^+,4^-)$ in \cite{Bern:1991aq,Bern2,Kharel,Mahlon:1993si}.





\section{Conclusion}
We have discussed the Ward identity in detail for off shell amplitudes in pure Yang-Mills theory. We can state that Ward identity is not only a constraint on gauge field amplitudes, but together with three point amplitudes it can be used to derive any point amplitudes. We explicitly prove that the Ward identity with two complexified external legs holds at tree and one loop levels using Feynman rules. Then we use the Ward identity to deduce recursion relations for off shell amplitudes at tree and one loop levels. In this technique, three steps are important to simplify the calculations. First, according to the complexfied Ward identity, we can convert the calculation of the amplitudes to the calculation of derivative of the amplitudes. Second, we decompose the three point vertex which contains the off shell leg into three terms, which simplifies many steps in our calculation. Thirdly, according to the cancellation details in the proof of complexified Ward identity, we find most terms from different diagrams cancel with each other. The number of remaining effective terms or diagrams are reduced. It turns out that the recursion relation we derive at tree level is equivalent to Berends-Giele recursion relation \cite{Berends:1987me}. However, our expressions at 1-loop level are new. And we present 1-loop off shell three and four point amplitudes as examples of applying our method at 1-loop level.



Comparing with the previous work \cite{Chen1}, we find the technique in this chapter is more universal. Here we can obtain a recursion relation for the total amplitudes instead of just the boundary terms of the amplitudes, and we do not need to use BCFW recursion relation. Furthermore,  for this technique, we do not need to avoid the unphysical poles from the polarization vectors of the shifted on shell leg which can also depend on $z$. Hence this technique works well for the amplitudes with any helicity structure and the momenta shifts are more general than the ones in \cite{Chen1}. In addition, this technique can be used for calculating one loop off shell amplitudes. 

In principle, it is possible to generalize our method to higher loop levels and to other theories such as QCD.  The only obstacle is to classify all the cancellation details for the Ward identity with complexfied external momenta. We leave this to future work. Another extension is to combine our technique with other methods, such as unitary cut \cite{Anastasiou:2006jv,Bern:1994zx,Bern:1994cg,Britto:2005ha,Britto:2004nj,Britto:2004nc,Cachazo:2004dr}, BCFW \cite{Britto:2004ap, Britto:2005fq}, OPP \cite{OPP}  etc. to further simplify the calculation in pure Yang-Mills theory.  

\chapter{Conclusion and Outlook}

I have described my research on two kinds of amplitudes in high energy physics: one on proton Compton scattering amplitudes and the other on YM gauge field scattering amplitudes.

On proton Compton scattering, we used an effective field theory approach. The work was based on a modified Rarita-Schwinger theory recently developed by Konstantin G. Savvidy which allowed for consistent coupling of electromagnetic field with a spin 3/2 particle \cite{Kostas10}, resolving the long standing superluminal propagation problem of the old Rarita-Schwinger theory. Proton and $\Delta^+$ were naturally unified in this theory, with proton being the spin 1/2 component and $\Delta^+$ being the spin 3/2 component. We introduced six form factors and two bare polarizabilities into an effective Lagrangian. We derived proton and $\Delta^+$ magnetic moments and studied the approximations of the amplitudes around the $\Delta^+$ peak. We then calculated the Compton amplitudes using the effective Lagrangian and fit to current experimental data. We extracted proton's static polarizabilities and $\Delta^+$ magnetic moment from the best fit values.

Although in this work we derived a value for $\Delta^+$ magnetic moment, we should not bet on it because $\Delta^+$ magnetic moment does not directly affect proton Compton scattering. A suitable experiment to measure $\Delta^+$ magnetic moment would be the pion photoproduction process: $\gamma\,p\to \gamma\,p\,\pi^0$. Since $\Delta^+$ magnetic moment has been poorly measured, it deserves more work. We can immediately extend our approach to deal with pion photoproduction process and obtain a more reliable fit for $\Delta^+$ magnetic moment.

Another future direction would be gaining deeper understandings of the physics behind the polarizabilities. In the current work we have noticed that if the polarizabilities have a suitable energy dependence we can achieve much better fitting. Especially, the static magnetic polarizability has a preference to change sign between about 200-300 MeV. We have investigated several possibilities that would result in an energy dependence of the polarizabilities but without much success. We would like to continue to hunt for the physics of polarizabilities, and in turn to better explain the experimental data on proton Compton scattering and extract more precise values of the proton polarizabilities.

In the first part of our study of YM gauge field amplitudes, we analyzed the boundary behavior of off-shell amplitudes with a pair of legs complexified. In Feynman gauge, we introduced a set of ''reduced vertices'' to simplify the analysis of boundary behavior. We then proved that amplitudes with non-adjacent legs deformed have better boundary properties than those with adjacent legs deformed. Based on the boundary behavior, we extended BCFW recursion relations to off-shell amplitudes, and uncovered relations between off-shell amplitudes. In the other work, using Ward identity for complexified amplitudes, we obtained recursion relations for off-shell amplitudes both at tree and one loop levels.


There are several directions deserving future investigations. Firstly, one should investigate the possibility of formulating BCJ relations for off-shell amplitudes. In our current work, we have seen some hint for off-shell BCJ relations at four point level. We would like to extend the study of BCJ relations to general N-point amplitudes. Secondly, we would like to further our analysis of boundary behavior of tree level off-shell amplitudes to the integrands of loop level amplitudes. Thirdly, as we have found in our current work that a permutation of legs for a fixed color ordering can improve boundary behavior, we intend to study the boundary behavior of amplitudes with more general permutations of the legs. If we were able to obtain better boundary behavior, e.g. as speculated by the authors of \cite{Boels}, there would be a wealth of powerful amplitude relations awaiting us. These could play important roles in the calculation and analysis of amplitudes. All these exciting developments would, in turn, enable us to have a precise mastery of QCD background in the hadronic colliders, paving the road to the discoveries of beyond the Standard Model physics!

\newpage
\addcontentsline{toc}{chapter}{Acknowledgement}
\centerline{\large{\textbf{Acknowledgement}}} \vspace{1cm}
At this time as I finish the thesis, I feel indebted to my advisors. Professor Yeuk-Kwan Edna Cheung has been a very responsible and kind advisor all through my Ph.D years. She has taught me what it means to be a physicist. And in this last year she also spent a lot of time in helping me find a job and revising this thesis. Dr. Konstantin G. Savvidy, my advisor on the proton Compton scattering work, is also very helpful. I am especially moved by his endless exploration into the depth of physics. And the last but not the least, Dr. Gang Chen, during our collaboration in the Yang-Mills amplitudes, he has provided me much hand on guidance in research.

I thank all the referees of the thesis and the whole defense committee for their valuable suggestions on improving the thesis, including Professors J. D. Vergados, George Georgiou, Andreas Brandhuber, Minxing Luo, Bin Chen, Chuanjie Zhu, Zhongzhou Ren and Hongshi Zong. Among them, Professors J. D. Vergados, George Georgiou, Andreas Brandhuber and Minxing Luo have made most of the criticisms.

Also I would like to thank all the group members, including Changhong Li, Feng Xu, Jens Fjelstad, Lingfei Wang, for the fun we have had together.

I thank my parents for their love, and many friends for spiritual support.


This thesis was funded in parts by the Priority Academic Program Development of Jiangsu Higher Education Institutions (PAPD), NSFC grant No.~10775067, Research Links Programme of Swedish Research Council under contract No.~348-2008-6049, the Chinese Central Government's 985 Project grants for Nanjing University.

\clearpage

\end{document}